\documentclass{amsbook}
\usepackage{color, xcolor, colortbl}
\usepackage{graphicx,epstopdf}
\usepackage{geometry}
\usepackage{amsmath,amssymb,amsthm,amsfonts}
\usepackage{algorithm}
\usepackage{algorithmic}
\usepackage{dcolumn}% Align table columns on decimal point
\usepackage{epstopdf}
\usepackage{bm}
\usepackage{comment}
\usepackage[caption=false]{subfig}
\usepackage{appendix}
\usepackage{multirow}
\usepackage{braket}
\usepackage[english]{babel}
\usepackage{hyperref}
\usepackage[capitalize]{cleveref}
\usepackage[square,numbers]{natbib}
\usepackage[T1]{fontenc}
\usepackage{xpatch}
\usepackage{tikz}
\usepackage{adjustbox}
\usepackage{xspace}
\usepackage[roman]{complexity}
\usepackage{bbm}

\newcommand{\bvec}[1]{\mathbf{#1}}

\newcommand{\va}{\bvec{a}}
\newcommand{\vb}{\bvec{b}}

\newcommand{\vr}{\bvec{r}}

\newcommand{\vx}{\bvec{x}}
\newcommand{\vy}{\bvec{y}}

\newcommand{\vX}{\bvec{X}}

\newcommand{\diag}{\operatorname{diag}}
\renewcommand{\Re}{\operatorname{Re}}
\renewcommand{\Im}{\operatorname{Im}}
\newcommand{\conj}[1]{\overline{#1}}
\renewcommand{\span}[1]{\operatorname{span}\{#1\}}
\newcommand{\Tr}{\operatorname{Tr}}

\DeclareMathOperator*{\argmin}{arg\,min}

\newcommand{\I}{\mathrm{i}}

\newcommand{\xprod}{\mathop{\prod\nolimits'}}

\newcommand{\opr}[1]{\ensuremath{\operatorname{#1}}}
\newcommand{\mc}[1]{\mathcal{#1}}
\newcommand{\mf}[1]{\mathfrak{#1}}
\newcommand{\wt}[1]{\widetilde{#1}}

\newcommand{\abs}[1]{\left\lvert#1\right\rvert}
\newcommand{\norm}[1]{\left\lVert#1\right\rVert}
\newcommand{\gives}{\quad\Rightarrow\quad}

\newcommand{\ud}{\,\mathrm{d}}
\newcommand{\Or}{\mathcal{O}}
\renewcommand{\EE}{\mathbb{E}}

\newcommand{\RR}{\mathbb{R}}
\newcommand{\CC}{\mathbb{C}}

\theoremstyle{plain}
\newtheorem{thm}{\protect\theoremname}[chapter]
\theoremstyle{plain}
\newtheorem{lem}[thm]{\protect\lemmaname}
\theoremstyle{remark}
\newtheorem{remx}[thm]{\protect\remarkname}
\newenvironment{rem}
  {\pushQED{\qed}\remx}
  {\popQED\endremx}
\theoremstyle{definition}
\newtheorem{examplex}[thm]{\protect\examplename}
\newenvironment{exam}
  {\pushQED{\qed}\examplex}
  {\popQED\endexamplex}
\theoremstyle{plain}
\newtheorem*{lem*}{\protect\lemmaname}
\theoremstyle{plain}
\newtheorem{prop}[thm]{\protect\propositionname}
\theoremstyle{plain}
\newtheorem{cor}[thm]{\protect\corollaryname}

\newtheorem{defn}[thm]{\protect\definitionname}

\theoremstyle{definition}
\newtheorem{exercisex}{\protect\exercisename}[chapter]
\newtheorem{exer}[exercisex]{\protect\exercisename}

\providecommand{\definitionname}{Definition}
\providecommand{\assumptionname}{Assumption}
\providecommand{\corollaryname}{Corollary}
\providecommand{\lemmaname}{Lemma}
\providecommand{\propositionname}{Proposition}
\providecommand{\remarkname}{Remark}
\providecommand{\examplename}{Example}
\providecommand{\theoremname}{Theorem}
\providecommand{\exercisename}{Exercise}

\usetikzlibrary{quantikz}
\newcommand{\qwb}{\qwbundle{}}
\usetikzlibrary{fit}
\tikzset{%
  highlight/.style={rectangle,rounded corners,fill=blue!15,draw,fill opacity=0.3,thick,inner sep=0pt}
}

\graphicspath{{./assets/}}
\numberwithin{equation}{chapter}
\numberwithin{figure}{chapter}
\numberwithin{table}{chapter}

\newcommand{\BE}{\operatorname{BE}}
\newcommand{\HBE}{\operatorname{HBE}}

\usepackage{amsbsy}

\begin{document}

%+Title
\title{
{\huge{Lecture Notes on \\
Quantum Algorithms for Scientific Computation}}}
\author{\vspace{5em}
{\huge Lin Lin}\\
\vspace{2em}
{
\Large
Department of Mathematics, University of California, Berkeley\\
Challenge Institute of Quantum Computation, University of California, Berkeley\\
Computational Research Division, Lawrence Berkeley National Laboratory\\
}
\vspace{5em}
\today
\vspace{5em}
\begin{center}
\includegraphics[width=0.3\textwidth]{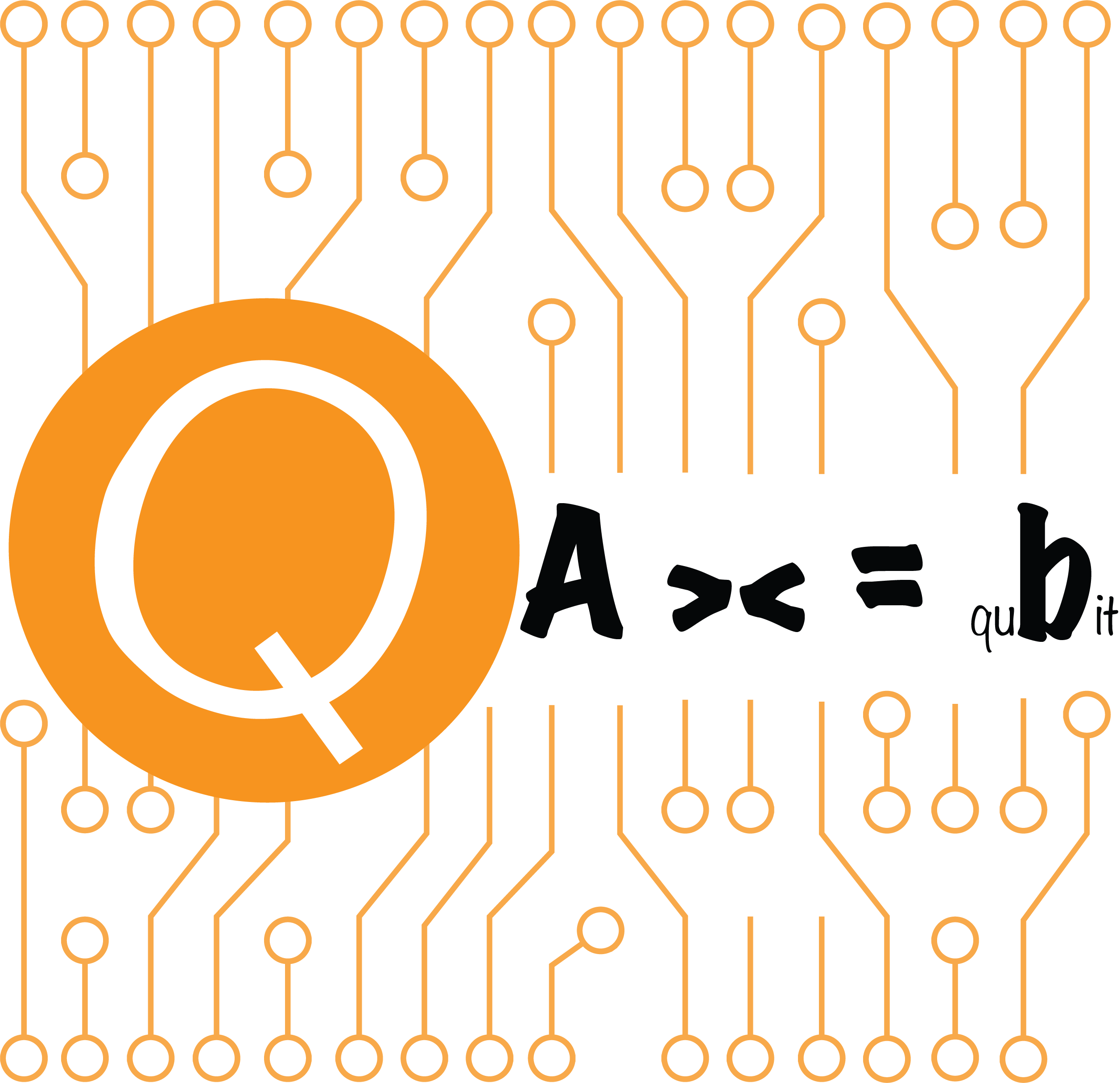}
\end{center}

%PRELIMINARY NOTES BEING CONTINUOUSLY UPDATED. \\
%DO NOT DISTRIBUTE WITHOUT EXPLICIT APPROVAL FROM THE LECTURER.
}
\maketitle

\tableofcontents

\chapter*{Preface}

With availability of near-term quantum devices and the breakthrough of quantum supremacy experiments, quantum computation has received an increasing amount of attention from a diverse range of scientific disciplines in the past few years. 
Despite the availability of excellent textbooks as well as lecture notes such as \cite{NielsenChuang2000,KitaevShenVyalyi2002,Nakahara2008,RieffelPolak2011,Aaronson2013,PreskillQuantumLec,DeWolfQuantumLec,ChildsQuantumLec}, these materials often cover \textit{all} aspects of quantum computation, including complexity theory, physical implementations of quantum devices, quantum information theory, quantum error correction, quantum algorithms etc. This leaves little room for introducing how a quantum computer is supposed to \textit{be used} to solve challenging computational problems in scientific and engineering. 
For instance, after the initial reading of (admittedly, selected chapters of) the classic textbook by Nielsen and Chuang~\cite{NielsenChuang2000}, I was both amazed by the potential power of a quantum computer, and baffled by its practical range of applicability: are we really trying to build a quantum computer,  either to perform a quantum Fourier transform or to perform a quantum search? 
Is quantum phase estimation the only bridge connecting a quantum computer on one side, and virtually \textit{all} scientific computing problems on the other, such as solving linear systems, eigenvalue problems, least squares problems, differential equations, numerical optimization etc.?

Thanks to the significant progresses in the development of quantum algorithms, it should be by now self-evident that the answer to both questions above is \textit{no}. 
This is a fast evolving field, and many important progresses have only been developed in the past few years. 
However, many such developments are theoretically and technically involved, and can be difficult to penetrate for someone with only basic knowledge of quantum computing.
I think it is worth delivering some of these exciting results, in a somewhat more accessible way, to a broader community interested in using future fault-tolerant quantum computers to solve scientific problems. 

This is a set of lecture notes used in a graduate topic class in applied
mathematics called ``Quantum Algorithms for Scientific Computation'' at the Department of Mathematics, UC Berkeley during the fall semester of 2021. 
These lecture  notes focus only on quantum algorithms closely related to scientific computation, and in particular, matrix computation.  In fact, this is only a small class of quantum algorithms viewed from the perspective of the ``quantum algorithm zoo''\footnote{\url{https://quantumalgorithmzoo.org/}}. This means that many important materials are consciously left out, such as
quantum complexity theory,  
applications in number theory and cryptography (notably, Shor's algorithm), applications in algebraic problems (such as the hidden subgroup problems) etc. 
Readers interested in these topics can consult some of the excellent aforementioned textbooks. 
Since the materials were designed to fit into the curriculum of one semester, several other topics relevant to scientific computation are not included, notably adiabatic quantum computation (AQC), and variational quantum algorithms (VQA). These materials may be added in future editions of the lecture notes.
To my knowledge, some of the materials in these lecture notes may be new and have not been presented in the literature. The sections marked by * can be skipped upon first reading without much detriment.

I would  like to thank Dong An, Yulong Dong, Di Fang, Fabian M. Faulstich, Cory Hargus, Zhen Huang, Subhayan Roy Moulik, Yu Tong, Jiasu Wang, Mathias Weiden, Jiahao Yao, Lexing Ying for useful discussions and for pointing out typos in the notes.
I would like also like to thank Nilin Abrahamsen, Di Fang, Subhayan Roy Moulik, Yu Tong for contributing some of the exercises, and Jiahao Yao for providing the cover image of the notes.
For errors / comments / suggestions / general thoughts on the lectures notes, please send me an email: \url{linlin@math.berkeley.edu}.

\chapter{Preliminaries of quantum computation}

\section{Postulates of quantum mechanics}\label{sec:postulate}

We introduce the four main postulates of quantum mechanics related to this course. 
For more details, we refer readers to \cite[Section 2.2]{NielsenChuang2000}.
All postulates concern closed quantum systems (i.e., systems isolated from environments) only.

\subsection{State space postulate}

The set of all quantum states of a quantum system forms a complex vector space with inner product structure (i.e., it is a Hilbert space, denoted by $\mc{H}$), called the state space. If the state space $\mc{H}$ is finite dimensional, it is isomorphic to some $\CC^N$, written as $\mc{H}\cong \CC^N$. Without loss of generality we may simply take $\mc{H}=\CC^N$. We always assume $N=2^n$ for some non-negative integer $n$, often called the number of quantum bits (or qubits). A quantum state $\psi\in\CC^N$ can be expressed in terms of its components as
\begin{equation}
\psi=\begin{pmatrix}
\psi_0\\ \psi_1 \\ \vdots \\ \psi_{N-1}
\end{pmatrix}.
 \quad 
\end{equation}
Its Hermitian conjugate is 
\begin{equation}
\psi^{\dag}=\begin{pmatrix}
\conj{\psi}_0& \conj{\psi}_1 & \cdots & \conj{\psi}_{N-1}
\end{pmatrix},
\end{equation}
where $\conj{c}$ is the complex conjugation of $c\in \CC$. We also use the Dirac notation, which uses $\ket{\psi}$ to denote a quantum state, $\bra{\psi}$ to denote its Hermitian conjugation $\psi^{\dag}$,
and the inner product
\begin{equation}
\braket{\psi|\varphi}:=\braket{\psi,\varphi}=\sum_{i\in[N]} \conj{\psi}_i \varphi_i.
\end{equation}
Here $[N]=\set{0,\ldots,N-1}$. Let $\{\ket{i}\}$ be the standard basis of $\CC^N$. The $i$-th entry of $\psi$ can be written as an inner product
$\psi_i=\braket{i|\psi}$. 
Then $\ket{\psi}\bra{\varphi}$ should be interpreted as an outer product, with $(i,j)$-th matrix element given by
\begin{equation}
\braket{i|(\ket{\psi}\bra{\varphi})|j}=\braket{i|\psi}\braket{\varphi|j}=\psi_i \conj{\varphi}_j.
\end{equation}

Two state vectors $\ket{\psi}$ and $c\ket{\psi}$ for some $0\ne c\in \CC$ always refer to the same physical state, i.e., $c$ has no observable effects. Hence without loss of generality we always assume $\ket{\psi}$ is normalized to be a unit vector, i.e., $\braket{\psi|\psi}=1$. Sometimes it is more convenient to write down an unnormalized state, which will be denoted by $\psi$ without the ket notation $\ket{\cdot}$. 
Restricting to normalized state vectors, the complex number $c=e^{\I\theta}$ for some $\theta\in [0,2\pi)$, called the global phase factor. 

\begin{exam}[Single qubit system]
A (single) qubit corresponds to a state space $\mc{H}\cong\CC^2$. We also define
\begin{equation}
\ket{0}=\begin{pmatrix}
1\\ 0
\end{pmatrix}, \quad
\ket{1}=
\begin{pmatrix}
0 \\ 1
\end{pmatrix}.
\end{equation}
Since the state space of the spin-$\frac12$ system is also isomorphic to $\CC^2$, this is also called the single spin system, where $\ket{0},\ket{1}$ are referred to as the spin-up and spin-down state, respectively. 
A general state vector in $\mc{H}$ takes the form 
\begin{equation}
\ket{\psi}=a\ket{0}+b\ket{1}=\begin{pmatrix}
a\\ b
\end{pmatrix}, \quad a,b\in\CC,
\end{equation}
and the normalization condition implies $\abs{a}^2+\abs{b}^2=1$.  
So we may rewrite $\ket{\psi}$ as
\begin{equation}
\ket{\psi}=e^{\I \gamma}\left(\cos \frac{\theta}{2} \ket{0}+e^{\I \varphi} \sin \frac{\theta}{2} \ket{1}\right), \quad \theta,\varphi,\gamma\in\RR.
\end{equation}
If we ignore the irrelevant global phase $\gamma$, 
the state is effectively
\begin{equation}
\ket{\psi}=\cos \frac{\theta}{2} \ket{0}+e^{\I \varphi} \sin \frac{\theta}{2} \ket{1}, \quad 0\le\theta<\pi,0\le\varphi<2\pi.
\end{equation}
So we may identify each single qubit quantum state with a unique point on the unit three-dimensional sphere (called the Bloch sphere) as
\begin{equation}
\va=(\sin\theta \cos\varphi, \sin\theta\sin\varphi,\cos\theta)^{\top}.
\end{equation}
\end{exam}

\subsection{Quantum operator postulate}

The evolution of a quantum state from $\ket{\psi}\to\ket{\psi'}\in \CC^N$ is always achieved via a unitary operator $U\in\CC^{N\times N}$, i.e.,
\begin{equation}
\ket{\psi'}=U\ket{\psi}, \quad U^{\dag} U=I_N.
\end{equation}
Here $U^{\dag}$ is the Hermitian conjugate of a matrix $U$, and $I_N$ is the $N$-dimensional identity matrix. When the dimension is apparent, we may also simply write $I\equiv I_N$. 
In quantum computation, a unitary matrix is often referred to as a gate.

\begin{exam}
For a single qubit, the Pauli matrices are
\begin{equation}
\sigma_x=X=\begin{pmatrix}
0 & 1\\
1 & 0
\end{pmatrix}
, \quad 
\sigma_y=Y=\begin{pmatrix}
0 & -\I\\
\I & 0
\end{pmatrix}, \quad
\sigma_z=Z=\begin{pmatrix}
1 & 0\\
0 & -1
\end{pmatrix}.
\end{equation}
Together with the two-dimensional identity matrix, they form a basis of all linear operators on $\CC^2$. 
\end{exam}

Some other commonly used single qubit operators include, to name a few:

\begin{itemize}

\item Hadamard gate
\begin{equation}
H=\frac{1}{\sqrt{2}}\begin{pmatrix}
1 & 1\\
1 & -1
\end{pmatrix}
\end{equation}

\item Phase gate
\begin{equation}
S=\begin{pmatrix}
1 & 0\\
0 & \I
\end{pmatrix}
\end{equation}

\item $\mathrm{T}$ gate:
\begin{equation}
T=\begin{pmatrix}
1 & 0\\
0 & e^{\I\pi/4}
\end{pmatrix}
\end{equation}
\end{itemize}
When there are notational conflicts, we will use the roman font such as $\mathrm{H},\mathrm{X}$ for these single-qubit gates (one common scenario is to distinguish the Hadamard gate $\mathrm{H}$ from a Hamiltonian $H$). An operator acting on an $n$-qubit quantum state space is  called an $n$-qubit operator.

Starting from an initial quantum state $\ket{\psi(0)}$, the quantum state can evolve in time, which gives a single parameter family of quantum states denoted by $\{\ket{\psi(t)}\}$. These quantum states are related to each other via a quantum evolution operator $U$:
\begin{equation}
\psi(t_2)=U(t_2,t_1)\psi(t_1),
\end{equation}
where $U(t_2,t_1)$ is unitary for any given $t_1,t_2$. Here $t_2>t_1$ refers to quantum evolution forward in time, $t_2<t_1$ refers to quantum evolution backward in time, and $U(t_1,t_1)=I$ for any $t_1$. 

The quantum evolution under a time-independent Hamiltonian $H$ satisfies the time-independent Schr\"odinger equation
\begin{equation}
\I \partial_t\ket{\psi(t)}=H\ket{\psi(t)}.
\end{equation}
Here $H=H^{\dag}$ is a Hermitian matrix. The corresponding time evolution operator is
\begin{equation}
U(t_2,t_1)=e^{-\I H(t_2-t_1)}, \quad \forall t_1,t_2.
\end{equation}
In particular, $U(t_2,t_1)=U(t_2-t_1,0)$. 

On the other hand, for any unitary matrix $U$, we can always find a Hermitian matrix $H$ such that $U=e^{\I H}$ (\cref{exer:unitary}).

\begin{exam}
Let the Hamiltonian $H$ be the Pauli-X gate. Then 
\begin{equation}
U(t,0)=e^{-\I Xt}=
\begin{pmatrix}
\cos t & -\I \sin t\\
-\I \sin t & \cos t
\end{pmatrix}=(\cos t)I-\I X (\sin t).
\end{equation}
Starting from an initial state $\ket{\psi(0)}=\ket{0}$, after time $t=\pi/2$, the state evolves into $\ket{\psi(\pi/2)}=-\I \ket{1}$, i.e., the $\ket{1}$ state (up to a global phase factor).
\end{exam}

\subsection{Quantum measurement postulate}

Without loss of generality, we only discuss a special type of quantum measurements called the projective measurement. 
For more general types of quantum measurements, see~\cite[Section 2.2.3]{NielsenChuang2000}.
All quantum measurements expressed as a positive operator-valued measure (POVM) can be expressed in terms of projective measurements in an enlarged Hilbert space via the Naimark dilation theorem.

In a finite dimensional setting, a quantum observable always corresponds to a Hermitian matrix $M$, which has the spectral decomposition
\begin{equation}
M=\sum_{m} \lambda_m P_m.
  \label{eqn:M_observable}
\end{equation}
Here $\lambda_m\in\RR$ are the eigenvalues of $M$, and $P_m$ is the projection operator onto the eigenspace associated with $\lambda_m$, i.e., $P_m^2=P_m$. 

When a quantum state $\ket{\psi}$ is measured by a quantum observable $M$, the outcome of the measurement is always an eigenvalue $\lambda_m$, with probability 
\begin{equation}
p_m=\braket{\psi|P_m|\psi}.
\end{equation}
After the measurement, the quantum state becomes
\begin{equation}
\ket{\psi}\to \frac{P_m \ket{\psi}}{\sqrt{p_m}}
\end{equation} 
Note that this is not a unitary process!

In order to evaluate the expectation value of a quantum observable $M$, we first use  the resolution of identity: 
\begin{equation}
\sum_{m} P_m=I.
\end{equation}
This implies the normalization condition,
\begin{equation}
\sum_{m} p_m=\sum_{m} \braket{\psi|P_m|\psi}=\braket{\psi|\psi}=1.
\end{equation}
Together with $p_m\ge 0$, we find that $\{p_m\}$ is indeed a probability distribution.

The expectation value of the measurement outcome is
\begin{equation}
\EE_{\psi}(M)=\sum_{m} \lambda_m p_m=\sum_{m} \lambda_m \braket{\psi|P_m|\psi}=\Braket{\psi|\left(\sum_{m} \lambda_m P_m \right)|\psi}=\braket{\psi|M|\psi}.
\end{equation}

\begin{exam}
Again let $M=X$. From the spectral decomposition of $X$:
\begin{equation}
X\ket{\pm}=\lambda_{\pm}\ket{\pm},
\end{equation}
where $\ket{\pm}:=\frac{1}{\sqrt{2}}(\ket{0}\pm\ket{1}), \quad \lambda_{\pm}=\pm 1$, we obtain the eigendecomposition
\begin{equation}\label{eqn:X_spectral}
M=X=\ket{+}\bra{+}-\ket{-}\bra{-}.
\end{equation}
Consider a quantum state $\ket{\psi}=\ket{0}=\frac{1}{\sqrt{2}}(\ket{+}+\ket{-})$, then
\begin{equation}
\braket{\psi|P_{+}|\psi}=\braket{\psi|P_{-}|\psi}=\frac12.
\end{equation}
Therefore the expectation value of the measurement is $\braket{\psi|M|\psi}=0.$
\end{exam}

\subsection{Tensor product postulate}

For a quantum state consists of $m$ components with state spaces $\{\mc{H}_i\}_{i=0}^{m-1}$, the state space is their tensor products denoted by $\mc{H}=\otimes_{i=0}^{m-1} \mc{H}_i$. Let $\ket{\psi_i}$ be a state vector in $\mc{H}_i$, then 
\begin{equation}
\ket{\psi}=\ket{\psi_0}\otimes\cdots\otimes\ket{\psi_{m-1}}
\end{equation}
in $\mc{H}$. However, not all quantum states in $\mc{H}$ can be written in the tensor product form above. Let $\{\ket{e^{(i)}_j}\}_{j\in [N_i]}$ be the basis of $\mc{H}_i$, then a general state vector in $\mc{H}$ takes the form
\begin{equation}
\ket{\psi}=\sum_{j_0\in[N_0],\ldots, j_{m-1}\in [N_{m-1}]} \psi_{j_0\cdots j_{m-1}} \ket{e^{(0)}_{j_0}}\otimes\cdots\otimes \ket{e^{(m-1)}_{j_{m-1}}}.
\label{eqn:tensor_component}
\end{equation}
Here $\psi_{j_0\cdots j_{m-1}}\in\CC$ is an entry of a $m$-way tensor, and the dimension of $\mc{H}$ is therefore $\prod_{i\in[m]} N_i$.

The state space of $n$-qubits is $\mc{H}=(\CC^2)^{\otimes n}\cong \CC^{2^n}$, rather than $\CC^{2n}$. We also use the notation
\begin{equation}
\ket{01}\equiv\ket{0,1}\equiv\ket{0}\ket{1}\equiv\ket{0}\otimes\ket{1}, \quad \ket{0^{\otimes n}}=\ket{0}^{\otimes n}. 
\end{equation}
Furthermore, $x\in\{0,1\}^n$ is called a classical bit-string, and $\{\ket{x}|x\in\{0,1\}^n\}$ is called the computational basis of $\CC^{2^n}$.

\begin{exam}[Two qubit system] The state space is $\mc{H}=(\CC^2)^{\otimes 2}\cong \CC^{4}$. The standard basis is (row-major order, i.e., last index is the fastest changing one)
\begin{equation}
\ket{00}=
\begin{pmatrix}
1\\ 0 \\ 0 \\ 0
\end{pmatrix}, \quad
\ket{01}=
\begin{pmatrix}
0\\ 1 \\ 0 \\ 0
\end{pmatrix}, \quad
\ket{10}=
\begin{pmatrix}
0\\ 0 \\ 1 \\ 0
\end{pmatrix}, \quad
\ket{11}=
\begin{pmatrix}
0\\ 0 \\ 0 \\ 1
\end{pmatrix}.
\end{equation}

The Bell state (also called the EPR pair) is defined to be
\begin{equation}
\ket{\psi}=\frac{1}{\sqrt{2}}(\ket{00}+\ket{11})=\frac{1}{\sqrt{2}}
\begin{pmatrix}
1 \\ 0 \\ 0 \\ 1
\end{pmatrix},
\label{eqn:bell_state}
\end{equation}
which cannot be written as any product state $\ket{a}\otimes\ket{b}$ (\cref{exer:bell}).

There are many important quantum operators on the two-qubit quantum system. One of them is the CNOT gate, with matrix representation
\begin{equation}
\opr{CNOT}=\begin{pmatrix}
{1} & {0} & {0} & {0} \\ 
{0} & {1} & {0} & {0} \\ 
{0} & {0} & {0} & {1} \\ 
{0} & {0} & {1} & {0}
\end{pmatrix}.
\end{equation}
In other words, when acting on the standard basis, we have
\begin{equation}
\opr{CNOT}\begin{cases}
\ket{00}&=\ket{00}\\
\ket{01}&=\ket{01}\\
\ket{10}&=\ket{11}\\
\ket{11}&=\ket{10}\\
\end{cases}.
\end{equation}

This can be compactly written as
\begin{equation}
\opr{CNOT}\ket{a}\ket{b}=\ket{a}\ket{a\oplus b}.
\end{equation}
Here $a\oplus b=(a+b)\mod 2$ is the ``exclusive or'' (XOR) operation.
\end{exam}

\begin{exam}[Multi-qubit Pauli operators]
For a $n$-qubit quantum system, the Pauli operator acting on the $i$-th qubit is denoted by $P_i$ ($P=X,Y,Z$). For instance
\begin{equation}
X_i:=I^{\otimes (i-1)}\otimes X\otimes I^{\otimes (n-i)}.
\end{equation} 
\end{exam}

\section{Density operator}

So far all quantum states encountered can be described by a single $\ket{\psi}\in\mc{H}$, called the pure state. More generally, if a quantum system is in one of a number of states $\ket{\psi_i}$ with respective probabilities $p_i$, then $\{p_i, \ket{\psi_i}\}$ is an ensemble of pure states. The density operator of the quantum system is 
\begin{equation}
\rho:=\sum_i p_i \ket{\psi_i}\bra{\psi_i}.
\end{equation}
For a pure state $\ket{\psi}$, we have
\begin{equation}
\rho=\ket{\psi}\bra{\psi}
\end{equation}
is a rank-$1$ matrix.

Consider a quantum observable in \cref{eqn:M_observable} associated with the projectors $\{P_m\}$. For a pure state, it can be verified that the probability result of returning $\lambda_m$, and the expectation value of the measurement are respectively,
\begin{equation}
p(m)=\Tr[P_m\rho], \quad \EE_{\rho}[M]=\Tr[M \rho]
\label{eqn:measure_rho}
\end{equation}
The expression \eqref{eqn:measure_rho} also holds for general density operators $\rho$.

An operator $\rho$ is the density operator associated to some ensemble $\{p_i, \ket{\psi_i}\}$ if and only if (1) $\Tr \rho=1$ (2) $\rho \succeq 0$, i.e., $\rho$ is a positive semidefinite matrix (also called a positive operator). All postulates in \cref{sec:postulate} can be stated in terms of density operators (see \cite[Section 2.4.2]{NielsenChuang2000}). Note that a pure state satisfies $\rho^2=\rho$. In general we have $\rho^2\preceq \rho$. If $\rho^2\prec \rho$, then $\rho$ is called (the density operator of) a mixed state. Furthermore, an ensemble of admissible density operators is also a density operator.

A quantum operator $U$ that transforms $\ket{\psi}$ to $U\ket{\psi}$ also transforms the density operator according to
\begin{equation}
  \rho=\sum_i p_i \ket{\psi_i}\bra{\psi_i}\xrightarrow{U}
\sum_i p_i U\ket{\psi_i}\bra{\psi_i}U^{\dag} = U\rho U^{\dag}:=U[\rho].
\end{equation}
However, not all quantum operations on density operators need to be unitary! 
See \cite[Section 8.2]{NielsenChuang2000} for more general discussions on quantum operations.

Most of the discussions in this course will be restricted to pure states, and unitary quantum operations. 
Even in this restricted setting, the density operator formalism can still be convenient, particularly for describing a subsystem of a composite quantum system. 
Consider a quantum system of $(n+m)$-qubits,
partitioned into a subsystem $A$ with $n$ qubits (the state space is $\mc{H}_A=\CC^{2^n}$) and a subsystem $B$ with $m$ qubits (the state space is $\mc{H}_B=\CC^{2^m}$) respectively. 
The quantum state is a pure state $\ket{\psi}\in \CC^{2^{n+m}}$ with density operator $\rho_{AB}$. 
Let $\ket{a_1},\ket{a_2}$ be two state vectors in $\mc{H}_A$, and $\ket{b_1},\ket{b_2}$ be two state vectors in $\mc{H}_B$. 
Then the partial trace over system $B$ is defined as
\begin{equation}
  \Tr_{B}[\ket{a_1}\bra{a_2}\otimes \ket{b_1}\bra{b_2}] = \ket{a_1}\bra{a_2} \Tr[\ket{b_1}\bra{b_2}] = \ket{a_1}\bra{a_2} \braket{b_2|b_1}.
  \label{eqn:partial_trace}
\end{equation}
Since we can expand the density operator $\rho_{AB}$ in terms of the basis of $\mc{H}_A,\mc{H}_B$, the definition of \eqref{eqn:partial_trace} can be extended to define
the reduced density operator for the subsystem $A$
\begin{equation}
  \rho_{A} = \Tr_{B}[\rho_{AB}].
  \label{eqn:reduced_density_operator}
\end{equation}
The reduced density operator for the subsystem $B$ can be similarly defined.  The reduced density operators $\rho_A,\rho_B$ are generally mixed states. 

\begin{exam}[Reduced density operator of tensor product states]
If $\rho_{AB}=\rho_1\otimes \rho_2$, then 
\begin{equation}
\Tr_{B}[\rho_{AB}] = \rho_1, \quad  \Tr_{A}[\rho_{AB}] = \rho_2.
\end{equation}
\end{exam}

If a quantum observable is defined only on the subsystem $A$, i.e., $M=M_A\otimes I$ where $M_A$ has the decomposition \eqref{eqn:M_observable}, 
then the success probability of returning $\lambda_m$, and the expectation value are respectively
\begin{equation}
  p(m)=\Tr[(P_m\otimes I) \rho] = \Tr[P_m \Tr_B[\rho]] = \Tr[P_m \rho_A], \quad  \EE_{\rho}[M]=\Tr[(M_A\otimes I) \rho] = \Tr[M_A \rho_A].
  \label{eqn:measure_reduced_density}
\end{equation}

\section{Quantum circuit}

Nearly all quantum algorithms operate on multi-qubit quantum systems. When quantum operators operate on two or more qubits, writing down quantum states in terms of its components as in \cref{eqn:tensor_component} quickly becomes cumbersome. The quantum circuit language offers a graphical and compact manner for writing down the procedure of applying a sequence of quantum operators to a quantum state. For more details see \cite[Section 4.2, 4.3]{NielsenChuang2000}.

In the quantum circuit language, time flows from the left to right, i.e., the input quantum state appears on the left, and the quantum operator appears on the right, and each ``wire'' represents a qubit i.e.,
\begin{center}
\begin{quantikz}
 \lstick{$\ket{\psi}$} & \gate{U}   & \rstick{$U\ket{\psi}$} \qw 
\end{quantikz}
\end{center}

Here are a few examples:
\begin{center}
\begin{quantikz}
 \lstick{$\ket{0}$} & \gate{X}   & \rstick{$\ket{1}$} \qw 
\end{quantikz}
\quad
\begin{quantikz}
 \lstick{$\ket{1}$} & \gate{Z}   &  \rstick{$-\ket{1}$} \qw
\end{quantikz}
\quad
\begin{quantikz}
 \lstick{$\ket{0}$} & \gate{H}   & \rstick{$\ket{+}$} \qw 
\end{quantikz}
\end{center}
which is a graphical way of writing
\begin{equation}
X\ket{0}=\ket{1}, \quad Z\ket{1}=-\ket{1}, \quad H\ket{0}=\ket{+}.
\end{equation}
The relation between these states can be expressed in terms of the following diagram
\begin{equation}
\begin{tikzcd}[column sep=3em,row sep=3em]
\ket{0} \arrow{r}{X} \arrow[swap]{d}{H} & \ket{1} \arrow{d}{H} \\
\ket{+} \arrow{r}{Z} & \ket{-}
\end{tikzcd}
\end{equation}

Also verify that
\begin{center}
\begin{quantikz}
 \lstick{$\ket{0}$} & \gate{X}   & \rstick{$\ket{1}$} \qw \\
 \lstick{$\ket{0}$} & \qw   & \rstick{$\ket{0}$} \qw 
\end{quantikz}
\end{center}
which is a graphical way of writing
\begin{equation}
(X\otimes I)\ket{00}=\ket{10}.
\end{equation}
Note that the input state can be general, and in particular does not need to be a product state. For example, if the input is a Bell state \eqref{eqn:bell_state}, we just apply the quantum operator to $\ket{00}$ and $\ket{11}$,  respectively and multiply the results by $1/\sqrt{2}$ and add together.
To distinguish with other symbols, these single qubit gates may be either written as $X,Y,Z,H$ or (using the roman font) $\mathrm{X,Y,Z,H}$.

The quantum circuit for the CNOT gate is

\begin{center}
\begin{quantikz}
 \lstick{$\ket{a}$}    & \ctrl{1}   & \rstick{$\ket{a}$} \qw    \\
 \lstick{$\ket{b}$}    & \targ{}      & \rstick{$\ket{a\oplus b}$}  \qw 
\end{quantikz}
\end{center}
Here the ``dot'' means that the quantum gate connected to the dot only becomes active if the state of the qubit $0$ (called the control qubit) is $a=1$. 
This justifies the name of the CNOT gate (controlled NOT).

Similarly, 
\begin{center}
\begin{quantikz}
 \lstick{$\ket{a}$}    & \ctrl{1}   & \rstick{$\ket{a}$} \qw    \\
 \lstick{$\ket{b}$}    & \gate{U}      & \rstick{$U^{a}\ket{b}$}  \qw 
\end{quantikz}
\end{center}
is the controlled $U$ gate for some unitary $U$. Here $U^a=I$ if $a=0$. The CNOT gate can be obtained by setting $U=X$.

Another commonly used two-qubit gate is the SWAP gate, which swaps the state in the $0$-th and the $1$-st qubits.
\begin{center}
\begin{quantikz}
 \lstick{$\ket{a}$}    & \swap{1}   &  \rstick{$\ket{b}$} \qw  \\
 \lstick{$\ket{b}$}    & \targX{}      & \rstick{$\ket{a}$} \qw  
\end{quantikz}
\end{center}

Quantum operators applied to multiple qubits can be written in a similar manner:
\begin{center}
\begin{quantikz}
 \lstick{qubit 0: $\ket{0}$}  & \gate[4]{U} & \qw   \\
 \lstick{qubit 1: $\ket{0}$}  &             & \qw   \\
 \lstick{qubit 2: $\ket{0}$}  &             & \qw   \\
 \lstick{qubit 3: $\ket{0}$}  &             & \qw   
\end{quantikz}
\end{center}
For a multi-qubit quantum circuit, unless stated otherwise, the first qubit will be referred to as the qubit 0, and the second qubit as the qubit 1, etc.

When the context is clear, we may also use a more compact notation for the multi-qubit quantum operators:
\begin{displaymath}
\begin{quantikz}
 \lstick{$\ket{0}^{\otimes{4}}$}  & \gate{U} \qwb & \qw    
\end{quantikz}
\Leftrightarrow
\begin{quantikz}
 \lstick{$\ket{0}^{\otimes{4}}$}  & \gate{U} \qwbundle[alternate]{} & \qwbundle[alternate]{}
\end{quantikz}
\Leftrightarrow
\begin{quantikz}
 \lstick{$\ket{0}^{\otimes{4}}$}  & \gate{U}  & \qw    
\end{quantikz}
\end{displaymath}
One useful multiple qubit gate is the Toffoli gate (or controlled-controlled-NOT, CCNOT gate).
\begin{center}
\begin{quantikz}
 \lstick{$\ket{a}$}    & \ctrl{1}   & \rstick{$\ket{a}$} \qw    \\
 \lstick{$\ket{b}$}    & \ctrl{1}   & \rstick{$\ket{b}$} \qw    \\ 
 \lstick{$\ket{c}$}    & \targ{}      & \rstick{$\ket{(ab)\oplus c}$}  \qw
\end{quantikz}
\end{center}
We may also want to apply a $n$-qubit unitary $U$ only when certain conditions are met
\begin{center}
\begin{quantikz}
 \lstick{$\ket{1}$}    & \ctrl{1}   & \rstick{$\ket{1}$} \qw    \\
 \lstick{$\ket{1}$}    & \ctrl{1}   & \rstick{$\ket{1}$} \qw    \\ 
 \lstick{$\ket{0}$}    & \octrl{1}  & \rstick{$\ket{0}$} \qw    \\ 
 \lstick{$\ket{x}$}    & \gate{U} \qwb  & \rstick{$U\ket{x}$}  \qw
\end{quantikz}
\end{center}
where the empty circle means that the gate being controlled only becomes active when the value of the control qubit is $0$. 
This can be used to write down the quantum ``if'' statements, i.e., when the qubits $0,1$ are at the $\ket{1}$ state and the qubit $2$ is at the $\ket{0}$ state, then apply $U$ to $\ket{x}$. 

A set of qubits is often called a register (or quantum variable). 
For example, in the picture above, the main quantum state of interest (an $n$ qubit quantum state $\ket{x}$) is called the system register. The first $3$ qubits can be called the control register.

\section{Copy operation and no-cloning theorem}

One of the most striking early results of quantum computation is the no-cloning theorem (by Wootters and Zurek, as well as Dieks in 1982), which forbids generic quantum copy operations (see also \cite[Section 12.1]{NielsenChuang2000}). The no-deleting theorem is a consequence of linearity of quantum mechanics.

Assume there is a unitary operator $U$ that acts as the copy operations, i.e.,
\begin{equation}
  U\ket{x}\otimes\ket{s}=\ket{x}\otimes\ket{x},
\end{equation}
for any black-box state $x$, and a chosen target state $\ket{s}$ (e.g. $\ket{0^n}$). Then take two states $\ket{x_1},\ket{x_2}$, we have
\begin{equation}
  U\ket{x_1}\otimes\ket{s}=\ket{x_1}\otimes\ket{x_1}, \quad U\ket{x_2}\otimes\ket{s}=\ket{x_2}\otimes\ket{x_2}.
\end{equation}
Taking the inner product of the two equations, we have
\begin{equation}
\braket{x_1|x_2}=\braket{x_1|x_2}^2,
\end{equation}
which implies $\braket{x_1|x_2}=0$ or $1$. When $\braket{x_1|x_2}=1$, $\ket{x_1},\ket{x_2}$ refer to the same physical state.
Therefore a cloning operator $U$ can at most copy states which are orthogonal to each other, and a general quantum copy operation is impossible.

Given the ubiquity of the copy operation in scientific computing like $y=x$, the no-cloning theorem has profound implications. 
For instance, all classical iterative algorithms for solving linear systems require storing some intermediate variables. 
This operation is generally not possible, or at least cannot be efficiently performed.

There are two notable exceptions to the range of applications of the no-cloning theorem. 
The first is that we know how a quantum state is prepared, i.e., $\ket{x}=U_x\ket{s}$ for a \emph{known} unitary $U_x$ and some $\ket{s}$. 
Then we can of course copy this specific vector $\ket{x}$ via 
\begin{equation}
(I\otimes U_x)\ket{x}\otimes\ket{s}=\ket{x}\otimes\ket{x}.
\end{equation}

The second is the copying of classical information. This is an application of the CNOT gate.
\begin{center}
\begin{quantikz}
 \lstick{$\ket{x}$}    & \ctrl{1}   & \rstick{$\ket{x}$} \qw    \\
 \lstick{$\ket{0}$}    & \targ{}      & \rstick{$\ket{x}$}  \qw 
\end{quantikz}
\end{center}
i.e.,
\begin{equation}
\opr{CNOT}\ket{x,0}=\ket{x,x}, \quad x\in\{0,1\}.
\end{equation}
The same principle applies to copying classical information from multiple qubits. \cref{fig:multiqubit_copy} gives an example of copying the classical information stored in 3 bits.
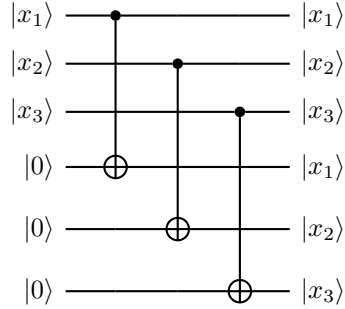
\begin{figure}[H]
\begin{center}
\begin{quantikz}
 \lstick{$\ket{x_1}$}    & \ctrl{3}   & \qw & \qw & \rstick{$\ket{x_1}$} \qw    \\
 \lstick{$\ket{x_2}$}    & \qw & \ctrl{3}   & \qw & \rstick{$\ket{x_2}$} \qw    \\
 \lstick{$\ket{x_3}$}    & \qw & \qw & \ctrl{3}   & \rstick{$\ket{x_3}$} \qw    \\
 \lstick{$\ket{0}$}      & \targ{}    & \qw & \qw & \rstick{$\ket{x_1}$}  \qw    \\
 \lstick{$\ket{0}$}    & \qw & \targ{}   & \qw & \rstick{$\ket{x_2}$} \qw    \\
 \lstick{$\ket{0}$}    & \qw & \qw & \targ{}   & \rstick{$\ket{x_3}$} \qw    \\
\end{quantikz}
\end{center}
\caption{Copying classical information using multi-qubit CNOT gates.}
\label{fig:multiqubit_copy}
\end{figure}
  In general, a multi-qubit CNOT operation can be used to perform the classical copying operation in the computational basis. Note that in the circuit model, this can be implemented with a depth 1 circuit, since they all act on different qubits.

\begin{exam}
  Let us verify that the CNOT gate does not violate the no-cloning theorem, i.e., it cannot be used to copy a superposition of classical bits $\ket{x}=a\ket{0}+b\ket{1}$.
  Direct calculation shows
  \begin{equation}
  \opr{CNOT}\ket{x}\otimes \ket{0}=a\ket{00}+b\ket{11} \ne \ket{x}\otimes \ket{x}.
  \end{equation}
  In particular, if $\ket{x}=\ket{+}$, then CNOT creates a Bell state.

  \end{exam}

The quantum no-cloning theorem implies that there does not exist a unitary $U$ that performs the deleting operation, which resets a black-box state $\ket{x}$ to $\ket{0^n}$. 
This is because such a deleting unitary can be viewed as a copying operation
\begin{equation}
U \ket{0^n}\otimes\ket{x} = \ket{0^n}\otimes\ket{0^n}.
\end{equation}
Then take $\ket{x_1},\ket{x_2}$ that are orthogonal to each other, apply the deleting gate, and compute the inner products, we obtain
\begin{equation}
0=\braket{x_1|x_2}=\braket{0^n|0^n}=1,
\end{equation}
which is a contradiction.

A more common way to express the no-deleting theorem is in terms of the time reversed dual of the no-cloning theorem:  in general, given two copies of some arbitrary quantum state, it is impossible to delete one of the copies.
More specifically, there is no unitary $U$ performing the following operation using known states $\ket{s},\ket{s'}$,
\begin{equation}\label{eqn:no_delete}
U\ket{x}\ket{x}\ket{s}=\ket{x}\ket{0^n}\ket{s'}
\end{equation}
for an arbitrary unknown state $\ket{x}$ (\cref{exer:no_delete}).

\section{Measurement}
The quantum measurement applied to any qubit, by default, measures the outcome in the computational basis. 
For example,
\begin{center}
\begin{quantikz}
 \lstick{$\ket{0}$} & \gate{H}   & \meter{} 
\end{quantikz}
\end{center}
outputs $0$ or $1$ each w.p. $1/2$. We may also measure some of the qubits in a multi-qubit system.
\begin{equation}
\begin{quantikz}
 \lstick{$\ket{0}$}  & \gate[3]{U} & \meter{}    \\
 \lstick{$\ket{0}$}  &             & \rstick[wires=2]{$\ket{\psi}$}\qw \\
 \lstick{$\ket{0}$}  &             & \qw   
\end{quantikz}
\equiv
\begin{quantikz}
 \lstick{$\ket{0}$}  & \gate[2]{U} & \meter{}    \\
 \lstick{$\ket{0}^{\otimes 2}$}  &             & \rstick{$\ket{\psi}$}\qw
\end{quantikz}
\end{equation}

There are two important principles related to quantum measurements: the principle of deferred measurement, and the principle of implicit measurement. At a first glance, both principles may seem to be counterintuitive.

The principle of deferred measurement states that measurement operations can always be moved from an intermediate stage of a quantum circuit to the end of the circuit. This is because even if a measurement is performed as an intermediate step in a quantum circuit, and the result of the measurement is used to conditionally control subsequent quantum gates, such classical controls can always be replaced by quantum controls, and the result of the quantum measurement is postponed to later.

\begin{exam}[Deferring quantum measurements]
  Consider the circuit
\begin{center}
\begin{quantikz}
  \lstick{$\ket{0}$}    & \gate{H} & \meter{} & \cwbend{1} \\
  \lstick{$\ket{0}$}    & \qw      & \qw      & \gate{X} & \qw
\end{quantikz}
\end{center}
Here the double line denotes the classical control operation. 
The outcome is that qubit $0$ has probability $1/2$ of outputting $0$, and the qubit $1$ is at state $\ket{0}$. 
Qubit $0$ also has probability $1/2$ of outputting $1$, and the qubit $1$ is at state $\ket{1}$. 

However, we may replace the classical control operation after the measurement by a quantum controlled $X$ (which is CNOT), and measure the qubit $0$ afterwards: 
\begin{center}
\begin{quantikz}
  \lstick{$\ket{0}$}    & \gate{H} & \ctrl{1} & \meter{}  \\
  \lstick{$\ket{0}$}    & \qw      & \targ{}  &  \qw
\end{quantikz}
\end{center}
It can be verified that the result is the same. Note that CNOT acts as the classical copying operation. 
So qubit $1$ really stores the classical information (i.e., in the computational basis) of qubit $0$. 
\end{exam}

\begin{exam}[Deferred measurement requires extra qubits]
The procedure of deferring quantum measurements using CNOTs is general, and important. Consider the following circuit: 
\begin{center}
\begin{quantikz}
  \lstick{$\ket{0}$}    & \gate{H} & \meter{} & \gate{H} & \meter{}
\end{quantikz}
\end{center}
The probability of obtaining $0,1$ is $1/2$, respectively. 
However, if we simply ``defer'' the measurement to the end by removing the intermediate measurement, we obtain
\begin{center}
\begin{quantikz}
  \lstick{$\ket{0}$}    & \gate{H} & \gate{H} & \meter{}
\end{quantikz}
\end{center}
The result of the measurement is deterministically $0$!
The correct way of deferring the intermediate quantum measurement is to introduce another qubit
\begin{center}
\begin{quantikz}
  \lstick{$\ket{0}$}    & \gate{H} & \ctrl{1} & \gate{H} & \meter{}\\
  \lstick{$\ket{0}$}    & \qw      & \targ{}  & \qw & \qw \\
\end{quantikz}
\end{center}
Measuring the qubit $0$, we obtain $0$ or $1$ w.p. $1/2$, respectively.
Hence when deferring quantum measurements, it is necessary to store the intermediate information in extra (ancilla) qubits, even if such information is not used afterwards.
\end{exam}

The principle of implicit measurements states that at the end of a quantum circuit, any unmeasured qubit may be assumed to be measured.
More specifically, assume the quantum system consists of two subsystems $A$ and $B$. 
If qubits $A$ are to be measured at the end of the circuits, the results of the measurements does not depend on whether the qubits $B$ are measured or not. 
Recall from \cref{eqn:measure_reduced_density} that a measurement on the subsystem $A$ only depends on the reduced density matrix $\rho_A$.
So we only need to show that $\rho_A$ does not depend on the measurement in $B$. 
To see why this is the case, let $\{P_i\}$ be the projectors onto the computational basis of $B$. 
Before the measurement, the density operator is $\rho$.
If we measure the subsystem $B$, the resulting density operator is transformed into
\begin{equation}\label{eqn:measure_rhoB}
  \rho'=\sum_{i} (I\otimes P_i) \rho (I\otimes P_i).
\end{equation}
Then it can be verified that 
\begin{equation}\label{eqn:measure_rhoB_rhoA}
\rho'_A = \Tr_B[\rho']=\Tr_B\left[\rho \sum_i(I\otimes P_i)\right] = \Tr_B[\rho] = \rho_A.
\end{equation}
This proves the principle of implicit measurements.

By definition, the output of all quantum algorithms must be obtained through measurements, and hence the measurement outcome is probabilistic in general.
If the goal is to compute the expectation value of a quantum observable $M_A$ acting on a subsystem $A$, then its variance is
\begin{equation}
\opr{Var}_{\rho}[M_A]=\Tr[M_A^2 \rho_A] - (\Tr[M_A \rho_A])^2.
\end{equation}
The number of samples $\mc{N}$ needed to estimate $\Tr[M_A \rho_A]$ to \emph{additive} precision $\epsilon$ satisfies
\begin{equation}
\sqrt{\frac{\opr{Var}_{\rho}[M_A]}{\mc{N}}}\le \epsilon \gives \mc{N}\ge \frac{\opr{Var}_{\rho}[M_A]}{\epsilon^2},
\end{equation}
which only depends on $\rho_A$. 

\begin{exam}[Estimating success probability on one qubit]\label{exam:prob_onequbit}
Let $A$ be the single qubit to be measured in the computational basis, and we are interested in the accuracy in estimating the success probability of obtaining $1$, i.e., $p$. 
This can be realized as an expectation value with $M_A=\ket{1}\bra{1}$, and $p=\Tr[M_A \rho_A]$.
Note that $M_A^2=M_A$, then
\begin{equation}
 \opr{Var}_{\rho}[M_A] = p-p^2=p(1-p).
\end{equation}
Hence to estimate $p$ to \emph{additive} error $\epsilon$, the number of samples needed satisfies
\begin{equation}
\mc{N}\ge \frac{p(1-p)}{\epsilon^2}.
\end{equation}
Note that if $p$ is close to $0$ or $1$, the number of samples needed is also very small: 
indeed, the outcome of the measurement becomes increasing deterministic in this case!

If we are interested in estimating $p$ to \emph{multiplicative} accuracy $\epsilon$, then the number of samples is
\begin{equation}
\mc{N}\ge \frac{p(1-p)}{p^2\epsilon^2}=\frac{1-p}{p\epsilon^2},
\end{equation}
and the task becomes increasingly more difficult when $p$ approaches $0$. 
\end{exam}

\section{Linear error growth and Duhamel's principle}

If a quantum algorithm denoted by a unitary $U$ can be decomposed into a series of simpler unitaries as $U=U_{K}\cdots U_1$, and if we can implement each $U_i$ to precision $\epsilon$, then what is the global error?
We now introduce a simple technique connecting the local error with the global error. 
In the context of quantum computation, this is often referred to as the ``hybrid argument''.

\begin{prop}[Hybrid argument]
Given unitaries $U_1,\wt{U}_1,\ldots, U_K,\wt{U}_K\in\CC^{N\times N}$ satisfying
\begin{equation}
\norm{U_i-\wt{U}_i}\le \epsilon, \quad \forall i=1,\ldots,K,
\end{equation}
we have
\begin{equation}
\norm{U_{K}\cdots U_1-\wt{U}_K \cdots \wt{U}_1}\le K \epsilon.
\end{equation}
\label{prop:hybridization_argument}
\end{prop}
\begin{proof}
Use a telescoping series
\begin{equation}
\begin{split}
&U_{K}\cdots U_1-\wt{U}_K \cdots \wt{U}_1\\
=&(U_{K}\cdots U_2U_1-U_K \cdots U_2\wt{U}_1)+
(U_{K}\cdots U_3 U_2\wt{U}_1-U_{K}\cdots U_3 \wt{U}_2 \wt{U}_1)+\cdots\\
&+(U_K\wt{U}_{K-1} \cdots \wt{U}_1-\wt{U}_K\wt{U}_K \cdots \wt{U}_1)\\
=& U_{K}\cdots U_2(U_1-\wt{U}_1)+U_{K}\cdots U_3(U_2-\wt{U}_2)+\cdots+(U_K-\wt{U}_K)\wt{U}_{K-1} \cdots \wt{U}_1.
\end{split}
\label{eqn:hybrid_telescope}
\end{equation}
Since all $U_i,\wt{U}_i$ are unitary matrices, we readily have
\begin{equation}
\norm{U_{K}\cdots U_1-\wt{U}_K \cdots \wt{U}_1}\le \sum_{i=1}^K \norm{U_i-\wt{U}_i}\le K\epsilon.
\end{equation}
\end{proof}

In other words, if we can implement each local unitary to precision $\epsilon$, the global error grows at most \textit{linearly} with respect to the number of gates and is bounded by $K\epsilon$.
The telescoping series \cref{eqn:hybrid_telescope}, as well as the hybrid argument can also be seen as a discrete analogue of the variation of constants method (also called Duhamel's principle).

\begin{prop}[Duhamel's principle for Hamiltonian simulation]
Let $U(t),\wt{U}(t)\in \CC^{N\times N}$ satisfy
\begin{equation}
\I \partial_t U(t)=H U(t), \quad \I \partial_t\wt{U}(t)=H\wt{U}(t)+B(t), \quad U(0)=\wt{U}(0)=I,
\end{equation}
where $H\in \CC^{N\times N}$ is a Hermitian matrix, and $B(t)\in \CC^{N\times N}$ is an arbitrary matrix. 
Then
\begin{equation}
\wt{U}(t)=U(t)-\I\int_0^t U(t-s) B(s) \ud s,
\label{eqn:duhamel}
\end{equation}
and
\begin{equation}
\norm{\wt{U}(t)-U(t)}\le \int_{0}^t \norm{B(s)}\ud s.
\end{equation}
\label{prop:duhamel}
\end{prop}
\begin{proof}
Directly verify that \cref{eqn:duhamel} is the solution to the differential equation.
\end{proof}
As a special case, consider $B(t)=E(t)\wt{U}(t)$, then \cref{eqn:duhamel} becomes
\begin{equation}
\wt{U}(t)=U(t)-\I\int_0^t U(t-s) E(s) \wt{U}(s) \ud s,
\end{equation}
and
\begin{equation}
\norm{\wt{U}(t)-U(t)}\le \int_{0}^t \norm{E(s)}\ud s.
\end{equation}
This is a direct analogue of the hybrid argument in the continuous setting.

\section{Universal gate sets and reversible computation}

In classical computation, there are many universal gate sets, in the sense that any classical gate can be represented as a combination of gates from the set.
For example, the NAND gate (``Not AND'') alone forms a universal gate set \cite[Section 3.1.2]{NielsenChuang2000}. The NOR gate (``Not OR'') is also a universal gate set.

In the quantum setting, any unitary operator on $n$ qubits can be implemented using $1$- and $2$-qubit gates\cite[Section 4.5]{NielsenChuang2000}. 
It is desirable to come up with a  set of discrete universal gates, but this means that we need to give up the notion that the unitary $U$ can be exactly represented. 
Instead, a set of quantum gates $\mc{S}$ is universal if given any unitary operator $U$ and desired precision $\epsilon$, we can find $U_1,\ldots,U_m\in \mc{S}$ such that
\begin{equation}
  \norm{U-U_m U_{m-1} \cdots U_1} \le \epsilon.
\end{equation}
Here $\norm{A}=\sup_{\braket{\psi|\psi}=1} \norm{A\ket{\psi}}$ is the operator norm (also called the spectral norm) of $A$, and $\norm{\ket{\psi}}=\sqrt{\braket{\psi|\psi}}$ is the vector $2$-norm).  There are many possible choices of universal gate sets, e.g. $\{H,T,\opr{CNOT}\}$. Another universal gate set is $\{H,\opr{Toffoli}\}$, which only involves real numbers.

Are some universal gate sets better than others? The Solovay-Kitaev theorem states that all choices of universal gate sets are asymptotically equivalent (see e.g. ~\cite[Chapter 2]{ChildsQuantumLec}):
\begin{thm}[Solovay-Kitaev]
  Let $\mc{S},\mc{T}$ be two universal gate sets that are closed under
  inverses. Then any $m$-gate circuit using the gate set $\mc{S}$ can be implemented to
  precision $\epsilon$ using a circuit of $\Or(m\cdot \polylog(m/\epsilon))$ gates from the gate set $\mc{T}$, 
  and there is a classical algorithm for finding this circuit in time $\Or(m\cdot \polylog(m/\epsilon))$.
  \label{thm:solovay_kitaev}
\end{thm}

Another natural question is about the computational power of quantum computers. 
Perhaps surprisingly, it is very difficult to prove that quantum computer is \emph{more powerful than} classical computer. 
But is quantum computer \emph{at least as powerful as} classical computers? 
The answer is yes! 
More specifically, any classical circuit can also be  asymptotically efficiently implemented using a quantum circuit.

The proof rests on that the classical universal gate can be efficiently simulated using quantum circuits. 
Note that this is not a straightforward process: NAND, and other classical gates (such as AND, OR etc.) are not reversible gates!
Hence the first step is to perform classical computation with \emph{reversible} gates.
More specifically, any irreversible classical gate $x\mapsto f(x)$ can be made into a reversible classical gate
\begin{equation}
(x,y)\mapsto (x,y\oplus f(x)).
\label{eqn:reversible_classical}
\end{equation}
In particular, we have $(x,0)\mapsto (x,f(x))$ computed in a reversible way. The key idea is to store all intermediate steps of the computation (see \cite[Section 3.2.5]{NielsenChuang2000} for more details).

On the quantum computer, storing all intermediate computational steps indefinitely creates two problems: (1) tremendous waste of quantum resources (2) the intermediate results stored in some extra qubits are still entangled to the quantum state of interest. So if the environments interfere with intermediate results, the quantum state of interest is also affected.

Fortunately, both problems can be solved by a step called ``uncomputation''. In order to implement a Boolean function
\(f :\{0,1\}^{n} \rightarrow\{0,1\}^{m}\), we assume there is an oracle
\begin{equation}\ket{0^m}\ket{x}\mapsto \ket{f(x)}\ket{x},\end{equation}
where \(\ket{0^m}\) comes from a \(m\)-qubit output register. The oracle
is often further implemented with the help of a working register (a.k.a.
``garbage'' register) such that
\begin{equation}
U_f:\ket{0^{w}}\ket{0^m}\ket{x}\mapsto \ket{g(x)}\ket{f(x)}\ket{x}.
\end{equation}

From the no-deleting theorem, there is no generic unitary operator
that can set a black-box state to \(\ket{0^w}\). In order to
set the working register back to \(\ket{0^w}\) while keeping the input and
output state, we introduce yet another \(m\)-qubit ancilla register
initialized at \(\ket{0^m}\). Then we can use an $n$-qubit CNOT controlled on the output register and obtain
\begin{equation}\label{eqn:cnot_working_copy}
\ket{0^m}\ket{g(x)}\ket{f(x)}\ket{x}\mapsto 
\underbrace{\ket{f(x)}}_{\text{ancilla}}\underbrace{\ket{g(x)}}_{\text{working}}\underbrace{\ket{f(x)}}_{\text{output}}\underbrace{\ket{x}}_{\text{input}}.
\end{equation}
It is important to remember that in the operation above, the multi-qubit CNOT gate only performs the classical copying operation in the computational basis, and does not violate the no-cloning theorem.

Recall that $U_f^{-1}=U_f^{\dag}$, we have
\begin{equation}
(I_m\otimes U_f^{\dag})\ket{f(x)}\ket{g(x)}\ket{f(x)}\ket{x}=
\ket{f(x)}(U_f^{\dag}\ket{g(x)}\ket{f(x)}\ket{x})=
\ket{f(x)}\ket{0^w}\ket{0^m}\ket{x}.\end{equation}
Finally we apply an $n$-qubit SWAP operator on the ancilla and output registers to obtain
\begin{equation}\label{eqn:swap_working_copy}
\ket{f(x)}\ket{0^w}\ket{0^m}\ket{x}\mapsto \ket{0^m}\ket{0^w}\ket{f(x)}\ket{x}.
\end{equation}
After this procedure, both the ancilla and the working register are set to the initial state. They are no longer entangled to the input or output register, and can be reused for other purposes. This procedure is called \emph{uncomputation}.
The circuit is shown in \cref{fig:circuit_uncompute}.

\begin{figure}[H]
\begin{displaymath}
\begin{quantikz}
\lstick{$\ket{0^m}$} & \qw            & \targ{}    & \qw                
& \swap{2} & \rstick{$\ket{0^m}$} \qw \\
\lstick{$\ket{0^w}$} & \gate[3]{U_f}  & \qw        & \gate[3]{U_f^{\dag}} & \qw & \rstick{$\ket{0^w}$}\qw\\
\lstick{$\ket{0^m}$} &                & \ctrl{-2}  & & \targX{} 
& \rstick{$\ket{f(x)}$}\qw\\
\lstick{$\ket{x}$}   &                & \qw        & & \qw & 
\rstick{$\ket{x}$}\qw\\
\end{quantikz}
\end{displaymath}
\caption{Circuit for uncomputation. The \opr{CNOT} and \opr{SWAP} operators indicate the multi-qubit copy and swap operations, respectively.}
\label{fig:circuit_uncompute}
\end{figure}
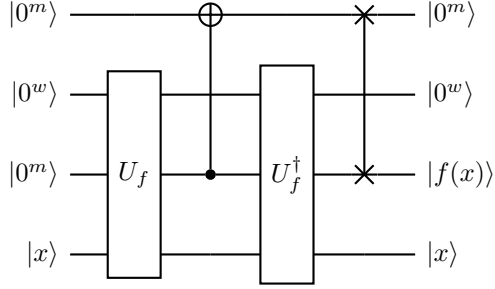

\begin{rem}[Discarding working registers]
After the uncomputation as shown in \cref{fig:circuit_uncompute}, the first two registers are unchanged before and after the application of the circuit (though they are changed during the intermediate steps). Therefore \cref{fig:circuit_uncompute} effectively implements a unitary
\begin{equation}
(I_{m+w}\otimes V_f)\ket{0^m}\ket{0^{w}}\ket{0^m}\ket{x}=\ket{0^m}\ket{0^{w}}\ket{f(x)}\ket{x}
\end{equation}
or equivalently
\begin{equation}
V_f\ket{0^m}\ket{x}=\ket{f(x)}\ket{x}.
\end{equation}
In the definition of $V_f$, all working registers have been discarded (on paper). This allows us to simplify the notation and focus on the essence of the quantum algorithms under study.
\end{rem}

Using the technique of uncomputation, if the map $x\mapsto f(x)$ can be efficiently implemented on a classical computer, then we can implement this map efficiently on a quantum computer as well. 
To do this, we first turn it into a reversible map \eqref{eqn:reversible_classical}.
All reversible single-bit and two-bit classical gates can be implemented using single-qubit and two-qubit quantum gates. 
So the reversible map can be made into a unitary operator
\begin{equation}
U_f: \ket{x,y}\mapsto \ket{x,y\oplus f(x)}
\label{eqn:reversible_quantum}
\end{equation}
on a quantum computer. This proves that a quantum computer is  at least as powerful as classical computers. 

The unitary transformation $U_f$ in \eqref{eqn:reversible_quantum} can be applied to any superposition of states in the computational basis, e.g. 
\begin{equation}
U_f: \frac{1}{\sqrt{2^n}}\sum_{x\in\{0,1\}^n} \ket{x,0^m} \mapsto \frac{1}{\sqrt{2^n}}\sum_{x\in\{0,1\}^n}
\ket{x,f(x)}.
\end{equation}
This does not necessarily mean that we can efficient implement the map $\ket{x}\mapsto \ket{f(x)}$. 
However, if $f$ is a bijection, and we have access to the inverse of the reversible circuit for computing $f^{-1}$, then we may use the technique of uncomputation to implement such a map (\cref{exer:implement_reversible}).

\section{Fixed point number representation and classical arithmetic operations}\label{sec:fixedpoint}

Let \([N]=\set{0,1,\ldots,N-1}\). Any integer \(k\in[N]\) where \(N=2^n\)
can be expressed as an \(n\)-bit string as \(k=k_{n-1}\cdots k_0\) with
\(k_i\in\{0,1\}\). 
This is called the binary representation of the integer $k$. 
It should be interpreted as
\begin{equation}k=\sum_{i\in[n]} k_i 2^{i}.\end{equation}
The number $k$ divided by $2^m$ ($0\le m\le n$) can be written as (note that the decimal is shifted to be after $k_m$):
\begin{equation}
a=\frac{k}{2^m}=\sum_{i\in[n]} k_i 2^{i-m}=:(k_{n-1}\cdots k_{m}.k_{m-1}\cdots k_0).
\end{equation}
The most common case is $m=n$, where
\begin{equation}
a=\frac{k}{2^n}=\sum_{i\in[n]} k_i 2^{i-n}=:(0.k_{n-1}\cdots k_0).
\end{equation}

Sometimes we may also write $a=0.k_{1}\cdots k_{n}$, so that $k_i$ is the $i$-th decimal of $a$ in the binary representation. For a given floating number $0\le a<1$ written as
\begin{equation}
a=(0.k_1\cdots k_n k_{n+1}\cdots),
\end{equation}
the number $(0.k_1\cdots k_n)$ is called the $n$-bit fixed point representation of $a$. Therefore to represent $a$ to additive precision $\epsilon$, we will need $n=\lceil\log_2 \epsilon\rceil$ qubits. If the sign of $a$ is also important, we may reserve one extra bit to indicate its sign.

Together with the reversible computational model, we can perform classical arithmetic operations, such as $(x,y)\mapsto x+y$, $(x,y)\mapsto xy$, $x\mapsto x^{\alpha}$, $x\mapsto \cos(x)$ etc. using reversible quantum circuits. The number of ancilla qubits, and the number of elementary gates needed for implementing such quantum circuits is $\Or(\poly(n))$ (see \cite[Chapter 6]{RieffelPolak2011} for more details). 

It is worth commenting that while quantum computer is \emph{theoretically} as powerful as classical computers, there is a \emph{very significant} overhead in implementing reversible classical circuits on quantum devices, both in terms of the number of ancilla qubits and the circuit depth. 

\section{Fault tolerant computation}\label{sec:fault_tolerant}

All previous discussions assume that quantum operations can be perfectly performed. Due to the immense technical difficulty for realizing quantum computers, both quantum gates and quantum measurements may involve (significant) errors, particularly on near-term quantum devices. However, the threshold theorem states that if the noise in individual quantum gates is below a certain constant threshold (around $10^{-4}$ or above), it is possible to efficiently perform an arbitrarily large quantum computation with any desired precision (see \cite[Section 10.6]{NielsenChuang2000}). This procedure requires quantum error correction protocols. 

This course will not discuss any details on quantum error corrections. We always assume fault-tolerant protocols have been implemented, and all errors come from either approximation errors at the mathematical level, or Monte Carlo errors in the readout process due to the probabilistic nature of the measurement process.

\section{Complexity of quantum algorithms}

Let $n$ be the number of qubits needed to represent the input. A quantum algorithm is efficient if the number of gates in the quantum circuit is $\Or(\poly(n))$. Due to the probabilistic nature of the measurement outcome, we are typically satisfied if a quantum algorithm can produce the correct answer with sufficiently high probability $p$.
For a decision problem that asks for a binary answer $0$ or $1$, we require $p>2/3$ (or at least $p > 1/2+1/\poly(n)$).
For other problems that we have an efficient procedure to check the correctness of the answer, we require $p=\Omega(1)$. Repeating this process many times and apply the Chernoff bound \cite[Box 3.4]{NielsenChuang2000}, we can make the probability of outputting an incorrect answer vanishingly small.

In quantum algorithms, the computational cost is often measured in terms of the \textit{query complexity}. Assume that we have access to black-box unitary operator $U_f$ (e.g. the one used in the reversible computation), which is often called a quantum \emph{oracle}. 
Our goal is to perform a given task using as few queries as possible to $U_f$.

\begin{exam}[Query access to a boolean function]
Let $f:\{0,1\}^n\to\{0,1\}$ be a boolean function, which can be queried via the following unitary
\begin{equation}
U_f\ket{x}=(-1)^{f(x)}\ket{x}.
\label{eqn:phase_kickback}
\end{equation} 
This is called a \emph{phase kickback}, i.e., the value of $f(x)$ is returned as a phase factor. 
The phase kickback is an important tool in many quantum algorithms, e.g. Grover's algorithm. 
Note that 1) $U_f$ can be applied to a superposition of states in the computational basis, and 2) Having query access to $f(x)$ does not mean that we know everything about $f(x)$, e.g. finding the set $\set{x|f(x)=0}$ can still be a difficult task.
\end{exam}

\begin{exam}[Partially specified quantum oracles]
When designing quantum algorithms, it is common that we are not interested in the behavior of the entire unitary matrix $U_f$, but only $U_f$ applied to certain vectors. For instance, for a $(n+1)$-qubit state space, we are only interested in 
\begin{equation}
U_f\ket{0}\ket{x}=\ket{0}(A\ket{x})+\ket{1}(B\ket{x}).
\label{eqn:partial_specify_U}
\end{equation}
This means that we have only defined the first block-column of $U_f$ as (remember that the row-major order is used)
\begin{equation}
U_f=\begin{pmatrix}
A & *\\
B & *
\end{pmatrix}
\end{equation}
Here $A,B$ are $N\times N$ matrices, and $*$ stands for an arbitrary $N\times N$ matrix so that $U_f$ is unitary. 
Of course in order to implement $U_f$ into quantum gates, we still need to specify the content of $*$. 
However, at the conceptual level, the partially specified unitary \eqref{eqn:partial_specify_U} simplifies the design process of quantum oracles.
\end{exam}

The concept of query complexity hides the implementation details of $U_f$, and in some cases we can prove lower bounds on the number of queries to solve a certain problem, e.g. in the case of Grover's search (proving a lower bound of the number of gates among all quantum algorithms can be much harder). Furthermore, once we have access to the number of elementary gates needed to implement $U_f$, we obtain immediately the \textit{gate complexity} of the total algorithm. However, some queries can be (provably) difficult to implement, and then there can be a large gap between the query complexity and gate complexity. In order to obtain a meaningful query complexity analysis, one should also make sure that other components of the quantum algorithm will not end up dominating the total gate complexity, when all factors are taken into account.

Another important measure of the complexity is the \textit{circuit depth}, i.e., the maximum number of gates along any path from an input to an output. Since quantum gates can be performed in parallel, the circuit depth is approximately equivalent to the concept of ``wall-clock time'' in classical computation, i.e., the real time needed for a quantum computer to carry out a certain task.
Since quantum states can only be preserved for a short period of time (called the \textit{coherence time}), the circuit depth also provides an approximate measure of whether the quantum algorithm exceeds the coherence limit of a given quantum computer. In many scenarios, the maximum coherence time is the most severe limiting factor. When possible, it is often desirable to reduce the circuit depth, even if it means that the quantum circuit needs to be carried out many more times.

Let us summarize the basic components of a typical quantum algorithm: the set of qubits can be separated into system registers (storing quantum states of interest) and ancilla registers (auxiliary registers needed to implement the unitary operation acting on system registers). Starting from an initial state, apply a series of one-/two-qubit gates, and perform measurements. Uncomputation should be performed whenever possible. Within the ancilla registers, if a register can be ``freed'' after the uncomputation, it is called a working register. Since  working registers can be reused for other purposes, the cost of  working registers is often not (explicitly) factored into the asymptotic cost analysis in the literature.

\section{Notation}

We use $\|\cdot\|$ to denote vector or matrix 2-norm: when $v$ is a vector we denote by $\|v\|$ its 2-norm, and when $A$ is matrix we denote by $\|A\|$ its operator norm. 
Other matrix and vector norms will be introduced when needed.
Unless otherwise specified, a vector $v\in\CC^N$ is an unnormalized vector, and a normalized vector (stored as a quantum state) is denoted by $\ket{v}=v/\norm{v}$.
A vector $v$ can be expressed in terms of $j$-th component as $v=(v_j)$ or $(v)_j=v_j$. We use a $0$-based indexing, i.e., $j=0,\ldots,N-1$ or $j\in [N]$. When $1$-based indexing is used, we will explicitly write $j=1,\ldots,N$.
We use the following asymptotic notations besides the usual $\Or$ (or ``big-O'') notation: 
we write $f=\Omega(g)$ if $g=\Or(f)$; $f=\Theta(g)$ if $f=\Or(g)$ and $g=\Or(f)$; $f=\wt{\Or}(g)$ if $f=\Or(g\operatorname{polylog}(g))$.

\vspace{2em}
\begin{exer}\label{exer:unitary}
  Prove that any unitary matrix $U\in \CC^{N\times N}$ can be written as $U=e^{\I H}$, where $H$ is an Hermitian matrix.
\end{exer}

\begin{exer}
  Prove \cref{eqn:X_spectral}.
\end{exer}

\begin{exer}
  Write down the matrix representation of the SWAP gate, as well as the $\sqrt{\opr{SWAP}}$ and $\sqrt{\opr{iSWAP}}$ gates.
\end{exer}

\begin{exer}\label{exer:bell}
  Prove that the Bell state \cref{eqn:bell_state} cannot be written as any product state $\ket{a}\otimes\ket{b}$.
\end{exer}

\begin{exer}
  Prove \cref{eqn:measure_rho} holds for a general mixed state $\rho$.
\end{exer}

\begin{exer}
Prove that an ensemble of admissible density operators is also a density operator.
\end{exer}

\begin{exer}\label{exer:no_delete}
Prove the no-deleting theorem in \cref{eqn:no_delete}.
\end{exer}

\begin{exer}
Work out the circuit for implementing \cref{eqn:cnot_working_copy} and \cref{eqn:swap_working_copy}.
\end{exer}

\begin{exer}\label{exer:implement_reversible}
  Prove that if $f:\{0,1\}^n\to\{0,1\}^n$ is a bijection, and we have access to the inverse mapping $f^{-1}$, then the mapping $U_f:\ket{x}\mapsto \ket{f(x)}$ can be implemented on a quantum computer.
\end{exer}

\begin{exer}
Prove \cref{eqn:measure_rhoB} and \cref{eqn:measure_rhoB_rhoA}.
\end{exer}

\chapter{Grover's algorithm}

Now we will introduce a few basic quantum algorithms, of which the ideas and variants are present in numerous other quantum algorithms. 
Hence they are called ``quantum primitives''.
There is no official ruling on which quantum algorithms qualify for being included in the set of quantum primitives, but the membership of Grover's algorithm, quantum Fourier transform,  quantum phase estimation, and Trotter based Hamiltonian simulation should not be controversial. We first introduce Deutsch's algorithm, which is arguably one of the simplest quantum algorithms carrying out a well-defined task.

\section{The first quantum algorithm: Deutsch's algorithm}

Assume we have two boxes, each of them may contain either an apple or an orange.
We would like to answer: whether the two boxes contain the same type of fruit (but do not need to answer whether it is apple or orange).

This seems to be a weird question. 
If the content of a box can only be checked by opening it, then to answer the question we would need to open \emph{two} boxes and check what is inside.
It is impossible to answer whether the fruit types are the same without knowing the types! 
Nevertheless, this is precisely the question to be addressed by Deutsch's algorithm.

\begin{center}
\includegraphics[width=0.4\textwidth]{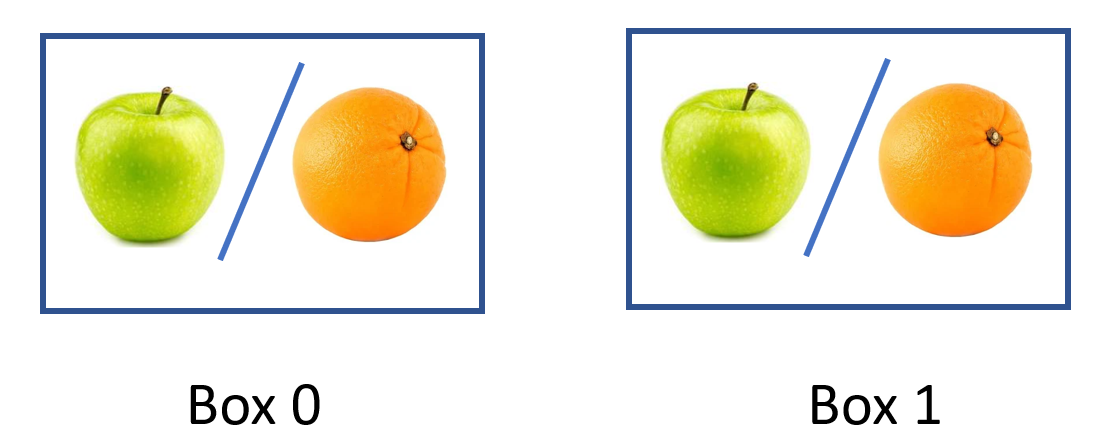}
\end{center}

Mathematically, consider a boolean function $f:\{0,1\}\to\{0,1\}$. 
The question is whether $f(0)=f(1)$ or $f(0)\ne f(1)$?
A quantum fruit-checker assumes the access to $f$ via the following quantum oracle:
\begin{equation}
U_f\ket{x,y}=\ket{x,y\oplus f(x)}, \quad x,y\in \{0,1\}.
\label{eqn:deutsch_oracle}
\end{equation}
A classical fruit-checker can only query $U_f$ in the computational basis as
\begin{equation}\label{eqn:deutsch_computationalbasis}
U_f\ket{0,0}=\ket{0,f(0)},\quad U_f\ket{1,0}=\ket{1,f(1)}.
\end{equation}
After these \emph{two} queries, we can measure qubit $1$ with a deterministic outcome, and answer whether $f(0)=f(1)$. However, a quantum fruit-checker can apply $U_f$ to a linear combination of states in the computational basis.

Let us first check that $U_f$ is unitary:

\begin{equation}
\begin{split}
\braket{x',y'|U_f^{\dag}U_f|x,y}=&\braket{x',y'\oplus f(x')|x,y\oplus f(x)}\\
=&\braket{x'|x}\braket{y'\oplus f(x')|y\oplus f(x)}\\
=&\delta_{x,x'}\delta_{y,y'}
\end{split}
\end{equation}
which gives $U_f^{\dag}U_f=I$.

The idea behind Deutsch's algorithm is to convert the oracle \eqref{eqn:deutsch_oracle} into a phase kickback in \cref{eqn:phase_kickback}. Take $\ket{y}=\ket{-}=\frac{1}{\sqrt{2}}(\ket{0}-\ket{1})$.
Then 
\begin{equation}\label{eqn:kickback_deutsch}
U_f \ket{x,y}=\frac{1}{\sqrt{2}}\left(\ket{x,f(x)}-\ket{x,1\oplus f(x)}\right)=(-1)^{f(x)} \ket{x,y}.
\end{equation}
Note that $\ket{y}=H X\ket{0}$, \cref{eqn:kickback_deutsch} can also be interpreted as
\begin{equation}
(I\otimes XH)U_f(I\otimes HX)\ket{x,0}=(-1)^{f(x)} \ket{x,0}.
\end{equation}
The application of $XH$ can be viewed as the step of uncomputation. Neglecting the qubit 1 which does not change before and after the application, we can focus on the first qubit only, which effectively defines a unitary
\begin{equation}
\wt{U}_f\ket{x}=(-1)^{f(x)}\ket{x}.
\end{equation}
Hence the information of $f(x)$ is stored as a phase factor ($0$ or $\pi$).
Recall that using a Hadamard gate $H\ket{0}=\ket{+}, H\ket{1}=\ket{-}$,
the quantum circuit of Deutsch's algorithm is
\begin{center}
\begin{quantikz}
\lstick{$\ket{0}$} & \gate{H} &\gate{\wt{U}_f} & \gate{H}  & \meter{}
\end{quantikz}
\end{center}
or in the commonly seen form in \cref{fig:circuit_deutsch}.
\begin{figure}[H]
\begin{center}
\begin{quantikz}
\lstick{$\ket{0}$} & \gate{H} &\gate[2][2cm]{U_f} \gateinput{$x$}\gateoutput{$x$}& \gate{H}  & \meter{}\\
\lstick{$\ket{1}$} & \gate{H} &\gateinput{$y$}\gateoutput{$y\oplus f(x)$}              & \qw 
\end{quantikz}
\end{center}
\caption{Quantum circuit for Deutsch's algorithm.}
\label{fig:circuit_deutsch}
\end{figure}
The answer is embedded in the measurement outcome of qubit $0$.
To verify this:
\begin{equation}
\begin{split}
\ket{0,1} \xrightarrow{H\otimes H}& \ket{+,-}=\frac{1}{\sqrt{2}}(\ket{0}+\ket{1})\otimes\ket{-}\\
\xrightarrow{U_f}&\frac{1}{\sqrt{2}}\left((-1)^{f(0)}\ket{0}+(-1)^{f(1)}\ket{1}\right)\otimes\ket{-}\\
\xrightarrow{H\otimes I} &\frac12\left((-1)^{f(0)}+(-1)^{f(1)}\right)\ket{0,-}\\
&+\frac12\left((-1)^{f(0)}-(-1)^{f(1)}\right)\ket{1,-}.
\end{split}
\end{equation}
So if $f(0)=f(1)$, the final state is $\pm\ket{0,-}$. Measuring qubit $0$ returns $0$ deterministically (the globally phase factor is irrelevant). Similarly if $f(0)\ne f(1)$, the final state is $\pm\ket{1,-}$. Measuring qubit $0$ returns $1$ deterministically. In summary, only \emph{one} query to $U_f$ is sufficient to answer whether the two boxes contain the same type of fruit. 
The procedure is equally counterintuitive. Note that a classical fruit checker as implemented in \cref{eqn:deutsch_computationalbasis} naturally receives the information by measuring qubit 1. 
On the other hand, Deutsch's algorithm only uses qubit 1 as a signal qubit, and all the information is retrieved by measuring qubit 0, which, at least from the classical perspective of \cref{eqn:deutsch_computationalbasis}, seems to contain no information at all!

Now we have seen that in terms of the \emph{query complexity}, a quantum fruit-checker is clearly more efficient.
However, it is a fair question how to implement the oracle $U_f$, especially in a way that somehow does not already reveal the values of $f(0),f(1)$.
We give the implementation of some cases of $U_f$ in \cref{exam:qiskit_deutsch}.
In general, proving the query complexity alone may not be convincing enough that quantum computers are better than classical computers, and the gate complexity matters. 
This will not be the last time we hide away such ``implementation details'' of quantum oracles.

\begin{rem}[Deutsch--Jozsa algorithm]
The single-qubit version of the Deutsch algorithm can be naturally generalized to the $n$-qubit version, called the Deutsch--Jozsa algorithm. Given $N=2^n$ boxes with an apple or an orange in each box, and the promise that either 1) all boxes contain the same type of fruit or 2) exactly half of the boxes contain apples and the other half contain oranges, we would like to distinguish the two cases. Mathematically, given the promise that a boolean function $f:\{0,1\}^n\to\{0,1\}$ is either a constant function (i.e., $|\set{x|f(x)=0}|=0$ or $2^n$) or a balanced function (i.e., $|\set{x|f(x)=0}|=2^{n-1}$), we would like to decide to which type $f$ belongs. We refer to \cite[Section 1.4.4]{NielsenChuang2000} for more details.
\end{rem}

\begin{exam}[Qiskit example of Deutsch's algorithm]\label{exam:qiskit_deutsch}
For $f(0)=f(1)=1$ (constant case), we can use $U_f=I\otimes X$. 
For $f(0)=0,f(1)=1$ (balanced case), we can use $U_f=\opr{CNOT}$. 
\begin{center}
\includegraphics[width=1.0\textwidth]{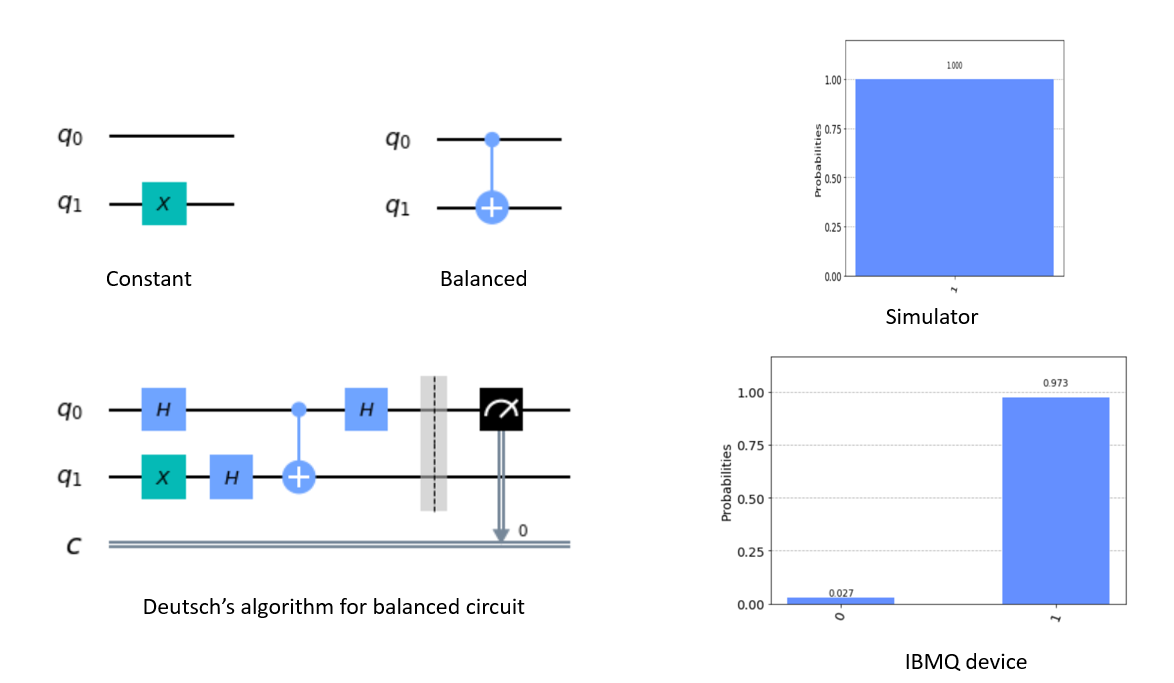}
\end{center}
\end{exam}

\section{Unstructured search problem}\label{sec:grover}

Assume we have $N=2^{n}$ boxes, and we are given the promise that only one of the boxes contains an orange, and each of the remaining boxes contains an apple. The goal is to find the box that contains the orange.

Mathematically, given a boolean function $f:\{0,1\}^n\to\{0,1\}$ and the promise that there exists a unique marked state $x_0$ that $f(x_0)=1$, we would like to find $x_0$.
This is called an unstructured search problem. 
Classically, there is no simpler methods than opening $(N-1)$ boxes in the worst case to determine $x_0$. 

The quantum algorithm below, called Grover's algorithm, relies on access to an oracle
\begin{equation}
U_f\ket{x,y}=\ket{x,y\oplus f(x)}, \quad x\in \{0,1\}^n,y\in \{0,1\},
\label{eqn:grover_oracle}
\end{equation}
and can find $x_0$ using $\Or(\sqrt{N})$ queries.
A classical computer again can only query $U_f$ in the computational basis, and Grover's algorithm achieves a \emph{quadratic speedup} in terms of the query complexity. 

The origin of the quadratic speedup can be summarized as follows: while classical probabilistic algorithms work with probability densities, quantum algorithms work with wavefunction amplitudes, of which the square gives the probability densities. More specifically, we start from a uniform superposition of all states as the initial state
\begin{equation}
\ket{\psi_0}=\frac{1}{\sqrt{N}}\sum_{x\in[N]} \ket{x}.
\end{equation}
This state can be prepared using Hadamard gates as 
\begin{equation}
\ket{\psi_0}=H^{\otimes n}\ket{0^n}.
\end{equation}
We would like to \emph{amplify} the desired amplitude corresponding to $\ket{x_0}$ from $1/\sqrt{N}$ to $\sqrt{p}=\Omega(1)$. 
We demonstrate below that this requires $\Or(\sqrt{N})$ queries to $U_f$.
After this procedure, by measuring the final state in the computational basis, we obtain some output state $\ket{x}$. We can check whether $x=x_0$ by applying another query of $U_f$ according to $U_f \ket{x,0}=\ket{x,f(x)}$. 
The probability of obtaining $f(x)=1$ is $p$.  If $f(x)\ne 1$, we repeat the process. 
Then after $\Or(1/p)$ times of repetition, we can obtain $x_0$ with high probability.

The first step of Grover's algorithm is to turn the oracle \eqref{eqn:grover_oracle} into a phase kickback. For this we take $\ket{y}=\ket{-}$, and for any $x\in\{0,1\}^n$,
\begin{equation}
U_f\ket{x,-}=\frac{1}{\sqrt{2}}\left(\ket{x,f(x)}-\ket{x,1\oplus f(x)}\right)=(-1)^{f(x)} \ket{x,-}.
\end{equation}
Any quantum state $\ket{\psi}$ can be decomposed as
\begin{equation}
\ket{\psi}=\alpha \ket{x_0}+\beta\ket{\psi^{\perp}},
\end{equation}
where $\ket{\psi^{\perp}}$ is the component of $\ket{\psi}$ orthogonal to $\ket{x_0}$, i.e., $\braket{\psi^{\perp}|x_0}=0$. 
We have
\begin{equation}
U_f\ket{\psi}\otimes\ket{-}=(-\alpha\ket{x_0}+\beta \ket{\psi^{\perp}})\otimes\ket{-}.
\end{equation}
Here the minus sign is gained through the phase kickback.
Discarding the $\ket{-}$ which is unchanged by applying $U_f$, we obtain an $n$-qubit unitary
\begin{equation}
R_{x_0} (\alpha \ket{x_0}+\beta\ket{\psi^{\perp}})=-\alpha\ket{x_0}+\beta \ket{\psi^{\perp}}.
\end{equation}
Therefore $R_{x_0}$ is a \emph{reflection operator} across the hyperplane orthogonal to $\ket{x_0}$, i.e., the Householder reflector 
\begin{equation}
R_{x_0}=I-2\ket{x_0}\bra{x_0}.
\end{equation}

Let us write
\begin{equation}
\ket{\psi_0}=\sin (\theta/2)\ket{x_0}+\cos(\theta/2)\ket{\psi_0^{\perp}},
\end{equation}
with $\theta=2\arcsin \frac{1}{\sqrt{N}}\approx \frac{2}{\sqrt{N}}$,
and $\ket{\psi_0^{\perp}}=\frac{1}{\sqrt{N-1}}\sum_{x\ne x_0} \ket{x}$ is a normalized state orthogonal to $\ket{x_0}$. 
Then
\begin{equation}
R_{x_0}\ket{\psi_0}=-\sin(\theta/2)\ket{x_0}+\cos(\theta/2)\ket{\psi_0^{\perp}}.
\end{equation}
So $\opr{span}\{\ket{x_0},\ket{\psi_0^{\perp}}\}$ is an invariant subspace of $R_{x_0}$. 

The next key step is to consider another Householder reflector with respect to $\ket{\psi_0}$. For later convenience we add a global phase factor $-1$ (which is irrelevant to the physical outcome):
\begin{equation}
R_{\psi_0}=-(I-2\ket{\psi_0}\bra{\psi_0}).
\end{equation}

Direct computation shows 
\begin{equation}
\begin{split}
R_{\psi_0}R_{x_0}\ket{\psi_0}=&R_{\psi_0}(\ket{\psi_0}-2\sin(\theta/2)\ket{x_0})\\
=&(\ket{\psi_0}-4\sin^2(\theta/2)\ket{\psi_0})+2\sin(\theta/2)\ket{x_0}\\
=&\sin (\theta/2) (3-4\sin^2 (\theta/2))\ket{x_0}+ 
\cos (\theta/2)(1-4\sin^2 (\theta/2))\ket{\psi_0^{\perp}}\\
=& \sin(3\theta/2)\ket{x_0}+\cos(3\theta/2)\ket{\psi_0^{\perp}}.
\end{split}
\end{equation}
So define $G=R_{\psi_0}R_{x_0}$ as the product of the two reflection operators (called the Grover operator), then it amplifies the angle from $\theta/2$ to $3\theta/2$.
The geometric picture is in fact even clearer in \cref{fig:grover_rotation} and the conclusion can be observed without explicit computation.

\begin{figure}[H]
\begin{center}
\includegraphics[width=0.4\textwidth]{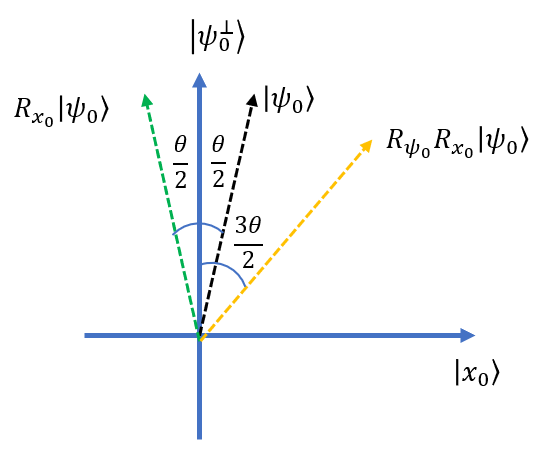}
\end{center}
\caption{Geometric interpretation of one Grover iteration.}
\label{fig:grover_rotation}
\end{figure}

Applying the Grover operator $k$ times, we obtain 
\begin{equation}\label{eqn:grover_k}
G^k \ket{\psi_0}=\sin((2k+1)\theta/2)\ket{x_0}+\cos((2k+1)\theta/2)\ket{\psi_0^{\perp}}.
\end{equation}
So for $\sin((2k+1)\theta/2)\approx 1$, we need $k\approx \frac{\pi}{2\theta}-\frac{1}{2}\approx \frac{\sqrt{N}\pi}{4}$.
This proves that Grover's algorithm can solve the unstructured search problem with $\Or(\sqrt{N})$ queries to $U_f$.

Another derivation of the Grover method is to focus on the operator, instead of the initial vector $\ket{\psi_0}$ at each step of the calculation.
In the orthonormal basis $\mc{B}=\{\ket{x_0},\ket{\psi_0^{\perp}}\}$, the matrix representation of the reflector $R_{x_0}=I-2\ket{x_0}\bra{x_0}$ is
\begin{equation}
[R_{x_0}]_{\mc{B}}=\begin{pmatrix}
-1 & 0 \\
0 & 1
\end{pmatrix}.
\end{equation}
The matrix representation for the Grover diffusion operator $R_{\psi_0}=2\ket{\psi_0}\bra{\psi_0}-I$ is
\begin{equation}
[R_{\psi_0}]_{\mc{B}}=\begin{pmatrix}
2a^2-1 & 2a\sqrt{1-a^2} \\
2a\sqrt{1-a^2} & 1-2a^2
\end{pmatrix}=\begin{pmatrix}
-\cos \theta & \sin\theta \\
\sin\theta & \cos\theta
\end{pmatrix}.
\end{equation}
Here $\sin(\theta/2)=a=1/\sqrt{N}$. Therefore for the matrix representation of the Grover iterate $G=R_{\psi_0}R_{x_0}$ is
\begin{equation}
\label{eqn:grover_step_matrix}
[G]_{\mc{B}}=\begin{pmatrix}
\cos \theta & \sin\theta \\
-\sin\theta & \cos\theta
\end{pmatrix},
\end{equation}
i.e., $G$ is a rotation matrix restricted to the two-dimensional space $\mc{H}=\opr{span}\mc{B}$.

The initial vector satisfies 
\begin{equation}
[\ket{\psi_0}]_{\mc{B}}=\begin{pmatrix}
\sin(\theta/2)\\
\cos(\theta/2)
\end{pmatrix},
\end{equation}
so Grover's search can be applied as before, via $G^k$ for $k\approx \frac{\pi\sqrt{N}}{4}$  times.

To draw the quantum circuit of Grover's algorithm, we need an implementation of $R_{\psi_0}$. 
Note that
\begin{equation}
R_{\psi_0}=H^{\otimes n}(2\ket{0^n}\bra{0^n}-I)H^{\otimes n}.
\end{equation}
This can be implemented via the following circuit using one ancilla qubit:
\begin{displaymath}
\begin{quantikz}
\lstick{$\ket{-}$}& \gate{X} & \targ{} & \qw & \qw\\
\lstick{$\ket{\psi}$}&  \gate{H^{\otimes n}}& \octrl{-1} & \gate{H^{\otimes n}} & \qw
\end{quantikz}
\end{displaymath}
Here the controlled-NOT gate is an $n$-qubit controlled-$X$ gate, and is only active if the system qubits are in the $0^n$ state.
Discarding the signal qubit, we obtain and implementation of $R_{\psi_0}$.
Since the signal qubit $\ket{-}$ only changes up to a sign, it can be reused for both $R_{\psi_0}$ and $R_{x_0}$.

The reflector $R_{\psi_0}$ can also be implemented without using the ancilla qubit as (use a 3-qubit system as an example)
\begin{figure}[H]
\begin{displaymath}
\begin{quantikz}
&  \gate{H}& \gate{X}&\gate{Z} & \gate{X}&\gate{H} & \qw\\
&  \gate{H}& \gate{X}&\ctrl{-1} & \gate{X}&\gate{H} & \qw\\
&  \gate{H}& \gate{X}&\ctrl{-1} & \gate{X}&\gate{H} & \qw
\end{quantikz}
\end{displaymath}
\caption{Implementing $R_{\psi_0}$ for a three qubit system.}
\label{fig:rpsi0}
\end{figure}
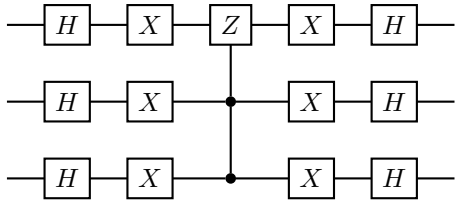

\begin{rem}[Multiple marked states] The Grover search algorithm can be naturally generalized to the case when there are $M>1$ marked states.
The query complexity is $\Or(\sqrt{N/M})$.
\end{rem}

\begin{exam}[Qiskit for Grover's algorithm]
\url{https://qiskit.org/textbook/ch-algorithms/grover.html}
\end{exam}

\section{Amplitude amplification}\label{sec:amplitudeamplification}

Grover's algorithm is not restricted to the problem of unstructured search.
One immediate application is called amplitude amplification (AA), which is used ubiquitously as a subroutine to achieve quadratic speedups.

Let \(\ket{\psi_0}\) be prepared by an oracle $U_{\psi_0}$, i.e., $U_{\psi_0}\ket{0^n}=\ket{\psi_0}$.
We have the knowledge that
\begin{equation}
\ket{\psi_0}=\sqrt{p_0} \ket{\psi_{\mathrm{good}}}+\sqrt{1-p_0} \ket{\psi_{\mathrm{bad}}},
\end{equation}
and $p_0\ll 1$. Here $\ket{\psi_{\mathrm{bad}}}$ is an orthogonal state to the desired state $\ket{\psi_{\mathrm{good}}}$.
We cannot get access to $\ket{\psi_{\mathrm{good}}}$ directly, but would like to obtain a state that has a large overlap with $\ket{\psi_{\mathrm{good}}}$, i.e., amplify the amplitude of $\ket{\psi_{\mathrm{good}}}$.

In the problem of unstructured search, $\ket{\psi_{\mathrm{good}}}=\ket{x_0}$,
and $p_0=1/N$. 
Although we do not have access to the answer $\ket{x_0}$, we assume access to the reflection oracle $R_{x_0}$.
Here we also assume access to the reflection oracle
\begin{equation}
R_{\mathrm{good}}=1-2\ket{\psi_{\mathrm{good}}}\bra{\psi_{\mathrm{good}}}.
\end{equation}
From $U_{\psi_0}$, we can construct the reflection with respect to the initial state
\begin{equation}\label{eqn:R_psi0}
R_{\psi_0}=2\ket{\psi_0}\bra{\psi_0}-I=U_{\psi_0}(2\ket{0^n}\bra{0^n}-I)U^{\dag}_{\psi_0}
\end{equation}
via the $n$-qubit controlled-$X$ gate.
So following exactly the same procedure as the unstructured search problem, we can construct the Grover iterate
\begin{equation}
G=R_{\psi_0}R_{\mathrm{good}}.
\end{equation}
Applying $G^k$ to $\ket{\psi_0}$ for some $k=\Or(1/\sqrt{p_0})$, we obtain a state that has $\Omega(1)$ overlap with $\ket{\psi_{\mathrm{good}}}$.

\begin{exam}[Reflection with respect to signal qubits]
  \label{exam:ref_signal}
One common scenario is that the implementation of $U_{\psi_0}$ requires $m$ ancilla qubits (also called signal qubits), i.e.,
\begin{equation}
U_{\psi_0}\ket{0^m}\ket{0^n}=\sqrt{p_0}\ket{0^m}\ket{\psi_0}+\sqrt{1-p_0}\ket{\perp},
\label{eqn:Upsi0_mark}
\end{equation}
where $\ket{\perp}$ is some orthogonal state satisfying
\begin{equation}
(\Pi\otimes I_n)\ket{\perp}=0, \quad \Pi=\ket{0^m}\bra{0^m}.
\end{equation}
Therefore
\begin{equation}
\ket{\psi_{\mathrm{good}}}=\ket{0^m}\ket{\psi_0}, \quad \ket{\psi_{\mathrm{bad}}}=\ket{\perp}.
\end{equation}
This setting is special since the ``good'' state can be verified by measuring the ancilla qubits after applying $U_{\psi_0}$ in \cref{eqn:Upsi0_mark}, and post-select the outcome $0^m$.
In particular, the expected number of measurements needed to obtain $\ket{\psi_{\mathrm{good}}}$ is $1/p_0$.

In order to employ the AA procedure, we first note that the reflection operator can be simplified as
\begin{equation}
R_{\mathrm{good}}=(1-2\ket{0^m}\bra{0^m})\otimes I_n=(1-2\Pi)\otimes I_n.
\end{equation}
This is because $\ket{\psi_{\mathrm{good}}}$ can be entirely identified by measuring the ancilla qubits. 
Meanwhile 
\begin{equation}
R_{\psi_0}=U_{\psi_0}(2\ket{0^{m+n}}\bra{0^{m+n}}-I)U^{\dag}_{\psi_0}.
\end{equation}
Let $G=R_{\psi_0}R_{\mathrm{good}}$, and applying $G^k$ to $\ket{\psi_0}$ for some $k=\Or(1/\sqrt{p_0})$ times, we obtain a state that has $\Omega(1)$ overlap with $\ket{\psi_{\mathrm{good}}}$. 
This achieves the desired quadratic speedup.
\end{exam}

\begin{exam}[Amplitude damping] Assuming access to an oracle in \cref{eqn:Upsi0_mark}, where $p_0$ is large, we can easily dampen the amplitude to any number $\alpha\le \sqrt{p_0}$.
% We introduce an additional signal qubit, and for simplicity let us assume $m=1$.
% Then  \cref{eqn:Upsi0_mark} becomes
% \begin{equation}
% (I\otimes U_{\psi_0})\ket{0}\ket{0}\ket{0^n}=\ket{0}\sqrt{p_0}\ket{0}\ket{\psi_0}+\sqrt{1-p_0}\ket{\perp}.
% \end{equation}
% Now for $\ket{b}=b_0\ket{0}+b_1\ket{1}$, define the controlled rotation
% \begin{equation}\label{eqn:controlled_rotation_simple}
% \opr{CR}\ket{0}\ket{b}=b_0 (\cos \theta \ket{0}+\sin \theta\ket{1})\ket{0}+b_1\ket{1},
% \end{equation}
% we have
% \begin{equation}
% (\opr{CR}\otimes I_n)(I\otimes U_{\psi_0})\ket{0}\ket{0}\ket{0^n}=\sqrt{p_0}\cos \theta\ket{0}\ket{0}\ket{\psi_0}+\sqrt{1-p_0}\ket{\perp'}.
% \end{equation}
% Then just choose $\sqrt{p_0}\cos \theta=\alpha$.
% 
% When the number of ancilla qubits $m>1$, we need to replace the controlled rotation by the multi-qubit controlled rotation.

We introduce an additional signal qubit. Then  \cref{eqn:Upsi0_mark} becomes
\begin{equation}
(I\otimes U_{\psi_0})\ket{0}\ket{0}\ket{0^n}=\ket{0}\sqrt{p_0}\ket{0}\ket{\psi_0}+\sqrt{1-p_0}\ket{\perp}.
\end{equation}
Define a single qubit rotation operation as 
\begin{equation}\label{eqn:controlled_rotation_simple}
R_{\theta}\ket{0}=\cos \theta \ket{0}+\sin \theta\ket{1},
\end{equation}
and we have
\begin{equation}
\begin{split}
&(R_{\theta}\otimes I_{m+n})(I\otimes U_{\psi_0})\ket{0}\ket{0^m}\ket{0^n}\\
=&\cos \theta\ket{0}(\sqrt{p_0}\ket{0^m}\ket{\psi_0}+\sqrt{1-p_0}\ket{\perp'})+
\sin \theta\ket{1}(\sqrt{p_0}\ket{0^m}\ket{\psi_0}+\sqrt{1-p_0}\ket{\perp'})\\
:=&\sqrt{p_0}\cos \theta\ket{0}\ket{0^m}\ket{\psi_0}+\sqrt{1-p_0\cos^2\theta}\ket{\perp'}.
\end{split}
\end{equation}
Here $(\ket{0^{m+1}}\bra{0^{m+1}}\otimes I_n)\ket{\perp'}=0$.
We only need to choose $\sqrt{p_0}\cos \theta=\alpha$. 

\end{exam}

\section{Lower bound of query complexity*}\label{sec:lowerbound_grover}

Recall that the unstructured search problem tries to find a marked state $x_0\in[N]$, using a reflection oracle
\begin{equation}
R_{x_0}=I-2\ket{x_0}\bra{x_0}.
\end{equation}
Grover's algorithm can find $x_0$ with constant probability (e.g. at least $1/2$) by making $\Or(\sqrt{N})$ times querying $R_{x_0}$. 
It turns out that this is \emph{asymptotically optimal}, i.e., no quantum algorithm can perform this task using fewer than $\Omega(\sqrt{N})$ access to $R_{x_0}$.

Any quantum search algorithm that starts from a universal initial state $\ket{\psi_0}$ and queries $R_{x_0}$ for $k$ steps can be written in the following form:
\begin{equation}
\ket{\psi^{x_0}_k}=\mc{U}^{x_0}_k\ket{\psi_0}=U_kR_{x_0}\cdots U_2R_{x_0} U_1R_{x_0}\ket{\psi_0},
\end{equation}
for some unitaries $U_1,\ldots,U_k$.
For simplicity we assume no ancilla qubits are used, and the proof can be generalized to the case in the presence of ancilla qubits.
The superscript $x_0$ indicates that the state depends on the marked state $x_0$.
Specifically, by ``solving'' the search problem, it means that there exists a $k$ so that for each marked state $x_0$, 
\begin{equation}\label{eqn:def_grover_success}
|\braket{\psi^{x_0}_k|x_0}|^2\ge \frac12.
\end{equation}
In other words, measuring $\ket{\psi^{x_0}_k}$ in the computational basis, the probability of obtaining $\ket{x_0}$ is at least $1/2$. 

To prove the lower bound, we compare the action of $\mc{U}^{x_0}_k$ with another ``fake algorithm'' $\mc{U}_k$, defined as
\begin{equation}
\ket{\psi_k}=\mc{U}_k\ket{\psi_0}=U_k\cdots U_2 U_1\ket{\psi_0}.
\end{equation}
In particular,  $\ket{\psi_k}$ does not involve any information of the solution $x_0$ and therefore cannot possibly solve the search problem.

For a set of vectors $\{f^{x_0}\}_{x_0\in [N]}$, and each $f^{x_0}\in \CC^N$, we will extensively use the following discrete $\ell_2$-norm:
\begin{equation}
\norm{f}_{\ell_2}:=\left(\sum_{x_0\in[N]} \norm{f^{x_0}}\right)^{\frac12}.
\end{equation}
In particular, we have the following triangle inequality
\begin{equation}
\norm{f}_{\ell_2}-\norm{g}_{\ell_2}\le\norm{f+g}_{\ell_2}\le\norm{f}_{\ell_2}+\norm{g}_{\ell_2}.
\end{equation}

The proof contains two steps. First, we show that the true solution and the fake solution differs significantly, in the sense that 
\begin{equation}\label{eqn:dk_lowerbound}
D_k:=\sum_{x_0\in[N]}\norm{\ket{\psi^{x_0}_k}-\ket{\psi_k}}^2=\Omega(N).
\end{equation}
Second, we prove that (define $D_0=0$):
\begin{equation}
D_k\le 4k^2, \quad k\ge 0.
\label{eqn:grover_dk}
\end{equation}
Therefore to satisfy \cref{eqn:dk_lowerbound}, we must have $k=\Omega(\sqrt{N})$.

In the first step, since multiplying a phase factor $e^{\I \theta}$ to $\ket{\psi^{x_0}_k}$ does not have any physical consequences, we may choose a particular phase $\theta$ so that \cref{eqn:def_grover_success} becomes
\begin{equation}
\braket{\psi^{x_0}_k|x_0}\ge \frac{1}{\sqrt{2}}.
\end{equation}
Therefore
\begin{equation}
\norm{\ket{\psi^{x_0}_k}-\ket{x_0}}^2= 2-2\braket{\psi^{x_0}_k|x_0}\le 2-\sqrt{2}.
\end{equation}
This means that
\begin{equation}
\sum_{x_0\in[N]}\norm{\ket{\psi^{x_0}_k}-\ket{x_0}}^2\le 2N-\sqrt{2}N.
\label{eqn:search_correct_cor}
\end{equation}

On the other hand, for the ``fake algorithm'', using the Cauchy-Schwarz inequality,
\begin{equation}
\sum_{x_0\in[N]}\norm{\ket{\psi_k}-\ket{x_0}}^2\ge 2N-2\sum_{x_0\in[N]}|\braket{x_0|\psi}|\ge2N-2\sqrt{N} \sum_{x_0\in[N]}|\braket{x_0|\psi}|^2=2N-2\sqrt{N}.
\label{eqn:search_fake_cor}
\end{equation}
This violates the bound in \cref{eqn:search_correct_cor}, and the fake algorithm cannot solve the search problem for arbitrarily large $k$.

So from \cref{eqn:search_correct_cor,eqn:search_fake_cor}, and the triangle inequality, we have
\begin{equation}
\begin{split}
D_k=\sum_{x_0\in[N]}\norm{\ket{\psi^{x_0}_k}-\ket{\psi_k}}^2&=\sum_{x_0\in[N]}\norm{(\ket{\psi^{x_0}_k}-\ket{x_0})-(\ket{\psi_k}-\ket{x_0})}^2\\
&\ge \left(\sqrt{\sum_{x_0\in[N]}\norm{\ket{\psi_k}-\ket{x_0}}^2}-\sqrt{\sum_{x_0\in[N]}\norm{\ket{\psi^{x_0}_k}-\ket{x_0}}^2}\right)^2\\
&\ge \left(\sqrt{2N-2\sqrt{N}}-\sqrt{2N-\sqrt{2}N}\right)^2\\
&=\Omega(N). 
\end{split}
\end{equation}
This proves \cref{eqn:dk_lowerbound}.
In other words, the true solution and the fake solution must be well separated in $\ell_2$-norm.

In the second step, we prove \cref{eqn:grover_dk} inductively. 
Clearly  \cref{eqn:grover_dk} is true for $k=0$. 
Assume this is true, then
\begin{equation}
\begin{split}
D_{k+1}&=\sum_{x_0\in[N]}\norm{\ket{\psi^{x_0}_{k+1}}-\ket{\psi_{k+1}}}^2\\
&=\sum_{x_0\in[N]}\norm{U_{k+1}R_{x_0}\ket{\psi^{x_0}_{k}}-U_{k+1}\ket{\psi_{k}}}^2\\
&=\sum_{x_0\in[N]}\norm{R_{x_0}\ket{\psi^{x_0}_{k}}-\ket{\psi_{k}}}^2\\
&=\sum_{x_0\in[N]}\norm{R_{x_0}(\ket{\psi^{x_0}_{k}}-\ket{\psi_{k}})+(R_{x_0}-I)\ket{\psi_{k}}}^2\\
&\le \left(\sqrt{\sum_{x_0\in[N]}\norm{R_{x_0}(\ket{\psi^{x_0}_{k}}-\ket{\psi_{k}})}^2}+\sqrt{\sum_{x_0\in[N]}\norm{(R_{x_0}-I)\ket{\psi_{k}}}^2}\right)^2.
\end{split}
\end{equation}
The last inequality uses the triangle inequality of the discrete $\ell_2$-norm.
Note that
\begin{equation}
\sqrt{\sum_{x_0\in[N]}\norm{(R_{x_0}-I)\ket{\psi_{k}}}^2}=\sqrt{4\sum_{x_0\in[N]}\abs{\braket{x_0|\psi_k}}^2}=2,
\end{equation}
and 
\begin{equation}
\sqrt{\sum_{x_0\in[N]}\norm{R_{x_0}(\ket{\psi^{x_0}_{k}}-\ket{\psi_{k}})}^2}=\sqrt{\sum_{x_0\in[N]}\norm{\ket{\psi^{x_0}_{k}}-\ket{\psi_{k}}}^2}=\sqrt{D_k},
\end{equation}
we have
\begin{equation}
\sqrt{D_{k+1}}\le \sqrt{D_k}+2\le 2(k+1), 
\end{equation}
which finishes the induction.

Finally, combining the lower bound \cref{eqn:dk_lowerbound,eqn:grover_dk}, we find that the necessary condition to solve the unstructured search problem is $4k^2=\Omega(N)$, or $k=\Omega(\sqrt{N})$.

\begin{rem}[Implication for amplitude amplification]
Due to the close relation between unstructred search and amplitude amplification, it means that given a state $\ket{\psi}$ of which the amplitude of the ``good'' component is $\alpha\ll 1$, no quantum algorithms can amplify the amplitude to $\Omega(1)$ using $o(\alpha^{-\frac12})$ queries to the reflection operators.
\end{rem}

\vspace{2em}

\begin{exer}
In Deutsch's algorithm, demonstrate why not assuming access to an oracle $V_f:\ket{x}\mapsto\ket{f(x)}$.
\end{exer}

\begin{exer}
For all possible mappings $f:\{0,1\}\to\{0,1\}$, draw the corresponding quantum circuit to implement $U_f:\ket{x,0}\mapsto\ket{x,f(x)}$.
\end{exer}

\begin{exer}
Prove \cref{eqn:grover_k}.
\end{exer}

\begin{exer}
Draw the quantum circuit for \cref{eqn:R_psi0}.
\end{exer}

\begin{exer}
Prove that when ancilla qubits are used, the complexity of the unstructured search problem is still $\Omega(\sqrt{N})$.
\end{exer}

\chapter{Quantum phase estimation}\label{chap:qpe}

The setup of the phase estimation problem is as follows. 
Let $U$ be a unitary, and $\ket{\psi}$ is an eigenvector, i.e.,
\begin{equation}
U\ket{\psi}=e^{\I 2\pi \varphi} \ket{\psi}, \quad \varphi\in[0,1).
\label{eqn:qpe_objective}
\end{equation}
The goal is to find $\varphi$ up to certain precision.
This is a quantum primitive with numerous applications: prime factorization (Shor's algorithm), linear system (HHL), eigenvalue problem, amplitude estimation, quantum counting, quantum walk, etc.

Using a classical computer, we can estimate $\varphi$ using $U\ket{\psi}\oslash\ket{\psi}$, where $\oslash$ stands for the element-wise division operation.
Specifically, if $\ket{\psi}$ is indeed an eigenvector and $\braket{j|\psi}\ne 0$ for any $j$ in the computational basis, then we can extract the phase from
\begin{equation}
\braket{j|U|\psi}/\braket{j|\psi}=e^{\I 2\pi \varphi}.
\end{equation}
Unfortunately, such a element-wise division operation cannot be efficiently implemented on quantum computers, and alternative methods are needed.

Quantum phase estimation has numerous variants, and still receives intensive research attention till today. This chapter only introduces some of the simplest variants.

\section{Hadamard test}

We first introduce the Hadamard test, which is a useful tool for computing the expectation value of an unitary operator with respect to a state, i.e., $\braket{\psi|U|\psi}$.
Since $U$ is generally not Hermitian, this does not correspond to the measurement of a physical observable.
Instead the real and imaginary part of the expectation value need to be measured separately.

The (real) Hadamard test is the quantum circuit in \cref{fig:hadamard_real} for estimating the real part of $\braket{\psi|U|\psi}$.
\begin{figure}[H]
\begin{displaymath}
\begin{quantikz}
\lstick{$\ket{0}$} & \gate{H} & \ctrl{1}  & \gate{H} &\meter{}\\
\lstick{$\ket{\psi}$} & \qwb           & \gate{U}  & \qw & \qw
\end{quantikz}
\end{displaymath}
\caption{Hadamard test for $\Re \braket{\psi|U|\psi}$.}
\label{fig:hadamard_real}
\end{figure}

To verify this, we find that the circuit transforms $\ket{0}\ket{\psi}$ as
\begin{displaymath}
\begin{split}
\ket{0}\ket{\psi} \xrightarrow{\mathrm{H}\otimes I} & \frac{1}{\sqrt{2}}(\ket{0}+\ket{1})\ket{\psi}\\
\xrightarrow{c\text{-}U} & \frac{1}{\sqrt{2}}(\ket{0}\ket{\psi}+\ket{1}U\ket{\psi})\\
\xrightarrow{\mathrm{H}\otimes I} & \frac{1}{2}\ket{0}(\ket{\psi}+U\ket{\psi})+
\frac{1}{2}\ket{1}(\ket{\psi}-U\ket{\psi}).
\end{split}
\end{displaymath}
The probability of measuring the qubit $0$ to be in state $\ket{0}$ is 
\begin{equation}
p(0)=\frac{1}{2}(1+\Re \braket{\psi|U|\psi}).
\end{equation}
This is well defined since $-1\le \Re \braket{\psi|U|\psi}\le 1$.

To obtain the imaginary part, we can use the circuit in \cref{fig:hadamard_imag} called the (imaginary) Hadamard test, where
\begin{equation} 
S=\begin{pmatrix}
1 & 0\\
0 & \I
\end{pmatrix}
\end{equation}
is called the phase gate.

\begin{figure}[H]
\begin{displaymath}
\begin{quantikz}
\lstick{$\ket{0}$} & \gate{H} & \gate{S^{\dag}} &\ctrl{1}  & \gate{H} &\meter{}\\
\lstick{$\ket{\psi}$} & \qwb           & \qw & \gate{U}  & \qw & \qw
\end{quantikz}
\end{displaymath}
\caption{Hadamard test for $\Im \braket{\psi|U|\psi}$.}
\label{fig:hadamard_imag}
\end{figure}

Similar calculation shows the circuit transforms $\ket{0}\ket{\psi}$ to the state
\begin{equation}
\frac{1}{2}\ket{0}(\ket{\psi}-\I U\ket{\psi})+
\frac{1}{2}\ket{1}(\ket{\psi}+\I U\ket{\psi}).
\end{equation}
Therefore the probability of measuring the qubit $0$ to be in state $\ket{0}$ is 
\begin{equation}
p(0)=\frac{1}{2}(1+\Im \braket{\psi|U|\psi}).
\end{equation}
Combining the results from the two circuits, we obtain the estimate to $\braket{\psi|U|\psi}$.

\begin{exam}[Overlap estimate using swap test]
An application of the Hadamard test is called the swap test, which is used to estimate the overlap of two quantum states $\abs{\braket{\varphi|\psi}}$.
The quantum circuit for the swap test is
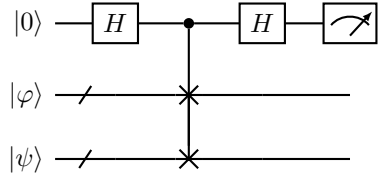
\begin{figure}[H]
\begin{displaymath}
\begin{quantikz}
\lstick{$\ket{0}$} & \gate{H} & \ctrl{2}  & \gate{H} &\meter{}\\
\lstick{$\ket{\varphi}$} & \qwb           & \swap{1}  & \qw & \qw\\
\lstick{$\ket{\psi}$} & \qwb           & \targX{}  & \qw & \qw
\end{quantikz}
\end{displaymath}
\caption{Circuit for the SWAP test.}
\label{fig:circuit_swap}
\end{figure}
Note that this is exactly the Hadamard test with $U$ being the $n$-qubit swap gate.
Direct calculation shows that the probability of measuring the qubit $0$ to be in state $\ket{0}$ is 
\begin{equation}
p(0)=\frac{1}{2}(1+\Re \braket{\varphi,\psi|\psi,\varphi})=\frac12(1+\abs{\braket{\varphi|\psi}}^2).
\end{equation}
\end{exam}

\begin{exam}[Overlap estimate with relative phase information]
\label{exam:swap_test}
In the swap test, the quantum states $\ket{\varphi},\ket{\psi}$ can be black-box states, and in such a scenario obtaining an estimate to $\abs{\braket{\varphi|\psi}}$ is the best one can do.
In order to retrieve the relative phase information and to obtain $\braket{\varphi|\psi}$, we need to have access to the unitary circuit preparing  $\ket{\varphi},\ket{\psi}$, i.e.,
\begin{equation}
U_{\varphi}\ket{0^n}=\ket{\varphi}, \quad U_{\psi}\ket{0^n}=\ket{\psi}.
\end{equation}
Then we have $\braket{\varphi|\psi}=\braket{0^n|U_{\varphi}^{\dag}U_{\psi}|0^n}$.
\end{exam}

\begin{exam}[Single qubit phase estimation]\label{exam:singlequbit_pe}
The Hadamard test can also be used to derive the simplest version of the phase estimate based on success probabilities.
Apply the Hadamard test in \cref{fig:hadamard_real} with $U,\psi$ satisfying \cref{eqn:qpe_objective}. Then the probability of measuring the qubit $0$ to be in state $\ket{1}$ is 
\begin{equation}
p(1)=\frac{1}{2}(1-\Re \braket{\psi|U|\psi})=\frac{1}{2}(1-\cos (2\pi\varphi)).
\end{equation}
Therefore
\begin{equation}
\varphi=\pm\frac{\arccos(1-2p(1))}{2\pi} \pmod{1}.
\end{equation}
In order to quantify the efficiency of the procedure, recall from \cref{exam:prob_onequbit} that if $p(1)$ is far away from  $0$ or $1$, i.e., $(2\varphi \mod 1)$ is far away from $0$, in order to approximate $p(1)$ (and hence $\varphi$) to additive precision $\epsilon$, the number of samples needed is $\Or(1/\epsilon^2)$.

Now assume $\varphi$ is very close to $0$ and we would like to estimate $\varphi$ to additive precision $\epsilon$. Note that
\begin{equation}
p(1)\approx (2\pi\varphi)^2 = \Or(\epsilon^2).
\end{equation}
Then $p(1)$ needs to be estimated to precision $\Or(\epsilon^2)$, and again the number of samples needed is $\Or(1/\epsilon^2)$. The case when $\varphi$ is close to $1/2$ or $1$ is similar.

Note that the circuit in \cref{fig:hadamard_real} cannot distinguish the sign of $\varphi$ (or  whether $\varphi\ge 1/2$ when restricted to the interval $[0,1)$).
To this end we need \cref{fig:hadamard_imag}, but replace $\mathrm{S}^{\dag}$ by $\mathrm{S}$, so that the success probability of measuring $1$ in the computational basis is
\begin{equation} 
p(1)=\frac{1}{2}(1+\sin(2\pi\varphi)).
\end{equation}
This gives
\begin{equation}
 \varphi
\begin{cases}
\in[0,1/2), & p(1)\ge \frac12,\\
\in(1/2,1),    & p(1)< \frac12.
\end{cases}
\end{equation}
Unlike the previous estimate, in order to correctly estimate the sign, we only require $\Or(1)$ accuracy, and run \cref{fig:hadamard_imag} for a constant number of times (unless $\varphi$ is very close to $0$ or $\pi$).
\end{exam}

\begin{exam}[Qiskit example for phase estimation using the Hadamard test]
\label{exam:qpe_hadamard}
Here is a qiskit example of the simple phase estimation for

\[
R|1\rangle = 
\begin{pmatrix}
1 & 0\\
0 & e^{i2\pi\varphi}\\ 
\end{pmatrix}
\begin{pmatrix}
0\\
1\\ 
\end{pmatrix}
= e^{\I2\pi\varphi}|1\rangle,
\]
with $\varphi=0.5+\frac{1}{2^d}\equiv 0.{10\cdots 01}$ (d bits in total). In this example, $d=4$ and the exact value is $\varphi=0.5625$. 
\begin{center}
\includegraphics[width=0.8\textwidth]{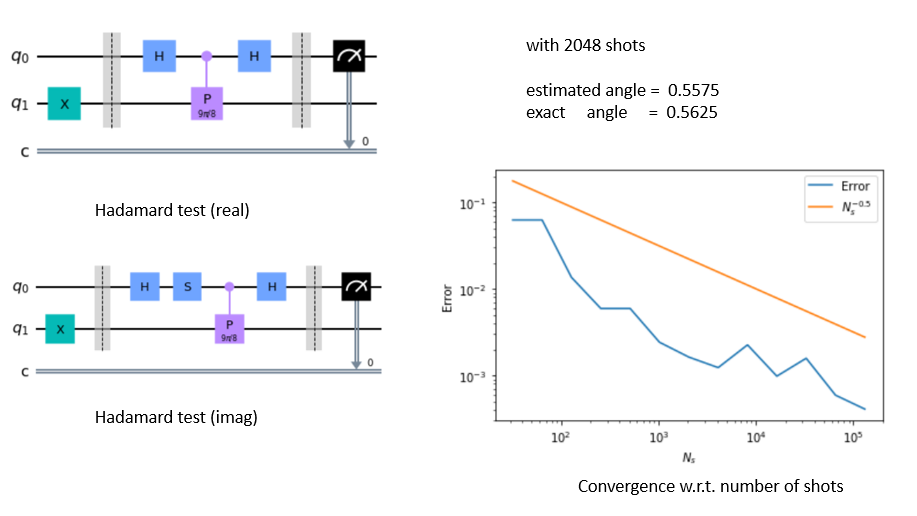}
\end{center}

\end{exam}

\section{Quantum phase estimation (Kitaev's method)*}

In \cref{exam:singlequbit_pe}, the number of measurements needed to estimate $\varphi$ to precision $\epsilon$ is $\Or(1/\epsilon^2)$.
A quadratic improvement in precision (i.e., $\Or(1/\epsilon)$) can be achieved by means of the quantum phase estimation. One such procedure is called Kitaev's method~\cite[Section 13.5]{KitaevShenVyalyi2002}.

In the fixed point representation in \cref{sec:fixedpoint}, for simplicity we assume that the eigenvalue can be exactly represented using $d$ bits, i.e.,
\begin{equation}
\varphi=(.\varphi_{d-1}\cdots\varphi_0).
\end{equation}
In the simplest scenario, we assume $d=1$, and $\varphi=.\varphi_0,\varphi_0\in\{0,1\}$.
Then $e^{\I2 \pi  \varphi}=e^{\I \pi \varphi_0}$.

Performing the real Hadamard test in \cref{exam:singlequbit_pe}, we have $p(1)=0$ if $\varphi_0=0$, and $p(1)=1$ if $\varphi_0=1$. 
In either case, the result is \emph{deterministic}, and one measurement is sufficient to determine the value of $\varphi_0$.

Next, consider $\varphi=.\underbrace{0\cdots 0\varphi_0}_{d \mbox{ bits}}$. To determine the value of $\varphi_0$, we need to reach precision $\epsilon<2^{-d}$. The method in \cref{exam:singlequbit_pe} requires $\Or(1/\epsilon^2)=\Or(2^{2d})$ repeated measurements, or number of queries to $U$.
The observation from Kitaev's method is that if we can have access to $U^j$ for a suitable power $j$, then the number of queries to $U$ can be reduced.
More specifically, if we can query $U^{2^{d-1}}$, then the circuit in \cref{fig:kitaev_example} with $j=d-1$ gives
\begin{equation} 
p(1)=\frac{1}{2}(1-\cos (2\pi.\varphi_0))=
\begin{cases}
0, &\varphi_0=0,\\
1, &\varphi_0=1.
\end{cases}
\end{equation}
The result is again deterministic. 
Therefore the total number of queries to $U$ becomes $\Or(2^{-d})=\Or(\epsilon^{-1})$.

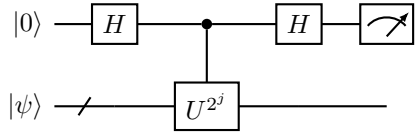
\begin{figure}[H]
\begin{displaymath}
\begin{quantikz}
\lstick{$\ket{0}$} & \gate{H} & \ctrl{1}  & \gate{H} &\meter{}\\
\lstick{$\ket{\psi}$} & \qwb           & \gate{U^{2^{j}}}  & \qw & \qw
\end{quantikz}
\end{displaymath}
\caption{Circuit used in Kitaev's method. Another one with a phase gate to determine the sign of $2^j \varphi$ may also be used.}
\label{fig:kitaev_example}
\end{figure}

This is the basic idea behind Kitaev's method: use a more complex quantum circuit (and in particular, with a larger circuit depth) to reduce the \emph{total} number of queries. 
As a general strategy, instead of estimating $\varphi$ from a single number, we assume access to $U^{2^j}$, and estimate $\varphi$ bit-by-bit.
In particular, changing $U\to U^{2^j}$ in the Hadamard test allows us to estimate
\begin{equation}
2^j \varphi=\varphi_{d-1}\cdots \varphi_{d-j}.\varphi_{d-j-1}\cdots\varphi_0= .\varphi_{d-j-1}\cdots\varphi_0 \pmod{1}
\end{equation}
One immediate difficulty of the bit-by-bit estimation is that we need to tell $0.0111\ldots$ apart from $0.1000\ldots$, and the two numbers can be arbitrarily close to each other (though the two numbers can also differ at some number of digits), and some careful work is needed. 
We will first describe the algorithm, and then analyze its performance.
The algorithm works for any $\varphi$, and then the goal is to estimate its $d$ bits.
For simplicity of the analysis, we assume $\varphi$ is exactly represented by $d$ bits. We will use extensively the distance
\begin{equation}\label{eqn:modone_dist}
\abs{x}_1\equiv\abs{x}_{\bmod 1}:=\min\{(x \bmod 1),1-(x \bmod 1)\},
\end{equation}
which is the distance on the unit circle.

First, by applying the circuit in \cref{fig:kitaev_example} (and the corresponding circuit to determine the sign) with $j=0,1,\ldots,d-3$, for each $j$ we can estimate $p(0)$, so that the error in $2^j\varphi$ is less than $1/16$ for all $j$ (this can happen with a sufficiently high success probability. For simplicity, let us assume that this happens with certainty).
The measured result is denoted by $\alpha_j$. 
This means that any perturbation must be due to the $5$th digit in the binary representation. For example, if $2^j\varphi=0.11100$, then $\alpha_j=0.11011$ is an acceptable result with an error $0.00001=1/32$, but $\alpha_j=0.11110$ is not acceptable since the error is $0.0001=1/16$. We then \textit{round} $\alpha_j$ $\pmod 1$ by its closest $3$-bit estimate denoted by $\beta_j$, i.e., $\beta_j$ is taken from the set $\set{0.000,0.001,0.010,0.011,0.100,0.101,0.110,0.111}$.  Consider the example of $2^j\varphi=0.11110$, if $\alpha_j=0.11101$, then $\beta_j=0.111$. But if $\alpha_j=0.11111$, then $\beta_j=0.000$. 
Another example is $2^j\varphi=0.11101$, if $\alpha_j=0.11110$, then both $\beta_j=0.111$ (rounded down) and $\beta_j=0.000$ (rounded up) are acceptable. We can pick one of them at random. We will show later that the uncertainty in $\alpha_j,\beta_j$ is not detrimental to the success of the algorithm.

Second, we perform some post-processing. Start from $j=d-3$, we can estimate $.\varphi_2\varphi_1\varphi_0$ to accuracy $1/16$, which recovers these three bits exactly.  The values of these three bits will be taken from $\beta_{d-3}$ directly. Then we proceed with the iteration: for $j=d-4,\ldots,0$, we assign
\begin{equation}\label{eqn:kitaev_decidephi}
\varphi_{d-j-1}=\begin{cases}
0,& \abs{.0\varphi_{d-j-2}\varphi_{d-j-3}-\beta_j}_{\bmod 1}<1/4,\\
1,& \abs{.1\varphi_{d-j-2}\varphi_{d-j-3}-\beta_j}_{\bmod 1}<1/4.
\end{cases}
\end{equation}
Here $\abs{\cdot}_{\bmod 1}$ is the periodic distance on $[0,1)$ and its value is always $\le 1/2$. 
Since the two possibilities are separated by $1/2$, for each $j$, there will be at most one case that is satisfied. We will also show that in all circumstances, there is always one case that is satisfied, regardless of the ambiguity of the choice of $\beta_j$ above.

After running the algorithm above, we recover $\varphi=.\varphi_{d-1}\cdots \varphi_0$ exactly. The total cost of Kitaev's method measured by the number of queries to $U$ is $\Or\left(\sum_{j=0}^{d-3} 2^{j}\right)=\Or(\epsilon^{-1})$.

If $\varphi$ is exactly represented by $d$ bits, we will obtain an estimate
\begin{equation}\label{eqn:phi_kitaev_estimate}
\abs{.\varphi_{d-1}\cdots \varphi_0-\varphi}_{\bmod 1}<2^{-d}=\epsilon.
\end{equation}

% 
% First at $k=0$, we run the (real and imaginary) Hadamard  test with $U$ to constant precision, and obtain the estimate for the two leading digits $\varphi_{d-1},\varphi_{d-2}$.
% 
% Next, for $k=1,\ldots,d-2$, we run the (real) Hadamard test with $U\to U^{2^k}$ to constant precision. This allows us to estimate the phase of the next \emph{three} bits:
% \begin{equation}
% .\underbracket{\varphi_{d-k-1}\varphi_{d-k-2}\varphi_{d-k-3}}\cdots\varphi_0 \pmod{1}.
% \end{equation}
% Note that the real Hadamard test cannot distinguish the first bit (which is a sign). So the value of $\varphi_{d-k-1}$ has already been determined from the previous stage. Then we also need the value of $\varphi_{d-k-3}$ in order to determine $\varphi_{d-k-2}$ in the presence of rounding. This is why each step requires us to focus on three bits at a time.
% 
% Repeating this procedure, we can determine all values in the fixed precision representation $\varphi=.\varphi_{d-1}\cdots\varphi_0$. 

\begin{exam}
Consider $\varphi=0.\varphi_4\varphi_3\varphi_2\varphi_1\varphi_0=0.11111$ and $d=5$. Running Kitaev's algorithm with $j=0,1,2$ gives the following possible choices of $\beta_j$:
\[
\begin{array}{c|c|c}\hline
j & 2^j \varphi & \mbox{possible }\beta_j \\\hline
0 & 0.11111 & \set{0.111,0.000} \\\hline
1 & 0.1111 & \set{0.111,0.000} \\\hline
2 & 0.111 & \set{0.111} \\\hline
\end{array}
\]
Start with $j=2$. We have only one choice of $\beta_j$, and can recover $0.\varphi_2\varphi_1\varphi_0=0.111$. Then for $j=1$,  we need to use \cref{eqn:kitaev_decidephi} to decide $\varphi_3$. If we choose $\beta_j=0.111$, we have $\varphi_3=1$. But if we choose $\beta_j=0.000$, we still need to choose $\varphi_3=1$, since 
$\abs{.011-0.000}_{\bmod 1}=0.100=1/2>1/4$, and $\abs{.111-0.000}_{\bmod 1}=0.001=1/8<1/4$. 
Similar analysis shows that for $j=0$ we have $\varphi_4=1$. This recovers $\varphi$ exactly.
\end{exam}

\begin{exam}[A variant of Kitaev's algorithm that does not work]
Let us modify Kitaev's algorithm as follows: for each $2^j \varphi$ is determined to precision $1/8$, and round the result to $\beta_j\in\set{0.00,0.01,0.10,0.11}$. Start from $j=d-2$, we estimate $.\varphi_1\varphi_0$ exactly. Then for $j=d-3,\ldots,0$, we assign
\begin{equation}\label{eqn:kitaev_decidephi_2}
\varphi_{d-j-1}=\begin{cases}
0,& \abs{.0\varphi_{d-j-2}-\beta_j}_{\bmod 1}<1/2,\\
1,& \abs{.1\varphi_{d-j-2}-\beta_j}_{\bmod 1}<1/2.
\end{cases}
\end{equation}
Now that the inequality $<1/2$ above can be equivalently written as $\le 1/4$.
  
Let us run the algorithm above for $\varphi=0.\varphi_4\varphi_3\varphi_2\varphi_1\varphi_0=0.1111$ and $d=4$. This gives:
\[
\begin{array}{c|c|c}\hline
j & 2^j \varphi & \mbox{possible }\beta_j \\\hline
0 & 0.1111 & \set{0.11,0.00} \\\hline
1 & 0.111 & \set{0.11,0.00} \\\hline
2 & 0.11 & \set{0.11} \\\hline
\end{array}
\]
Start with $j=2$. We have only one choice of $\beta_j$, and can recover $0.\varphi_1\varphi_0=0.11$. Then for $j=1$, if we choose $\beta_j=0.11$, we have $\varphi_2=1$. But if we choose $\beta_j=0.00$, then 
$\abs{.01-0.00}_{\bmod 1}=0.0.01=1/4$, and $\abs{.11-0.000}_{\bmod 1}=0.01=1/4$. So the algorithm cannot distinguish the two possibilities and fails.
\end{exam}

Let us now  inductively show why Kitaev's algorithm works. Again assume $\varphi$ is exactly represented by $d$ bits. For $j=d-3$, we know that $\varphi_2\varphi_1\varphi_0$ can be recovered exactly. Then assume $\varphi_{d-j-2}\cdots\varphi_0$ have all been exactly computed, at step $j$ we would like to determine the value of $\varphi_{d-j-1}$. From
\begin{equation}
\abs{\alpha_j-2^j\varphi}_{\bmod 1}<1/16, \quad \abs{\alpha_j-\beta_j}_{\bmod 1}\le 1/16,
\end{equation}
we know
\begin{equation}
\abs{2^j\varphi-\beta_j}_{\bmod 1}< 1/8.
\end{equation}
Then 
\begin{equation}
\begin{split}
 &\abs{.\varphi_{d-j-1}\varphi_{d-j-2}\varphi_{d-j-3}-\beta_j}_{\bmod 1}\\
\le &
\abs{.\varphi_{d-j-1}\varphi_{d-j-2}\varphi_{d-j-3}-2^j\varphi}_{\bmod 1}
+
\abs{2^j\varphi-\beta_j}_{\bmod 1}\\
\le & 1/16+1/8<1/4.
\end{split}
\end{equation}
The wrong choice of $\varphi_{d-j-1}$ denoted by  $\wt{\varphi}_{d-j-1}$ then satisfies
\begin{equation}
\begin{split}
 &\abs{.\wt{\varphi}_{d-j-1}\varphi_{d-j-2}\varphi_{d-j-3}-\beta_j}_{\bmod 1}\\
\ge &
\abs{.\wt{\varphi}_{d-j-1}\varphi_{d-j-2}\varphi_{d-j-3}-.\varphi_{d-j-1}\varphi_{d-j-2}\varphi_{d-j-3}}_{\bmod 1}-
\abs{.\varphi_{d-j-1}\varphi_{d-j-2}\varphi_{d-j-3}-\beta_j}_{\bmod 1}\\
> & 1/2-1/4=1/4.
\end{split}
\end{equation}
This proves the validity of \cref{eqn:kitaev_decidephi}, and hence that of Kitaev's algorithm.

% \begin{exam}[Qiskit example for phase estimation using Kitaev's algorithm]
% The setup is the same as in \cref{exam:qpe_hadamard}.
% \begin{center}
% \includegraphics[width=1.0\textwidth]{phase_result_kitaev}
% \end{center}
% \end{exam}

\section{Quantum Fourier transform}

Fourier transform is used ubiquitously in scientific computing, and fast Fourier transform (FFT) is the backbone for many fast algorithms in classical computation.
Similarly, the quantum Fourier transform is also an important component in many quantum algorithms, such as phase estimation, Shor's algorithm, and inspires other fast algorithms such as fast Fermionic Fourier transform (FFFT) etc.

For any $j$ in the computational basis, the (discrete) forward Fourier transform is defined as
\begin{equation}
U_{\mathrm{FT}}\ket{j}=\frac{1}{\sqrt{N}}\sum_{k\in [N]} e^{\I 2\pi\frac{kj}{N}} \ket{k}.
\end{equation}
In particular
\begin{equation}
  U_{\mathrm{FT}}\ket{0^n}=\frac{1}{\sqrt{N}}\sum_{k\in [N]} \ket{k}=H^{\otimes n}\ket{0^n}.
  \label{eqn:UFT_zeron}
\end{equation}

The (discrete) inverse Fourier transform is
\begin{equation}
U^{\dag}_{\mathrm{FT}}\ket{j}=\frac{1}{\sqrt{N}}\sum_{k\in [N]} e^{-\I 2\pi\frac{kj}{N}} \ket{k}.
\end{equation}

Using the binary representation of integers 
\begin{equation}
k=(k_{n-1}\cdots k_0.), \quad j=(j_{n-1}\cdots j_0.)
\end{equation}
we have
\begin{displaymath}
\begin{split}
\frac{kj}{N}=&k_0\frac{j}{2^{n}}+k_1\frac{j}{2^{n-1}}+\cdots+k_{n-1}\frac{j}{2}\\
=&k_0(.j_{n-1}\cdots j_0)+k_1(j_{n-1}.j_{n-2}\cdots j_0)+\cdots+k_{n-1}(j_{n-1}\cdots j_1.j_0).
\end{split}
\end{displaymath}
Therefore the  exponential can be written as
\begin{equation} 
e^{\I 2\pi\frac{kj}{N}}=e^{\I 2\pi k_0(.j_{n-1}\cdots j_0)} e^{\I 2\pi k_1(.j_{n-2}\cdots j_0)}\cdots e^{\I 2\pi k_{n-1}(.j_0)}.
\end{equation}
The most important step of QFT is the following direct calculation, which requires some patience with the manipulation of indices:
\begin{equation}
\begin{split}
U_{\mathrm{FT}}\ket{j_{n-1}\cdots j_0}=&\frac{1}{\sqrt{2^n}}\sum_{k_{n-1},\ldots,k_0} e^{\I 2\pi k_0(.j_{n-1}\cdots j_0)} e^{\I 2\pi k_1(.j_{n-2}\cdots j_0)}\cdots e^{\I 2\pi k_{n-1}(.j_0)}\ket{k_{n-1}\cdots k_0}\\
=&\frac{1}{\sqrt{2^n}} \left(\sum_{k_{n-1}}e^{\I 2\pi k_{n-1}(.j_0)}\ket{k_{n-1}}\right)\otimes
 \left(\sum_{k_{n-2}}e^{\I 2\pi k_{n-2}(.j_1j_0)}\ket{k_{n-1}}\right)\\
 &\otimes\cdots\otimes \left(\sum_{k_{0}}e^{\I 2\pi k_0(.j_{n-1}\cdots j_0)}\ket{k_{0}}\right)\\
=&\frac{1}{\sqrt{2^n}} \left(\ket{0}+e^{\I 2\pi (.j_0)}\ket{1}\right)\otimes
 \left(\ket{0}+e^{\I 2\pi (.j_1j_0)}\ket{1}\right)\otimes\cdots\otimes \left(\ket{0}+e^{\I 2\pi (.j_{n-1}\cdots j_0)}\ket{1}\right).
\end{split}
\label{eqn:QFT_compute}
\end{equation}

\cref{eqn:QFT_compute} involves a series of controlled rotations of the form \begin{equation}
\ket{0}\to \frac{1}{\sqrt{2}}\left(\ket{0}+e^{\I 2\pi (.j_{n-1}\cdots j_0)}\ket{1}\right).
\end{equation}
Hence before discussing the quantum circuit for QFT, let us first work out the circuit for implementing this controlled rotation. 
We use the relation 
\begin{equation}
e^{\I 2\pi(.j_{n-1}\cdots j_0)}=
e^{\I 2\pi(.j_{n-1})}e^{\I 2\pi(.0j_{n-2})}\cdots e^{\I 2\pi(.0\cdots 0 j_0)}.
\end{equation}

\begin{exam}[Implementation of controlled rotation]
Consider the implementation of 
\begin{equation}
\ket{0}\ket{j}\to \frac{1}{\sqrt{2}}\left(\ket{0}+e^{\I 2\pi (.j_{n-1}\cdots j_0)}\ket{1}\right)\ket{j},
\end{equation}
Let 
\begin{equation}
R_z(\theta)=\begin{pmatrix}
1 & 0 \\
0 & e^{\I \theta}
\end{pmatrix},
\end{equation}
and $R_j=R_z(\pi/2^{j-1})$. In particular, $R_1=Z$. 
The quantum circuit is

\begin{center}
\begin{quantikz}
  \lstick{$\ket{0}$}& \gate{H} & \gate{R_1} & \gate{R_2} &\cdots \qw  & \gate{R_{n}}&\qw\\
\lstick{$\ket{j_{n-1}}$}& \qw& \ctrl{-1} & \qw&\cdots\qw&\qw&\qw\\
\lstick{$\ket{j_{n-2}}$}& \qw& \qw& \ctrl{-2} & \cdots\qw&\qw&\qw\\
\lstick{$\cdots$}&\\
\lstick{$\ket{j_{0}}$}& \qw& \qw & \qw&\cdots\qw & \ctrl{-4} & \qw\\
\end{quantikz}
\end{center}
\end{exam}

The implementation of QFT follows the same principle, but \emph{does not} require the signal qubit to store the phase information. 
Let us see a few examples.

When $n=1$, we need to implement
\begin{equation}
\ket{j_0}\to \frac{1}{\sqrt{2}} \left(\ket{0}+e^{\I 2\pi (.j_0)}\ket{1}\right).
\end{equation}
This is the Hadamard gate:
\begin{equation}
\ket{j_0}\to H\ket{j_0}.
\end{equation}

When $n=2$, we need to implement
\begin{equation}
\ket{j}\to\frac{1}{\sqrt{2^2}} \left(\ket{0}+e^{\I 2\pi (.j_0)}\ket{1}\right)\otimes
 \left(\ket{0}+e^{\I 2\pi (.j_1j_0)}\ket{1}\right).
\label{eqn:QFT_n2}
\end{equation}
This can be implemented using the following circuit:
\begin{displaymath}
\begin{quantikz}
\lstick{$\ket{j_{1}}$}& \gate{H} & \gate{R_2} &\qw&\qw&\rstick{$\frac{1}{\sqrt{2}}(\ket{0}+e^{\I 2\pi (.j_{1} j_0)}\ket{1})$}\qw\\
\lstick{$\ket{j_{0}}$}& \qw& \ctrl{-1} & \qw&\gate{H}&\rstick{$\frac{1}{\sqrt{2}}(\ket{0}+e^{\I 2\pi (.j_0)}\ket{1})$}\qw
\end{quantikz}
\end{displaymath}
Comparing the result with that in \cref{eqn:QFT_n2}, we find that the ordering of the qubits is reversed. 
To recover the desired result in QFT, we can apply a SWAP gate to the outcome, i.e., 
\begin{displaymath}
\begin{quantikz}
\lstick{$\ket{j_{1}}$}& \gate{H} & \gate{R_2} &\qw&\qw&\swap{1}&\rstick{$\frac{1}{\sqrt{2}}(\ket{0}+e^{\I 2\pi (.j_0)}\ket{1})$}\qw\\
\lstick{$\ket{j_{0}}$}& \qw& \ctrl{-1} & \qw&\gate{H}&\targX{}&\rstick{$\frac{1}{\sqrt{2}}(\ket{0}+e^{\I 2\pi (.j_{1} j_0)}\ket{1})$}\qw
\end{quantikz}
\end{displaymath}

In order to implement the inverse Fourier transform, we only need to apply the Hermitian conjugate as
\begin{displaymath}
\begin{quantikz}
\lstick{$\ket{j_{1}}$}& \swap{1}&\qw&\gate{R_2^{\dag}} &\gate{H} & \rstick{$\frac{1}{\sqrt{2}}(\ket{0}+e^{-\I 2\pi (.j_0)}\ket{1})$}\qw\\
\lstick{$\ket{j_{0}}$}& \targX{}&\gate{H}&\ctrl{-1} &\qw& \rstick{$\frac{1}{\sqrt{2}}(\ket{0}+e^{-\I 2\pi (.j_{1} j_0)}\ket{1})$}\qw
\end{quantikz}
\end{displaymath}
Similarly one can construct the circuit for $U_{\mathrm{FT}}$ and its inverse for $n=3$. 

In general, the QFT circuit is given by \cref{fig:circ_QFT}. 
Compare the circuit in \cref{fig:circ_QFT} with \cref{eqn:QFT_compute}, we find again that the ordering is reversed in the output.
To restore the correct order as defined in QFT, we can use $\Or(n/2)$ swaps operations. The total gate complexity of QFT is $\Or(n^2)$.

\begin{figure}[H]
\begin{displaymath}
\begin{quantikz}
\lstick{$\ket{j_{n-1}}$}& \gate{H} & \gate{R_2} & \gate{R_3} &\cdots \qw  & \gate{R_{n}}&\qw&\qw&\qw&\cdots\qw&\qw&\rstick{$\frac{1}{\sqrt{2}}(\ket{0}+e^{\I 2\pi (.j_{n-1}\cdots j_0)}\ket{1})$}\qw\\
\lstick{$\ket{j_{n-2}}$}& \qw& \ctrl{-1} & \qw&\cdots\qw&\qw&\qw&\gate{H}&\gate{R_2}&\cdots\qw&\qw&\rstick{$\frac{1}{\sqrt{2}}(\ket{0}+e^{\I 2\pi (.j_{n-2}\cdots j_0)}\ket{1})$}\qw\\
\lstick{$\ket{j_{n-3}}$}& \qw& \qw& \ctrl{-2} & \cdots\qw&\qw&\qw&\qw&\ctrl{-1}&\cdots\qw&\qw&\rstick{$\frac{1}{\sqrt{2}}(\ket{0}+e^{\I 2\pi (.j_{n-3}\cdots j_0)}\ket{1})$}\qw\\
\lstick{$\cdots$}&\\
\lstick{$\ket{j_{0}}$}& \qw& \qw & \qw&\cdots\qw & \ctrl{-4} & \qw &\qw&\qw&\cdots\qw&\gate{H}&\rstick{$\frac{1}{\sqrt{2}}(\ket{0}+e^{\I 2\pi (.j_0)}\ket{1})$}\qw\\
\end{quantikz}
\end{displaymath}
\caption{Quantum circuit for quantum Fourier transform (before applying swap operations).}
\label{fig:circ_QFT}
\end{figure}
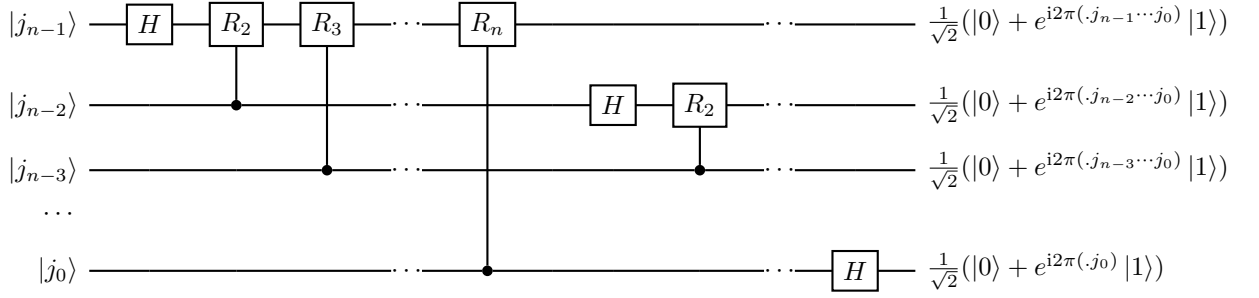

% \begin{center}
% \includegraphics[width=0.9\textwidth]{qft_placeholder}
% \end{center}

\begin{exam}[Qiskit example for QFT]
https://qiskit.org/textbook/ch-algorithms/quantum-fourier-transform.html
\end{exam}

\section{Quantum phase estimation using quantum Fourier transform}

In Kitaev's method, we use $1$ ancilla qubit but $d$ \emph{different} circuits of various circuit depths to perform phase estimation.
In this section, we introduce the (standard) quantum phase estimation (QPE), which uses one signal quantum circuit based on QFT, but requires $d$-ancilla qubits to store the phase information in the quantum computer. From now, we assume $\varphi=.\varphi_{d-1}\cdots \varphi_0$ is exact.
From the availability of $U^j$ we can define a controlled unitary operation
\begin{equation}
\mc{U}=\sum_{j\in[2^d]} \ket{j}\bra{j}\otimes U^j.
\end{equation}
When $d=1$, $\mc{U}$ is simply the controlled $U$ operation.
For a general $d$, it seems that we need to implement all $2^d$ different $U^j$.
However, this is not necessary. 
Using the binary representation of integers $j=(j_{d-1}\cdots j_0.)=\sum_{i=0}^{d-1} j_i 2^{i}$, we have
$U^{j}=U^{\sum_{i=0}^{d-1} j_i 2^i}=\prod_{i=0}^{d-1}U^{j_i2^i}$.
Therefore similar to the operations in QFT,
\begin{equation}
\begin{split}
\mc{U}=&\sum_{j\in[2^d]} \ket{j}\bra{j}\otimes U^j\\
=&\sum_{j_{d-1},\ldots,j_0} (\ket{j_{d-1}}\bra{j_{d-1}})\otimes\cdots\otimes(\ket{j_{0}}\bra{j_{0}})\otimes\prod_{i=0}^{d-1}U^{j_i2^i}\\
=&\xprod_{i=0}^{d-1}\left(\sum_{j_i}\ket{j_{i}}\bra{j_{i}}\otimes U^{j_i 2^i}\right)\\
=&\xprod_{i=0}^{d-1}\left(\ket{0}\bra{0}\otimes I_n+\ket{1}\bra{1}\otimes U^{2^i}\right).
\end{split}
\end{equation}
Here the primed product $\xprod$ is a slightly awkward notation, which means the tensor product for the first register, and the regular matrix product for the second register. 
It is in fact much clearer to observe the structure in the quantum circuit in \cref{fig:circuit_controlled_Upower}.
\begin{figure}[H]
\begin{displaymath}
  \begin{quantikz}
    \lstick{$\ket{j_{d-1}}$} & \qw & \qw & \cdots \qw& \ctrl{4}& \qw \\
    \lstick{$\cdots$}&\\
    \lstick{$\ket{j_{1}}$}   & \qw & \ctrl{2} & \cdots \qw & \qw & \qw\\
    \lstick{$\ket{j_{0}}$}   & \ctrl{1} & \qw & \qw \cdots \qw & \qw & \qw\\
    \lstick{$\ket{\psi}$}    & \gate{U}\qwb & \gate{U^2} & \cdots \qw & \gate{U^{2^{d-1}}} &\rstick{$\mc{U}\ket{\psi}$}\qw 
  \end{quantikz}
\end{displaymath}
\caption{Circuit for controlled matrix power of $U$.}
\label{fig:circuit_controlled_Upower}
\end{figure}
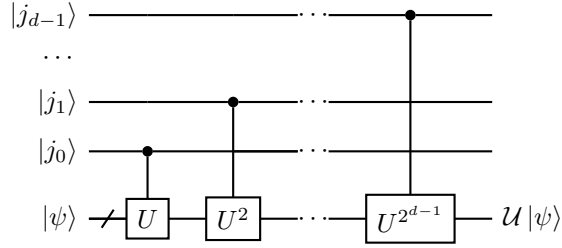

\begin{rem}
At first glance, the saving due to the usage of the circuit in \cref{fig:circuit_controlled_Upower} may not seem to be large, since we still need to implement matrix powers as high as $U^{2^{d-1}}$.
However, the alternative would be to implement $\sum_{j\in[2^d]} \ket{j}\bra{j}$, which requires very complicated multi-qubit control operations.
Another scenario when significant advantage can be gained is when $U$ can be \emph{fast-forwarded}, i.e., $U^j$ can be implemented at a cost that is independent of $j$. This is for instance, if $U=R_z(\theta)$ is a single-qubit rotation. Then the circuit \cref{fig:circuit_controlled_Upower} is exponentially better than the direct implementation of $\mc{U}$.
\end{rem}

Now let the initial state in the ancilla qubits be $\ket{0^n}$. 
Use  QFT and $\mc{U}$, we transform the initial states according to
\begin{equation}
\begin{split}
\ket{0^d}\ket{\psi_0} \xrightarrow{U_{\mathrm{FT}}\otimes I} & \frac{1}{\sqrt{2^d}}\sum_{j\in[2^d]}\ket{j}\ket{\psi_0}\\
\xrightarrow{\mc{U}} & \frac{1}{\sqrt{2^d}}\sum_{j\in[2^d]}\ket{j}U^j\ket{\psi_0}= \frac{1}{\sqrt{2^d}}\sum_{j\in[2^d]}\ket{j}e^{\I 2\pi \varphi j}\ket{\psi_0}\\
\xrightarrow{U_{\mathrm{FT}}^{\dag}\otimes I} &  \sum_{k'\in[2^d]}\left(\frac{1}{2^d}\sum_{j\in[2^d]}
e^{\I 2\pi j\left(\varphi -\frac{k'}{2^d}\right)}\right)\ket{k'}\ket{\psi_0}.
\end{split}
\label{eqn:qpe_exactvec}
\end{equation}
Since we have $\varphi=\frac{k}{2^d}$ for \emph{some} $k\in[2^d]$, measuring the ancilla qubits, and we will obtain the state $\ket{k}\ket{\psi_0}$ with certainty, and we obtain the phase information.
Therefore the quantum circuit for QFT based QPE is given by \cref{fig:circ_qpe_qft}. Here we have used \cref{eqn:UFT_zeron}.
We should note that $U_{\mathrm{FT}}^{\dag}$ includes the swapping operations.
(Exercise: 1. what if the swap operation is not implemented? 2. Is it possible to modify the circuit and remove the need of implementing the swap operations?)

%  the measurement of the ancilla register gives a state $\ket{k_0}\ket{k_1}\cdots\ket{k_{d-1}}$, instead of the standard ordering
% $\ket{k_{d-1}}\ket{k_{d-2}}\cdots\ket{k_{0}}$

 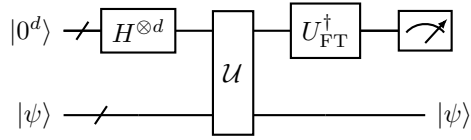
\begin{figure}[H]
 \begin{displaymath}
 \begin{quantikz}
   \lstick{$\ket{0^d}$} & \gate{H^{\otimes d}}\qwb& 
   \gate[2]{\mc{U}}&\gate{U^{\dag}_{\mathrm{FT}}}&\meter{}\qw\\
   \lstick{$\ket{\psi}$} & \qwb  &  & \qw & \rstick{$\ket{\psi}$}\qw 
 \end{quantikz}
 \end{displaymath}
 \caption{Quantum circuit for quantum phase estimation using quantum Fourier transform.}
 \label{fig:circ_qpe_qft}
 \end{figure}

% old circuit
% conceptual
% \begin{quantikz}
% \lstick{$\ket{0}^{\otimes d}$} & \gate{U_{\mathrm{FT}}}\qwbundle[alternate]{}\slice[style={xshift=-1cm}]{$\ket{j}\otimes\ket{\psi}$}& 
% \ctrlbundle{1}&\gate{U^{\dag}_{\mathrm{FT}}}\qwbundle[alternate]{}&\meter{}\qwbundle[alternate]{}\\
% \lstick{$\ket{\psi}$} & \qwb  & \gate{U^j} & \qw
% \end{quantikz}
%
%practical
%\begin{displaymath}
%\begin{quantikz}
%\lstick{$j_{d-1}\to\ket{0}$} & \gate[4]{U_{\mathrm{FT}}}& \ctrl{4}\qw\qw&\qw&\push{\cdots}
%&\qw&\gate[4]{U^{\dag}_{\mathrm{FT}}}&\meter{}\\
%\lstick{$j_{d-2}\to\ket{0}$} &  & \qw & \ctrl{3}\qw&\push{\cdots}&\qw&\qw&\meter{}\\
%\lstick{$\cdots$} &  & \qw&\qw&\push{\cdots}&\qw&\qw&\push{\cdots}\\
%\lstick{$j_{0}\to\ket{0}$} &  & \qw&\qw&\push{\cdots}&\ctrl{1}&\qw&\meter{}\\
%\lstick{$\ket{\psi}$} & \qwb & \gate{U^{2^{d-1}}} & \gate{U^{2^{d-2}}} &\push{\cdots}&\gate{U}
%&\rstick{$\ket{\psi}$}\qw
%\end{quantikz}
%\end{displaymath}

\begin{exam}[Hadamard test viewed as QPE]
\label{exam:hadamard_qpe}
When $d=1$, note that $U^{\dag}=U^{\dag}_{\mathrm{FT}}=H$, the QFT based QPE in \cref{fig:circ_qpe_qft} is exactly the Hadamard test in \cref{fig:hadamard_real}. Note that $\varphi$ does not need to be exactly represented by a one bit number! 
\end{exam}

\begin{exam}[Qiskit example for QPE]
https://qiskit.org/textbook/ch-algorithms/quantum-phase-estimation.html
\end{exam}

\section{Analysis of quantum phase estimation}\label{sec:analysis_qpe}

In order to apply QPE (standard or Kitaev), we have assumed that 
\begin{enumerate}
  \item $\ket{\psi_0}$ is an eigenstate.
  \item $\varphi_0$ has an $d$-bit binary representation.
\end{enumerate}
In general practical calculations, neither condition can be \emph{exactly} satisfied, and we need to analyze the effect on the error of the QPE.
Recall the discussion in \cref{sec:fault_tolerant}, we assume the only sources of errors are at the mathematical level (instead of noisy implementation of quantum devices).
In this context, the error can be due to an inexact eigenstate $\ket{\psi}$, or Monte Carlo errors in the readout process due to the probabilistic nature of the measurement process.
In this section, we relax these constraints and study what happens when the conditions are not exactly met. 
We assume $U$ has the eigendecomposition.
\begin{equation}
U\ket{\psi_j}=e^{\I 2\pi  \varphi_j}\ket{\psi_j}.
\end{equation}
Without loss of generality we assume $0\le \varphi_0\le \varphi_1\cdots \le\varphi_{N-1}<1$. 
We are interested in using QPE to find the value of $\varphi_0$.

We first relax the condition (1), i.e., assume \emph{all} $\varphi_i$'s have an exact $d$-bit binary representation, but the quantum state is given by a linear combination
\begin{equation}
\ket{\phi}=\sum_{k\in[N]} c_k \ket{\psi_k}.
\end{equation}
Here the overlap $p_0=\abs{\braket{\phi|\psi_0}}^2=\abs{c_0}^2<1$. 

Applying the QPE circuit in \cref{fig:circ_qpe_qft} to $\ket{0^t}\ket{\phi}$ with $t=d$, and measure the ancilla qubits, we obtain the binary representation of $\varphi_0$ with probability $p_0$. 
Furthermore, the system register returns the eigenstate $\ket{\psi_0}$ with probability $p_0$.
Of course, in order to recognize that $\varphi_0$ is the desired phase, we need some \emph{a priori} knowledge of the location of $\varphi_0$, e.g. $\varphi_0\in(a,b)$ and $\varphi_i>b$ for all $i\ne 0$.
It would be desirable to relax both conditions (1) and (2). 
However, the analysis can be rather involved and additional assumptions are needed. 
We will discuss some of the implications in the context of estimating the ground state energy in \cref{sec:groundenergy}.

For now, to simplify the analysis, we focus on the case that only the condition (2) is violated, i.e., $\varphi_0$ cannot be exactly represented by a $d$-bit number, and we need to apply the QPE circuit to an initial state $\ket{0^t}\ket{\phi}$ with $t>d$.
The exact relation between the $t$ and the desired accuracy $d$ will be determined later.
Similar to \cref{eqn:qpe_exactvec}, we obtain the state
\begin{equation}
\begin{split}
\ket{0^t}\ket{\psi_0} \to&\sum_{k'\in[T]}\left(\frac{1}{T}\sum_{j\in[T]}
e^{\I 2\pi j\left(\varphi_0 -\frac{k'}{T}\right)}\right)\ket{k'}\ket{\psi_0}\\
=& \sum_{k'} \gamma_{0,k'}\ket{k'}\ket{\psi_0}.
\end{split}
\label{eqn:qpe_inexactvec}
\end{equation}
Here 
\begin{equation}
\gamma_{0,k'}=\frac{1}{T}\sum_{j\in[T]}
e^{\I 2\pi j\left(\varphi_0 -\frac{k'}{T}\right)}=\frac{1}{T}\frac{1-e^{\I 2\pi T\left(\varphi_0 - \wt{\varphi}_{k'}\right)}}{1-e^{\I 2\pi \left(\varphi_0 - \wt{\varphi}_{k'}\right)}}, \quad \wt{\varphi}_{k'}=\frac{k'}{T}.
\label{eqn:gamma_0k}
\end{equation}
 
Therefore if $\varphi_0$ has an exact $d$-bit representation, i.e., $\varphi_0=\wt{\varphi}_{k'_0}$ for some $k'_0$, then $\gamma_{0,k'}=\delta_{k',k'_0}$. 
We recover the previous result that one run of the QPE circuit gives the value $\varphi_0$ deterministically.

Now assume that $\varphi_0 \ne \wt{\varphi}_{k'}$ for any $k'$. 
Note that $e^{\I 2\pi x}$ is a periodic function with period $1$, we can only determine the value of $x \bmod 1$. 
Therefore  we use the periodic distance \cref{eqn:modone_dist}.
In terms of the phase, we would like to find $k_0'$ such that
\begin{equation}
\abs{\varphi_{0} - \wt{\varphi}_{k'_0}}_1 < \epsilon.
\end{equation}
Here $\epsilon=2^{-d}=2^{t-d}/T$ is the precision parameter. 
In particular, for any $k'$ we have
\begin{equation}
\abs{\varphi_{0} - \wt{\varphi}_{k'}}_1 \le \frac12.
\end{equation}

Using the relation that for any $\theta\in[-\pi,\pi]$,
\begin{equation}
\abs{1-e^{\I\theta}}=\sqrt{2(1-\cos\theta)}=2 \abs{\sin(\theta/2)}\ge \frac{2}{\pi }\abs{\theta},
\end{equation}
we obtain
\begin{equation}
\abs{\gamma_{0,k'}}\le \frac{2}{T \frac{2}{\pi} 2\pi \abs{\varphi_{0} - \wt{\varphi}_{k'}}_1}
=\frac{1}{2T \abs{\varphi_{0} - \wt{\varphi}_{k'}}_1}.
\end{equation}

\begin{figure}[H]
\begin{center}
\includegraphics[width=0.4\textwidth]{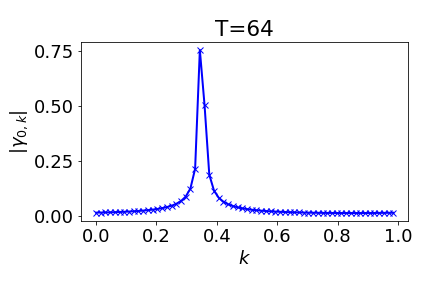}
\includegraphics[width=0.4\textwidth]{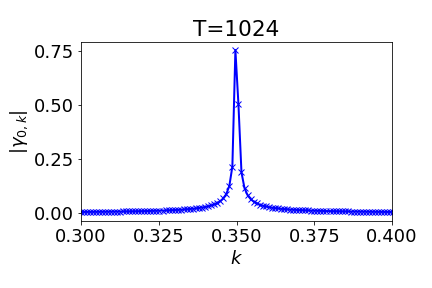}
\end{center}
\caption{For $\varphi_0=0.35$, the shape of $\abs{\gamma_{0,k}}$ with $T=64$ and $T=1024$.}
\label{fig:qpe_gamma}
\end{figure}

Let $k'_0$ be the measurement outcome, which can be viewed as a random variable. The probability of obtaining some $\wt{\varphi}_{k'_0}$ that is at least $\epsilon$ distance away from $\varphi_0$ is
\begin{equation}
\begin{split}
P(\abs{\varphi_0-\wt{\varphi}_{k'_0}}_1\ge \epsilon)=&\sum_{\abs{\varphi_0-\wt{\varphi}_{k'}}_1\ge \epsilon} \abs{\gamma_{0,k'}}^2\\
\le &\sum_{\abs{\varphi_0-\wt{\varphi}_{k'}}_1\ge \epsilon}\frac{1}{4T^2 \abs{\varphi_{0} - \wt{\varphi}_{k'}}^2_1}\\
\le &\frac{2}{4T} \int_{\epsilon}^{\infty} \frac{1}{x^2}\ud x + \frac{2}{4T^2 \epsilon^2}=\frac{1}{2T\epsilon}+\frac{1}{2(T\epsilon)^2}.
\end{split}
\end{equation}
Set $t-d=\lceil\log_2 \delta^{-1}\rceil$, then $T\epsilon=2^{t-d}\ge \delta^{-1}$.
Hence for any $0<\delta<1$, the failure probability
\begin{equation}
P(\abs{\varphi_0-\wt{\varphi}_{k'_0}}_1\ge \epsilon)\le \frac{\delta+\delta^2}{2}\le \delta.
\end{equation}
In other words, in order to obtain the phase $\varphi_0$ to accuracy $\epsilon=2^{-d}$ with a success probability at least $1-\delta$, we need $d+\lceil\log_2 \delta^{-1}\rceil$ ancilla qubits to store the value of the phase.
On top of that, the simulation time needs to be $T=(\epsilon\delta)^{-1}$.

%\section{Quantum median method and high fidelity quantum phase estimation*}

\begin{rem}[Quantum median method]
One problem with QPE is that in order to obtain a success probability $1-\delta$, we must use $\log_2 \delta^{-1}$ ancilla qubits, and the maximal simulation time also needs to be increased by a factor $\delta^{-1}$. 
The increase of the maximal simulation time is particularly undesirable since it increases the circuit depth and hence the required coherence time of the quantum device.
When $\ket{\psi}$ is an exact eigenstate, this can be improved by the median method, which uses $\log \delta^{-1}$ copies of the result from QPE without using ancilla qubits or increasing the circuit depth.
When $\ket{\psi}$ is a linear combination of eigenstates, the problem of the aliasing effect becomes more difficult to handle.
One possibility is to generalize the median method into the quantum median method~\cite{NagajWocjanZhang2009}, which uses classical arithmetics to evaluate the median using a quantum circuit.
To reach success probability $1-\delta$, we still need $\log_2 \delta^{-1}$ ancilla qubits, but the maximal simulation time does not need to be increased.
\end{rem}

\vspace{2em}

\begin{exer}
Write down the quantum circuit for the overlap estimate in \cref{exam:swap_test}.
\end{exer}

\begin{exer}
For $\varphi=0.111111$, we run Kitaev's algorithm to estimate its first $4$ bits. Check that the outcome satisfies \cref{eqn:phi_kitaev_estimate}. Note that $0.0000$ and $0.1111$ are both acceptable answers.
Prove the validity of Kitaev's algorithm in general in the sense of \cref{eqn:phi_kitaev_estimate}.
\end{exer}

\begin{exer}
For a $3$ qubit system, explicitly construct the circuit for $U_{\mathrm{FT}}$ and its inverse. 
\end{exer}

\begin{exer}
For a $n$-qubit system, write down the quantum circuit for the swap operation used in QFT.
\end{exer}

\begin{exer}
Similar to the Hadamard test in \cref{exam:hadamard_qpe}, develop an algorithm to perform QPE using the circuit in \cref{fig:circ_qpe_qft} with only $d=2$, while the phase $\varphi$ can be any number in $[0,1/2)$. 
\end{exer}

\chapter{Applications of quantum phase estimation}\label{chap:app_qpe}

\section{Ground state energy estimation}\label{sec:groundenergy}

As an application, we use QPE to solve the problem of estimating the ground state energy of a Hamiltonian. 
Let $H$ be a Hermitian matrix with eigendecomposition
\begin{equation}
H\ket{\psi_j}=\lambda_j \ket{\psi_j}.
\end{equation}
Below are two examples of Hamiltonians commonly encountered in quantum many-body physics.
\begin{exam}[Transverse field Ising model]
The Hamiltonian for the one dimensional transverse field Ising model (TFIM) with nearest neighbor interaction of length $n$ is
\begin{equation}\label{eqn:ham_tfim}
H=-\sum_{i=1}^{n-1} Z_iZ_{i+1}-g\sum_{i=1}^n X_i.
\end{equation}
The dimension of the Hamiltonian matrix $H$ is $2^n$. 
\end{exam}

\begin{exam}[Fermionic system in second quantization]
\label{exam:fermion_hamiltonian}
For a fermionic system (such as electrons), the Hamiltonian can be expressed in terms of the creation and annihilation operators as
\begin{equation}
H=\sum_{ij=1}^{n} T_{ij} \hat{a}_{i}^{\dagger} \hat{a}_{j}+\sum_{ijkl=1}^{n} V_{ijkl} \hat{a}_i^{\dagger}\hat{a}_j^{\dagger}\hat{a}_l \hat{a}_k.
\label{eqn:fermion_hamiltonian}
\end{equation}
The creation and annihilation operators $\hat{a}^{\dag}_i,\hat{a}_i$ can be converted into Pauli operators via e.g. the Jordan-Wigner transform as
\begin{equation}
\hat{a}_{i}=Z^{\otimes (i-1)}\otimes\frac12(X+\I Y)\otimes I^{\otimes (N-i)}, \quad \hat{a}^{\dagger}_i=Z^{\otimes (i-1)}\otimes\frac12(X-\I Y)\otimes I^{\otimes (N-i)}.
\label{eqn:jordan_wigner}
\end{equation}
Here $X,Y,Z,I$ are single-qubit Pauli-matrices.  
The dimension of the Hamiltonian matrix $\hat{H}$ is thus $2^n$.

The number operator takes the form
\begin{equation}
    \hat{n}_{i}:=\hat{a}^{\dag}_{i}\hat{a}_{i}=\frac12 (I-Z_{i}).
    \label{eq:ndecomp}
\end{equation}
For a given state $\ket{\Psi}$, the total number of particles is
\begin{equation}
N_e=\Braket{\Psi|\sum_{i=1}^n \hat{n}_i|\Psi}.
\end{equation}
The Hamiltonian $H$ preserves the total number of particles $N_e$.
\end{exam}

Without loss of generality we assume $0<\lambda_0\le\lambda_1\le \cdots  <\lambda_{N-1}<\frac12$. Note that for the purpose of estimating the ground state energy, we do not necessarily require a positive energy gap. 
For simplicity of the presentation, we still assume that the ground state is non-degenerate, i.e., $\lambda_0<\lambda_1$. 
We are also provided an approximate eigenstate
\begin{equation}
\ket{\phi}=\sum_{k\in[N]} c_k \ket{\psi_k},
\end{equation}
of which the overlap with the ground state is  $p_0=\abs{\braket{\phi|\psi_0}}^2$. Our goal is to estimate $\lambda_0$ to precision $\epsilon=2^{-d}$. We  assume $\epsilon<\lambda_0$. This appears in many problems in quantum many-body physics, quantum chemistry, optimization etc.

In order to use QPE (based on QFT), we assume access to the unitary evolution operator $U=e^{\I 2\pi H}$. This is called a Hamiltonian simulation problem, which will be discussed in detail in later chapters.
For now we assume $U$ can be implemented exactly.
Then
\begin{displaymath}
U\ket{\psi_0}=e^{\I 2\pi \lambda_0}\ket{\psi_0}.
\end{displaymath}
This becomes a phase estimation problem, where the input vector is not an exact eigenstate. 
Following the discussion in \cref{sec:analysis_qpe}, if \emph{all} eigenvalues $\lambda_j$ can be exactly represented by $d$-bit numbers, we obtain both the ground state and the ground state energy with probability $p_0$.
Therefore repeating the process for $\Or(p_0^{-1})$ times we obtain the ground state energy.

Now we relax both conditions (1) and (2) in \cref{sec:analysis_qpe}, and apply the QPE circuit in \cref{fig:circ_qpe_qft} to an initial state $\ket{0^t}\ket{\phi}$ for some $t>d$.
Similar to \cref{eqn:qpe_exactvec}, we have
\begin{equation}
\begin{split}
\ket{0^t}\ket{\phi} \xrightarrow{U_{\mathrm{FT}}\otimes I} & \sum_{k}c_k\frac{1}{\sqrt{T}}\sum_{j\in[T]}\ket{j}\ket{\psi_k}\\
\xrightarrow{\mc{U}} & \sum_{k} c_k \frac{1}{\sqrt{T}}\sum_{j\in[T]}\ket{j}U^j\ket{\psi_k}= \sum_{k} c_k \frac{1}{\sqrt{T}}\sum_{j\in[T]}\ket{j}e^{\I 2\pi \lambda_k j}\ket{\psi_k}\\
\xrightarrow{U_{\mathrm{FT}}^{\dag}\otimes I} &  
\sum_{k} c_k \sum_{k'\in[T]}\left(\frac{1}{T}\sum_{j\in[T]}
e^{\I 2\pi j\left(\lambda_k -\frac{k'}{T}\right)}\right)\ket{k'}\ket{\psi_k}\\
=& \sum_{k} \sum_{k'\in[T]} c_k \gamma_{k,k'}\ket{k'}\ket{\psi_k}.
\end{split}
\label{eqn:qpe_inexactvec}
\end{equation}
Here 
\begin{equation}
\gamma_{k,k'}=\frac{1}{T}\sum_{j\in[T]}
e^{\I 2\pi j\left(\lambda_k -\frac{k'}{T}\right)}=\frac{1}{T}\frac{1-e^{\I 2\pi T\left(\lambda_k - \wt{\varphi}_{k'}\right)}}{1-e^{\I 2\pi \left(\lambda_k - \wt{\varphi}_{k'}\right)}}, \quad \wt{\varphi}_{k'}=\frac{k'}{T}.
\end{equation}
Therefore the definition in \cref{eqn:gamma_0k} is a special case.

Our algorithm is simple: we would like to run QPE $M$ times, and denote the output of the $\ell$-th run by $\wt{\varphi}^{(\ell)}_{k'}$. 
Then we take the minimum of the measured output $\min_{\ell}\wt{\varphi}^{(\ell)}_{k'}$ as the estimate to the ground state energy.
The hope is to obtain the ground state energy to accuracy $\epsilon$ with success probability at least $1-\delta$ for any $\delta>0$. 
Let us now analyze this algorithm.

If $\varphi_0$ has an exact $d$-bit representation, i.e., $\lambda_0=\wt{\varphi}_{k'_0}$ for some integer $k'_0$, then $\gamma_{k'_0,k'}=\delta_{k'_0,k'}$. 
It may seem that this implies with probability $p_0$, we obtain the exact estimate of $\varphi_0$, and correspondingly the eigenstate $\ket{\psi_0}$ is stored in the system register. This is much better than the previous assumption that \emph{all} eigenvalues $\lambda_j$ need to be represented by a $d$-bit number.

Unfortunately, this analysis is not correct.
In fact, for any $\lambda_k$ that does not have an exact $t$-bit representation (note that $t>d$), we may have $\gamma_{k,k''}\ne 0$ and $\wt{\varphi}_{k''}<\lambda_0$, i.e., we obtain an energy estimate that is lower than the ground state energy! Therefore the probability of ending up in the state $\ket{k''}\ket{\psi_k}$ is $\abs{c_k\gamma_{k,k''}}^2$, i.e., it is still possible to obtain a wrong ground state energy. This is called the \emph{aliasing effect}.

We demonstrate below that if $T$ is large enough, we can control the probability of underestimating the ground state energy. 
Since $\lambda_0$ is the ground state energy, and all eigenvalues are in $(0,1/2)$, when $\wt{\varphi}_{k'}\le \lambda_0-\epsilon$, we have 
\begin{equation}
\abs{\lambda_k-\wt{\varphi}_{k'}}_1\ge \abs{\lambda_0-\wt{\varphi}_{k'}}_1=\lambda_0-\wt{\varphi}_{k'}\ge \epsilon, \quad \forall k\in[N].
\end{equation}
Then the probability of under estimating the ground state energy by $\epsilon$ is
\begin{equation}
\begin{split}
P(\wt{\varphi}_{k'}\le \lambda_0-\epsilon)=&\sum_k \sum_{\lambda_0-\wt{\varphi}_{k'}\ge \epsilon} \abs{c_k \gamma_{k,k'}}^2\\
\le & \sum_{k} \sum_{\lambda_0-\wt{\varphi}_{k'}\ge \epsilon} p_k \frac{1}{4T^2 \abs{\lambda_k-\wt{\varphi}_{k'}}_1^2}\\
\le & \sum_{k} \sum_{\lambda_0-\wt{\varphi}_{k'}\ge \epsilon} p_k \frac{1}{4T^2 \abs{\lambda_0-\wt{\varphi}_{k'}}^2}\\
=   & \sum_{\lambda_0-\wt{\varphi}_{k'}\ge \epsilon}\frac{1}{4T^2 \abs{\lambda_0-\wt{\varphi}_{k'}}^2}\\
\le & \frac{1}{4T\epsilon}+\frac{1}{4(T\epsilon)^2}.
\end{split}
\end{equation}
Let $T\epsilon=\delta'^{-1}$ with $\delta'<1/2$, we have
\begin{equation}
P(\wt{\varphi}_{k'}\le \lambda_0-\epsilon)\le \frac{\delta'}{2}.
\end{equation}
Therefore after $M$ repetitions, we have
\begin{equation}
P\left(\min_{\ell}\wt{\varphi}^{(\ell)}_{k'}\le \lambda_0-\epsilon\right)\le \frac{M\delta'}{2}.
\end{equation}
In order to obtain $P\left(\min_{\ell}\wt{\varphi}^{(\ell)}_{k'}\le \lambda_0-\epsilon\right)<\frac{\delta}{2}$, we need to set $\delta'=M^{-1} \delta$.

On the other hand, we also would like to have bound $P\left(\min_{\ell}\wt{\varphi}^{(\ell)}_{k'}\ge \lambda_0+\epsilon\right)$.
To this end, we first note that when $\wt{\varphi}_{k'}\ge \lambda_0+\epsilon$, we have $\abs{\wt{\varphi}_{k'}- \lambda_0}_1\ge\epsilon$. 
Moreover,
\begin{equation}
\begin{split}
P(\abs{\wt{\varphi}_{k'}- \lambda_0}_1<\epsilon)=&\sum_k \sum_{\abs{\wt{\varphi}_{k'}- \lambda_0}_1<\epsilon} \abs{c_k \gamma_{k,k'}}^2\\
\ge& p_0 \sum_{\abs{\wt{\varphi}_{k'}- \lambda_0}_1<\epsilon} \abs{\gamma_{0,k'}}^2\\
=& p_0\left(1- \sum_{\abs{\wt{\varphi}_{k'}- \lambda_0}_1\ge\epsilon} \abs{\gamma_{0,k'}}^2\right)\\
\ge & p_0\left(1-\frac{1}{2T\epsilon}-\frac{1}{2(T\epsilon)^2}\right)\\
\ge & p_0(1-\delta').
\end{split}
\end{equation} 
Here we have used the normalization condition that
\begin{equation}
\sum_{k'} \abs{\gamma_{k,k'}}^2=1, \quad \forall k.
\end{equation}
Therefore 
\begin{equation}
P(\abs{\wt{\varphi}_{k'}- \lambda_0}_1\ge\epsilon)=1-P(\abs{\wt{\varphi}_{k'}- \lambda_0}_1<\epsilon)\le 1-p_0(1-\delta')\le 1-\frac{p_0}{2}.
\end{equation} 
This means that
\begin{equation}
P(\min_{\ell}\wt{\varphi}^{(\ell)}_{k'}\ge \lambda_0+\epsilon)\le P(\abs{\wt{\varphi}_{k'}- \lambda_0}_1\ge\epsilon)^M=(1-p_0/2)^M\le e^{-\frac{p_0 M}{2}}.
\end{equation}
We can then take $M=\lceil \frac{2}{p_0} \log \frac{2}{\delta} \rceil$ so that
\begin{equation}
P(\min_{\ell}\wt{\varphi}^{(\ell)}_{k'}\ge \lambda_0+\epsilon)\le \frac{\delta}{2}.
\end{equation}
To summarize, according to the relation $\delta'=M^{-1} \delta$, in order to estimate the ground state energy to precision $\epsilon=2^{-d}$ with success probability $1-\delta$, we need 
\begin{equation}
t=d+\lceil \log \delta'^{-1}\rceil=d+\Or(\log p^{-1}_0+\log\log \delta^{-1})
\end{equation}
ancilla qubits in QPE.
The circuit depth is
\begin{equation}
T=\Or\left((\epsilon \delta p_0)^{-1}\log \delta^{-1}\right).
\end{equation}
Taking the number of repetitions $M$ into account, the total cost of the method is
\begin{equation}
  \Or\left(\epsilon^{-1} \delta^{-1} p_0^{-2}(\log \delta^{-1})^2\right).
\end{equation}

\begin{rem}[Dependence on the initial overlap $p_0$]
The analysis above shows that QPE has a nontrivial dependence on the initial overlap $p_0$, which has a rather indirect source. 
First, in order to reduce the possibility of overshooting, the number of repetitions $M$ needs to be large enough and is $\Or(p_0^{-1})$. 
However, this also increases the probability of undershooting the eigenvalue, and hence $\delta'$ needs to be chosen to be $\Or(M^{-1})=\Or(p_0)$. 
This means that the circuit depth should be $\Or(\delta')=\Or(p_0^{-1})$. The total complexity is thus $TM=\Or(p_0^{-2})$. Therefore, when the initial overlap $p_0$ is small, using QPE to find the ground state energy can be very costly. 
On the other hand, using different techniques, the dependence on $p_0$ can be drastically improved to $\Or(p_0^{-\frac12})$~\cite{LinTong2020}. See also the discussions in \cref{sec:qsp_groundstate}.

\end{rem}

\begin{rem}[QPE for preparing the ground state]
The estimate of the ground state energy does not necessarily require the energy gap $\lambda_g:=\lambda_1-\lambda_0$ to be positive. 
However, if our goal is to prepare the ground state $\ket{\psi_0}$ from an initial state $\ket{\phi}$ using QPE, then we need stronger assumptions.
In particular, we cannot afford to obtain $\ket{k'}\ket{\psi_k}$, where $\abs{\wt{\varphi}_{k'}-\lambda_0}<\epsilon$ but $k\ne 0$. This at least requires $\epsilon=2^{-d}<\lambda_g$, and introduces a natural dependence on the inverse of the gap.
\end{rem} 

Through the analysis above, we see that although the analysis of QPE is very clean when 1) all eigenvalues (properly scaled to be represented as phase factors) are exactly given by $d$-bit numbers 2) the input vector is an eigenstate, the analysis can become rather complicated and tedious when such conditions are relaxed. 
Such difficulty does not merely show up at the theoretical level, but can seriously impact the robust performance of QPE in practical applications. 
To simplify the discussion of the applications below, we will be much more cavalier about the usage of QPE and assume all eigenvalues are exactly represented by $d$-bit numbers whenever necessary.
But we should keep such caveats in mind. 
Furthermore, when we move beyond QPE, the issue of having exact $d$-bit numbers will become much less severe in techniques based on quantum signal processing, i.e., quantum eigenvalue transformation (QET) and quantum singular value transformation (QSVT). 

\section{Amplitude estimation}\label{sec:amplitude_estimate}

Let \(\ket{\psi_0}\) be prepared by an oracle $U_{\psi_0}$, i.e., $U_{\psi_0}\ket{0^n}=\ket{\psi_0}$.
We have the knowledge that
\begin{equation}
\ket{\psi_0}=\sqrt{p_0} \ket{\psi_{\mathrm{good}}}+\sqrt{1-p_0} \ket{\psi_{\mathrm{bad}}},
\end{equation}
and $p_0\ll 1$. Here $\ket{\psi_{\mathrm{bad}}}$ is an orthogonal state to the desired state $\ket{\psi_{\mathrm{good}}}$, and
 $\sqrt{p_0}=\sin\frac{\theta}{2}$.
The problem of amplitude estimation is to estimate $p_0$ to $\epsilon$-precision.
If $p_0$ is away from $0$ or $1$ and is to be estimated directly from the Monte Carlo method, the number of samples needed is $\mc{N}=\Or(\epsilon^{-2})$.

Let $G=R_{\psi_0}R_{\mathrm{good}}$ be the Grover operator as in \cref{sec:amplitudeamplification}.
Then in the basis $\mc{B}=\{\ket{\psi_{\mathrm{good}}},\ket{\psi_{\mathrm{bad}}}\}$, the subspace $\mc{H}=\opr{span}{B}$ is an invariant subspace of $G$.
Recall the computation of \cref{eqn:grover_step_matrix}, the matrix representation is
\begin{equation}
[G]_{\mc{B}}=\begin{pmatrix}
\cos \theta & \sin\theta \\
-\sin\theta & \cos\theta
\end{pmatrix}.
\end{equation}
Its two eigenstates are
\begin{equation}
\ket{\psi_\pm}=\frac{1}{\sqrt{2}} \left(\ket{\psi_{\mathrm{good}}}\pm\I \ket{\psi_{\mathrm{bad}}}\right),
\end{equation}
with eigenvalues $e^{\pm \I \theta}$, respectively.

Therefore the problem of estimating $\theta$ can be solved with  phase estimation with an imperfect initial state.
Note that
\begin{equation}
\abs{\braket{\psi_0|\psi_+}}^2=\frac{1}{2}\abs{\sin\frac{\theta}{2}+\I \cos\frac{\theta}{2}}^2=\frac{1}{2}=\abs{\braket{\psi_0|\psi_{-}}}^2.
\end{equation}
Consider a QPE circuit with $t$ ancilla qubits and querying $G$ for $T=2^t$ times. 
Then each execution with the system register will be in $\ket{\psi_+}$ or $\ket{\psi_{-}}$ states each with probability $0.5$.
Let 
\begin{equation}
t=d+\lceil \log \delta^{-1}\rceil
\end{equation}
be the number of ancilla qubits with $\epsilon'=2^{-d}$.
Then QPE obtains an estimate denoted by $\wt{\theta}$, which approximates $\theta$ to precision $\epsilon'$ with success probability $1-\delta$.
Note that $p_0=\sin^2\frac{\theta}{2}$, and
\begin{equation}
\begin{split}
&\sin^2\frac{\wt{\theta}}{2}-\sin^2\frac{\theta}{2}\\
=&\sin^2\frac{\wt{\theta}-\theta}{2}\cos^2\frac{\theta}{2}+\cos^2\frac{\wt{\theta}-\theta}{2}\sin^2\frac{\theta}{2}
+2\sin\frac{\theta}{2}\cos\frac{\theta}{2}\sin\frac{\wt{\theta}-\theta}{2}\cos\frac{\wt{\theta}-\theta}{2}-\sin^2\frac{\theta}{2}
\\
=&\sin\frac{\theta}{2}\cos\frac{\theta}{2}\sin(\wt{\theta}-\theta)+
\left(1-2\sin^2\frac{\theta}{2}\right)\sin^2\frac{\wt{\theta}-\theta}{2}.
\end{split}
\end{equation}
Using the fact that $\abs{\sin (\wt{\theta}-\theta)}\le \abs{\wt{\theta}-\theta}\le \epsilon'$, we have
\begin{equation}
\abs{\wt{p}-p}\le \sqrt{p_0(1-p_0)} \epsilon'+(1-2p_0)\frac{\epsilon'^2}{4}.
\end{equation}
Let $\epsilon'$ be sufficiently small. 
Now if $p_0(1-p_0)=\Omega(1)$, we can choose $\epsilon'=\Or(\epsilon)$, and the total complexity  of QPE is $\Or(\epsilon^{-1})$. 

If $p_0$ is small, then we should estimate $p_0$ to multiplicative accuracy $\epsilon$ instead. Use
\begin{equation}
\abs{\wt{p}-p}\approx\sqrt{p_0}\epsilon'<p_0 \epsilon,
\end{equation}
we have $\epsilon'=\sqrt{p_0}\epsilon$. 
Therefore the runtime of QPE is $\Or(p_0^{-\frac12}\epsilon^{-1})$.
If $p_0$ is to be estimated to precision $\epsilon'$ using the Monte Carlo method, the number of samples would be $\mc{N}=\Or(p_0^{-1}\epsilon^{-2})$.

So, the amplitude estimation method achieves quadratic speedup in the total complexity, but the circuit depth is increased to $\Or(\epsilon'^{-1}\delta^{-1})$.

\begin{exam}[Amplitude estimation to accelerate Hadamard test]
Consider the circuit for the Hadamard test in \cref{fig:hadamard_real} to estimate $\Re\braket{\psi|U|\psi}$. Let the initial state $\ket{\psi}$ be prepared by a unitary $U_{\psi}$, then the following combined circuit 
\begin{displaymath}
\begin{quantikz}
\lstick{$\ket{0}$} & \gate{H} & \ctrl{1}  & \gate{H} &\meter{}\\
\lstick{$\ket{0^n}$} & \gate{U_{\psi}}& \gate{U}  & \qw & \qw
\end{quantikz}
\end{displaymath}
maps $\ket{0}\ket{0^n}$ to 
\begin{equation}
\ket{\psi_0}=\frac{1}{2}\ket{0}(\ket{\psi}+U\ket{\psi})+
\frac{1}{2}\ket{1}(\ket{\psi}-U\ket{\psi}):=\sqrt{p_0} \ket{\psi_{\mathrm{good}}}+\sqrt{1-p_0} \ket{\psi_{\mathrm{bad}}},
\end{equation}
and the goal is to estimate $p_0$.  This also gives the implementation of the reflector $R_{\psi_0}$.

Note that $R_{\mathrm{good}}$ can be implemented by simply reflecting against the signal qubit, i.e.,
\begin{equation}
R_{\mathrm{good}}=(I-2\ket{0}\bra{0})\otimes I_n=-Z\otimes I_n.
\end{equation}
Then we can run QPE to the Grover unitary $G=R_{\psi_0}R_{\mathrm{good}}$ to estimate $p(0)$, and the circuit depth is $\Or(\epsilon^{-1})$.
\end{exam}

\section{HHL algorithm for solving linear systems of equations}\label{sec:HHL}

In this section, we consider the solution of linear systems of equations
\begin{equation}
Ax=b,
\end{equation}
where $A\in \CC^{N\times N}$ is a non-singular  matrix.
Without loss of generality we assume $A$ is Hermitian. 
Otherwise, we can  solve a dilated Hermitian matrix, which enlarges the matrix dimension by a factor of $2$ (i.e., use one ancilla qubit)
\begin{equation}
\label{eqn:dilation_hermitian}
    \wt{A}= \begin{bmatrix} 0 & A^\dagger\\ A & 0 \end{bmatrix}
    =\ket{1}\bra{0}\otimes A+\ket{0}\bra{1}\otimes A^{\dag},
\end{equation}
and solve the enlarged problem
\begin{equation}
\wt{A}\ket{0,x}=\ket{1,b}.
\end{equation}

We assume $b\in\CC^N$ is a normalized vector and hence can be stored in a quantum state. 
More specifically, we have a unitary $U_b$ such that (may require some work registers)
\begin{equation}
\ket{b}=U_b\ket{0^n}.
\end{equation}
On classical computers, the solution is simply $x=A^{-1}b$.
However, the solution vector is in general not a normalized vector, and hence cannot be directly stored as a quantum state.
Therefore the goal of solving the quantum linear system problem (QLSP) is to find a quantum state $\ket{\wt{x}}$ so that
\begin{equation}
\norm{\ket{\wt{x}}-\ket{x}}\le \epsilon, \quad \ket{x}=\frac{A^{-1}\ket{b}}{\norm{A^{-1}\ket{b}}}.
\end{equation}
The normalization constant $\norm{A^{-1}\ket{b}}$ should be recovered separately.

One useful application of QLSP solvers is to evaluate the many-body Green's function~\cite{NegeleOrland1988}, based on the quantum many-body Hamiltonian in \cref{exam:fermion_hamiltonian}. 
We will omit the details here.

\subsection{Algorithmic description}

The HHL algorithm~\cite{HarrowHassidimLloyd2009}, based on QPE, is the first quantum algorithm for solving QLSP. 
The algorithm can be summarized as follows. Let $A$ have the following eigendecomposition
\begin{equation}
A\ket{v_j}=\lambda_j\ket{v_j}.
\end{equation}
To simplify the analysis, we assume $0< \lambda_0\le \lambda_1\le \ldots\le \lambda_{N-1}< 1$ and all eigenvalues have an exact $d$-bit representation. 
The analysis can also be generalized to the case when $A$ has negative eigenvalues, but the interpretation of the result from QPE needs to be modified accordingly.

The matrix $A$ can be queried using Hamiltonian simulation as $U=e^{\I 2\pi A}$.
Then if the input state is already one of the
eigenvectors, i.e., \(\ket{b}=\ket{v_j}\), then QPE can be applied to implement the mapping
\begin{equation}
U_{\mathrm{QPE}}\ket{0^d}\ket{v_j} =\ket{\lambda_j}\ket{v_j}.
\end{equation}
In general, let the input state
\(|b\rangle=\sum_{j} \beta_{j}\left|v_{j}\right\rangle\) be expanded
using the eigendecomposition of $A$. Then by linearity, 
\begin{equation}
U_{\mathrm{QPE}}\ket{0^d}\ket{b}=\sum_{j}\beta_j\ket{\lambda_j}\ket{v_j}.
\end{equation}
Note that the unnormalized solution satisfies
\begin{equation}
A^{-1}\ket{b}=\sum_{j}\frac{\beta_j}{\lambda_j}\ket{v_j},
\end{equation}
so all we need to do is to use the information of the eigenvalue $\ket{\lambda_j}$ stored in the ancilla register, and perform a controlled rotation  to multiply the factor $\lambda_j^{-1}$ to \emph{each} $\beta_j$.
To this end, we see that it is crucial to store all eigenvalues in the quantum computer coherently, as achieved by QPE. We would like to implement the following controlled rotation unitary (see \cref{sec:controlled_rotation})
\begin{equation}
U_{\mathrm{CR}}\ket{0}\ket{\lambda_j}=\left(\sqrt{1-\frac{C^{2}}{\wt{\lambda}_{j}^{2}}}\ket{0}+\frac{C}{\wt{\lambda}_{j}}\ket{1}\right)\ket{\lambda_j}.
\label{eqn:U_CR}
\end{equation}
where each $\wt{\lambda}_{j}$ approximates $\lambda_j$.

Finally, perform uncomputation by applying $U^{\dag}_{\mathrm{QPE}}$, we convert the information from the ancilla register $\ket{\lambda_j}$ back to $\ket{0^d}$. 
Therefore the quantum circuit for the HHL algorithm is in \cref{fig:circuit_hhl}.
\begin{figure}[H]
\begin{center}
\begin{quantikz}
\lstick{$\ket{0}$}& \qw &\gate[2]{U_{\mathrm{CR}}} & \qw &\rstick{$\sqrt{1-\frac{C^{2}}{\wt{\lambda}_{j}^{2}}}\ket{0}+\frac{C}{\wt{\lambda}_{j}}\ket{1}$}\qw\\
\lstick{$\ket{0^{d}}$}& \gate[2]{U_{\mathrm{QPE}}}& \qw & \gate[2]{U^{\dag}_{\mathrm{QPE}}}& \rstick{$\ket{0^{d}}$}\qw\\
\lstick{$\ket{v_j}$} & & \qw& &\rstick{$\ket{v_j}$}\qw\\
\end{quantikz}
\end{center}
\caption{Circuit for the HHL algorithm.}
\label{fig:circuit_hhl}
\end{figure}
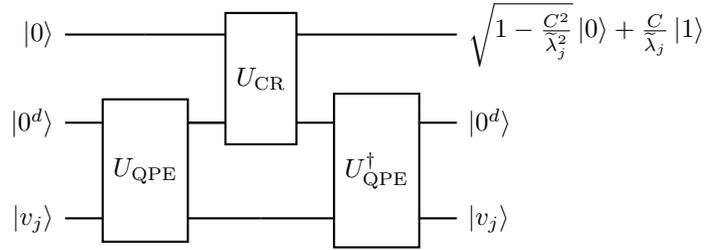

Note that through the uncomputation, the $d$ ancilla qubits for storing the eigenvalues also becomes a working register.
Discarding all working registers, the resulting unitary denoted by $U_{\mathrm{HHL}}$ satisfies
\begin{equation}
U_{\mathrm{HHL}}\ket{0}\ket{b}=\sum_{j}\left(\sqrt{1-\frac{C^{2}}{\wt{\lambda}_{j}^{2}}}\ket{0}+\frac{C}{\wt{\lambda}_{j}}\ket{1}\right)
\beta_j\ket{v_j}.
\label{eqn:hhl_output}
\end{equation}
Finally, measuring the signal qubit (the only ancilla qubit left), if the outcome is $1$, we obtain the (unnormalized) vector
\begin{equation}
\wt{x}=\sum_{j}\frac{C\beta_j}{\wt{\lambda}_{j}}\ket{v_j}
\end{equation}
stored as a normalized state in the system register is
\begin{equation}
\ket{\wt{x}}=\frac{\wt{x}}{\norm{\wt{x}}}\approx \ket{x},
\end{equation}
which is the desired approximate solution to QLSP.
In particular, the constant $C$ does not appear in the solution.

\begin{rem}[Recovering the norm of the solution]
The HHL algorithm returns the solution to QLSP in the form of a normalized state $\ket{x}$ stored in the quantum computer. In order to recover the magnitude of the unnormalized solution $\norm{\wt{x}}$, we note that the success probability of measuring the signal qubit in \cref{eqn:hhl_output} is
\begin{equation}
p(1)=\norm{\sum_{j}\frac{C\beta_j}{\wt{\lambda}_{j}}\ket{v_j}}^2=\norm{\wt{x}}^2.
\end{equation} 
Therefore sufficient repetitions of running the circuit in \cref{eqn:hhl_output} and estimate $p(1)$, we can obtain an estimate of $\norm{\wt{x}}$. 
\end{rem}
More general discussion of the readout problem of the HHL algorithm will be given in \cref{rem:readout_hhl}.

\subsection{Implementation of controlled rotation}\label{sec:controlled_rotation}

\begin{prop}[Controlled rotation given rotation angles]\label{prop:controlled_rotation}
Let \(0\le\theta<1\) has exact $d$-bit fixed point representation \(\theta=.\theta_{d-1}\cdots\theta_0\) be its \(d\)-bit fixed point representation. Then there is a $(d+1)$-qubit unitary \(U_{\theta}\) such that
\begin{equation}
U_{\theta}: |0\rangle|\theta\rangle \mapsto
(\cos (\pi\theta)|0\rangle+\sin (\pi\theta)|1\rangle)|\theta\rangle.
\end{equation}
\end{prop}

\begin{proof}

First (by e.g. Taylor expansion)
\begin{equation}
\exp \left(-\I \tau \sigma_{y}\right)=\left(\begin{array}{cc}{\cos (\tau)} & {-\sin (\tau)} \\ {\sin (\tau)} & {\cos (\tau)}\end{array}\right)=: R_y(2\tau).
\end{equation}
Here $R_y(\cdot)$ perform a single-qubit rotation  around the $y$-axis.
For any $j\in [2^d]$ with its binary representation
$j=j_{d-1}\cdots j_0$, we have
\begin{equation}
j/2^d=(.j_{d-1}\cdots j_0).
\end{equation}
So choose $\tau=\pi(.j_{d-1}\cdots j_0)$, and define
\begin{equation}
U_{\theta}=\sum_{j\in [2^d]}\exp \left(-\I \pi(.j_{d-1}\cdots j_0) \sigma_{y}\right)\otimes \ket{j}\bra{j}.
\end{equation}
Applying \(U_{\theta}\) to \(\ket{0}\ket{\theta}\) gives the
desired results. 
\end{proof}

The quantum circuit for the controlled rotation circuit is
\begin{center}
\begin{quantikz}
  \lstick{$\ket{0}$}& \gate{R_{y}(\pi)} & \gate{R_{y}(\pi/2)} &\cdots \qw  & \gate{R_{y}(\pi/2^{d-1})}&\qw\\
\lstick{$\ket{\theta_{d-1}}$}& \ctrl{-1} & \qw&\qw&\qw&\qw\\
\lstick{$\ket{\theta_{d-2}}$}& \qw& \ctrl{-2} & \qw&\qw&\qw\\
\lstick{$\cdots$}&\\
\lstick{$\ket{\theta_{0}}$}& \qw & \qw&\qw & \ctrl{-4} & \qw\\
\end{quantikz}
\end{center}
This is a sequence of single-qubit rotations on the signal qubit, each controlled by a single qubit.

In order to use the controlled rotation operation, we need to store the information of $\lambda_j$ in term of an angle $\theta_j$. 
Let $C>0$ be a lower bound to $\lambda_0$, so that $0<C/\lambda_j<1$ for all $j$.
Define 
\begin{equation}
\theta_j=\frac{1}{\pi}\arcsin (C / \lambda_j),
\end{equation}
and $\wt{\theta}_j$ be its $d'$-bit representation. 
Then 
\begin{equation}
\sin \pi \wt{\theta}_j\equiv \frac{C}{\wt{\lambda}_j}\approx \frac{C}{\lambda_j}.
\end{equation}
Again for simplicity we assume $d'$ is large enough so that the error of the fixed point representation is negligible in this step.
The mapping
\begin{equation}
U_{\mathrm{angle}}\ket{0^{d'-d}}\ket{\lambda_j}=\ket{\wt{\theta}_j}
\end{equation}
can be implemented using \emph{classical arithmetics circuits} in \cref{sec:fixedpoint}, which may require $\poly(d')$ gates and an additional working register of $\poly(d')$ qubits, which are not displayed here.
Therefore the entire controlled rotation operation needed for the HHL algorithm is given by the circuit in \cref{fig:circuit_controlrotation_HHL}.
\begin{figure}[H]
\begin{center}
\begin{quantikz}
\lstick{$\ket{0}$}& \qw &\gate[3]{U_{\theta}} & \qw &\rstick{$\cos (\pi\wt{\theta}_j)\ket{0}+\sin (\pi\wt{\theta}_j)\ket{1}$}\qw\\
\lstick{$\ket{0^{d'-d}}$}& \gate[2]{U_{\mathrm{angle}}}& \qw & \gate[2]{U^{\dag}_{\mathrm{angle}}}& \rstick{$\ket{0^{d'-d}}$}\qw\\
\lstick{$\ket{\lambda_j}$} & & & &\rstick{$\ket{\lambda_j}$}\qw\\
\end{quantikz}
\end{center}
\caption{Circuit for the controlled rotation step used by the HHL algorithm (not including additional working register for classical arithmetic operations).}
\label{fig:circuit_controlrotation_HHL}
\end{figure}
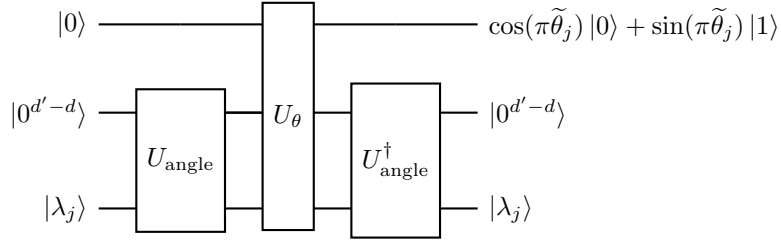
Therefore through the uncomputation $U^{\dag}_{\mathrm{angle}}$, the $d'-d$ ancilla qubits also become a working register.
Discard the working register, and we obtain a unitary $U_{\mathrm{CR}}$ satisfying
\begin{equation}
U_{\mathrm{CR}}\ket{0}\ket{\lambda_j}=
\left(\cos(\pi \wt{\theta}_j)\ket{0}+\sin(\pi \wt{\theta}_j)\ket{1}\right)\ket{\lambda_j}=
\left(\sqrt{1-\frac{C^{2}}{\wt{\lambda}_{j}^{2}}}\ket{0}+\frac{C}{\wt{\lambda}_{j}}\ket{1}\right)\ket{\lambda_j}.
\end{equation}
This is the unitary used in the HHL circuit in \cref{fig:circuit_hhl}.

\subsection{Complexity analysis of the HHL algorithm}\label{sec:analysis_hhl}

Although the choice of the constant $C$ does not appear in the normalized quantum state, it does directly affect the success probability. 
From \cref{eqn:hhl_output} we immediately obtain the success probability for measuring the signal qubit with outcome $1$ is the square of the norm of the unnormalized solution
\begin{equation}
p(1)=\norm{\wt{x}}^2\approx C^2 \norm{A^{-1}\ket{b}}^2.
\end{equation}
Therefore the success probability is determined by 
\begin{enumerate}

\item the choice of the normalization constant $C$,

\item the norm of the true solution $\norm{x}=\norm{A^{-1}\ket{b}}$.

\end{enumerate}

To maximize the success probability, $C$ should be chosen to be as large as possible (without exceeding $\lambda_0$).
So assuming the exact knowledge of $\lambda_0$, we can choose $C=\lambda_0$.
For a Hermitian positive definite matrix $A$, $\norm{A}=\lambda_{N-1}$, and $\norm{A^{-1}}=\lambda_0^{-1}$.
For simplicity, assume the largest eigenvalue of $A$ is $\lambda_{N-1}=1$.
Then the condition number of $A$ is
\begin{equation}
\kappa:=\norm{A}\norm{A^{-1}}=\frac{\lambda_{N-1}}{\lambda_0}=C^{-1}.
\end{equation}
Furthermore, 
\begin{equation}
\norm{A^{-1}\ket{b}}\ge \frac{1}{\norm{A}}\norm{\ket{b}}= 1.
\end{equation}
Therefore 
\begin{equation}
p(1)=\Omega(\kappa^{-2}).
\end{equation}
In other words, in the worst case, we need to repeatedly run the HHL algorithm for $\Or(\kappa^2)$ times to obtain the outcome $1$ in the signal qubit.

Assuming the number of system qubits $n$ is large, the circuit depth and the gate complexity of $U_{\mathrm{HHL}}$ is mainly determined by those of $U_{\mathrm{QPE}}$. 
Therefore we can measure the complexity of the HHL algorithm in terms of the number of queries to $U=e^{\I 2\pi A}$.
In order to solve QLSP to precision $\epsilon$, we need to estimate the eigenvalues to \emph{multiplicative accuracy} $\epsilon$ instead of the standard additive accuracy.

To see why this is the case, assume $\wt{\lambda}_j=\lambda_j(1+e_j)$ and $\abs{e_j}\le \frac{\epsilon}{4}\le \frac{1}{2}$. 
Then the unnormalized solution satisfies
\begin{equation}
\norm{\wt{x}-x}=\norm{\sum_{j} \beta_j\left(\frac{1}{\wt{\lambda}_j}-\frac{1}{\lambda_j}\right)\ket{v_j}}\le \norm{\sum_{j} \frac{\beta_j}{\lambda_j}\left(\frac{-e_j}{1+e_j}\right)\ket{v_j}}
\le \frac{\epsilon}{2} \norm{x}.
\end{equation}
Hence
\begin{equation}
\abs{1-\frac{\norm{\wt{x}}}{\norm{x}}}\le \frac{\epsilon}{2}.
\end{equation}
Then the normalized solution satisfies
\begin{equation}
\norm{\ket{\wt{x}}-\ket{x}}=\norm{\frac{\wt{x}}{\norm{\wt{x}}}-\frac{x}{\norm{x}}}\le \abs{1-\frac{\norm{\wt{x}}}{\norm{x}}}+\frac{\norm{\wt{x}-x}}{\norm{x}}\le \epsilon.
\end{equation}

The discussion requires the QPE to be run to additive precision $\epsilon'=\lambda_0 \epsilon=\epsilon/\kappa$. 
Therefore the query complexity of QPE is $\Or(\kappa/\epsilon)$. 
Counting the number of times needed to repeat the HHL circuit, the worst case query complexity of the HHL algorithm is $\Or(\kappa^3/\epsilon)$.

The above analysis the worst case analysis, because we assume $p(1)$ attains the lower bound $\Omega(\kappa^{-2})$.
In practical applications, the result may not be so pessimistic.
For instance, if $\beta_j$ concentrates around the smallest eigenvalues of $A$, then we may have $\norm{x}\sim \Theta(\lambda_0^{-1})=\Theta(\kappa^{-1})$. Then $p(1)=\Theta(1)$.
In such a case, we only need to repeat the HHL algorithm for a constant number of times to yield outcome $1$ in the ancilla qubit.
This does not reduce the query complexity of each run of the algorithm. 
Then in this \emph{best} case, the query complexity is $\Or(\kappa/\epsilon)$.

\subsection{Additional considerations}

Below we discuss a few more aspects of the HHL algorithm.
The first observation is that the asymptotic worst-case complexity of the HHL algorithm can be generally improved using amplitude amplification.
\begin{rem}[HHL with amplitude amplification]
We may view \cref{eqn:hhl_output} as 
\begin{equation}
U_{\mathrm{HHL}}\ket{0}\ket{b}=\sqrt{p(1)}\ket{1}\ket{\psi_{\mathrm{good}}}+\sqrt{1-p(1)}\ket{0}\ket{\psi_{\mathrm{bad}}}, \quad \ket{\psi_{\mathrm{good}}}=\ket{\wt{x}}.
\end{equation}
Since $\ket{\psi_{\mathrm{good}}}$ is marked by a single signal qubit, we may use \cref{exam:ref_signal} to construct a reflection operator with respect to the signal qubit. 
This is simply given by
\begin{equation}
R_{\mathrm{good}}=Z\otimes I_n.
\end{equation}
The reflection with respect to the initial vector is
\begin{equation}
R_{\psi_0}=U_{\psi_0}(2\ket{0^{1+n}}\bra{0^{1+n}}-I)U^{\dag}_{\psi_0},
\end{equation}
where $U_{\psi_0}=U_{\mathrm{HHL}}(I_1\otimes U_b)$.
Let $G=R_{\psi_0}R_{\mathrm{good}}$ be the Grover iterate. Then amplitude amplification allows us to apply $G$ for $\Theta(p(1)^{-\frac12})$ times to boost the success probability of obtaining $\ket{\psi_{\mathrm{good}}}$ with constant success probability. Therefore in the worst case when $p(1)=\Theta(\kappa^{-2})$, the number of repetitions is reduced to $\Or(\kappa)$, and the total runtime is $\Or(\kappa^2/\epsilon)$.
This query complexity is the commonly referred query complexity for the HHL algorithm. Note that as usual, amplitude amplification increases the circuit depth. 
However, the tradeoff is that the circuit depth \emph{increases} from $\Or(\kappa/\epsilon)$ to $\Or(\kappa^2/\epsilon)$.
\end{rem}

So far our analysis, especially that based on QPE relies on the assumption that all $\lambda_j$ all eigenvalues have an exact $d$-bit representation. From the discussion in \cref{sec:analysis_qpe}, we know that such an assumption is unrealistic and causes theoretical and practical difficulties. 
The full analysis of the HHL algorithm is thus more involved. We refer to e.g. \cite{DervovicHerbsterMountneyEtAl2018} for more details.

\begin{rem}[Comparison with classical iterative linear system solvers]
Let us now compare the cost of the HHL algorithm to that of classical iterative
algorithms.
If $A$ is $n$-qubit Hermitian positive definite with condition number $\kappa$, and is $d$-sparse (i.e., each row/column of $A$ has at most $d$ nonzero entries), then each matrix vector multiplication $Ax$ costs $\Or(dN)$ floating point operations. The number of iterations for the steepest descent (SD) algorithm is $\Or(\kappa\log \epsilon^{-1})$, and the this number can be significantly reduced to $\Or(\sqrt{\kappa}\log \epsilon^{-1})$ by the renowned conjugate gradient (CG) method.
Therefore the total cost (or wall clock time) of SD and CG is $\Or(dN\kappa\log \epsilon^{-1})$
and $\Or(dN\sqrt{\kappa}\log \epsilon^{-1})$, respectively.

On the other hand, the query complexity of the HHL algorithm, even after using the AA algorithm, is still $\Or(\kappa^2/\epsilon)$.
Such a performance is terrible in terms of both $\kappa$ and $\epsilon$.
Hence the power of the HHL algorithm, and other QLSP solvers is based on that each application of $A$ (in this case, using the unitary $U$) is \emph{much} faster.
In particular, if $U$ can be implemented with $\poly(n)$ gate complexity (also can be measured by the wall clock time), then the total gate complexity of the HHL algorithm (with AA) is $\Or(\poly(n)\kappa^2/\epsilon)$.
When $n$ is large enough, we expect that $\poly(n)\ll N=2^n$ and the HHL algorithm would eventually yield an advantage.
Nonetheless, for a realistic problem, the assumption that $U$ can be implemented with $\poly(n)$ cost, and \emph{no} classical algorithm can implement $Ax$ with $\poly(n)$ cost should be taken with a grain of salt and carefully examined.
\end{rem}

\begin{rem}[Readout problem of QLSP]
\label{rem:readout_hhl}
By solving the QLSP, the solution is stored as a quantum state in the quantum computer.
Sometimes the QLSP is only a subproblem of a larger application, so it is sufficient to treat the HHL algorithm (or other QLSP solvers) as a ``quantum subroutine'', and leave $\ket{x}$ in the quantum computer.
However, in many applications (such as the solution of Poisson's equation in \cref{sec:poisson_hhl}, the goal is to solve the lienar system.
Then the information in $\ket{x}$ must be converted to a measurable classical output.

The most common case is to compute the expectation of some observable $\braket{O}=\braket{x|O|x}\approx \braket{\wt{x}|O|\wt{x}}$.
Assuming $\braket{O}=\Theta(1)$. Then to reach additive precision $\epsilon$ of the observable, the number of samples needed is $\Or(\epsilon^{-2})$.
On the other hand, in order to reach precision $\epsilon$, the solution vector $\ket{\wt{x}}$ must be solved to precision $\epsilon$.
Assuming the worst case analysis for the HHL algorithm, the total query complexity needed is 
\begin{equation}
\Or(\kappa^2/\epsilon)\times \Or(\epsilon^{-2})=\Or(\kappa^2/\epsilon^3).
\end{equation}
\end{rem}

\begin{rem}[Query complexity lower bound]
The cost of a quantum algorithm for solving a generic QLSP scales at least as $\Omega(\kappa(A))$, where $\kappa(A):=\norm{A}\norm{A^{-1}}$ is the condition number of $A$.
The proof is based on converting the QLSP into a Hamiltonian simulation problem, and the lower bound with respect to $\kappa$ is proved via the ``no-fast-forwarding'' theorem for simulating generic Hamiltonians~\cite{HarrowHassidimLloyd2009}.
Nonetheless, for specific classes of Hamiltonians, it may still be possible to develop fast algorithms to overcome this lower bound.
\end{rem}

\section{Example: Solve Poisson's equation}\label{sec:poisson_hhl}

As an application of the HHL algorithm, let us consider a toy problem of solving the Poisson's equation in one-dimension with Dirichlet boundary conditions
\begin{equation}
-u''(r)=b(r), \quad r\in \Omega=[0,1], \quad u(0)=u(1)=0.
\end{equation}
We use the central finite difference formula to discretize the Laplacian operator on a uniform grid $r_i=(i+1)h,i\in[N]$ and $h=1/(N+1)$.
Let $u_i=u(r_i), b_i=b(r_i)$, then  Poisson's equation is discretized into a linear system of equations $Au=b$, with a tridiagonal matrix
\begin{equation}
A=\frac{1}{h^2}\begin{pmatrix}
2 & -1 & 0 & \cdots & 0 & 0 & 0\\
-1 &  2 & -1 & \cdots & 0 & 0 & 0\\
\vdots &&  &&       \vdots\\
0 & 0 & 0  &  \cdots & -1 & 2 & -1\\
0 & 0 & 0  &  \cdots & 0 & -1 & 2
\end{pmatrix}.
\label{eqn:A_tridiagonal}
\end{equation}
Our goal here is not to address the quality of the spatial discretization, but to study the cost for solving the linear system.
To this end we need to compute the condition number of $A$.
We also assume the right hand vector $b$ is already normalized so that $\ket{b}=b$.

\begin{prop}[Diagonalization of tridiagonal matrices]
A Hermitian Toeplitz tridiagonal matrix
\begin{equation}
A=\begin{pmatrix}
a & \conj{b} & 0 & \cdots & 0 & 0 & 0\\
b &  a & \conj{b} & \cdots & 0 & 0 & 0\\
\vdots &&  &&       \vdots\\
0 & 0 & 0  &  \cdots & b & a & \conj{b} \\
0 & 0 & 0  &  \cdots & 0 & b & a
\end{pmatrix}\in\CC^{N\times N}
\end{equation}
can be analytically diagonalized as
\begin{equation}
Av_k=\lambda_k v_k, \quad k=1,\ldots,N.
\end{equation}
where $(v_k)_j=v_{j,k}, j=1,\ldots,N$, $b=\abs{b}e^{\I \theta}$, and 
\begin{equation}
\lambda_k=a+2\abs{b}\cos \frac{k\pi}{N+1}, \quad v_{j,k}=\sin \frac{jk\pi}{N+1} e^{\I j\theta}.
\end{equation}
\label{prop:diag_tridiagonal}
\end{prop}
\begin{proof}
Note that formally $v_{0,k}=v_{N+1,k}=0$. Then direct matrix vector multiplication shows that for any $j=1,\ldots,N$,
\begin{equation}
(A v_k)_j=a \sin \frac{jk\pi}{N+1} e^{\I j\theta} + 2\abs{b} \sin \frac{jk\pi}{N+1}\cos \frac{k\pi}{N+1}e^{\I j\theta}=\lambda_k (v_k)_{j}.
\end{equation}
\end{proof}
 
Using \cref{prop:diag_tridiagonal} with $a=2/h^2,b=-1/h^2$, the largest eigenvalue of $A$ is $\lambda_{\max}=\norm{A}\approx 4/h^2$, and the smallest eigenvalue $\lambda_{\min}=\norm{A^{-1}}^{-1}\approx \pi^2$. So 
\begin{equation}
\kappa(A)\approx \frac{4}{h^2\pi^2}=\Or(N^2).
\end{equation}
The circuit depth of the HHL algorithm is $\Or(N^2/\epsilon)$, and the worst case query complexity (using AA) is $\Or(N^4/\epsilon)$.
So when $N$ is large, there is \emph{little benefit} in employing the quantum computer to solve this problem.

Let us now consider solving a $d$-dimensional Poisson's equation with Dirichlet boundary conditions
\begin{equation}
-\Delta u(\vr)=b(\vr), \quad \vr\in \Omega=[0,1]^d, \quad u|_{\partial \Omega}=0.
\end{equation}
The grid is the Cartesian product of the uniform grid in 1D with $N$ grid points per dimension and $h=1/(N+1)$.
The total number of grid points is $\mc{N}=N^d$.
After discretization, we obtain a linear system $\mc{A}u=b$, where
\begin{equation}
\mc{A}=A\otimes I\otimes\cdots I+\cdots+I\otimes\cdots I\otimes A,
\label{eqn:A_ddimension}
\end{equation}
where $I$ is an identity matrix of size $N$.
Since $\mc{A}$ is Hermitian and positive definite, we have $\norm{\mc{A}}\approx 4d/h^2$, and $\norm{\mc{A}^{-1}}^{-1}\approx d\pi^2$.
So $\kappa(\mc{A})\approx 4/(h^2\pi^2)\approx \kappa(A)$.
Therefore the condition number is independent of the spatial dimension $d$.

The worst case query complexity of the HHL algorithm is $\Or(N^2/\epsilon)$.
So when the number of grid points per dimension $N$ is fixed and the spatial dimension $d$ increases, and if $U=e^{\I \mc{A}\tau}$ can be implemented efficiently for some $\tau\norm{\mc{A}}<1$ with $\poly(d)$ cost, then the HHL algorithm will have an advantage over classical solvers, for which each matrix-vector multiplication scales linearly with respect to $\mc{N}$ and is therefore exponential in $d$.

\section{Solve linear differential equations*}\label{sec:linear_ode}

Consider the solution of a time-dependent linear differential equation
\begin{equation}
\begin{split}
\frac{\ud}{\ud t}x(t) &= A(t) x(t)+b(t), \quad t\in [0,T],\\
x(0)&=x_0.
\end{split}
\label{eqn:linear_ode}
\end{equation}
Here $b(t),x(t)\in \CC^d$ and $A(t)\in \CC^{d\times d}$. 
\cref{eqn:linear_ode} can be used to represent a very large class of ordinary differential equations (ODEs) and spatially discretized partial differential equations (PDEs).
For instance, if $A(t)=-\I H(t)$ for some Hermitian matrix $H(t)$ and $b(t)=0$, this is the time-dependent Schr\"odinger equation
\begin{equation}
\I \frac{\ud}{\ud t}x(t) = H(t) x(t).
\label{eqn:td_schrodinger}
\end{equation}
A special case when $H(t)\equiv H$ is often called the Hamiltonian simulation problem, and the solution can be written as
\begin{equation}
x(T)=e^{-\I HT}x(0),
\end{equation}
which can be viewed as the problem of evaluating the matrix function $e^{-\I HT}$. This  will be discussed separately in later chapters.
  
In this section we consider the general case of \cref{eqn:linear_ode}, and for simplicity discretize the equation in time using the forward Euler method with a uniform grid $t_k=k\Delta t$ where $\Delta t=T/N,k=0,\ldots,N$. Let $A_k=A(t_k), b_k=b(t_k)$.
The resulting discretized system becomes
\begin{equation}
x_{k+1}-x_k=\Delta t(A_k x_k+b_k), \quad k=1,\ldots,N,
\end{equation}
which can be rewritten as
\begin{equation}
\begin{pmatrix}
I & 0 & 0 & \cdots & 0 & 0\\
-(I+\Delta t A_1) & I & 0 & \cdots & 0 & 0\\
0 & -(I+\Delta t A_2) & I & \cdots & 0 & 0\\
\vdots & & & & &\vdots\\
0& 0 & 0 & \cdots & -(I+\Delta t A_{N-1}) & I
\end{pmatrix}
\begin{pmatrix}
x_1\\
x_2\\
x_3\\
\vdots\\
x_{N}
\end{pmatrix}
=
\begin{pmatrix}
(I+\Delta t A_0)x_0+\Delta t b_0\\
\Delta t b_1\\
\Delta t b_2\\
\vdots\\
\Delta t b_{N-1}
\end{pmatrix},
\label{eqn:qlsp_ode}
\end{equation}
or more compactly as a linear systems of equations
\begin{equation}
\mc{A} \vx=\vb.
\end{equation}
Here $I$ is the identity matrix of size $d$, and $\vx\in\RR^{Nd}$ encodes the entire history of the states.

To solve \cref{eqn:qlsp_ode} as a QLSP, the right hand side $\vb$ needs to be a normalized vector. This means we need to properly normalize $x_0,b_k$ so that
\begin{equation}
\norm{\vb}^2=\norm{(I+\Delta t A_0)x_0+\Delta t b_0}^2+\sum_{k\in[N-1]} (\Delta t)^2\norm{b_k}^2=1.
\end{equation}
In the limit $\Delta t \to 0$, this can be written as
\begin{equation}
\norm{\vb}^2=\norm{x_0}^2+\int_{0}^T \norm{b(t)}^2\ud t=1,
\end{equation}
which is not difficult to satisfy as long as $x_0,b(t)$ can be prepared efficiently using unitary circuits and $\norm{x_0},\norm{b(t)}=\Theta(1)$.

To solve \cref{eqn:qlsp_ode} using the HHL algorithm, we need to estimate the condition number of $\mc{A}$. 
Note that $\mc{A}$ is a block-bidiagonal matrix and in particular is not Hermitian. So we need to use the dilation method in \cref{eqn:dilation_hermitian} and solve the corresponding Hermitian problem.

\subsection{Scalar case}
In order to estimate the condition number, for simplicity we first assume $d=1$ (i.e., this is a scalar ODE problem), and $A(t)\equiv a\in\CC$ is a constant. Then
\begin{equation}
\mc{A}=\begin{pmatrix}
1 & 0 & 0 & \cdots & 0 & 0\\
-(1+\Delta t a) & 1 & 0 & \cdots & 0 & 0\\
0 & -(1+\Delta t a) & 1 & \cdots & 0 & 0\\
\vdots & & & & &\vdots\\
0& 0 & 0 & \cdots & -(1+\Delta t a) & 1
\end{pmatrix}\in\CC^{N\times N}.
\end{equation} 
Let $\xi=1+\Delta t a$. When $\Re a\le 0$, the absolute stability condition of the forward Euler method requires 
$\abs{1+\Delta t a}=\abs{\xi}<1$. 
In general we are interested in the regime $\Delta t \abs{a}\ll 1$, and in particular $\abs{1+\Delta t a}=\abs{\xi}<2$.

\begin{prop}
For any $A\in \CC^{N\times N}$, $\norm{A}^2,(1/\norm{A^{-1}})^2$ are given by the largest and smallest eigenvalue of $A^{\dag}A$, respectively.  
\label{prop:sv_estimate}
\end{prop}

From \cref{prop:sv_estimate}, we need to first compute
\begin{equation}
\mc{A}^{\dag}\mc{A}=\begin{pmatrix}
1+\abs{\xi}^2 &  -\conj{\xi} & \cdots & 0 & 0\\
-\xi & 1+\abs{\xi}^2 & -\conj{\xi} & \cdots & 0 & 0\\
0 & -\xi & 1+\abs{\xi}^2 & \cdots & 0 & 0\\
\vdots & & & & &\vdots\\
0& 0 & 0 & \cdots & 1+\abs{\xi}^2 & -\conj{\xi}\\
0& 0 & 0 & \cdots & -\xi & 1
\end{pmatrix}.
\end{equation}

\begin{thm}[Gershgorin circle theorem, see e.g. {\cite[Theorem 7.2.1]{GolubVan2013}}]
Let $A\in\CC^{N\times N}$ with entries $a_{ij}$.
For each $i=1,\ldots,N$, define
\begin{equation}
R_i=\sum_{j\ne i}\abs{a_{ij}}.
\end{equation}
Let $D(a_{ii},R_i)\subseteq \CC$ be a closed disc centered at $a_{ii}$ with radius $R_i$, which is called a Gershgorin disc.
Then every eigenvalue of $A$ lies within at least one of the Gershgorin discs $D(a_{ii},R_i)$
\label{thm:gershgorin_circle}
\end{thm}

Since $\mc{A}^{\dag}\mc{A}$ is Hermitian, we can restrict the Gershgorin discs to the real line so that $D(a_{ii},R_i)\subseteq \RR$. 
Then Gershgorin discs of the matrix $\mc{A}^{\dag}\mc{A}$ satisfy the bound
\begin{equation}
\begin{split}
D(a_{11},R_1)&\subseteq \left[1+\abs{\xi}^2-\abs{\xi},1+\abs{\xi}^2+\abs{\xi}\right],\\
D(a_{ii},R_i)&\subseteq \left[1+\abs{\xi}^2-2\abs{\xi},1+\abs{\xi}^2+2\abs{\xi}\right], \quad i=2,\ldots, N-1\\
D(a_{NN},R_N)&\subseteq \left[1-\abs{\xi},1+\abs{\xi}\right].
\end{split}
\end{equation}

Applying \cref{thm:gershgorin_circle}  we have
\begin{equation}
\lambda_{\max}(\mc{A}^{\dag}\mc{A}) \le (1+\abs{\xi})^2< 9.
\end{equation}
for all values of $a$ such that $\abs{\xi}<2$.

To obtain a meaningful lower bound of $\lambda_{\min}(\mc{A}^{\dag}\mc{A})$, we need $\Re a<0$ and hence $\abs{\xi}<1$. Use the inequality
\begin{equation}
\sqrt{1+x}\le 1+\frac12 x, \quad x>-1, 
\end{equation}
we have
\begin{equation}
1-\abs{\xi}= 1-\sqrt{(1+\Delta t\Re a )^2+\Delta t^2(\Im a )^2}\ge -\Delta t\Re a-\frac{(\Delta t \abs{a})^2}{2}\ge -\frac{\Delta t}{2}\Re a,
\end{equation}
when 
\begin{equation}
\Delta t\abs{a}<(-\Re a)/\abs{a}
\label{eqn:condition_dt_rea}
\end{equation}
is satisfied.
Then we have
\begin{equation}
1>1-\abs{\xi}\ge -\frac{\Delta t}{2}\Re a.
\end{equation}
Then according to \cref{thm:gershgorin_circle}, under the condition \cref{eqn:condition_dt_rea},
\begin{equation}
\lambda_{\min}(\mc{A}^{\dag}\mc{A})\ge (1-\abs{\xi})^2\ge \frac{(\Delta t \Re a)^2}{4}.
\end{equation}
Therefore
\begin{equation}
\norm{\mc{A}}=\sqrt{\lambda_{\max}(\mc{A}^{\dag}\mc{A})} \le \sqrt{1+\abs{\xi}^2+2\abs{\xi}}< 3,
\label{eqn:A_scalar_max}
\end{equation}
and
\begin{equation}
 \norm{\mc{A}^{-1}}^{-1}=\sqrt{\lambda_{\min}(\mc{A}^{\dag}\mc{A})}\ge  \frac{\Delta t (-\Re a)}{2}.
\label{eqn:A_scalar_min}
\end{equation}
As a result, when $(-\Re a)=\Theta(1)$, the condition number satisfies
\begin{equation}
\kappa(\mc{A})=\norm{\mc{A}}\norm{\mc{A}^{-1}}=\Or\left(1/\Delta t\right).
\label{eqn:kappaA_scalar_ode}
\end{equation}

In summary, for the scalar problem $d=1$ and $\Re a<0$, the query complexity of the HHL algorithm is $\Or((\Delta t)^{-2}\epsilon^{-2})$.

\begin{exam}[Growth of the condition number when $\Re a\ge0$]
The Gershgorin circle theorem does not provide a meaningful bound of the condition number when $\Re a>0$ and $1<\abs{\xi}<2$. 
This is a correct behavior.
To see this, just consider $a=1,b=0$, and the solution should grow exponentially as $x(T)=e^{T}$. 
If $\kappa(\mc{A})=\Or\left((\Delta t) ^{-1}\right)$ holds and in particular is independent of the final $T$, then the norm of the solution can only grow polynomially in $T$, which is a contradiction. See \cref{fig:ode_growth_T} for an illustration. Note that when $a=-1.0$, $\Delta t=0.1$, the condition number is less than $20$, which is smaller than the upper bound above, i.e.,
$3\times\frac{2}{\Delta t (-a)}=60$. 

\begin{figure}[H]
\begin{center}
\includegraphics[width=0.4\textwidth]{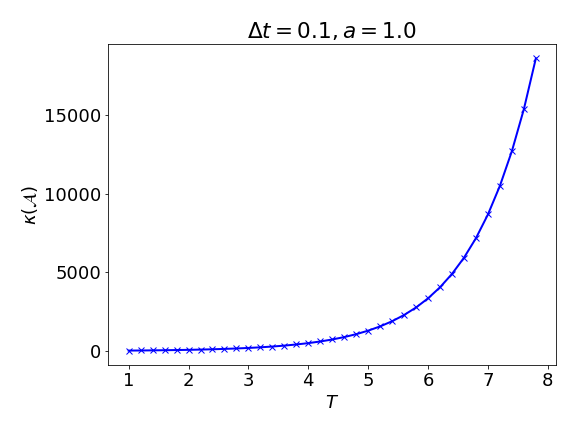}
\includegraphics[width=0.4\textwidth]{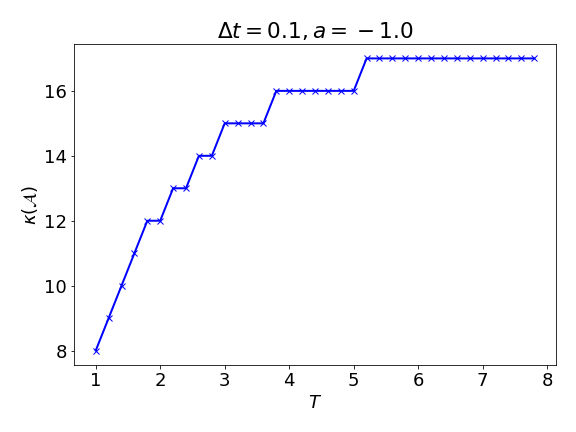}
\end{center}
\caption{Growth of the condition number $\kappa(\mc{A})$ with respect to $T$ with a fixed step size $\Delta t=0.1$ for $a=1.0$ and $a=-1.0$.}
\label{fig:ode_growth_T}
\end{figure}
\end{exam}

\begin{rem}
It may be tempting to modify the $(N,N)$-th entry of $\mc{A}$ to be $1+\abs{\xi}^2$ to obtain
\begin{equation}
\mc{G}=\begin{pmatrix}
1+\abs{\xi}^2 &  -\conj{\xi} & \cdots & 0 & 0\\
-\xi & 1+\abs{\xi}^2 & -\conj{\xi} & \cdots & 0 & 0\\
0 & -\xi & 1+\abs{\xi}^2 & \cdots & 0 & 0\\
\vdots & & & & &\vdots\\
0& 0 & 0 & \cdots & 1+\abs{\xi}^2 & -\conj{\xi}\\
0& 0 & 0 & \cdots & -\xi & 1+\abs{\xi}^2
\end{pmatrix}.
\end{equation}
Here  $\mc{G}$ is a Toeplitz tridiagonal matrix satisfying the requirement of \cref{prop:diag_tridiagonal}.
The eigenvalues of $\mc{G}$ take the form 
\begin{equation}
\lambda_k=1+\abs{\xi}^2+2\abs{\xi} \cos \frac{k\pi}{N+1}, \quad k=1,\ldots,N.
\end{equation}
If we take the approximation $\lambda_{\min}(\mc{A}^{\dag}\mc{A})\approx \lambda_{\min}(\mc{G})$, we would find that \cref{eqn:kappaA_scalar_ode} holds even when $\Re a>0$.
This behavior is however incorrect,  despite that the matrices $\mc{A}^{\dag}\mc{A}$ and $\mc{G}$ only differ by a single entry!
\end{rem}

\subsection{Vector case}
Here we consider a general $d>1$, but for simplicity assume $A(t)\equiv A\in\CC^{d\times d}$ is a constant matrix. We also assume $A$ is diagonalizable with eigenvalue decomposition
\begin{equation}
A=V\Lambda V^{-1},
\end{equation}
and $\Lambda=\diag(\lambda_1,\ldots,\lambda_N)$. We only consider the case $\Re \lambda_k< 0$ for all $k$.

\begin{prop}
For any diagonalizable $A\in\CC^{N\times N}$ with eigenvalue decomposition $Av_k=\lambda_k v_k$, we have
\begin{equation}
\norm{A^{-1}}^{-1}\le \min_k \abs{\lambda_k}\le \max_k \abs{\lambda_k}\le \norm{A}.
\end{equation}
\end{prop}
\begin{proof}
Use the Schur form $A=QTQ^{\dag}$, where $Q$ is an unitary matrix and $T$ is an upper triangular matrix (see e.g. \cite[Theorem 7.13]{GolubVan2013}).
The diagonal entries of $T$ encodes all eigenvalues of $A$, and the eigenvalues can appear in any order along the diagonal of $T$. The proposition follows by arranging the eigenvalue of $A$ with the smallest and largest absolute values to the $(N,N)$-th entry, respectively.
\end{proof}

The absolute stability condition of the forward Euler method requires $\Delta t \norm{A}<1$, and we are interested in the regime $\Delta t \norm{A}\ll 1$. Therefore $\Delta t \abs{\lambda_k}\ll 1$ for all $k$.

Let $I$ be an identity matrix of size $d$, and denote by $B=-(I+\Delta t A)$, then
\begin{equation}
\mc{A}^{\dag}\mc{A}=\begin{pmatrix}
I+B^{\dag}B &  B^{\dag} & \cdots & 0 & 0\\
B & I+B^{\dag}B & B^{\dag} & \cdots & 0 & 0\\
0 & B & I+B^{\dag}B & \cdots & 0 & 0\\
\vdots & & & & &\vdots\\
0& 0 & 0 & \cdots & I+B^{\dag}B & B^{\dag}\\
0& 0 & 0 & \cdots & B & I
\end{pmatrix}.
\end{equation}

Note that
\begin{equation}
\norm{B}\le \norm{I+B^{\dag}B}\le 1+(1+\Delta t \norm{A})^2\le 5.
\end{equation}
For any $\vx\in\CC^{Nd}$, 
\begin{equation}
\begin{split}
\norm{\mc{A}^{\dag}\mc{A}\vx}^2 &\le \norm{I+B^{\dag}B}\Big[(\norm{x_1}^2+\norm{x_2}^2)+(\norm{x_1}^2+\norm{x_2}^2+\norm{x_3}^2)+
\cdots\\
&+(\norm{x_{N-1}}^2+\norm{x_{N}}^2) \Big]\le 15\norm{\vx}^2.
\end{split}
\end{equation}
So $\lambda_{\max}(\mc{A}^{\dag}\mc{A})\le 15$, and 
\begin{equation}
\norm{\mc{A}}=\sqrt{\lambda_{\max}(\mc{A}^{\dag}\mc{A})}\le\sqrt{15}=\Or(1).
\end{equation}

To bound $\norm{\mc{A}^{-1}}$, we first note that from the eigenvalue decomposition of $A$, we have
\begin{equation}
\mc{A}=
\begin{pmatrix}
V \\
& V \\
& & \ddots \\
& & & V
\end{pmatrix}
\begin{pmatrix}
I & 0 & 0 & \cdots & 0 & 0\\
-(I+\Delta t \Lambda) & I & 0 & \cdots & 0 & 0\\
0 & -(I+\Delta t \Lambda) & I & \cdots & 0 & 0\\
\vdots & & & & &\vdots\\
0& 0 & 0 & \cdots & -(I+\Delta t \Lambda) & I
\end{pmatrix}
\begin{pmatrix}
V^{-1} \\
& V^{-1} \\
& & \ddots \\
& & & V^{-1}
\end{pmatrix}.
\end{equation}
Hence
\begin{equation}
\norm{\mc{A}^{-1}}\le \norm{V}\norm{V^{-1}}\max_{k} \norm{\mc{A}_k^{-1}}=\kappa(V) \norm{\mc{A}_k^{-1}}.
\end{equation}
Here $\kappa(V)=\norm{V}\norm{V^{-1}}$ is the condition number of the eigenvector matrix
\begin{equation}
\mc{A}_k=\begin{pmatrix}
1 & 0 & 0 & \cdots & 0 & 0\\
-(1+\Delta t \lambda_k) & 1 & 0 & \cdots & 0 & 0\\
0 & -(1+\Delta t \lambda_k) & 1 & \cdots & 0 & 0\\
\vdots & & & & &\vdots\\
0& 0 & 0 & \cdots & -(1+\Delta t \lambda_k) & 1
\end{pmatrix}.
\end{equation}
This reduces the problem to the scalar case.
Let $\xi_k=1+\Delta t \lambda_k$, and assume
\begin{equation}
\Re \lambda_k <0, \quad \Delta t\abs{\lambda_k}<(-\Re \lambda_k)/\abs{\lambda_k}, \quad \forall k.
\end{equation}
Then from \cref{eqn:A_scalar_min} we have
\begin{equation}
 \min_{k}\norm{\mc{A}^{-1}_k}^{-1}\ge \frac{\Delta t \min_{k}(-\Re \lambda_k)}{2}.
 \end{equation}
So if $\min_{k}(-\Re \lambda_k)=\Theta(1)$, we have 
\begin{equation}
\kappa(\mc{A})=\Or(\kappa(V)/\Delta t),
\end{equation}
and the cost of the HHL algorithm is $\Or((\Delta t)^{-2}\epsilon^{-2}\kappa(V))$.
Hence compared to the scalar case, the condition number of the eigenvector matrix $V$ can play an important role. 

% We may add $B^{\dag}B$ to the bottom-right block to make it a block Toeplitz tridiagonal matrix, and obtain a modified matrix
% \begin{equation}
% \mc{G}=\begin{pmatrix}
% I+B^{\dag}B &  B^{\dag} & \cdots & 0 & 0\\
% B & I+B^{\dag}B & B^{\dag} & \cdots & 0 & 0\\
% 0 & B & I+B^{\dag}B & \cdots & 0 & 0\\
% \vdots & & & & &\vdots\\
% 0& 0 & 0 & \cdots & I+B^{\dag}B & B^{\dag}\\
% 0& 0 & 0 & \cdots & B & I+B^{\dag}B
% \end{pmatrix}.
% \end{equation}
% This modification can be similarly justified as in the scalar case.
% 

\subsection{Computing observables}

The solution of \cref{eqn:qlsp_ode} means that the normalized state $\ket{\vx}$ is computed to precision $\epsilon$ and stored in the quantum computer.
In order to evaluate observables at the final time $T$, i.e., $\braket{x(T)|O|x(T)}$, we find that by the normalization condition, $\norm{x(T)}$ is on average $\Or(N^{-\frac12})=\Or((\Delta t)^{\frac12})$, and $\braket{x(T)|O|x(T)}=\Or(\Delta t)$. 
Therefore instead of reaching accuracy $\epsilon$, the Monte Carlo procedure must reach precision $\Or(\epsilon\Delta t)$. 
This increases the number of samples by another factor of $\Or((\Delta t)^{-2})$.

There is however a simple way to overcome this problem.
Instead of solving \cref{eqn:qlsp_ode}, we can redefine $\vx$ by artificially padding the vector with $N$ copies of the final state $x_N$. This can be written as
and can be reinterpreted as
\begin{equation}
\vX=\ket{0}\otimes \vx+\ket{1} \otimes \vy,
\end{equation}
with the unnormalized vector
\begin{equation}
\vx=\begin{pmatrix}
x_1\\
x_2\\
\vdots\\
x_{N}
\end{pmatrix},
\quad
\vy=\begin{pmatrix}
x_{N+1}\\
x_{N+2}\\
\vdots\\
x_{2N}
\end{pmatrix}=\begin{pmatrix}
x_{N}\\
x_{N}\\
\vdots\\
x_{N}
\end{pmatrix},
\end{equation}
and the corresponding linear systems of equation becomes
\begin{equation}
\begin{pmatrix}
I & \\
-(I+\Delta t A_1) & I \\
& & &\ddots\\
&  & & -(I+\Delta t A_{N-1}) & I\\
&&&&-I & I\\
&&&& &  -I& I \\
&&&& &  & & \ddots\\
&&&&&&&-I & I
\end{pmatrix}
\begin{pmatrix}
x_1\\
x_2\\
\vdots\\
x_{N}\\
x_{N+1}\\
x_{N+2}\\
\vdots\\
x_{2N}
\end{pmatrix}
=
\begin{pmatrix}
(I+\Delta t A_0)x_0+\Delta t b_0\\
\Delta t b_1\\
\vdots\\
\Delta t b_{N-1}\\
0\\
0\\
\vdots\\
0
\end{pmatrix}.
\label{eqn:modify_qsp_ode}
\end{equation}
Note that the solution vector only requires one ancilla qubit, after solving the equation
The condition number of this modified equation is still $\kappa=\Or((\Delta t)^{-1})$, so the total query complexity is $\Or((\Delta t)^{-2}\epsilon^{-2})$. 
By solving the modified equation \cref{eqn:modify_qsp_ode} to precision $\epsilon$, we can estimate $\braket{x(T)|O|x(T)}$ by
\begin{equation}
\braket{\vy|I\otimes O|\vy}
\end{equation}
of which the magnitude does not scale with $\Delta t$. 
Here $I$ is the identity matrix of size $N$, and $\vy$ can be obtained by measuring the ancilla qubit and obtain $1$. 
If the norm of $\norm{x(t)}$ is comparable for all $t\in[0,T]$, then the success probability will be $\Theta(1)$ after $\ket{\vX}$ is obtained.

\section{Example: Solve the heat equation*}\label{sec:heatequation}

As an application of the differential equation solver in \ref{sec:linear_ode}, let us consider a toy problem of solving the heat equation in one-dimension with Dirichlet boundary conditions
\begin{equation}
\partial_t u(r,t)=u''(r), \quad r\in \Omega=[0,1], \quad u(0,t)=u(1,t)=0.
\end{equation}
After spatial discretization using the central finite difference method with $d$ grid points, this becomes a linear ODE system
\begin{equation}
\partial_t u=-Au,
\end{equation}
where $A\in \RR^{N\times N}$ is a tridiagonal matrix given by \cref{eqn:A_tridiagonal}.
After applying the forward Euler method and discretize the simulation time $T$ into $L$ intervals with $\Delta t=T/L$, we obtain a linear system of size $NL$. The eigenvalues of $-A$, denoted by $\lambda_k$, are all negative and satisfy
\begin{equation}
-\frac{4}{h^2}\approx -\norm{A}\le \lambda_k\le -\norm{A^{-1}}^{-1}\approx -\frac{1}{\pi^2}.
\end{equation}The absolute stability condition requires $\abs{1+\Delta t \lambda_k}<1$, or $\Delta t<h^2/4=\Or(N^{-2})$, which implies $L=\Or(N^2)$.   Since $A$ is Hermitian, we have $\kappa(V)=1$. 
So for $T=\Or(1)$, the query complexity for solving the heat equation is $\Or(N^2\epsilon^{-2})$, which is the same as solving Poisson's equation.

Again, the potential advantage of the quantum solver only appears when solving the $d$-dimensional heat equation
\begin{equation}
\partial_t u(\vr,t)=\Delta u(r), \quad r\in \Omega=[0,1]^d, \quad u(\cdot,t)|_{\partial \Omega}=0.
\end{equation}
This can be written as a linear system of equations
\begin{equation}
\partial_t u=-\mc{A} u, 
\end{equation}
where $\mc{A}$ is given in \cref{eqn:A_ddimension}. 
The eigenvalues of $\mc{A}$ are all negative.
Note that $\norm{\mc{A}}=\Theta(dN^2)$, then $h=\Or(d^{-1}N^{-2})$, and the query complexity of the HHL solver is $\Or(dN^2 \epsilon^{-2})$.
This could potentially have an exponential advantage over classical solvers.

\vspace{2em}

\begin{exer}[Quantum counting] Given query access to a function $f:\{0,1\}^N \rightarrow \{0,1\}$ design a quantum algorithm that computes the size of its kernel, i.e.,, total number of $x$'s that satisfy $f(x)=1$. 
\end{exer}

\begin{exer}
Consider the initial value problem of the linear differential equation~\cref{eqn:linear_ode}. 
\begin{enumerate}
    \item Construct the linear system of equations 
    \[
    \bvec{\mathcal{A}} \vx = \vb
    \]
    like \cref{eqn:qlsp_ode} using the backward Euler method.
    \item In the scalar case when $A(t) \equiv a \in \mathbb{C}$ is a constant satisfying $\Re(a) \leq 0$, estimate the query complexity of the HHL algorithm applying to the linear system constructed in (1).
\end{enumerate}
\end{exer}

\chapter{Trotter based Hamiltonian simulation}

The Hamiltonian simulation problem with a time-independent Hamiltonian, or the Hamiltonian simulation problem for short is the following problem: given an initial state $\ket{\psi_0}$ and a Hamiltonian $H$, evaluate the quantum state at time $t$ according to $\ket{\psi(t)}=e^{-\I t H}\ket{\psi_0}$. Hamiltonian simulation is of immense importance in characterizing quantum dynamics for a diverse range of systems and situations in quantum physics, chemistry and materials science. Simulation of one quantum Hamiltonian by another quantum system was also one of the motivations of Feynman's 1982 proposal for design of quantum computers~\cite{Feynman1982}. We have also seen that Hamiltonian simulation appears as a quantum subroutine in numerous other quantum algorithms, such as QPE and its various applications. 

The Hamiltonian simulation problem can also be viewed as a linear ODE:
\begin{equation}
\partial_t \psi(t)=-\I H\psi(t), \quad \psi(0)=\psi_0.
\end{equation}
However, thanks to the unitarity of the operator $e^{-\I t H}$ for any $t$, we do not need to store the full history of the quantum states as in \cref{sec:linear_ode}, and can instead focus on the quantum state at time $t$ of interest.

Following the conceptualization of a universal quantum simulator using a Trotter decomposition of the time evolution operator $e^{-\I tH}$ \cite{Lloyd1996},
many new quantum algorithms for Hamiltonian simulation have been proposed.
%  a number of new ~\cite{BerryAhokasCleveEtAl2007,BerryChildsCleveEtAl2015,LowChuang2017,LowWiebe2019,Campbell2019}.  Detailed assessment of these algorithms, with continued improvement of theoretical error bounds, has since emerged as a very active area of research~\cite{BerryChilds2012,BerryCleveGharibian2014,BerryChildsKothari2015,ChildsMaslovNamEtAl2018,ChildsOstranderSu2019,Low2019,ChildsSu2019,ChildsSuTranEtAl2020,ChenHuangKuengEtAl2020,SahinogluSomma2020,AnFangLin2021,SuBerryWiebeEtAl2021}.
We will discuss some more advanced methods in later chapters.
This chapter focuses on the Trotter based Hamiltonian simulation method (also called the product formula).

\section{Trotter splitting}

Consider the Hamiltonian simulation problem for $H=H_1+H_2$, where $e^{-\I H_1 \Delta t}$ and $e^{-\I H_2 \Delta t}$ can be efficiently computed at least for some $\Delta t$.
In general $[H_1,H_2]\ne 0$, and the splitting of the evolution of $H_1,H_2$  needs to be implemented via the Lie product formula
\begin{equation}
 e^{-\I t H}=\lim_{L\to \infty}\left(e^{-\I \frac{t}{L} H_1}e^{-\I \frac{t}{L} H_2}\right)^L.
\end{equation}
When taking $L$ to be a finite number, and let $\Delta t=t/L$, this gives the simplest first order Trotter method with
\begin{equation}
\norm{e^{-\I \Delta t H}-e^{-\I \Delta t H_1}e^{-\I \Delta t H_2}}=\Or(\Delta t^2),
\label{eqn:trotter_coarse_error}
\end{equation}
 
Therefore to perform Hamiltonian simulation to time $t$, the error is 
\begin{equation}
\norm{e^{-\I t H}-\left(e^{-\I \frac{t}{L} H_1}e^{-\I \frac{t}{L} H_2}\right)^L}=\Or(\Delta t^2 L)=\Or\left(\frac{t^2}{L}\right).
\end{equation}
So to reach precision $\epsilon$ in the operator norm,  we need 
\begin{equation}
L=\Or(t^2 \epsilon^{-1}).
\end{equation}
This can be improved to the second order Trotter method (also called the symmetric Trotter splitting, or Strang splitting) 
\begin{equation}
\norm{e^{-\I \Delta t H}-e^{-\I \Delta t/2 H_2}e^{-\I \Delta t H_1}e^{-\I \Delta t/2 H_2}}=\Or(\Delta t^3).
\end{equation}
Following a similar analysis to the first order method, we find that to reach precision $\epsilon$ we need 
\begin{equation}
L=\Or(t^{3/2}\epsilon^{-1/2}).
\end{equation}
Higher order Trotter methods are also available, such as the $p$-th order Suzuki formula.
The local truncation error is $(\Delta t)^{p+1}$. 
Therefore to reach precision $\epsilon$, we need
\begin{equation}
L=\Or(t^{\frac{p+1}{p}}\epsilon^{-1/p}).
\end{equation}
This is often written as $L=\Or(t^{1+o(1)}\epsilon^{-o(1)})$ as $p\to \infty$.

\begin{exam}[Simulating transverse field Ising model]
For the one dimensional transverse field Ising model (TFIM) with nearest neighbor interaction in \cref{eqn:ham_tfim}, wince all Pauli-$Z_i$ operators commute, we have
\begin{equation}
e^{-\I t H_1}:=e^{\I t\sum_{i=1}^{n-1} Z_iZ_{i+1}}=\prod_{i=1}^{n-1} e^{\I t Z_i Z_{i+1}}.
\end{equation}
Each $e^{\I t Z_i Z_{i+1}}$ is a rotation involving only the qubits $i,j$, and the splitting has no error. 
Similarly
\begin{equation}
e^{-\I t H_2}:=e^{g\sum_{i} X_i}=\prod_{i} e^{\I t g X_i},
\end{equation}
and each $e^{\I t g X_i}$ can be implemented independently without error.
\end{exam}

\begin{exam}[Particle in a potential]
Let $H=-\Delta_{\vr}+V(\vr)=H_1+H_2$ be the Hamiltonian of a particle in a potential field $V(\vr)$, where $\vr\in \Omega=[0,1]^d$ with periodic boundary conditions.
After discretization using Fourier modes,  $e^{\I H_1t}$ can be efficiently performed by diagonalizing $H_1$ in the Fourier space, and $e^{\I H_2t}$ can be efficiently performed since $V(\vr)$ is diagonal in the real space.
\end{exam}

\section{Commutator type error bound}

In this section we try to refine the error bounds in \cref{eqn:trotter_coarse_error} by evaluating the preconstant explicitly.
For simplicity we only focus on the first order Trotter formula. 
The Trotter propagator $\wt{U}(t)=e^{-\I t H_1}e^{-\I t H_2}$ satisfies the equation
\begin{equation}
\begin{aligned}
\I\partial_t \wt{U}(t)&=H_1 e^{-\I t H_1}e^{-\I t H_2}+e^{-\I t H_1}H_2e^{-\I t H_2}\\
&=(H_1+H_2)e^{-\I t H_1}e^{-\I t H_2}+e^{-\I t H_1}H_2e^{-\I t H_2}-H_2 e^{-\I t H_1}e^{-\I t H_2}\\
&=H\wt{U}(t)+[e^{-\I t H_1},H_2]e^{-\I t H_2},
\end{aligned}
\end{equation}
with initial condition $\wt{U}(0)=I$. 
By Duhamel's principle, and let $U(t)=e^{-\I t H}$, we have
\begin{equation}
\wt{U}(t)=U(t)-\I\int_{0}^{t} e^{-\I H (t-s)} [e^{-\I s H_1},H_2]e^{-\I s H_2}\ud s.
\end{equation}
So we have
\begin{equation}
\norm{\wt{U}(t)-U(t)}\le \int_{0}^{t} \norm{[e^{-\I s H_1},H_2]} \ud s.
\label{eqn:U_Trotter1bound_1}
\end{equation}

Now consider $G(t)=[e^{-\I t H_1},H_2]e^{\I t H_1}=e^{-\I t H_1}H_2e^{\I t H_1}-H_2$, which satisfies $G(0)=0$ and 
\begin{equation}
\I\partial_t G(t)=e^{-\I t H_1}[H_1,H_2]e^{+\I t H_1}.
\end{equation}
Hence
\begin{equation}
\norm{[e^{-\I t H_1},H_2]}=\norm{G(t)}\le t\norm{[H_1,H_2]}.
\end{equation}
Plugging this back to \cref{eqn:U_Trotter1bound_1}, we have
\begin{equation}
\norm{\wt{U}(t)-U(t)}\le \int_0^t s\norm{[H_1,H_2]} \ud s\le \frac{t^2}{2}\norm{[H_1,H_2]}\le t^2\nu^2.
\end{equation}
In the last equality, we have used the relation $\norm{[H_1,H_2]}\le 2\nu^2$ with $\nu=\max\{\norm{H_1},\norm{H_2}\}$.
Therefore \cref{eqn:trotter_coarse_error} can be replaced by a sharper inequality
\begin{equation}
\norm{e^{-\I \Delta t H}-e^{-\I \Delta t H_1}e^{-\I \Delta t H_2}}\le \frac{\Delta t^2}{2}\norm{[H_1,H_2]}\le (\Delta t)^2\nu^2.
\label{eqn:trotter_sharp_error}
\end{equation}
Here the first inequality is called the commutator norm error estimate, and the second inequality the operator norm error estimate.

For the transverse field Ising model with nearest neighbor interaction, we have $\norm{H_1},\norm{H_2}=\Or(n)$, and hence $\nu^2=\Or(n^2)$.
On the other hand, since $[Z_iZ_j,X_k]\ne 0$ only if $k=i$ or $k=j$, the commutator bound satisfies $\norm{[H_1,H_2]}=\Or(n)$.
Therefore to reach precision $\epsilon$, the scaling of the total number of time steps $L$ with respect to the system size is $\Or(n^2/\epsilon)$ according to the estimate based on the operator norm, but is only $\Or(n/\epsilon)$ according to that based on the commutator norm.

For the particle in a potential, for simplicity consider $d=1$ and the domain $\Omega=[0,1]$ is discretized using a uniform grid of size $N$. 
For smooth and bounded potential, we have $\norm{H_1}=\Or(N^2)$, and $\norm{V}=\Or(1)$.
Therefore the operator norm bound gives $\nu^2=\Or(N^4)$. 
This is too pessimistic. Reexamining the second inequality of \cref{eqn:trotter_sharp_error} shows that in this case, the error bound should be $\Or((\Delta t)^2\nu)$ instead of $(\Delta t)^2\nu^2$. So according to the operator norm error estimate, we have $L=\Or(N^2/\epsilon)$. 
On the other hand, in the continuous space, for any smooth function $\psi(r)$, we have
\begin{equation}
[H_1,H_2]\psi=\left[-\frac{\ud^2}{\ud r^2},V\right]\psi=-V''\psi-V'\psi'.
\end{equation}
So 
\begin{equation}
\norm{[H_1,H_2]\psi}\le \norm{V''}+\norm{V'}\norm{\psi'}=\Or(N).
\end{equation}
Here we have used that $\norm{V'}=\norm{V''}=\Or(1)$, and $\norm{\psi'}=\Or(N)$ in the worst case scenario.
Therefore $\norm{[H_1,H_2]}=\Or(N)$, and we obtain a significantly improved estimate $L=\Or(N/\epsilon)$ according to the commutator norm.

The commutator scaling of the Trotter error is an important feature of the method. We refer readers to~\cite{JahnkeLubich2000} for analysis of the second order Trotter method, and~\cite{Thalhammer2008,ChildsSuTranEtAl2021} for the analysis of the commutator scaling of high order Trotter methods.

\begin{rem}[Vector norm bound]
The Hamiltonian simulation problem of interest in practice often concerns the solution with particular types of initial conditions, instead of arbitrary initial conditions.
Therefore the operator norm bound in \cref{eqn:trotter_sharp_error} can still be too loose. 
Taking the initial condition into account, we readily obtain
\begin{equation}
\norm{e^{-\I \Delta t H}\psi(0)-e^{-\I \Delta t H_1}e^{-\I \Delta t H_2}\psi(0)}\le \frac{\Delta t^2}{2}\max_{0\le s\le \Delta t}\norm{[H_1,H_2]\psi(s)}.
\end{equation}
For the example of the particle in a potential, we have
\begin{equation}
\max_{0\le s\le \Delta t}\norm{[H_1,H_2]\psi(s)}\le \norm{V''}+\norm{V'}\max_{0\le s\le \Delta t}\norm{\psi'(s)}. 
\end{equation}
Therefore if we are given the \emph{a priori} knowledge that $\max_{0\le s\le t}\norm{\psi'(s)}=\Or(1)$, we may even have $L=\Or(\epsilon^{-1})$, i.e., the number of time steps is independent of $N$.
\end{rem}

\vspace{2em}

\begin{exer}
Consider the Hamiltonian simulation problem for $H = H_1 + H_2 + H_3$. Show that the first order Trotter formula
\begin{equation*}
    \wt{U}(t) = e^{-\I t H_1} e^{-\I t H_2} e^{-\I t H_3}
\end{equation*}
has a commutator type error bound.
\end{exer}

\begin{exer}
Consider the time-dependent Hamiltonian simulation problem for the following controlled Hamiltonian
\[
H(t) = a(t) H_1 + b(t) H_2,
\]
where $a(t)$ and $b(t)$ are smooth functions bounded together with all derivatives. We focus on the following Trotter type splitting, defined as
    \begin{equation*}
    \wt{U}(t) :=   \wt{U}(t_{n}, t_{n-1}) \cdots \wt{U}(t_1, t_0), \quad
        \wt{U}(t_{j+1}, t_j) = e^{-\I \Delta t a(t_j) H_1} e^{-\I \Delta t b(t_j) H_2},
    \end{equation*}
    where the intervals $[t_j, t_{j+1}]$ are equidistant and of length $\Delta t$ on the interval $[0, t]$ with $t_n = t$. Show that this method has first-order accuracy, but does \textit{not} exhibit a commutator type error bound in general.
\end{exer}

\chapter{Block encoding}

In order to perform matrix computations, we must first address the problem of the \textit{input model}: how to get access to information in a matrix $A\in\CC^{N\times N}$ ($N=2^n$) which is generally a non-unitary matrix, into the quantum computer? 
One possible input model is given via the unitary $e^{\I \tau A}$ (if $A$ is not Hermitian, in some scenarios we can consider its Hermitian version via the dilation method). 
This is particularly useful when $e^{\I \tau A}$ can be constructed using simple circuits, e.g. Trotter splitting.

A more general input model, as will be discussed in this chapter, is called ``block encoding''.
Of course, if $A$ is a dense matrix without obvious structures, any input model will be very expensive (e.g. exponential in $n$) to implement.
Therefore a commonly assumed input model is $s$-sparse, i.e., there are at most $s$ nonzero entries in each row / column of the matrix.
Furthermore, we have an efficient procedure to get access to the location, as well as the value of the nonzero entries.
This in general can again be a difficult task given that the number of nonzero entries can still be exponential in $n$ for a sparse matrix.
Some dense matrices may also be efficiently block encoded on quantum computers.
This chapter will illustrate the block encoding procedure for via a number of detailed examples. 

\section{Query model for matrix entries}

The query model for sparse matrices is based on certain quantum oracles. In some scenarios, these quantum oracles can be implemented all the way to the elementary gate level. 

Throughout the discussion we assume $A$ is an $n$-qubit, square matrix, and
\begin{equation}
\norm{A}_{\max}:=\max_{ij} \abs{A_{ij}}< 1.
\end{equation}
If the $\norm{A}_{\max}\ge1$, we can simply consider the rescaled matrix $\wt{A}/\alpha$ for some $\alpha>\norm{A}_{\max}$.

To query the entries of a matrix, the desired oracle takes the following general form
\begin{equation}
O_A \ket{0}\ket{i}\ket{j}=\left(A_{ij}\ket{0}+\sqrt{1-\abs{A_{ij}}^2}\ket{1}\right)\ket{i}\ket{j}.
\label{eqn:entry_oracle}
\end{equation}
In other words, given $i,j\in[N]$ and a signal qubit $0$, $O_A$ performs a controlled rotation (controlling on $i,j$) of the signal qubit, which encodes the information in terms of amplitude of $\ket{0}$.

However, the classical information in $A$ is usually not stored natively
in terms of such an oracle $O_A$. Sometimes it is more natural to assume that there is an oracle
\begin{equation}
\wt{O}_A\ket{0^{d'}}\ket{i}\ket{j}=\ket{\wt{A}_{ij}}\ket{i}\ket{j},
\end{equation}
where $\wt{A}_{ij}$ is a $d'$-bit fixed point representation of $A_{ij}$,
and the value of $\wt{A}_{ij}$ is either computed on-the-fly with a quantum computer, or obtained through an external database.
In either case, the implementation of $\wt{O}_A$ may be challenging, and we will only consider the query complexity with respect to this oracle.

Using \textit{classical arithmetic operations}, we can convert this oracle into an oracle
\begin{equation}
O'_A \ket{0^d}\ket{i}\ket{j}=\ket{\wt{\theta}_{ij}}\ket{i}\ket{j},
\end{equation}
where $0\le \wt{\theta}_{ij}< 1$, and $\wt{\theta}_{ij}$ is a $d$-bit representation of $\theta_{ij}=\arccos(A_{ij})/\pi$~. This step may require some additional work registers not shown here.

Now using the controlled rotation in \cref{prop:controlled_rotation}, the information of $\wt{A}_{ij}$, $\wt{\theta}_{ij}$ has now been transferred to the phase of the signal qubit. 
We should then perform uncomputation and free the work register storing such intermediate information $\wt{A}_{ij}$, $\wt{\theta}_{ij}$. The procedure is as follows 
\begin{equation}\label{eqn:matrixentry_access}
\begin{split}
\ket{0}\underbrace{\ket{0^d}}_{\text{work register}}\ket{i}\ket{j}\xrightarrow{O'_A} &\ket{0}\ket{\wt{\theta}_{ij}}\ket{i}\ket{j}\\
\xrightarrow{\opr{CR}}& \left(A_{ij}\ket{0}+\sqrt{1-\abs{A_{ij}}^2}\ket{1}\right)\ket{\wt{\theta}_{ij}}\ket{i}\ket{j}\\
\xrightarrow{(O'_A)^{-1}}& \left(A_{ij}\ket{0}+\sqrt{1-\abs{A_{ij}}^2}\ket{1}\right)\ket{0^d}\ket{i}\ket{j}
\end{split}
\end{equation}

From now on, we will always assume that the matrix entries of $A$ can be queried using the phase oracle $O_A$ or its variants.

\section{Block encoding}

The simplest example of block encoding is the following: 
assume we can find a $(n+1)$-qubit unitary matrix $U$ (i.e., $U\in\CC^{2N\times 2N}$) such that
\begin{equation*}
U_A=\begin{pmatrix}
{A} & {*} \\
{*} & {*}
\end{pmatrix}
\end{equation*}
where $*$ means that the corresponding matrix entries are irrelevant,
then for any $n$-qubit quantum state $\ket{b}$, we can consider the state
\begin{equation}
\ket{0,b}=\ket{0}\ket{b}=\begin{pmatrix}
b\\ 0
\end{pmatrix},
\end{equation}
and
\begin{equation}
U_A\ket{0,b}=\begin{pmatrix}
Ab\\
*
\end{pmatrix}
=:\ket{0}A\ket{b}+\ket{\perp}.
\end{equation}
Here the (unnormalized) state $\ket{\perp}$ can be written as $\ket{1}\ket{\psi}$ for some (unnormalized) state $\ket{\psi}$, that is irrelevant to the computation of $A\ket{b}$.
In particular, it satisfies the orthogonality relation.
\begin{equation}
(\bra{0}\otimes I_n)\ket{\perp}=0.
\end{equation}
In order to obtain $A\ket{b}$, we need to \textit{measure} the qubit $0$ and only keep the state if it returns $0$. This can be summarized into the following quantum circuit: 

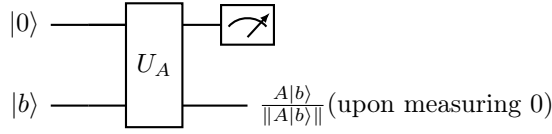
\begin{figure}[H]
\begin{center}
\begin{quantikz}
  \lstick{$\ket{0}$}&  \qw & \gate[2]{U_{A}}   \qw & \meter{} \\
  \lstick{$\ket{b}$}& \qw & &\qw \rstick{$\frac{A\ket{b}}{\norm{A\ket{b}}} (\mbox{upon measuring 0})$}\\
\end{quantikz}
\end{center}
\caption{Circuit for block encoding of $A$ using one ancilla qubit.}
\label{fig:circuit_be_onequbit}
\end{figure}
Note that the output state is normalized after the measurement takes place.
The success probability of obtaining $0$ from the measurement can be computed as
\begin{equation}
p(0)=\norm{A\ket{b}}^2=\braket{b|A^{\dag}A|b}.
\end{equation}
So the missing information of norm $\norm{A\ket{b}}$ can be recovered via the success probability $p(0)$ if needed.
We find that the success probability is only determined by $A,\ket{b}$, and is independent of other irrelevant components of $U_A$.

Note that we may not need to restrict the matrix $U_A$ to be a $(n+1)$-qubit matrix. If we can find any $(n+m)$-qubit matrix $U_A$ so that
\begin{equation}
U_A=\begin{pmatrix}
A & * & \cdots & *\\
* & * & \cdots & *\\
\vdots & & \vdots\\
* & * & \cdots & *
\end{pmatrix}
\end{equation}
Here each $*$ stands for an $n$-qubit matrix, and there are $2^m$ block rows / columns in $U_A$. The relation above can be written compactly using the braket notation as
\begin{equation}
A=\left(\langle 0^m | \otimes I_n\right) U_A \left( | 0^m \rangle \otimes I_n
\right)
\end{equation}

A necessary condition for the existence of $U_A$ is that $\norm{A}\le 1$. (Note: $\norm{A}_{\max}\le 1$ does not guarantee that $\norm{A}\le 1$, see \cref{exer:A_spec_max_norm}).
However, if we can find sufficiently large $\alpha$ and $U_A$ so that
\begin{equation}
A/\alpha=\left(\langle 0^m | \otimes I_n\right) U_A \left( | 0^m \rangle \otimes I_n
\right).
\end{equation}
Measuring the $m$ ancilla qubits and all $m$-qubits return $0$, we still obtain the normalized state $\frac{A\ket{b}}{\norm{A\ket{b}}}$.
The number $\alpha$ is hidden in the success probability:
\begin{equation}
p(0^m)=\frac{1}{\alpha^2}\norm{A\ket{b}}^2=\frac{1}{\alpha^2}\braket{b|A^{\dag}A|b}.
\end{equation}
So if $\alpha$ is chosen to be too large, the probability of obtaining all $0$'s from the measurement can be vanishingly small.

Finally, it can be difficult to find $U_A$ to block encode $A$ exactly. This is not a problem, since it is sufficient if we can find $U_A$ to block encode $A$ up to some error $\epsilon$. We are now ready to give the definition of block encoding in \cref{def:blockencode}.

\begin{defn}[Block encoding] Given an $n$-qubit matrix $A$, if we can find $\alpha, \epsilon \in \mathbb{R}_+$, and an $(m+n)$-qubit unitary matrix $U_A$ so that 
\begin{equation}
\Vert A - \alpha \left(\langle 0^m | \otimes I_n\right) U_A \left( | 0^m \rangle \otimes I_n \right) \Vert \leq \epsilon,
\end{equation}
then $U_A$ is called an $(\alpha, m, \epsilon)$-block-encoding of $A$.
When the block encoding is exact with $\epsilon=0$, $U_A$ is called an $(\alpha, m)$-block-encoding of $A$. The set of all $(\alpha, m, \epsilon)$-block-encoding of $A$ is denoted by $\BE_{\alpha,m}(A,\epsilon)$, and we define $\BE_{\alpha,m}(A)=\BE(A,0)$.
\label{def:blockencode}
\end{defn}

Assume we know each matrix element of the $n$-qubit matrix $A_{ij}$, and we are given an $(n+m)$-qubit unitary $U_A$. In order to verify that $U_A\in\BE_{1,m}(A)$, we only need to verify that
\begin{equation}
\braket{0^m,i|U_A|0^m,j}=A_{ij},
\end{equation}
and $U_A$ applied to any vector $\ket{0^m,b}$ can be obtained via the superposition principle.

Therefore we may first evaluate the state $U_A\ket{0^m,j}$, and perform inner product with $\ket{0^m,i}$ and verify the resulting the inner product is $A_{ij}$.
We will also use the following technique frequently. Assume $U_A=U_B U_C$, and then
\begin{equation}
\braket{0^m,i|U_A|0^m,j}=\braket{0^m,i|U_B U_C|0^m,j}=(U_B^{\dag}\ket{0^m,i})^{\dag}(U_C\ket{0^m,j}).
\end{equation}
So we can evaluate the states $U_B^{\dag}\ket{0^m,i},U_C\ket{0^m,j}$ independently, and then verify the inner product is $A_{ij}$.
Such a calculation amounts to running the circuit \cref{fig:circuit_be_mqubit}, and if the ancilla qubits are measured to be $0^m$, the system qubits return the normalized state $\sum_{i}A_{ij}\ket{i}/\norm{\sum_{i}A_{ij}\ket{i}}$.
\begin{figure}[H]
\begin{center}
\begin{quantikz}
  \lstick{$\ket{0^m}$}&  \qw & \gate[2]{U_{A}}   \qw & \meter{} \\
  \lstick{$\ket{j}$}& \qw & &\qw \\
\end{quantikz}
\end{center}
\caption{Circuit for general block encoding of $A$.}
\label{fig:circuit_be_mqubit}
\end{figure}
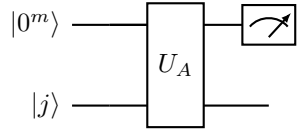

\begin{exam}[$(1,1)$-block-encoding is general]\label{exam:be11}
For any $n$-qubit matrix $A$ with $\norm{A}_2\le 1$, the singular value decomposition (SVD) of $A$ is denoted by $W\Sigma V^{\dag}$, where all singular values in the diagonal matrix $\Sigma$ belong to $[0,1]$. Then we may construct an $(n+1)$-qubit unitary matrix
\begin{equation}
\begin{split}
U_A:=&\left(\begin{array}{cc}
W & 0 \\
0 & I_{n}
\end{array}\right)\left(\begin{array}{cc}
\Sigma & \sqrt{I_{n}-\Sigma^{2}} \\
\sqrt{I_{n}-\Sigma^{2}} & -\Sigma
\end{array}\right)\left(\begin{array}{cc}
V^{\dagger} & 0 \\
0 & I_{n}
\end{array}\right)\\
=&\left(\begin{array}{cc}
A &  W \sqrt{I_{n}-\Sigma^{2}} \\
\sqrt{I_{n}-\Sigma^{2}} V^{\dagger} & -\Sigma
\end{array}\right)\\
\end{split}
\end{equation}
which is a $(1,1)$-block-encoding of $A$. 
\end{exam}

\begin{exam}[Random circuit block encoded matrix]
In some scenarios, we may want to construct a pseudo-random non-unitary matrix on quantum computers. Note that it would be highly inefficient if we first generate a dense pseudo-random matrix $A$ classically and then feed it into the quantum computer using e.g. quantum random-access memory (QRAM). Instead we would like to work with matrices that are \emph{inherently easy} to generate on quantum computers. 
This inspires the random circuit based block encoding matrix (RACBEM) model~\cite{DongLin2021}. Instead of first identifying $A$ and then finding its block encoding $U_A$, we reverse this thought process: we first identify a unitary $U_A$ that is easy to implement on a quantum computer, and then ask which matrix can be block encoded by $U_A$.

\cref{exam:be11} shows that in principle, any matrix $A$ with $\norm{A}_2\le 1$ can be accessed via a $(1,1,0)$-block-encoding. In other words, $A$ can be block encoded by an $(n+1)$-qubit random unitary $U_A$, and $U_A$ can be constructed using only basic one-qubit unitaries and CNOT gates.
The layout of the two-qubit operations can be designed to be compatible with the coupling map of the hardware.  A cartoon is shown in \cref{fig:racbem_cartoon}, and an example is given in \cref{fig:racbem_3qubit}.

\begin{figure}[htbp]
\begin{center}
\includegraphics[width=0.2\textwidth]{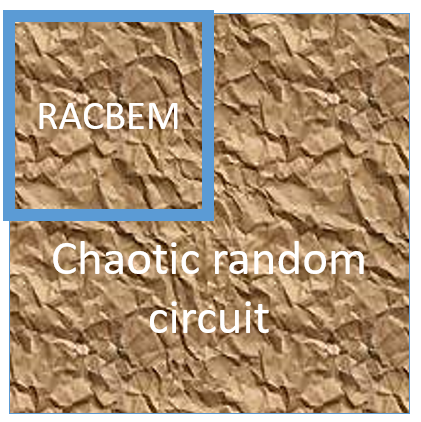}
\caption{A cartoon illustration of the RACBEM model.}
\end{center}
\label{fig:racbem_cartoon}
\end{figure}

\begin{figure}[htbp]
    \centering
    \includegraphics[width=\textwidth]{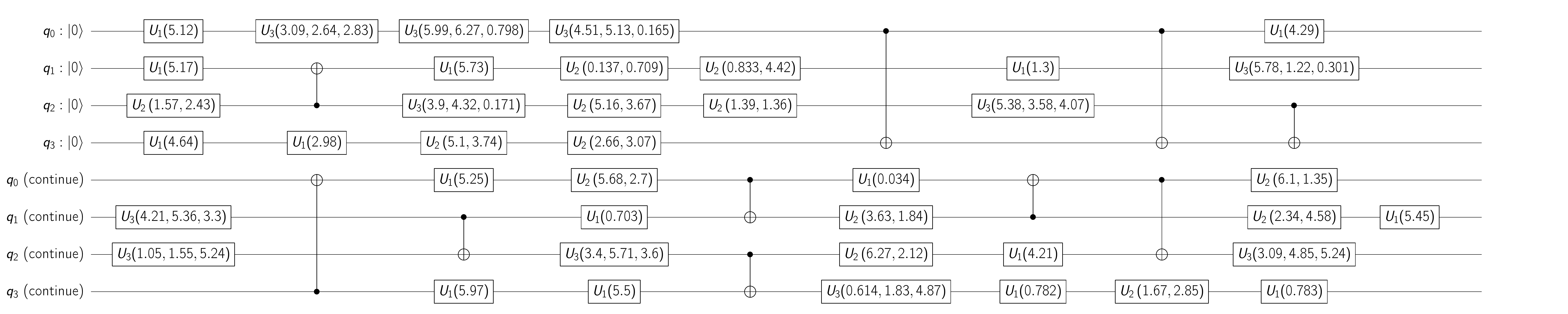}
    
        \centering 
        \begin{center}
            \scalebox{0.575}{$A = \left(\begin{array}{*{8}{r}}
0.096+0.256\I & -0.041+0.058\I & 0.096-0.224\I & 0.120+0.061\I & -0.138-0.054\I & 0.013+0.052\I & 0.189-0.099\I & 0.152-0.166\I \\ 
0.143-0.023\I & 0.001-0.335\I & 0.046-0.237\I & -0.056+0.007\I & 0.063+0.016\I & 0.079-0.063\I & 0.017+0.276\I & -0.046+0.007\I \\ 
0.054-0.017\I & -0.073+0.149\I & -0.002-0.063\I & -0.128+0.128\I & -0.371+0.048\I & -0.163-0.102\I & -0.069-0.069\I & 0.126+0.037\I \\ 
-0.043-0.208\I & -0.156-0.170\I & 0.189-0.080\I & -0.090+0.142\I & -0.057+0.075\I & 0.252+0.080\I & 0.150+0.057\I & 0.098-0.043\I \\ 
0.145+0.178\I & -0.325+0.125\I & 0.114+0.242\I & -0.136-0.316\I & 0.145+0.255\I & -0.120-0.335\I & -0.046+0.295\I & -0.142-0.184\I \\ 
-0.117+0.149\I & -0.101+0.338\I & -0.213-0.018\I & -0.474+0.081\I & -0.036-0.121\I & 0.444+0.147\I & -0.198+0.035\I & -0.091-0.054\I \\ 
-0.063+0.305\I & 0.001-0.145\I & -0.177+0.045\I & -0.209-0.150\I & -0.041+0.296\I & 0.046+0.082\I & 0.387-0.051\I & -0.430+0.233\I \\ 
-0.093-0.127\I & 0.254+0.307\I & -0.144-0.265\I & -0.048-0.353\I & 0.023+0.060\I & 0.085-0.156\I & 0.011+0.225\I & 0.249+0.420\I \\ 
\end{array}\right)$}
        \end{center}
    \caption{A RACBEM circuit constructed using the basic gate set $\{ \mathrm{U}_1, \mathrm{U}_2, \mathrm{U}_3, \mathrm{CNOT} \}$. The circuit at the bottom is a continuation of the top circuit.
$A$ is the 3-qubit matrix block encoded as the upper-left block, namely, identifying $q_0$ as the block encoding qubit.}
\label{fig:racbem_3qubit}
\end{figure}

\end{exam}

\begin{exam}[Block encoding of a diagonal matrix]
\label{exam:block_encode_diagonal}
As a special case, let us consider the block encoding of a diagonal matrix.
Since the row and column indices are the same, we may simplify the oracle \cref{eqn:entry_oracle} into
\begin{equation}
O_A \ket{0}\ket{i}=\left(A_{ii}\ket{0}+\sqrt{1-\abs{A_{ii}}^2}\ket{1}\right)\ket{i}.
\end{equation}
In the case when the oracle $\wt{O}_A$ is used, we may assume accordingly
\begin{equation}
\wt{O}_A \ket{0^{d'}}\ket{i}=\ket{A_{ii}}\ket{i}.
\end{equation}
Let $U_A=O_A$. Direct calculation shows that for any $i,j\in [N]$,
\begin{equation}
\bra{0}\bra{i}U_A\ket{0}\ket{j}=A_{ii}\delta_{ij}.
\end{equation}
This proves that $U_A\in\BE_{1,1}(A)$, i.e., $U_A$ is a $(1,1)$-block-encoding of the diagonal matrix $A$.
\end{exam}

\section{Block encoding of $s$-sparse matrices}

We now give a few examples of block encodings of more general sparse matrices.
We start from a $1$-sparse matrix, i.e., there is only one nonzero entry in each row or column of the matrix. This means that for each $j\in [N]$, there is a unique $c(j)\in[N]$ such that $A_{c(j),j}\ne 0$, and the mapping $c$ is a permutation. Then there exists a unitary $O_c$ such that
\begin{equation}
O_{c}\ket{j}=\ket{c(j)}.
\end{equation}
The implementation of $O_c$ may require the usage of some work registers that are omitted here. We also have
\begin{equation}
O_c^{\dag}\ket{c(j)}=\ket{j}.
\end{equation}

We assume the matrix entry $A_{c(j),j}$ can be queried via
\begin{equation}
O_A \ket{0}\ket{j}=\left(A_{c(j),j}\ket{0}+\sqrt{1-\abs{A_{c(j),j}}^2}\ket{1}\right)\ket{j}.
\end{equation}
Now we construct $U_A=(I\otimes O_c)O_A$, and compute
\begin{equation}
\bra{0}\bra{i}U_A\ket{0}\ket{j}=\bra{0}\bra{i}\left(A_{c(j),j}\ket{0}+\sqrt{1-\abs{A_{c(j),j}}^2}\ket{1}\right)\ket{c(j)}
=A_{c(j),j} \delta_{i,c(j)}.
\end{equation}
This proves that $U_A\in\BE_{1,1}(A)$.

For a more general $s$-sparse matrix, WLOG we assume each row and column has exactly $s$ nonzero entries (otherwise we can always treat some zero entries as nonzeros). 
For each column $j$, the row index for the $\ell$-th nonzero entry is denoted by $c(j,\ell)\equiv c_{j,\ell}$. 
For simplicity, we assume that there exists a unitary $O_c$ such that
\begin{equation}
O_c\ket{\ell}\ket{j}=\ket{\ell}\ket{c(j,\ell)}.
\label{eqn:Oc_sparsecol}
\end{equation}
Here we assume $s=2^\mf{s}$ and the first register is an $\mf{s}$-qubit register. 
A necessary condition for this query model is that $O_c$ is reversible, i.e., we can have $O_c^{\dag}\ket{\ell}\ket{c(j,\ell)}=\ket{\ell}\ket{j}$. This means that for each row index $i=c(j,\ell)$, we can recover the column index $j$ given the value of $\ell$. 
This can be satisfied e.g. 
\begin{equation}
c(j,\ell)=j+\ell-\ell_0 \pmod N,
\end{equation}
where $\ell_0$ is a fixed number. This corresponds to a banded matrix. 
This assumption is of course somewhat restrictive. We shall discuss more general query models in \cref{sec:query_general}. 

Corresponding to \cref{eqn:Oc_sparsecol}, the matrix entries can be queried via
\begin{equation}
O_{A}\ket{0}\ket{\ell}\ket{j}=\left(A_{c(j,\ell),j}\ket{0}+\sqrt{1-\abs{A_{c(j,\ell),j}}^2}\ket{1}\right)\ket{\ell}\ket{j}.
\end{equation}
In order to construct a unitary that encodes all row indices at the same time, we define $D=H^{\otimes \mf{s}}$ (sometimes called a diffusion operator, which is a term originated from Grover's search) satisfying
\begin{equation}
D\ket{0^{\mathfrak{s}}}=\frac{1}{\sqrt{s}}\sum_{\ell\in[s]} \ket{\ell}.
\end{equation}

Consider $U_A$ given by the circuit in \cref{fig:UA_s_sparse}.
The measurement means that to obtain $A\ket{b}$, the ancilla register should all return the value $0$.

\begin{figure}[H]
\begin{center}
\begin{quantikz}
\lstick{$\ket{0}$}&  \qw & \gate[3]{O_{A}}   & \qw& \qw & \meter{}\\
\lstick{$\ket{0^\mathfrak{s}}$}&  \gate{D} &    & \gate[2]{O_c}& \gate{D} & \meter{} \\
\lstick{$\ket{b}$}& \qw &  & & \qw & \qw
\end{quantikz}
\end{center}
\caption{Quantum circuit for block encoding an $s$-sparse matrix.}
\label{fig:UA_s_sparse}
\end{figure}
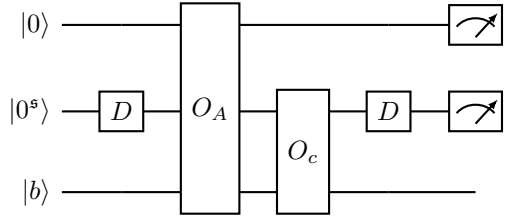

\begin{prop}
The circuit in \cref{fig:UA_s_sparse} defines $U_A\in \BE_{s,\mathfrak{s}+1}(A)$. 
\end{prop}

\begin{proof}
We call $\ket{0}\ket{0^{\mathfrak{s}}}\ket{j}$ the source state, and $\ket{0}\ket{0^{\mathfrak{s}}}\ket{i}$ the target state. In order to compute the inner product $\bra{0}\bra{0^{\mathfrak{s}}}\bra{i} U_A \ket{0}\ket{0^{\mathfrak{s}}}\ket{j}$, we apply $D,O_A,O_c$ to the source state accordingly as
\begin{equation}
\begin{split}
\ket{0}\ket{0^{\mathfrak{s}}}\ket{j}\xrightarrow{D} & \frac{1}{\sqrt{s}}\sum_{\ell\in[s]} \ket{0}\ket{\ell}\ket{j}\\
\xrightarrow{O_A} & \frac{1}{\sqrt{s}}\sum_{\ell\in[s]} \left(A_{c(j,\ell),j}\ket{0}+\sqrt{1-\abs{A_{c(j,\ell),j}}^2}\ket{1}\right)\ket{\ell}\ket{j}\\
\xrightarrow{O_c} & \frac{1}{\sqrt{s}}\sum_{\ell\in[s]} \left(A_{c(j,\ell),j}\ket{0}+\sqrt{1-\abs{A_{c(j,\ell),j}}^2}\ket{1}\right)\ket{\ell}\ket{c(j,\ell)}.
\end{split}
\end{equation}
Since we are only interested in the final state when all ancilla qubits are the $0$ state, we may apply $D$ to target state $\ket{0}\ket{0^{\mathfrak{s}}}\ket{i}$ as (note that $D$ is Hermitian)
\begin{equation}
\ket{0}\ket{0^{\mathfrak{s}}}\ket{i}\xrightarrow{D}  \frac{1}{\sqrt{s}}\sum_{\ell'\in[s]} \ket{0}\ket{\ell'}\ket{i}.
\end{equation}
Hence the inner product
\begin{equation}
\bra{0}\bra{0^{\mathfrak{s}}}\bra{i} U_A \ket{0}\ket{0^{\mathfrak{s}}}\ket{j}=\frac{1}{s}\sum_{\ell}A_{c(j,\ell),j} \delta_{i,c(j,\ell)}=\frac{1}{s}A_{ij}.
\end{equation}
\end{proof}

\section{Hermitian block encoding}\label{sec:herm_be}

So far we have considered general $s$-sparse matrices. 
Note that if $A$ is a Hermitian matrix, its $(\alpha, m, \epsilon)$-block-encoding $U_A$ does not need to be Hermitian. 
Even if $\epsilon=0$, we only have that the upper-left $n$-qubit block of $U_A$ is Hermitian. 
For instance, even the block encoding of a Hermitian, diagonal matrix in \cref{exam:block_encode_diagonal} may not be Hermitian (exercise).
On the other hand, there are indeed cases when $U_A=U_A^{\dag}$ is indeed a Hermitian matrix, and hence the definition:

\begin{defn}[Hermitian block encoding] Let $U_A$ be an $(\alpha, m, \epsilon)$-block-encoding of $A$. If $U_A$ is also Hermitian, then it is called an  $(\alpha, m, \epsilon)$-Hermitian-block-encoding of $A$. When $\epsilon=0$, it is called an $(\alpha, m)$-Hermitian-block-encoding.
The set of all $(\alpha, m, \epsilon)$-Hermitian-block-encoding of $A$ is denoted by $\HBE_{\alpha,m}(A,\epsilon)$, and we define $\HBE_{\alpha,m}(A)=\HBE(A,0)$.
\end{defn}

The Hermitian block encoding provides the simplest scenario of the qubitization process in \cref{sec:qubitize_hermbe}.

\section{Query models for general sparse matrices*}\label{sec:query_general}

If we query the oracle \eqref{eqn:Oc_sparsecol}, the assumption that for each $\ell$ the value of $c(j,\ell)$ is unique for all $j$ seems unnatural for constructing general sparse matrices.
So we consider an altnerative method for construct the block encoding of a general sparse matrix as below.

Again WLOG we assume that each row / column has at most $s=2^{\mf{s}}$ nonzero entries, and that we have access to the following two $(2n)$-qubit oracles
\begin{equation}
\begin{aligned}
O_r\ket{\ell}\ket{i}=&\ket{r(i,\ell)}\ket{i},\\
O_c\ket{\ell}\ket{j}=&\ket{c(j,\ell)}\ket{j}.
\end{aligned}
\end{equation}
Here $r(i,\ell),c(j,\ell)$ gives the $\ell$-th nonzero entry in the $i$-th row and $j$-th column, respectively.
It should be noted that although the index $\ell\in[s]$, we should expand it into an $n$-qubit state (e.g. let $\ell$ take the last $\mf{s}$ qubits of the $n$-qubit register following the binary representation of integers).
The reason for such an expansion, and that we need two oracles $O_r,O_c$ will be seen shortly.

Similar to the discussion before, we need a diffusion operator satisfying
\begin{equation}
D\ket{0^{n}}=\frac{1}{\sqrt{s}}\sum_{\ell\in[s]} \ket{\ell}.
\end{equation}
This can be implemented using Hadamard gates as 
\begin{equation}
D=I_{n-\mf{s}}\otimes H^{\otimes \mf{s}}. 
\label{eqn:diffusion_sparse}
\end{equation}

We assume that the matrix entries are queried using the following oracle using controlled rotations
\begin{equation}
O_A\ket{0}\ket{i}\ket{j}=\left(A_{ij}\ket{0}+\sqrt{1-\abs{A_{ij}}^2}\right)\ket{i}\ket{j},
\label{eqn:entry_oracle_general}
\end{equation}
where the rotation is controlled by both row and column indices.
However, if $A_{ij}=0$ for some $i,j$, the rotation can be arbitrary, as there will be no contribution due to the usage of $O_r,O_c$.

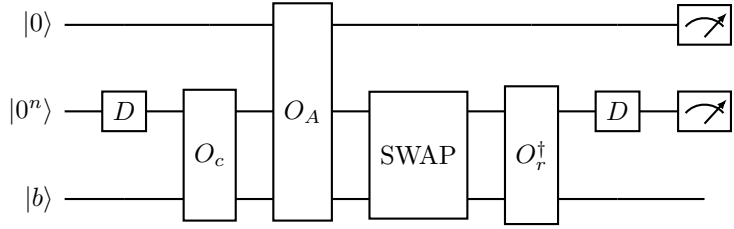
\begin{figure}

\begin{center}
\begin{quantikz}
\lstick{$\ket{0}$}&  \qw & \qw&\gate[3]{O_{A}} & \qw& \qw &\qw& \meter{}\\
\lstick{$\ket{0^n}$}&  \gate{D} & \gate[2]{O_c} & & \gate[2]{\opr{SWAP}}&\gate[2]{O_r^{\dag}}& \gate{D} & \meter{} \\
\lstick{$\ket{b}$}& \qw & &  & & & \qw & \qw\\
\end{quantikz}
\end{center}
\caption{Quantum circuit for block encoding of general sparse matrices.}
\label{fig:UA_general_sparse}
\end{figure}

\begin{prop}
\cref{fig:UA_general_sparse} defines $U_A\in\BE_{s,n+1}(A)$. 
\end{prop}

\begin{proof}
We apply the first four gate sets to the source state
\begin{equation}
\begin{split}
&\ket{0}\ket{0^n}\ket{j}\xrightarrow{D}\xrightarrow{O_c}\\
\xrightarrow{O_A}& \frac{1}{\sqrt{s}}\sum_{\ell\in[s]} \left(A_{c(j,\ell),j}\ket{0}+\sqrt{1-\abs{A_{c(j,\ell),j}}^2}\ket{1}\right)\ket{c(j,\ell)}\ket{j}\\
\xrightarrow{\opr{SWAP}}& \frac{1}{\sqrt{s}}\sum_{\ell\in[s]} \left(A_{c(j,\ell),j}\ket{0}+\sqrt{1-\abs{A_{c(j,\ell),j}}^2}\ket{1}\right)\ket{j}\ket{c(j,\ell)}.
\end{split}
\end{equation}
We then apply $D$ and $O_r$ to the target state
\begin{equation}
\ket{0}\ket{0^n}\ket{i}\xrightarrow{D}\xrightarrow{O_r}\frac{1}{\sqrt{s}}\sum_{\ell'\in[s]} \ket{0}\ket{r(i,\ell')}\ket{i}.
\end{equation}
Then the inner product gives
\begin{equation}
\begin{split}
\bra{0}\bra{0^n}\bra{i}U_A\ket{0}\ket{0^n}\ket{j}=&
\frac{1}{s}\sum_{\ell,\ell'} A_{c(j,\ell),j}\delta_{i,c(j,\ell)}\delta_{r(j,\ell'),j}\\
=&\frac{1}{s}\sum_{\ell} A_{c(j,\ell),j}\delta_{i,c(j,\ell)}=\frac{1}{s}A_{ij}.
\end{split}
\end{equation}
Here we have used that there exists a unique $\ell$ such that $i=c(j,\ell)$, and a unique $\ell'$ such that $j=r(i,\ell')$.
\end{proof}

We remark that the quantum circuit in \cref{fig:UA_general_sparse} is essentially the construction in \cite[Lemma 48]{GilyenSuLowEtAl2018}, which gives a $(s,n+3)$-block-encoding. The construction above slightly simplifies the procedure and saves two extra qubits (used to mark whether $\ell\ge s$).

Next we consider the Hermitian block encoding of a $s$-sparse Hermitian matrix.
Since $A$ is Hermitian, we only need one oracle to query the location of the nonzero entries
\begin{equation}
O_c\ket{\ell}\ket{j}=\ket{c(j,\ell)}\ket{j}.
\end{equation}
Here $c(j,\ell)$ gives the $\ell$-th nonzero entry in the $j$-th column.
It can also be interpreted as the $\ell$-th nonzero entry in the $i$-th column.
Again the first register needs to be interpreted as an $n$-qubit register.
The diffusion operator is the same as in \cref{eqn:diffusion_sparse}.

Unlike all discussions before, we introduce \textit{two} signal qubits, and a quantum state in the computational basis takes the form $\ket{a}\ket{i}\ket{b}\ket{j}$, where $a,b\in\{0,1\},i,j\in[N]$. In other words, we may view $\ket{a}\ket{i}$ as the first register, and $\ket{b}\ket{j}$ as the second register.
The $(n+1)$-qubit SWAP gate is defined as
\begin{equation}
\opr{SWAP}\ket{a}\ket{i}\ket{b}\ket{j}=\ket{b}\ket{j}\ket{a}\ket{i}.
\end{equation}
To query matrix entries, we need access to the square root of $A_{ij}$ as
(note that act on the second single-qubit register)
\begin{equation}
O_{A}\ket{i}\ket{0}\ket{j}=\ket{i}\left(\sqrt{A_{ij}}\ket{0}+\sqrt{1-\abs{A_{ij}}}\ket{1}\right)\ket{j}.
\end{equation}
The square root operation is well defined if $A_{ij}\ge 0$ for all entries.
If $A$ has negative (or complex) entries, we first write $A_{ij}=\abs{A_{ij}}e^{\I \theta_{ij}},\theta_{ij}\in[0,2\pi)$, and the square root is uniquely defined as $\sqrt{A_{ij}}=\sqrt{\abs{A_{ij}}}e^{\I \theta_{ij}/2}$.

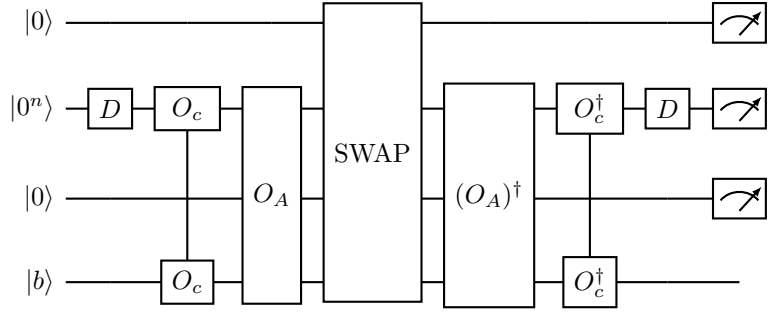
\begin{figure}

\begin{center}
\begin{quantikz}[column sep=0.3cm]
\lstick{$\ket{0}$}&  \qw &  \qw &  \qw & \gate[4]{\opr{SWAP}}&\qw&\qw&\qw&\meter{}\\
\lstick{$\ket{0^n}$}&  \gate{D} & \gate{O_c}{2}& \gate[3]{O_{A}}& & \gate[3]{(O_{A})^{\dag}}& \gate{O_c^{\dag}}{2} &\gate{D} &\meter{}\\
\lstick{$\ket{0}$}& \qw & \qw&  & & &\qw&\qw & \meter{} \\
\lstick{$\ket{b}$}& \qw &\gate{O_c}\vqw{-2} & & & &\gate{O_c^{\dag}}\vqw{-2}&\qw&\qw\\
\end{quantikz}
\end{center}
\caption{Quantum circuit for Hermitian block encoding of a general Hermitian matrix}
\label{fig:UA_hermitian_general_sparse}
\end{figure}

\begin{prop}
\cref{fig:UA_hermitian_general_sparse} defines $U_A\in\HBE_{s,n+2}(A)$.
\end{prop}

\begin{proof}
Apply the first four gate sets to the source state gives
\begin{equation}
\begin{split}
&\ket{0}\ket{0^n}\ket{0}\ket{j}\xrightarrow{D}\xrightarrow{O_c}\\
\xrightarrow{O_A}& \frac{1}{\sqrt{s}}\sum_{\ell\in[s]} 
\ket{0}\ket{c(j,\ell)}\left(\sqrt{A_{c(j,\ell),j}}\ket{0}+\sqrt{1-\abs{A_{c(j,\ell),j}}}\ket{1}\right)\ket{j}\\
\xrightarrow{\opr{SWAP}}& \frac{1}{\sqrt{s}}\sum_{\ell\in[s]} \left(\sqrt{A_{c(j,\ell),j}}\ket{0}+\sqrt{1-\abs{A_{c(j,\ell),j}}}\ket{1}\right)\ket{j}\ket{0}\ket{c(j,\ell)}
\end{split}
\end{equation}
Apply the last three gate sets to the target state
\begin{equation}
\begin{split}
&\ket{0}\ket{0^n}\ket{0}\ket{i}\xrightarrow{D}\xrightarrow{O_c}\\
\xrightarrow{O_A}& \frac{1}{\sqrt{s}}\sum_{\ell'\in[s]} 
\ket{0}\ket{c(i,\ell')}\left(\sqrt{A_{c(i,\ell'),i}}\ket{0}+\sqrt{1-\abs{A_{c(i,\ell'),i}}}\ket{1}\right)\ket{i}
\end{split}
\end{equation}
Finally, take the inner product as
\begin{equation}
\begin{split}
&\bra{0}\bra{0^n}\bra{0}\bra{i}U_A\ket{0}\ket{0^n}\ket{0}\ket{j}\\
=&\frac{1}{s}\sum_{\ell,\ell'} \sqrt{A_{c(j,\ell),j}} \sqrt{A^*_{c(i,\ell'),i}}\delta_{i,c(j,\ell)}\delta_{c(i,\ell'),j}\\
=&\frac{1}{s}\sqrt{A_{ij}A_{ji}^*}\sum_{\ell,\ell'} \delta_{i,c(j,\ell)}\delta_{c(i,\ell'),j}=\frac{1}{s}A_{ij}.
\end{split}
\end{equation}
In this equality, we have used that $A$ is Hermitian: $A_{ij}=A_{ji}^*$, and there exists a unique $\ell$ such that $i=c(j,\ell)$, as well as a unique $\ell'$ such that $j=c(i,\ell')$.
\end{proof}

The quantum circuit in \cref{fig:UA_hermitian_general_sparse} is essentially the construction in~\cite{ChildsKothariSomma2017}. The relation with quantum walks will be further discussed in \cref{sec:szegedy}.

\vspace{2em}

\begin{exer}
Construct a query oracle $O_A$ similar to that in \cref{eqn:matrixentry_access}, when $A_{ij}\in\CC$ with $\abs{A_{ij}}< 1$.
\end{exer}

\begin{exer}\label{exer:A_spec_max_norm}
  Let $A\in\CC^{N\times N}$ be a $s$-sparse matrix. Prove that $\norm{A}\le s \norm{A}_{\max}$. For every $1\le s \le N$, provide an example that the equality can reached.
\end{exer}

\begin{exer}
  Construct an $s$-sparse matrix so that the oracle in \cref{eqn:Oc_sparsecol} does not exist.
\end{exer}

\begin{exer}
Let $A\in\mathbb{C}^{N\times N}$ ($N=2^n$) be a Hermitian matrix with entries on the complex unit circle $A_{ij}=z_{ij}$, $|z_{ij}|=1$.
\begin{enumerate}
    \item
    Construct a $2n$ qubit block-diagonal unitary $V\in\mathbb C^{N^2\times N^2}$ such that 
\[
V\ket{0}\ket{j}=\frac1{\sqrt N}\sum_{i\in[N]}\sqrt{\bar z_{ij}}\ket{i}\ket{j}, \quad j\in [N].
\]
Here, block-diagonal means $(\bra{x}\otimes I)V(\ket{y}\otimes I)=0^{N\times N}$ for $x\ne y$.
    \item
    Draw a circuit which uses $V$ to implement a block encoding $U$ of $A$ with $n$ ancilla qubits . What is the prefactor $\alpha$ for the block encoding?
    \item
    Give an explicit expression for the entries of the block encoding $U$.
    \end{enumerate}
\end{exer}

\chapter{Matrix functions of Hermitian matrices}\label{chap:hermfunc}

Let $A$ be an $n$-qubit Hermitian matrix.
Then $A$ has the eigenvalue decomposition 
\begin{equation}
A=V\Lambda V^{\dag}.
\label{eqn:eig_decomposition}
\end{equation}
Here $\Lambda=\diag(\{\lambda_i\})$ is a diagonal matrix, and $\lambda_0\le \cdots\le \lambda_{N-1}$. 
Let the scalar function $f$ be well defined on all $\lambda_i$'s.
Then the matrix function $f(A)$ can be defined in terms of the eigendecomposition:

\begin{defn}[Matrix function of Hermitian matrices]
\label{def:matrix_function}
 Let $A\in \CC^{N\times N}$ be a Hermitian matrix with eigenvalue decomposition \cref{eqn:eig_decomposition}. Let $f: \mathbb{R} \rightarrow \mathbb{\CC}$ be a scalar function such that $f\left(\lambda_{i}\right)$ is defined for all $i\in[N]$. The  matrix function is defined as
\begin{equation}
f(A):=Vf(\Lambda)V^{\dag},
\label{eqn:matrix_function}
\end{equation}
where 
\begin{equation}
f\left(\Lambda\right)=\operatorname{diag}\left(f\left(\lambda_{0}\right), f\left(\lambda_{1}\right), \ldots, f\left(\lambda_{N-1}\right)\right).
\end{equation}
\end{defn}

This chapter introduces techniques to construct an efficient quantum circuit to compute $f(A)\ket{b}$ for any state $\ket{b}$. 
Throughout the discussion we assume $A$ is queried in the block encoding model denoted by $U_A$. 
For simplicity we assume that there is no error in the block encoding, i.e.,
$U_A\in\BE_{\alpha,m}(A)$, and WLOG we can take $\alpha=1$.

Many tasks in scientific computation can be expressed in terms of matrix functions. Here are a few examples:

\begin{itemize}

\item Hamiltonian simulation: $f(A)=e^{\I At}$.

\item Gibbs state preparation $f(A)=e^{-\beta A}$.

\item Solving linear systems of equation $f(A)=A^{-1}$.

\item Eigenstate filtering $f(A)=\mathbbm{1}_{(-\infty,0)}(A-\mu I)$.

\end{itemize}

A key technique for representing matrix functions is called the qubitization.

\section{Qubitization of Hermitian matrices with Hermitian block encoding}\label{sec:qubitize_hermbe}

We first introduce some heuristic idea behind qubitization. 
For any $-1< \lambda\le 1$, we can consider a $2\times 2$ rotation matrix, \begin{equation}
O(\lambda)=\begin{pmatrix}
\lambda & -\sqrt{1-\lambda^2}\\
\sqrt{1-\lambda^2} & \lambda
\end{pmatrix}
=\begin{pmatrix}
\cos\theta & -\sin\theta\\
\sin\theta & \cos\theta
\end{pmatrix}.
\label{eqn:rotation_matrix}
\end{equation}
where we have performed the change of variable $\lambda=\cos\theta$ with $0\le \theta<\pi$.

Now direct computation shows
\begin{equation}
O^k(\lambda)=\begin{pmatrix}
\cos(k\theta) & -\sin(k\theta)\\
\sin(k\theta) & \cos(k\theta)
\end{pmatrix}.
\end{equation}
Using the definition of Chebyshev polynomials (of first and second kinds, respectively)
\begin{equation}
T_k(\lambda)=\cos(k\theta)=\cos(k\arccos \lambda), \quad U_{k-1}(\lambda)=\frac{\sin(k\theta)}{\sin\theta}=
\frac{\sin(k \arccos \lambda)}{\sqrt{1-\lambda^2}},
\end{equation}
we have
\begin{equation}
O^k(\lambda)=\begin{pmatrix}
T_k(\lambda) & -\sqrt{1-\lambda^2}U_{k-1}(\lambda)\\
\sqrt{1-\lambda^2}U_{k-1}(\lambda) & T_k(\lambda)
\end{pmatrix}.
\end{equation}
Note that if we can somehow replace $\lambda$ by $A$, we immediately obtain a $(1,1)$-block-encoding for the Chebyshev polynomial $T_k(A)$! 
This is precisely what qubitization aims at achieving, though there are some small twists.

In the simplest scenario, we assume that $U_A\in\HBE_{1,m}(A)$.
Start from the spectral decomposition
\begin{equation}
A=\sum_{i} \lambda_i \ket{v_i}\bra{v_i},
\label{eqn:eigdecompose_herm}
\end{equation}
we have that for each eigenstate $\ket{v_i}$,
\begin{equation}
U_A\ket{0^m}\ket{v_i}=\ket{0^m}A\ket{v_i}+\ket{\wt{\perp}_i}=\lambda_i\ket{0^m}\ket{v_i}+\ket{\wt{\perp}_i}.
\label{eqn:UA_apply_vi}
\end{equation}
Here $\ket{\wt{\perp}_i}$ is an unnormalized state that is orthogonal to all states of the form $\ket{0^m}\ket{\psi}$, i.e.,
\begin{equation}
\Pi\ket{\wt{\perp}_i}=0.
\end{equation}
where
\begin{equation}
\Pi=\ket{0^m}\bra{0^m}\otimes I_n
\label{eqn:projection}
\end{equation}
is a projection operator.

Since the right hand side of \cref{eqn:UA_apply_vi} is a normalized state,
we may also write
\begin{equation}
\ket{\wt{\perp}_i}=\sqrt{1-\lambda_i^2}\ket{\perp_i},
\end{equation}
where $\ket{\perp_i}$ is a normalized state.

Now if $\lambda_i=\pm 1$, then $\mc{H}_i=\span{\ket{0^m}\ket{v_i}}$ is already an invariant subspace of $U_A$, and $\ket{\perp_i}$ can be any state.
Otherwise, use the fact that $U_A=U_A^{\dag}$, we can apply $U_A$ again to both sides of \cref{eqn:UA_apply_vi} and obtain
\begin{equation}
U_A\ket{\perp_i}=\sqrt{1-\lambda_i^2}\ket{0^m}\ket{v_i}-\lambda_i\ket{\perp_i}.
\label{eqn:UA_apply_perpi}
\end{equation}
Therefore $\mc{H}_i=\span{\ket{0^m}\ket{v_i},\ket{\perp_i}}$ is an invariant subspace of $U_A$. 
Furthermore, the matrix representation of $U_A$ with respect to the basis $\mc{B}_i=\{\ket{0^m}\ket{v_i},\ket{\perp_i}\}$ is
\begin{equation}
[U_A]_{\mc{B}_i}=\begin{pmatrix}
\lambda_i & \sqrt{1-\lambda_i^2}\\
\sqrt{1-\lambda_i^2} & -\lambda_i
\end{pmatrix},
\end{equation}
i.e., $U_A$ restricted to $\mc{H}_i$ is a reflection operator.
This also leads to the name ``qubitization'', which means that each eigenvector $\ket{v_i}$ is ``qubitized'' into a two-dimensional space $\mc{H}_i$.

In order to construct a block encoding for $T_k(A)$, we need to turn $U_A$ into a rotation. 
For this note that $\mc{H}_i$ is also an invariant subspace for the projection operator $\Pi$:
\begin{equation}
[\Pi]_{\mc{B}_i}=\begin{pmatrix}
1 & 0 \\
0 & 0
\end{pmatrix}.
\end{equation}
Similarly define $Z_{\Pi}=2\Pi-1$, since
\begin{equation}
[Z_{\Pi}]_{\mc{B}_i}=\begin{pmatrix}
1 & 0 \\
0 & -1
\end{pmatrix},
\end{equation}
$Z_{\Pi}$ acts as a reflection operator restricted to each subspace $\mc{H}_i$.
Then $\mc{H}_i$ is the invariant subspace for the \emph{iterate}
\begin{equation}
O=U_A Z_{\Pi}
\end{equation}
and
\begin{equation}
[O]_{\mc{B}_i}=\begin{pmatrix}
\lambda_i & -\sqrt{1-\lambda_i^2}\\
\sqrt{1-\lambda_i^2} & \lambda_i
\end{pmatrix}
\end{equation}
is the desired rotation matrix.
Therefore
\begin{equation}
[O^k]_{\mc{B}_i}=[(U_A Z_{\Pi})^k]_{\mc{B}_i}=\begin{pmatrix}
T_k(\lambda_i) & -\sqrt{1-\lambda_i^2}U_{k-1}(\lambda_i)\\
\sqrt{1-\lambda_i^2}U_{k-1}(\lambda_i) & T_k(\lambda_i)
\end{pmatrix}.
\end{equation}
Since $\{\ket{0^m}\ket{v_i}\}$ spans the range of $\Pi$, we have
\begin{equation}
O^k=\begin{pmatrix}
T_k(A) & *\\
* & *
\end{pmatrix}
\end{equation}
i.e., $O^k=(U_A Z_{\Pi})^k$ is a $(1,m)$-block-encoding of the Chebyshev polynomial $T_k(A)$.

In order to implement $Z_{\Pi}$, note that if $m=1$, then $Z_{\Pi}$ is just the Pauli $Z$ gate. When $m>1$, the circuit 
\begin{displaymath}
\begin{quantikz}
\lstick{$\ket{1}$}  & \targ{} & \gate{Z} & \targ{}    & \qw \\
\lstick{$\ket{b}$}& \octrl{-1} & \qw   & \octrl{-1} & \qw
\end{quantikz}
\end{displaymath}
returns $\ket{1}\ket{0^m}$ if $b=0^m$, and $-\ket{1}\ket{b}$ if $b\ne 0^m$.
So this precisely implements $Z_{\Pi}$ where the signal qubit $\ket{1}$ is used as a work register.
We may also discard the signal qubit, and resulting unitary is denoted by $Z_{\Pi}$.

In other words, the circuit in \cref{fig:circuit_qubitization_hbe} implements the operator $O$. 
Repeating the circuit $k$ times gives the $(1,m+1)$-block-encoding of $T_k(A)$.

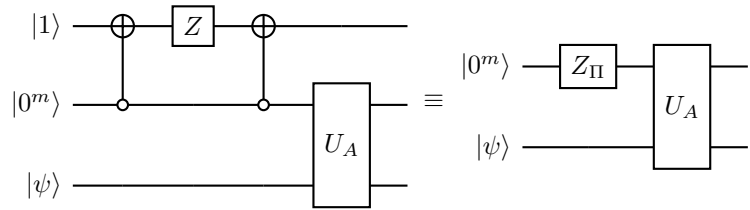
\begin{figure}[H]
\begin{displaymath}
\begin{quantikz}
\lstick{$\ket{1}$}  & \targ{} & \gate{Z} & \targ{}    & \qw&\qw \\
\lstick{$\ket{0^m}$}& \octrl{-1} & \qw   & \octrl{-1} &\gate[2]{U_A} &\qw\\
\lstick{$\ket{\psi}$} & \qw & \qw & \qw & &\qw
\end{quantikz}
\equiv\begin{quantikz}
\lstick{$\ket{0^m}$}& \gate{Z_{\Pi}}  &\gate[2]{U_A} &\qw\\
\lstick{$\ket{\psi}$} & \qw & & \qw 
\end{quantikz}
\end{displaymath}
\caption{Circuit implementing one step of qubitization with a Hermitian block encoding of a Hermitian matrix. Here $U_A\in \HBE_{1,m}(A)$.}
\label{fig:circuit_qubitization_hbe}
\end{figure}

\begin{rem}[Alternative perspectives of qubitization]
The fact that an arbitrarily large block encoding matrix $U_A$ can be partially block diagonalized into $N$ subblocks of size $2\times 2$ may seem a rather peculiar algebraic structure. 
In fact there are other alternative perspectives and derivations of the qubitization result. 
Some noticeable ones include the use of Jordan's Lemma, and the use of the cosine-sine (CS) decomposition. 
Throughout this chapter and the next chapter, we will adopt the more ``elementary'' derivations used above.
\end{rem}

\section{Application: Szegedy's quantum walk*}\label{sec:szegedy}

Quantum walk is one of the major topics in quantum algorithms. 
Roughly speaking, there are two versions of quantum walks. 
The continuous time quantum walk is the closely analogous to its classical counterpart, i.e., continuous time random walk. 
The other version, the discrete time quantum walk, or Szegedy's quantum walk~\cite{Szegedy2004} is not so obviously connected to the classical random walks.
We will not introduce the motivations behind the continuous and discrete time random walks, and refer readers to~\cite[Chapter 16,17]{ChildsQuantumLec} for detailed discussions. 

\subsection{Basics of Markov chain}

Let $G=(V,E)$ be a graph of size $N$. A Markov chain (or a random walk) is given by a transition matrix $P$, with its entry $P_{ij}$ denoting the probability of the transition from vertex $i$ to vertex $j$. 
The matrix $P$ is a stochastic matrix satisfying
\begin{equation}
P_{ij}\ge 0, \quad \sum_{j} P_{ij}=1.
\end{equation}
Let $\pi$ be the stationary state, which is a left eigenvector of $P$ with  eigenvalue $1$:
\begin{equation}
\sum_i\pi_i P_{ij}=\pi_j, \quad \pi_i\ge 0, \quad \sum_{i}\pi_i =1.
\end{equation}
A Markov chain is \emph{irreducible} if any state can be reached from any other state in a finite number of steps. An irreducible Markov chain is \emph{aperiodic} if there exists no
integer greater than one that divides the length of every directed cycle of the graph. A Markov chain is \emph{ergodic} if it is both irreducible and aperiodic. By the Perron--Frobenius Theorem, any ergodic Markov chain $P$ has a unique stationary state $\pi$, and $\pi_i>0$ for all $i$. A Markov chain is \emph{reversible} if the following detailed balance condition is satisfied
\begin{equation}
\pi_i P_{ij}=\pi_j P_{ji}.
\end{equation}

Now we define the \emph{discriminant matrix} associated with a Markov chain as
\begin{equation}
\label{eqn:discriminant_matrix}
D_{ij}=\sqrt{P_{ij}P_{ji}},
\end{equation}
which is real symmetric and hence Hermitian. For a reversible Markov chain, the stationary state can be encoded as an eigenvector of $D$ (the proof is left as an exercise).
\begin{prop}[Reversible Markov chain]
If a Markov chain is reversible, then the coherent version of the stationary state
\begin{equation}
\ket{\pi}=\sum_{i}\sqrt{\pi_i}\ket{i}
\end{equation}
is a normalized eigenvector of the discriminant matrix $D$ satisfying
\begin{equation}
D\ket{\pi}=\ket{\pi}.
\end{equation}
Furthermore, when $\pi_i>0$ for all $i$, we have
\begin{equation}
D=\diag(\sqrt{\pi}) P \diag(\sqrt{\pi})^{-1}.
\end{equation}
Therefore the set of (left) eigenvalues of $P$ and the set of the eigenvalues of $D$ are the same.
\label{prop:discriminant_reversiblewalk}
\end{prop}

\subsection{Block encoding of the discriminant matrix}
Our first goal is to construct a Hermitian block encoding of $D$.
Assume that we have access to an oracle $O_P$ satisfying
\begin{equation}
O_P\ket{0^n}\ket{j}=\sum_{k}\sqrt{P_{jk}}\ket{k}\ket{j}.
\end{equation}
Thanks to the stochasticity of $P$, the right hand side is already a normalized vector, and no additional signal qubit is needed.

We also introduce the $n$-qubit SWAP operator
Swap operator:
\begin{equation}
\opr{SWAP}\ket{i}\ket{j}=\ket{j}\ket{i},
\end{equation}
which swaps the value of the two registers in the computational basis, and can be directly implemented using $n$ two-qubit SWAP gates.

We claim that the following circuit gives $U_D\in \HBE_{1,n}(D)$.
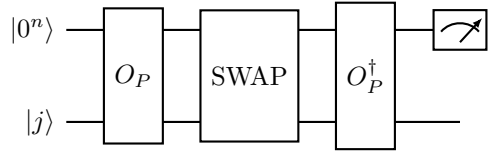
\begin{figure}[H]
\begin{center}
\begin{quantikz}
\lstick{$\ket{0^n}$}&  \gate[2]{O_P} & \gate[2]{\opr{SWAP}}&\gate[2]{O_{P}^{\dag}} & \meter{} \\
\lstick{$\ket{j}$}&  & & &\qw \\
\end{quantikz}
\end{center}
\caption{Circuit for the Hermitian block encoding of a discriminant matrix.}
\label{fig:circuit_hbe_discriminant}
\end{figure}

\begin{prop}
\cref{fig:circuit_hbe_discriminant} defines $U_D\in\HBE_{1,n}(D)$.
\end{prop}
\begin{proof}
Clearly $U_D$ is unitary and Hermitian. 
Now we compute as before
\begin{equation}
\ket{0^{n}}\ket{j}\xrightarrow{O_P} \sum_{k}\sqrt{P_{jk}}\ket{k}\ket{j}\xrightarrow{\opr{SWAP}} \sum_{k}\sqrt{P_{jk}}\ket{j}\ket{k}.
\end{equation}
Meanwhile
\begin{equation}
\ket{0^{n}}\ket{i}\xrightarrow{O_P} \sum_{k'}\sqrt{P_{ik'}}\ket{k'}\ket{i}.
\end{equation}
So the inner product gives
\begin{equation}
\bra{0^{n}}\bra{i} U_D \ket{0^n}\ket{j}=\sum_{k,k'}\sqrt{P_{ik'}P_{jk}} \delta_{j,k'}
\delta_{i,k}=\sqrt{P_{ij}P_{ji}}=D_{ij}.
\end{equation}
This proves the claim. 
\end{proof}

\subsection{Szegedy's quantum walk}
For a Markov chain defined on a graph $G=(V,E)$, Szegedy's quantum walk implements a qubitization of the discriminant matrix $D$ in \cref{eqn:discriminant_matrix}.
Let $U_D$ be the Hermitian block encoding defined by the circuit in 
\cref{fig:circuit_hbe_discriminant}, we may readily plug it into \cref{fig:circuit_qubitization_hbe}, and obtain the circuit in \cref{fig:szegedy_onestep} denoted by $O_Z$.
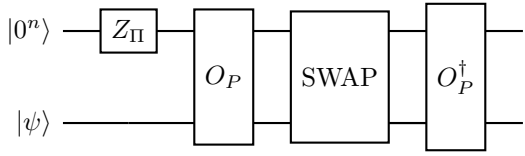
\begin{figure}[H]
\begin{quantikz}
\lstick{$\ket{0^n}$}& \gate{Z_{\Pi}}  &\gate[2]{O_P} & \gate[2]{\opr{SWAP}}&\gate[2]{O_{P}^{\dag}}&\qw \\
\lstick{$\ket{\psi}$} & \qw & & &\qw&\qw 
\end{quantikz}
\caption{Circuit implementing one step of Szegedy's quantum walk operator.}
\label{fig:szegedy_onestep}
\end{figure}
Let the eigendecomposition of $D$ be denoted by
\begin{equation}
D\ket{v_i}=\lambda_i \ket{v_i}.
\end{equation}
For each $\ket{v_i}$, the associated basis in the 2-dimensional subspace is
$\mc{B}_i=\{\ket{0^m}\ket{v_i},\ket{\perp_i}\}$.
Then the qubitization procedure gives
\begin{equation}
[O_Z]_{\mc{B}_i}=\begin{pmatrix}
\lambda_i & -\sqrt{1-\lambda_i^2}\\
\sqrt{1-\lambda_i^2} & \lambda_i
\end{pmatrix}.
\end{equation}
The eigenvalues of $O_Z$ in the $2\times 2$ matrix block are 
\begin{equation}
e^{\pm\I\arccos(\lambda_i)}.
\end{equation}
This relation is important for the following reasons. 
By \cref{prop:discriminant_reversiblewalk}, if a Markov chain is reversible and ergodic, the eigenvalues of $D$ and $P$ are the same. In particular, the largest eigenvalue of $D$ is unique and is equal to $1$, and the second largest eigenvalue of $D$ is $1-\delta$, where $\delta>0$ is called the spectral gap. Since $\arccos(1)=0$, and $\arccos(1-\delta)\approx \sqrt{2 \delta}$, we find that the spectral gap of $O_Z$ on the unit circle is in fact $\Or(\sqrt{\delta})$ instead of $\Or(\delta)$. 
This is called the spectral gap amplification, which leads to e.g. the quadratic quantum speedup of the hitting time.

\begin{exam}[Determining whether there is a marked vertex in a complete graph]
Let $G=(V,E)$ be a complete,  graph of $N=2^n$ vertices. We would like to distinguish the following two scenarios:
\begin{enumerate}

\item All vertices are the same, and the random walk is given by the transition matrix
\begin{equation}
P=\frac{1}{N} e e^{\top}, \quad e=(1,\ldots,1)^{\top}.
\end{equation} 

\item There is one \textit{marked} vertex. Without loss of generality we may assume this is the $0$-th vertex (of course we do not have access to this information). 
In this case, the transition matrix is
\begin{equation}
\wt{P}_{ij}=\begin{cases}
\delta_{ij}, & i=0,\\
P_{ij}, & i>0.
\end{cases}
\end{equation} 
\end{enumerate}
In other words, in the case (2), the random walk will stop at the marked index.
The transition matrix can also be written in the block partitioned form as
\begin{equation}
\wt{P}=\begin{pmatrix}
1 & 0\\
\frac{1}{N} \wt{e} & \frac{1}{N} \wt{e}\wt{e}^{\top}
\end{pmatrix}.
\end{equation}
Here $\wt{e}$ is an all $1$ vector of length $N-1$. 

For the random walk defined by $P$, the stationary state is $\pi=\frac{1}{N}e$, and the spectral gap is $1$. For the random walk defined by $\wt{P}$, the stationary state is $\wt{\pi}=(1,0,\ldots,0)^{\top}$, and the spectral gap of is $\delta=N^{-1}$. Starting from the uniform state $\pi$, the probability distribution after $k$ steps of random walk is $\pi^{\top}\wt{P}^{k}$. 
This converges to the stationary state of $\wt{P}$, and hence reach the marked vertex after $\Or(N)$ steps of walks (exercise).

These properties are also inherited by the discriminant matrices, with $D=P$ and 
\begin{equation}
\wt{D}=\begin{pmatrix}
1 & 0\\
0 & \frac{1}{N} \wt{e}\wt{e}^{\top}
\end{pmatrix}.
\end{equation}

To distinguish the two cases, we are given a Szegedy quantum walk operator called $O$, which can be either $O_Z$ or $\wt{O}_Z$, which is associated with $D,\wt{D}$, respectively. The initial state is 
\begin{equation}
\ket{\psi_0}=\ket{0^n}(H^{\otimes n} \ket{0^n}).
\end{equation}
Our strategy is to measure the expectation 
\begin{equation}
m_k=\braket{\psi_0|O^k |\psi_0},
\end{equation}
which can be obtained via Hadamard's test.

Before determining the value of $k$, first notice that if $O=O_Z$, then $O_{Z}\ket{\psi_0}=\ket{\psi_0}$. Hence $m_k=1$ for all values of $k$. 

On the other hand, if $O=\wt{O}_Z$, we use the fact that $\wt{D}$ only has two nonzero eigenvalues $1$ and $(N-1)/N=1-\delta$, with associated eigenvectors denoted by $\ket{\wt{\pi}}$ and $\ket{\wt{v}}=\frac{1}{\sqrt{N-1}}(0,1,1\ldots,1)^{\top}$, respectively. 
Furthermore,
\begin{equation}
\ket{\psi_0}=\frac{1}{\sqrt{N}}\ket{0^n}\ket{\wt{\pi}}+\sqrt{\frac{N-1}{N}}\ket{0^n}\ket{\wt{v}}.
\end{equation}
Due to qubitization, we have
\begin{equation}\label{eqn:tmp_OZk_eval}
\wt{O}_Z^k\ket{\psi_0}=\frac{1}{\sqrt{N}}\ket{0^n}T_k(1)\ket{\wt{\pi}}+\sqrt{\frac{N-1}{N}}\ket{0^n}T_k(1-\delta)\ket{\wt{v}}+\ket{\perp},
\end{equation}
where $\ket{\perp}$ is an unnormalized state satisfying $(\ket{0^n}\bra{0^n})\otimes I_n \ket{\perp}=0$. 
Then using $T_k(1)=1$ for all $k$, we have
\begin{equation}
m_k=\frac{1}{N}+\left(1-\frac{1}{N}\right)T_k(1-\delta).
\end{equation}
Use the fact that $T_k(1-\delta)=\cos(k\arccos(1-\delta))$, in order to have $T_k(1-\delta)\approx 0$, the smallest $k$ satisfies
\begin{equation}\label{eqn:deg_cheb_marked}
k\approx \frac{\pi}{2\arccos(1-\delta)}\approx \frac{\pi}{2\sqrt{2\delta}}=\frac{\pi\sqrt{N}}{2\sqrt{2}}.
\end{equation}
Therefore taking $k=\lceil\frac{\pi\sqrt{N}}{2\sqrt{2}} \rceil$, we have $m_k\approx 1/N$. Running Hadamard's test to constant accuracy allows us to distinguish the two scenarios.
\end{exam}

\begin{rem}[Without using the Hadamard test]
Alternatively, we may evaluate the success probability of obtaining $0^n$ in the ancilla qubits, i.e.,
\begin{equation}
p(0^n)=\norm{(\ket{0^n}\bra{0^n}\otimes I_n)O^k\ket{\psi_0}}^2.
\end{equation}
When $O=O_Z$, we have $p(0^n)=1$ with certainty. 
When $O=\wt{O}_Z$, according to \cref{eqn:tmp_OZk_eval},
\begin{equation}
p(0^n)=\frac{1}{N}+\left(1-\frac{1}{N}\right)T^2_k(1-\delta).
\end{equation}
So running the problem with $k=\lceil\frac{\pi\sqrt{N}}{2\sqrt{2}} \rceil$, we can distinguish between the two cases.
\end{rem}

\begin{rem}[Comparison with Grover's search]
It is natural to draw comparisons between Szegedy's quantum walk and Grover's search. The two algorithms make queries to different oracles, and both yield quadratic speedup compared to the classical algorithms.
The quantum walk is slightly weaker, since it only tells whether there is one marked vertex or not. 
On the other hand, Grover's search also finds the location of the marked vertex.
Both algorithms consist of repeated usage of the product of two reflectors.
The number of iterations need to be carefully controlled. 
Indeed, choosing a polynomial degree four times as large as \cref{eqn:deg_cheb_marked} would result in $m_k\approx 1$ for the case with a marked vertex. 
\end{rem}

\begin{rem}[Comparison with QPE]
Another possible solution of the problem of finding the marked vertex is to perform QPE on the Szegedy walk operator $O$ (which can be $O_Z$ or $\wt{O}_Z$).
The effectiveness of the method rests on the spectral gap amplification discussed above. We refer to ~\cite[Chapter 17]{ChildsQuantumLec} for more details.
\end{rem}

\subsection{Comparison with the original version of Szegedy's quantum walk}

The quantum walk procedure can also be presented as follows.
Using the $O_P$ oracle and the multi-qubit $\opr{SWAP}$ gate, we can define two set of quantum states
\begin{equation}
\begin{split}
\ket{\psi_j^{1}}&=O_P\ket{0^n}\ket{j}=\sum_{k}\sqrt{P_{jk}}\ket{k}\ket{j},\\ \ket{\psi_j^{2}}&=\opr{SWAP}( O_P\ket{0^n}\ket{j})=\sum_{k}\sqrt{P_{jk}}\ket{j}\ket{k}. \end{split}
\end{equation}
This defines two projection operators
\begin{equation}
\Pi_l=\sum_{j\in[N]}\ket{\psi_j^{l}}\bra{\psi_j^{l}}, \quad l=1,2,
\end{equation}
from which we can define two $2n$-qubit reflection operators $R_{\Pi_l}=2 \Pi_l-I_{2n}$. Let us write down the reflection operators more explicitly. Using the resolution of identity,
\begin{equation}
R_{\Pi_1}=O_P ((2\ket{0^n}\bra{0^n}-I)\otimes I_n)O_P^{\dag}=O_P (Z_{\Pi}\otimes I_n) O_P^{\dag}.
\end{equation}
Similarly
\begin{equation}
R_{\Pi_2}=\opr{SWAP}O_P (Z_{\Pi}\otimes I_n) O_P^{\dag}\opr{SWAP}.
\end{equation}

Then Szegedy's quantum walk operator takes the form
\begin{equation}
\mc{U}_Z=R_{\Pi_2}R_{\Pi_1},
\end{equation}
which is a rotation operator that resembles Grover's algorithm.
Note that
\begin{equation}
\mc{U}_Z=\opr{SWAP}O_P (Z_{\Pi}\otimes I_n) O_P^{\dag}\opr{SWAP}O_P (Z_{\Pi}\otimes I_n) O_P^{\dag},
\end{equation}
so
\begin{equation}
O^{\dag}_P \mc{U}_Z (O^{\dag}_P)^{-1}=O_Z^2,
\end{equation}
so the walk operator is the same as a block encoding of $T_2(D)$ using qubitization, up to a matrix similarity transformation, and the eigenvalues are the same. In particular, consider the matrix power $O_Z^k$, which provides a block encoding of the Chebyshev matrix polynomial $T_k(D)$.
Then the difference between $O_Z^{2k}$ and $\mc{U}_Z^k$ appears only at the beginning and end of the circuit.

\section{Linear combination of unitaries}\label{sec:lcu}

In practical applications, we are often not interested in constructing the Chebyshev polynomial of $A$, but a linear combination of Chebyshev polynomials.
For instance, the matrix inversion problem can be solved by expanding $f(x)=x^{-1}$ using a linear combination of Chebyshev polynomials, on the interval $[-1,-\kappa^{-1}]\cup [\kappa^{-1},1]$. Here we assume $\norm{A}=1$ and $\kappa=\norm{A}\norm{A^{-1}}$ is the condition number. 
One way to implement this is via a quantum primitive called the linear combination of unitaries (LCU).

Let \(T=\sum_{i=0}^{K-1}  \alpha_i U_i\) be the linear combination of
unitary matrices $U_i$. For simplicity let
\(K=2^a\), and $\alpha_i>0$ (WLOG we can absorb the phase of $\alpha_i$ into the unitary $U_i$).  Then
\begin{equation}
U:=\sum_{i\in[K]} \ket{i}\bra{i}\otimes U_i,
\end{equation}
implements the selection of \(U_i\) conditioned on the value of the
\(a\)-qubit ancilla states (also called the control register). \(U\) is
called a \emph{select oracle}. 

Let \(V\) be a unitary operation
satisfying
\begin{equation}\label{eqn:prepare_oracle}
V\ket{0^a}=\frac{1}{\sqrt{\norm{\alpha}_1}}\sum_{i\in[K]} \sqrt{\alpha_i}\ket{i},
\end{equation}
and \(V\) is called the \emph{prepare oracle}. 
The 1-norm of the coefficients is given by
\begin{equation}
\norm{\alpha}_1=\sum_{i} \abs{\alpha_i}.
\end{equation}

In the matrix form

\begin{equation}
V=\frac{1}{\sqrt{\norm{\alpha}_1}}\begin{pmatrix}
\sqrt{\alpha_0} & * & \cdots & *\\
\vdots & *& \ddots & \vdots \\
\sqrt{\alpha_{K-1}} & *& \cdots & *
\end{pmatrix}.
\end{equation}
where the first basis is \(\ket{0^m}\), and all other basis functions
are orthogonal to it. Then
\begin{equation}
V^{\dagger}=\frac{1}{\sqrt{\norm{\alpha}_1}}\begin{pmatrix}
\sqrt{\alpha_0} & \cdots & \sqrt{\alpha_{K-1}}\\
* & \cdots & * \\
\vdots & \ddots & \vdots \\
* & \cdots & *
\end{pmatrix}.
\label{eqn:VmatrixDag}
\end{equation}

Then $T$ can be implemented using the unitary given in \cref{lem:LCU} (called the LCU lemma).
\begin{lem}[LCU]
\label{lem:LCU}
Define
\(W=(V^{\dagger}\otimes I_n) U(V\otimes I_n)\), then for any
\(\ket{\psi}\),
\begin{equation}W\ket{0^a}\ket{\psi}=\frac{1}{\norm{\alpha}_1}\ket{0^a} T \ket{\psi} + \ket{\wt{\perp}},\end{equation}
where $\ket{\wt{\perp}}$ is an unnormalized state satisfying
\begin{equation}(\ket{0^a}\bra{0^a}\otimes I_n)\ket{\wt{\perp}}=0.\end{equation}
In other words, $W\in\BE_{\norm{\alpha}_1,a}(T)$.
\end{lem}

\begin{proof}

First
\begin{equation}
U(V\otimes I_n)\ket{0^a}\ket{\psi}=U\frac{1}{\sqrt{\norm{\alpha}_1}}\sum_i \sqrt{\alpha_i}\ket{i}\ket{\psi}=\frac{1}{\sqrt{\norm{\alpha}_1}}\sum_i \sqrt{\alpha_i}\ket{i} U_i\ket{\psi}.
\end{equation}
Then using the matrix representation \eqref{eqn:VmatrixDag}, and let
the state \(\ket{\wt{\perp}}\) collect all the states marked by \(*\) orthogonal to \(\ket{0^m}\),

\begin{equation}(V^{\dagger}\otimes I_n) U(V\otimes I_n)\ket{0^a}\ket{\psi}=\frac{1}{\norm{\alpha}_1}\ket{0^a}\sum_i \alpha_i U_i\ket{\psi} + \ket{\wt{\perp}} = \frac{1}{\norm{\alpha}_1}\ket{0^a} T \ket{\psi} + \ket{\wt{\perp}}.\end{equation}
\end{proof}

The LCU Lemma is a useful quantum primitive, as it states that the number of ancilla qubits needed only depends logarithmically on $K$, the number of terms in the linear combination. 
Hence it is possible to implement the linear combination of a very large number of terms efficiently.
From a practical perspective, the select and prepare oracles uses multi-qubit controls, and can be difficult to implement. If implemented directly, the number of multi-qubit controls again depends linearly on $K$ and is not desirable. Therefore an efficient implementation using LCU (in terms of the gate complexity) also requires additional structures in the prepare and select oracles.

If we apply \(W\) to \(\ket{0^a}\ket{\psi}\) and measure the ancilla qubits, then the probability of obtaining the outcome \(0^a\) in the ancilla qubits (and therefore obtaining the state 
\(T\ket{\psi}/\norm{T\ket{\psi}}\) in the system register) is
\(\left(\norm{T\ket{\psi}}/\norm{\alpha}_1\right)^2\). The expected
number of repetition needed to succeed is
\(\left(\norm{\alpha}_1/\norm{T\ket{\psi}}\right)^2\). Now we
demonstrate that using amplitude amplification
(AA) in \cref{sec:amplitudeamplification}, this number can be reduced to \(\Or\left(\norm{\alpha}_1/\norm{T\ket{\psi}}\right)\).

\begin{rem}[Alternative construction of the prepare oracle]
In some applications it may not be convenient to absorb the phase of $\alpha_i$ into the select oracle. In such a case, we may modify the prepare oracle instead. 
If $\alpha_i=r_i e^{\I \theta_i}$ with $r_i>0,\theta_i\in[0,2\pi)$, we can define $\sqrt{\alpha_i}=\sqrt{r_i} e^{\I \theta_i/2}$, and $V$ is defined as in \cref{eqn:prepare_oracle}. 
However, instead of $V^{\dag}$, we need to introduce
\begin{equation}
\wt{V}=\frac{1}{\sqrt{\norm{\alpha}_1}}\begin{pmatrix}
\sqrt{\alpha_0} & \cdots & \sqrt{\alpha_{K-1}}\\
* & \cdots & * \\
\vdots & \ddots & \vdots \\
* & \cdots & *
\end{pmatrix}.
\end{equation}
Then following the same proof as \cref{lem:LCU}, we find that 
$W=(\wt{V}\otimes I_n) U(V\otimes I_n)\in\BE_{\norm{\alpha}_1,a}(T)$.

\end{rem}

\begin{rem}[Linear combination of non-unitaries]
Using the block encoding technique, we may immediately obtain linear combination of general matrices that are not unitaries. 
However, with some abuse of notation, the term ``LCU'' will be used whether the terms to be combined are unitaries or not. 
In other words, the term ``linear combination of unitaries'' should be loosely interpreted as ``linear combination of things'' (LCT) in many contexts.
\end{rem}

\begin{exam}[Linear combination of two matrices]
Let $U_A,U_B$ be two $n$-qubit unitaries, and we would like to construct a block encoding of $T=U_A+U_B$.

There are two terms in total, so one ancilla qubit is needed.
The prepare oracle needs to implement 
\begin{equation}
V\ket{0}=\frac{1}{\sqrt{2}}(\ket{0}+\ket{1}),
\end{equation}
so this is the Hadamard gate. The circuit is given by \cref{fig:lcu_twounitary}, which constructs $W\in \BE_{\sqrt{2},1}(T)$.
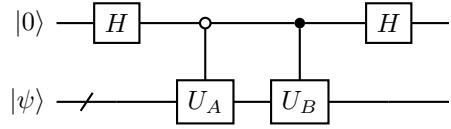
\begin{figure}[H]
\begin{center}
\begin{quantikz}
\lstick{$\ket{0}$} & \gate{H} & \octrl{1} & \ctrl{1} & \gate{H} & \qw\\
\lstick{$\ket{\psi}$} & \qwb & \gate{U_A} & \gate{U_B} & \qw & \qw
\end{quantikz}
\end{center}
\caption{Circuit for linear combination of two unitaries. }
\label{fig:lcu_twounitary}
\end{figure}

A special case is the linear combination of two block encoded matrices.
Given two $n$-qubit matrices $A,B$, for simplicity let $U_A\in\BE_{1,m}(A),U_B\in\BE_{1,m}(B)$. We would like to construct a block encoding of $T=A+B$. 
The circuit is given by \cref{fig:lcu_twobem}, which constructs $W\in \BE_{\sqrt{2},1+m}(T)$. This is also an example of a linear combination of non-unitary matrices.
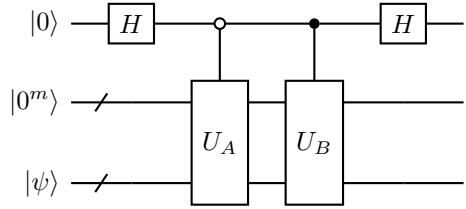
\begin{figure}[H]
\begin{center}
\begin{quantikz}
\lstick{$\ket{0}$} & \gate{H} & \octrl{1} & \ctrl{1} & \gate{H} & \qw\\
\lstick{$\ket{0^m}$} & \qwb & \gate[2]{U_A} & \gate[2]{U_B} & \qw & \qw\\
\lstick{$\ket{\psi}$} & \qwb &  &  & \qw & \qw
\end{quantikz}
\end{center}
\caption{Circuit for linear combination of two block encoded matrices. }
\label{fig:lcu_twobem}
\end{figure}
 
\end{exam}

\begin{exam}[Transverse field Ising model]
Consider the following TFIM model with periodic boundary conditions ($Z_{n}=Z_0$), and $n=2^\mf{n}$,
\begin{equation}
\hat{H}=-\sum_{i\in[n]} Z_iZ_{i+1}-\sum_{i\in[n]} X_i.
\end{equation}
In order to use LCU, we need $(\mf{n}+1)$ ancilla qubits.
The prepare oracle can be simply constructed from the Hadamard gate
\begin{equation}
V=H^{\otimes(\mf{n}+1)},
\end{equation}
and the select oracle implements
\begin{equation}
U=\sum_{i\in[n]}\ket{i}\bra{i}\otimes (-Z_iZ_{i+1}) + \sum_{i\in[n]}\ket{i+n}\bra{i+n}\otimes (-X_i).
\end{equation}
The corresponding $W\in\BE_{\sqrt{2n},\mf{n}+1}(\hat{H})$.
\end{exam}

\begin{exam}[Block encoding of a matrix polynomial]
Let us use the LCU lemma to construct the block encoding for an arbitrary matrix polynomial for a Hermitian matrix $A$ in \cref{sec:qubitize_hermbe}.
\begin{equation}
f(A)=\sum_{k\in[K]} \alpha_k T_k(A),
\end{equation}
with $\norm{\alpha}_1=\sum_{k\in[K]} |\alpha_k|$ and we set $K=2^{a}$. For simplicity assume $\alpha_k\ge 0$. 

We have constructed $U_k:=(U_A Z_{\Pi})^k$ as the $(1,m)$-block-encoding of $T_k(A)$. From each $U_k$ we can implement the select oracle
\begin{equation}
U:=\sum_{k\in[K]} \ket{k}\bra{k}\otimes U_k
\end{equation}
via multi-qubit controls.
Also given the availability of the prepare oracle
\begin{equation}
V\ket{0^{a}}=\frac{1}{\sqrt{\norm{\alpha}_1}}\sum_{k\in[K]} \sqrt{\alpha_k}\ket{k},
\end{equation}
we obtain a $(\norm{\alpha}_1,m+a)$-block-encoding of $f(A)$.

The need of using $a$ ancilla qubits, and even more importantly the need to implement the prepare and select oracles is undesirable.
We will see later that the quantum signal processing (QSP) and quantum singular value transformation (QSVT) can drastically reduce both sources of difficulties.
\end{exam}

\begin{exam}[Matrix functions given by a matrix Fourier series]
Instead of block encoding, LCU can also utilize a different query model based on Hamiltonian simulation.
Let $A$ be an $n$-qubit Hermitian matrix. 
Consider $f(x)\in\RR$ given by its Fourier expansion (up to a normalization factor)
\begin{equation}
  f(x)=\int \hat{f}(k) e^{\I kx}\ud k,
\end{equation}
and we are interested in computing the matrix function via numerical quadrature
\begin{equation}
  f(A)=\int \hat{f}(k) e^{\I kA}\ud k \approx \Delta k \sum_{k\in \mc{K}} \hat{f}(k) e^{\I kA}.
\end{equation}
Here $\mc{K}$ is a uniform grid discretizing the interval $[-L,L]$ using $\abs{\mc{K}}=2^{\mf{k}}$ grid points, and the grid spacing is $\Delta k = 2L/\abs{\mc{K}}$. 
The prepare oracle is given by the coefficients $c_k=\Delta k \hat{f}(k)$, and the corresponding subnormalization factor is
\begin{equation}
  \norm{c}_1=\sum_{k\in\mc{K}}\Delta k \abs{\hat{f}(k)} \approx \int \abs{\hat{f}(k)}\ud k.
\end{equation}
The select oracle is
\begin{equation}
  U=\sum_{k\in\mc{K}} \ket{k}\bra{k}\otimes e^{\I kA}.
\end{equation}
This can be efficiently implemented using the controlled matrix powers as in \cref{fig:circuit_controlled_Upower}, where the basic unit is the short time Hamiltonian simulation $e^{\I \Delta k A}$.
This can be used to block encode a large class of matrix functions.
\end{exam}

\section{Qubitization of Hermitian matrices with general block encoding}\label{sec:qubitize_genbe}

In \cref{sec:qubitize_hermbe} we assume that $U_A=U_A^{\dag}$ to block encode a Hermitian matrix $A$.
For instance, $s$-sparse Hermitian matrices, such Hermitian block encodings can be constructed following the recipe in \cref{sec:query_general}.
However, this can come at the expense of requiring additional structures and oracles.
In general, the block encoding of a Hermitian matrix may not be Hermitian itself.
In this section we demonstrate that the strategy of qubitization can be modified to accommodate general block encodings.

Again start from the eigendecomposition \cref{eqn:eigdecompose_herm}, we apply $U_A$ to $\ket{0^m}\ket{v_i}$ and obtain 
\begin{equation}
U_A\ket{0^m}\ket{v_i}=\lambda_i\ket{0^m}\ket{v_i}+\sqrt{1-\lambda_i^2}\ket{\perp_i'},
\label{eqn:UA_apply_vi_2}
\end{equation}
where $\ket{\perp_i'}$ is a normalized state satisfying $\Pi\ket{\perp_i'}=0$.

Since $U_A$ block-encodes a Hermitian matrix $A$, we have
\begin{equation}
U_A^{\dag}=\begin{pmatrix}
{A} & {*} \\
{*} & {*}
\end{pmatrix},
\end{equation}
which implies that there exists another normalized state $\ket{\perp_i}$ satisfying $\Pi\ket{\perp_i}=0$ and
\begin{equation}
U^{\dag}_A\ket{0^m}\ket{v_i}=\lambda_i\ket{0^m}\ket{v_i}+\sqrt{1-\lambda_i^2}\ket{\perp_i}.
\label{eqn:Udag_apply_vi}
\end{equation}
Now apply $U_A$ to both sides of \cref{eqn:Udag_apply_vi}, we obtain
\begin{equation}
\ket{0^m}\ket{v_i}=\lambda^2_i\ket{0^m}\ket{v_i}+\lambda_i\sqrt{1-\lambda_i^2}\ket{\perp_i'} +\sqrt{1-\lambda_i^2}U_A\ket{\perp_i},
\end{equation}
which gives
\begin{equation}
U_A\ket{\perp_i}=\sqrt{1-\lambda_i^2}\ket{0^m}\ket{v_i}-\lambda_i\ket{\perp_i'}.
\label{eqn:UA_apply_perpi_2}
\end{equation}
Define
\begin{equation}
\mc{B}_i=\{\ket{0^m}\ket{v_i},\ket{\perp_i}\}, \quad \mc{B}'_i=\{\ket{0^m}\ket{v_i},\ket{\perp_i'}\},
\end{equation}
and the associated two-dimensional subspaces $\mc{H}_i=\opr{span}{B_i}, \mc{H}'_i=\opr{span}{B_i'}$, we find that $U_A$ maps $\mc{H}_i$ to $\mc{H}_i'$.
Correspondingly $U_A^{\dag}$ maps $\mc{H}_i'$ to $\mc{H}_i$.

Then \cref{eqn:UA_apply_vi_2,eqn:UA_apply_perpi_2} give the matrix representation
\begin{equation}
[U_A]_{\mc{B}_i}^{\mc{B}_i'}=\begin{pmatrix}
\lambda_i & \sqrt{1-\lambda_i^2}\\
\sqrt{1-\lambda_i^2} & -\lambda_i
\end{pmatrix}.
\end{equation}
Similar calculation shows that
\begin{equation}
[U^{\dag}_A]_{\mc{B}_i'}^{\mc{B}_i}=\begin{pmatrix}
\lambda_i & \sqrt{1-\lambda_i^2}\\
\sqrt{1-\lambda_i^2} & -\lambda_i
\end{pmatrix}.
\end{equation}
Meanwhile both $\mc{H}_i$ and $\mc{H}_i'$ are the invariant subspaces of the projector $\Pi$, with matrix representation
\begin{equation}
[\Pi]_{\mc{B}_i}=[\Pi]_{\mc{B}_i'}=\begin{pmatrix}
1 & 0 \\
0 & 0
\end{pmatrix}.
\end{equation}
Therefore
\begin{equation}
[Z_\Pi]_{\mc{B}_i}=[Z_\Pi]_{\mc{B}_i'}=\begin{pmatrix}
1 & 0 \\
0 & -1
\end{pmatrix}.
\end{equation}
Hence $\mc{H}_i$ is an invariant subspace of $\wt{O}=U_A^{\dag}Z_{\Pi}U_A Z_{\Pi}$, with matrix representation
\begin{equation}
[\wt{O}]_{\mc{B}_i}=\begin{pmatrix}
\lambda_i & -\sqrt{1-\lambda_i^2}\\
\sqrt{1-\lambda_i^2} & \lambda_i
\end{pmatrix}^2.
\end{equation}
Repeating $k$ times, we have
\begin{equation}
\begin{split}
[\wt{O}^k]_{\mc{B}_i}=&(U_A^{\dag}Z_{\Pi}U_A Z_{\Pi})^{k}=\begin{pmatrix}
\lambda_i & -\sqrt{1-\lambda_i^2}\\
\sqrt{1-\lambda_i^2} & \lambda_i
\end{pmatrix}^{2k}\\
=&\begin{pmatrix}
T_{2k}(\lambda_i) & -\sqrt{1-\lambda_i^2}U_{2k-1}(\lambda_i)\\
\sqrt{1-\lambda_i^2}U_{2k-1}(\lambda_i) & T_{2k}(\lambda_i)
\end{pmatrix}.
\end{split}
\end{equation}
Since any vector $\ket{0^m}\ket{\psi}$ can be expanded in terms of the eigenvectors $\ket{0^m}\ket{v_i}$, we have
\begin{equation}
(U_A^{\dag}Z_{\Pi}U_A Z_{\Pi})^{k}=\begin{pmatrix}
T_{2k}(A) & *\\
* & *
\end{pmatrix}.
\end{equation}
Therefore if we would like to construct an \emph{even} order Chebyshev polynomial $T_{2k}(A)$, the circuit $(U_A^{\dag}Z_{\Pi}U_A Z_{\Pi})^{k}$ straightforwardly gives a $(1,m)$-block-encoding.

In order to construct the block-encoding of an \emph{odd} polynomial $T_{2k+1}(A)$, we note that
\begin{equation}
[U_A Z_{\Pi}(U_A^{\dag}Z_{\Pi}U_A Z_{\Pi})^{k}]_{\mc{B}_i}^{\mc{B}_{i}'}=\begin{pmatrix}
T_{2k+1}(\lambda_i) & -\sqrt{1-\lambda_i^2}U_{2k}(\lambda_i)\\
\sqrt{1-\lambda_i^2}U_{2k}(\lambda_i) & T_{2k+1}(\lambda_i)
\end{pmatrix}.
\end{equation}
Using the fact that $\mc{B}_i,\mc{B}_i'$ share the common basis $\ket{0^m}\ket{v_i}$, we still have the block-encoding
\begin{equation}
U_A Z_{\Pi}(U_A^{\dag}Z_{\Pi}U_A Z_{\Pi})^{k}=\begin{pmatrix}
T_{2k+1}(A) & *\\
* & *
\end{pmatrix}.
\end{equation}
Therefore $U_A Z_{\Pi}(U_A^{\dag}Z_{\Pi}U_A Z_{\Pi})^{k}$ is a $(1,m)$-block-encoding of $T_{2k+1}(A)$. 

In summary, the block-encoding of $T_{l}(A)$ is given by applying $U_A Z_{\Pi}$ and $U^{\dag}_A Z_{\Pi}$ alternately. 
If $l=2k$, then there are exactly $k$ such pairs. 
The quantum circuit for each pair $\wt{O}$ is
% \begin{displaymath}
% \begin{quantikz}
% \lstick{$\ket{1}$}  & \targ{} & \gate{Z} & \targ{}    & \qw &\targ{} & \gate{Z} & \targ{}    & \qw&\qw \\
% \lstick{$\ket{0^m}$}& \octrl{-1} & \qw   & \octrl{-1} & \gate[2]{U_A} &\octrl{-1} & \qw   & \octrl{-1} &\gate[2]{U^{\dag}_A} &\qw\\
% \lstick{$\ket{\psi}$} & \qw & \qw & \qw & &\qw & \qw & \qw & &\qw
% \end{quantikz}
% \end{displaymath}

\begin{figure}[H]
\begin{displaymath}
\begin{quantikz}
\lstick{$\ket{0^m}$}& \gate{Z_{\Pi}} & \gate[2]{U_A} & \gate{Z_{\Pi}} &\gate[2]{U^{\dag}_A} &\qw\\
\lstick{$\ket{\psi}$} & \qw & &\qw & &\qw
\end{quantikz}
\end{displaymath}
\caption{Circuit implementing one step of qubitization with a general block encoding of a Hermitian matrix. This block encodes $T_{2}(A)$. Here $U_A\in \BE_{1,m}(A)$.}
\label{fig:}
\end{figure}
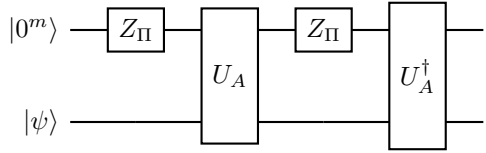

Otherwise if $l=2k+1$, then there is an extra $U_A Z_{\Pi}$.
The effect is to map each eigenvector $\ket{0^m}\ket{v_i}$ back and forth between the two-dimensional subspaces $\mc{H}_i$ and $\mc{H}_i'$.
In \cref{sec:qsp}, we shall see that the separate treatment even/odd polynomials will play a more prominent role.

Now that we have obtained $O_k\in\BE_{1,m}(T_k(A))$ for all $k$, we can use the LCU lemma again to construct block encodings for linear combination of Chebyshev polynomials. 
We omit the details here.

\section{Quantum eigenvalue transformation}\label{eqn:qsp_matrix}

Let us briefly recap what we have done so far: 
(1) construct a block encoding of an Hermitian matrix $A$ (the block encoding matrix itself can be non-Hermitian); 
(2) use qubitization to block encode $T_k(A)$; 
(3) use LCU to block encode an arbitrary polynomial function of $A$ (up to a subnormalization factor).
This framework is conceptually appealing, but the practical implementation of the select and prepare oracles are by no means straightforward, and can come at significant costs.

\subsection{Hermitian block encoding}

Note that $\I Z=e^{\I \frac{\pi}{2}Z}$, if $A$ is given by a Hermitian block encoding $U_A$, the block encoding of the Chebyshev polynomial in \cref{sec:qubitize_hermbe} can be written as
\begin{equation}
O^d=(-\I)^d  \prod_{j=1}^d (U_A e^{\I \frac{\pi}{2} Z_{\Pi}}).
\end{equation}
This is a special case of the following representation.
Note that $(-\I)^d$ is an irrelevant phase factor and can be discarded.
\begin{equation}
\label{eqn:qet_hbe}
U_{\wt{\Phi}}=e^{\I \wt{\phi}_0 Z_{\Pi}} \prod_{j=1}^d (U_A e^{\I \wt{\phi}_j Z_{\Pi}}).
\end{equation}
The representation \cref{eqn:qet_hbe} is called a quantum eigenvalue transformation (QET).

Due to qubitization, $U_{\wt{\Phi}}$ should block encode some matrix polynomial of $A$. We first state the following theorem without proof.
\begin{thm}[Quantum signal processing, a simplified version]\label{thm:qsp_simple}
Consider 
\begin{equation}
U_A=\begin{pmatrix}
x & \sqrt{1-x^2}\\
\sqrt{1-x^2} & -x
\end{pmatrix}.
\end{equation}
For any $\wt{\Phi} := (\wt{\phi}_0, \cdots, \wt{\phi}_d) \in \RR^{d+1}$,
\begin{equation}
  U_{\wt{\Phi}} := e^{\I \wt{\phi}_0 Z} \prod_{j=1}^d (U_A e^{\I \wt{\phi}_j Z}) = \begin{pmatrix}
P(x) & *\\
* & *
\end{pmatrix},
\end{equation}
where $P\in\CC[x]$ satisfy
\begin{enumerate}

\item $\deg(P) \leq d$,

\item $P$ has parity $(d \bmod 2)$, 

\item $|P(x)|\le 1, \forall x \in [-1, 1]$.
\end{enumerate}
Also define $-\wt{\Phi} := (-\wt{\phi}_0, \cdots, -\wt{\phi}_d) \in \RR^{d+1}$, then
\begin{equation}\label{eqn:qsp_conj}
  U_{-\wt{\Phi}} := e^{-\I \wt{\phi}_0 Z} \prod_{j=1}^d (U_A e^{-\I \wt{\phi}_j Z}) = U^*_{\wt{\Phi}}= \begin{pmatrix}
P^*(x) & *\\
* & *
\end{pmatrix}.
\end{equation}
Here $P^*(x)$ is the complex conjugation of $P(x)$, and $U^*_{\wt{\Phi}}$ is the complex conjugation (without transpose) of $U_{\wt{\Phi}}$.
\end{thm}

\begin{rem}
This theorem can be proved inductively. However, this is a special case of the quantum signal processing in \cref{thm:qsp}, so we will omit the proof here. In fact, \cref{thm:qsp} will also state the converse of the result, which describes precisely the class of matrix polynomials that can be described by such phase factor modulations. 
In \cref{thm:qsp_simple}, the condition (1) states that the polynomial degree is upper bounded by the number of $U_A$'s, and the condition (3) is simply a consequence of that $U_{\wt{\Phi}}$ is a unitary matrix. 
The condition (2) is less obvious, but should not come at a surprise, since we have seen the need of treating even and odd polynomials separately in the case of qubitization with a general block encoding.
\cref{eqn:qsp_conj} can be proved directly by taking the complex conjugation of $U_{\wt{\Phi}}$.
\end{rem}

Following the qubitization procedure, we immediately have \cref{thm:qet_hbe}.

\begin{thm}[Quantum eigenvalue transformation with Hermitian block encoding]\label{thm:qet_hbe}
Let $U_A\in\HBE_{1,m}(A)$. Then for any $\wt{\Phi} := (\wt{\phi}_0, \cdots, \wt{\phi}_d) \in \RR^{d+1}$,
\begin{equation}
  U_{\wt{\Phi}} = e^{\I \wt{\phi}_0 Z_{\Pi}} \prod_{j=1}^d (U_A e^{\I \wt{\phi}_j Z_{\Pi}}) = \begin{pmatrix}
P(A) & *\\
* & *
\end{pmatrix}\in \BE_{1,m}(P(A)),
\end{equation}
where $P\in\CC[x]$ satisfies the requirements in \cref{thm:qsp_simple}.
\end{thm}
Using \cref{thm:qet_hbe}, we may construct the block encoding of a matrix polynomial without invoking LCU. The cost is essentially the same as block encoding a Chebyshev polynomial. 

In order to implement $e^{\I \phi Z_{\Pi}}$, we note that the quantum circuit denoted by $\opr{CR}_{\phi}$ is in \cref{fig:circuit_cr_qsp}
returns $e^{\I\phi}\ket{0}\ket{0^m}$ if $b=0^m$, and $e^{-\I\phi}\ket{0}\ket{b}$ if $b\ne 0^m$.
So omitting the signal qubit, this is precisely $e^{\I \phi Z_{\Pi}}$.

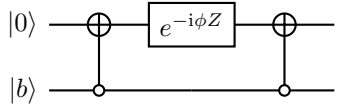
\begin{figure}[H]
\begin{displaymath}
% \begin{quantikz}
% \lstick{$\ket{1}$}  & \targ{} & \gate{e^{\I \phi Z}} & \targ{}    & \qw \\
% \lstick{$\ket{b}$}& \octrl{-1} & \qw   & \octrl{-1} & \qw
% \end{quantikz}
%\equiv
\begin{quantikz}
\lstick{$\ket{0}$}  & \targ{} & \gate{e^{-\I \phi Z}} & \targ{}    & \qw \\
\lstick{$\ket{b}$}& \octrl{-1} & \qw   & \octrl{-1} & \qw
\end{quantikz}
\end{displaymath}
\caption{Implementing the controlled rotation circuit for quantum eigenvalue transformation.}
\label{fig:circuit_cr_qsp}
\end{figure}

Therefore, if $A$ is given by a Hermitian block encoding $U_A$, we can follow the argument in \cref{sec:qubitize_hermbe} and construct the following unitary
The corresponding quantum circuit is in \cref{fig:qet_circuit_hbe}, which uses one extra ancilla qubit.
When measuring the $(m+1)$ ancilla qubits and obtain $\ket{0}\ket{0^m}$, the corresponding (unnormalized) state in the system register is $ P(A)\ket{\psi}$. 
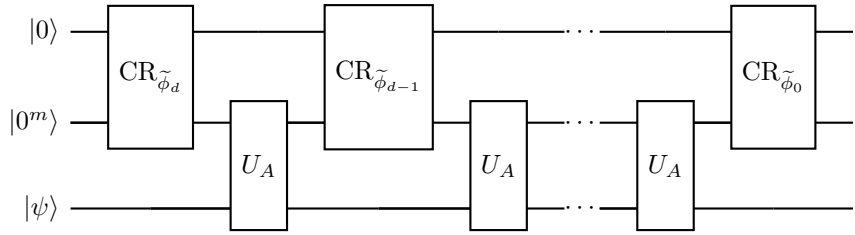
\begin{figure}[H]
\begin{displaymath}
%\begin{adjustbox}{width=.97\textwidth} 
\begin{quantikz}
  \lstick{$\ket{0}$} & \gate[2]{\opr{CR}_{\wt{\phi}_d}} & \qw & \gate[2]{\opr{CR}_{\wt{\phi}_{d-1}}} & \qw & \qw \raisebox{0em}{$\cdots$}&\qw& \gate[2]{\opr{CR}_{\wt{\phi}_0}}&\qw\\
\lstick{$\ket{0^m}$} &\qw& \gate[2]{U_A} &  \qw  & \gate[2]{U_A} &\qw\raisebox{0em}{$\cdots$} &\gate[2]{U_A}&\qw&\qw\\
\lstick{$\ket{\psi}$}& \qw& \qw& \qw& \qw& \qw\raisebox{0em}{$\cdots$}&\qw&\qw&\qw
\end{quantikz}
%\end{adjustbox}
\end{displaymath}
  \caption{Circuit of quantum eigenvalue transformation to construct $U_{P(A)}\in\BE_{1,m+1}(P(A))$, using $U_A\in\HBE_{1,m}(A)$.  }
  \label{fig:qet_circuit_hbe}
\end{figure}

The QET described by the circuit in \cref{fig:qet_circuit_hbe} generally constructs a block encoding of $P(A)$ for some complex polynomial $P$. 
In practical applications (such as those later in this chapter), we would like to construct a block encoding of $P_{\Re}(A)\equiv (\Re P)(A)=\frac12(P(A)+P^*(A))$ instead.
Below we demonstrate that a simple modification of \cref{fig:qet_circuit_hbe} allows us to achieve this goal.

To this end, we use \cref{eqn:qsp_conj}. Qubitization allows us to construct
\begin{equation}
U_{-\wt{\Phi}} = e^{-\I \wt{\phi}_0 Z_{\Pi}} \prod_{j=1}^d (U_A e^{-\I \wt{\phi}_j Z_{\Pi}}) = \begin{pmatrix}
P^*(A) & *\\
* & *
\end{pmatrix}. 
\end{equation} 
So all we need is to negate all phase factors in $\wt{\Phi}$.
In order to implement $\opr{CR}_{-\phi}$, we do not actually need to implement a new circuit.
Instead we may simply change the signal qubit from $\ket{0}$ to $\ket{1}$:
\begin{displaymath}
\begin{quantikz}
\lstick{$\ket{1}$}  & \targ{} & \gate{e^{-\I \phi Z}} & \targ{}    & \qw \\
\lstick{$\ket{b}$}& \octrl{-1} & \qw   & \octrl{-1} & \qw
\end{quantikz}
\end{displaymath}
which returns $e^{-\I\phi}\ket{1}\ket{0^m}$ if $b=0^m$, and $e^{\I\phi}\ket{1}\ket{b}$ if $b\ne 0^m$. 
In other words, the circuit for $U_{P^*(A)}$ and $U_{P(A)}$ are \emph{exactly the same} except that the input signal qubit is changed from $\ket{0}$ to $\ket{1}$.

Now we claim the circuit in \cref{fig:qet_circuit_general_real} implements a block encoding $U_{P_{\Re}(A)}\in \BE_{1,m+1}(P_{\Re}(A))$. 
This circuit  can be viewed as an implementation of the linear combination of unitaries $\frac12 (U_{P^*(A)}+U_{P(A)})$. 

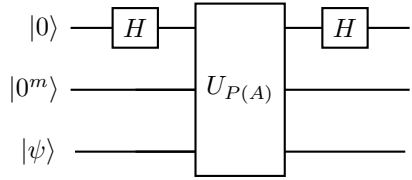
\begin{figure}[H]
  \begin{quantikz}
    \lstick{$\ket{0}$}   & \gate{H} & \gate[3]{U_{P(A)}} & \gate{H}    & \qw \\
    \lstick{$\ket{0^m}$}   & \qw & \qw   &  \qw &  \qw\\\
    \lstick{$\ket{\psi}$}& \qw & \qw   &  \qw &  \qw\
  \end{quantikz}
  \caption{Circuit of quantum eigenvalue transformation for constructing a $(1,m+1)$-block-encoding of $P_{\Re}(A)$. }
  \label{fig:qet_circuit_general_real}
\end{figure}
\begin{displaymath} 
\end{displaymath}
To verify this, we may evaluate
\begin{equation}
\begin{split}
  \ket{0}\ket{0^m}\ket{\psi} &\xrightarrow{H\otimes I_{m+n}} \frac{1}{\sqrt{2}}(\ket{0}+\ket{1})\ket{0^m}\ket{\psi}\\
  &\xrightarrow{U_{P(A)}} \frac{1}{\sqrt{2}}\ket{0}(\ket{0^m}P(A)\ket{\psi}+\ket{\perp})
  + \frac{1}{\sqrt{2}}\ket{1}(\ket{0^m}P^*(A)\ket{\psi}+\ket{\perp'})\\
  &\xrightarrow{H\otimes I_{m+n}} \ket{0}\left(\ket{0^m}\frac{P(A)+P^*(A)}{2}\ket{\psi}\right)+\ket{\wt{\perp}}\\
  &=\ket{0}\ket{0^m}P_{\Re}(A)\ket{\psi}+\ket{\wt{\perp}}
\end{split}
\end{equation}
Here $\ket{\perp},\ket{\perp'}$ are two $(m+n)$-qubit state orthogonal to any state $\ket{0^m}\ket{x}$, while  
$\ket{\wt{\perp}}$ is a $(m+n+1)$-qubit state orthogonal to any state $\ket{0}\ket{0^m}\ket{x}$.
In other words, by measuring all $(m+1)$ ancilla qubits and obtain $0^{m+1}$, the corresponding (unnormalized) state in the system register is $P_{\Re}(A)\ket{\psi}$.

\subsection{General block encoding}

If $A$ is given by a general block encoding $U_A$, the quantum eigenvalue transformation should consist of an alternating sequence of $U_A,U_A^{\dag}$ gates. The circuit is given by \cref{fig:qet_circuit_general}, and 
the corresponding block encoding is described in \cref{thm:qet_general}.
Note that the Hermitian block encoding becomes a special case with $U_A=U_A^{\dag}$.

\begin{figure}[H]
  \begin{quantikz}
    \lstick{$\ket{0}$} & \gate[2]{\opr{CR}_{\wt{\phi}_d}} & \qw & \gate[2]{\opr{CR}_{\wt{\phi}_{d-1}}} & \qw & \qw \raisebox{0em}{$\cdots$}&\qw& \gate[2]{\opr{CR}_{\wt{\phi}_0}}&\qw\\
    \lstick{$\ket{0^m}$} &\qw& \gate[2]{U_A} &  \qw  & \gate[2]{U^{\dag}_A} &\qw\raisebox{0em}{$\cdots$} &\gate[2]{U_A}&\qw&\qw\\
    \lstick{$\ket{\psi}$}& \qw& \qw& \qw& \qw& \qw\raisebox{0em}{$\cdots$}&\qw&\qw&\qw
  \end{quantikz}
  \caption{Circuit of quantum eigenvalue transformation to construct $U_{P(A)}\in\BE_{1,m+1}(P(A))$, using $U_A\in\BE_{1,m}(A)$.
  Here $U_A,U_A^{\dag}$ should be applied alternately.  When $d$ is even, the last $U_A$ gate should be replaced $U_A^{\dag}$. }
  \label{fig:qet_circuit_general}
\end{figure}
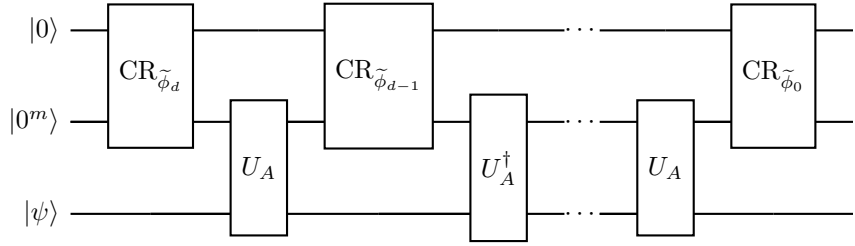

\begin{thm}[Quantum eigenvalue transformation with general block encoding]\label{thm:qet_general}
Let $U_A\in\BE_{1,m}(A)$. Then for any $\wt{\Phi} := (\wt{\phi}_0, \cdots, \wt{\phi}_d) \in \RR^{d+1}$,
let 
\begin{equation}
U_{\wt{\Phi}}= e^{\I \wt{\phi}_0 Z_{\Pi}} 
\prod_{j=1}^{d/2}\left[ U^{\dag}_A e^{\I \wt{\phi}_{2j-1} Z_{\Pi}}U_A e^{\I \wt{\phi}_{2j} Z_{\Pi}} \right]
\end{equation}
when $d$ is even, and
\begin{equation}
U_{\wt{\Phi}}=(-\I)^d e^{\I \wt{\phi}_0 Z_{\Pi}} (U_A e^{\I \wt{\phi}_{1} Z_{\Pi}})
\prod_{j=1}^{(d-1)/2}\left[ U^{\dag}_A e^{\I \wt{\phi}_{2j} Z_{\Pi}}U_A e^{\I \wt{\phi}_{2j+1} Z_{\Pi}} \right]
\end{equation}
when $d$ is odd.
Then
\begin{equation}
  U_{\wt{\Phi}} = \begin{pmatrix}
P(A) & *\\
* & *
\end{pmatrix}\in \BE_{1,m}(P(A)),
\end{equation}
where $P\in\CC[x]$ satisfy the conditions in \cref{thm:qsp_simple}.
\end{thm}

Following exactly the same procedure, we find that the circuit in \cref{fig:qet_circuit_general_real}, with $U_{P(A)}$ given by \cref{fig:qet_circuit_general} implements a $U_{P_{\Re}(A)}\in\BE_{1,m+1}(P_{\Re}(A))$.
This is left as an exercise.

\subsection{General matrix polynomials}\label{eqn:qet_general_polynomial}

In practical applications, we may be interested in matrix polynomials $f(A)$, where $f(x)\in\RR[x]$ does not have a definite parity. This violates the parity requirement of \cref{thm:qsp_simple}.
This can be solved by using the LCU technique.

Note that
\begin{equation}
f(x)=f_{\mathrm{even}}(x)+f_{\mathrm{odd}}(x),
\end{equation}
where $f_{\mathrm{even}}(x)=\frac12(f(x)+f(-x)), f_{\mathrm{odd}}(x)=\frac12(f(x)-f(-x))$. 
If $|f(x)|\le 1$ on $[-1,1]$, then  $|f_{\mathrm{even}}(x)|,|f_{\mathrm{odd}}(x)|\le 1$ on $[-1,1]$, and $f_{\mathrm{even}}(x)$, $f_{\mathrm{odd}}(x)$ can be each constructed using the circuit in \cref{fig:qet_circuit_general_real}. 
Introducing another ancilla qubit and using the LCU technique, we obtain a $(1,m+2)$-block-encoding of $(f_{\mathrm{even}}(A)+f_{\mathrm{odd}}(A))/2$. In other words, we obtain a circuit $U_f\in \BE_{2,m+2}(f(A))$.
Note that unlike the case of the block encoding of $P_{\Re}(A)$, we lose a subnormalization factor of $2$ here.

Following the same principle, if $f(x)=g(x)+\I h(x)\in\CC[x]$ is a given complex polynomial, and $g,h\in\RR[x]$ do not have a definite parity, we can construct $U_{g(A)}\in\BE_{2,m+2}(g(A)),U_{h(A)}\in\BE_{2,m+2}(h(A))$.
Then applying another layer of LCU, we obtain $U_{f(A)}\in \BE_{4,m+3}(f(A))$.

On the other hand, if the real and imaginary parts $g,h$ have definite parity, then $U_{g(A)}\in\BE_{1,m+1}(g(A)),U_{h(A)}\in\BE_{1,m+1}(h(A))$.
Applying LCU, we obtain $U_{f(A)}\in \BE_{2,m+2}(f(A))$.

The construction circuits in the cases above is left as an exercise.

\section{Quantum signal processing}\label{sec:qsp}

In terms of implementing matrix polynomials of Hermitian matrices, quantum eigenvalue transform provides a much simpler circuit than the method based on LCU and qubitization (i.e., linear combination of Chebyshev polynomials).
The simplification is clear both in terms of the number of ancilla qubits and of the circuit architecture.
However, it is not clear so far for which polynomials (either a complex polynomial $P\in\CC[x]$ or a real polynomial $P_{\Re}\in\RR[x]$) we can apply the QET technique, and how to obtain the phase factors. 
Quantum signal processing (QSP) provides a complete answer to this question.

Due to qubitization, all these questions can be answered in the context of SU(2) matrices. QSP is the theory of QET for SU(2) matrices, or the \emph{unitary representation} of a scalar (real or complex) polynomial $P(x)$. 
Let $A=x\in[-1,1]$ be a scalar with a one-qubit Hermitian block encoding
\begin{equation}
U_A(x)=\begin{pmatrix}
x & \sqrt{1-x^2}\\
\sqrt{1-x^2} & -x
\end{pmatrix}.
\end{equation}
Then
\begin{equation}
O(x)=U_A(x)Z=\begin{pmatrix}
x & -\sqrt{1-x^2}\\
\sqrt{1-x^2} & x
\end{pmatrix}
\end{equation}
is a rotation matrix.

Similar to \cref{eqn:qet_hbe}, the QSP representation takes the following form
\begin{equation}\label{eqn:qsp_o_convention}
U(x)= e^{\I \phi_0 Z} O(x) e^{\I \phi_1 Z} O(x) \cdots e^{\I \phi_{d-1} Z} O(x) e^{\I \phi_{d} Z}.
\end{equation}
By setting $\phi_0=\cdots=\phi_d=0$, we immediately obtain the block encoding of the Chebyshev polynomial $T_d(x)$.
The representation power of this formulation is characterized by \cref{thm:qsp}, which is based on slight modification of \cite[Theorem 4]{GilyenSuLowEtAl2019}.
In the following discussion, even functions have parity $0$ and odd functions have parity $1$.

\begin{thm}[Quantum signal processing]\label{thm:qsp}
There exists a set of phase factors $\Phi := (\phi_0, \cdots, \phi_d) \in \RR^{d+1}$ such that
\begin{equation}
  U_{\Phi}(x) = e^{\I \phi_0 Z} \prod_{j=1}^{d} \left[ O(x) e^{\I \phi_j Z} \right] = \left( \begin{array}{cc}
    P(x) & -Q(x) \sqrt{1 - x^2}\\
    Q^*(x) \sqrt{1 - x^2} & P^*(x)
  \end{array} \right)
  \label{eqn:QSP_representation}
\end{equation}
if and only if $P,Q\in\CC[x]$ satisfy
\begin{enumerate}

\item $\deg(P) \leq d, \deg(Q) \leq d-1$,

\item $P$ has parity $d \bmod 2$ and $Q$ has parity $d-1 \bmod 2$, and

\item $|P(x)|^2 + (1-x^2) |Q(x)|^2 = 1, \forall x \in [-1, 1]$.
\end{enumerate}
Here $\deg Q=-1$ means $Q=0$.
\end{thm}
\begin{proof}

$\Rightarrow:$

Since both $e^{\I \phi Z}$ and $O(x)$ are unitary, the matrix $U_{\Phi}(x)$ is always a unitary matrix, which immediately implies the condition (3).
Below we only need to verify the conditions (1), (2).

When $d=0$, $U_{\Phi}(x)=e^{\I \phi_0 Z}$, which gives $P(x)=e^{\I \phi_0}$ and $Q=0$, satisfying all three conditions. 
For induction, suppose $U_{(\phi_0,\cdots,\phi_{d-1})}(x)$ takes the form in \cref{eqn:QSP_representation} with degree $d-1$, then
for any $\phi\in\RR$, we have
\begin{equation}
\begin{split}
U_{(\phi_0,\cdots,\phi_{d-1})}(x)=& \left( \begin{array}{cc}
        P(x) & -Q(x) \sqrt{1 - x^2}\\
        Q^*(x) \sqrt{1 - x^2} & P^*(x)
        \end{array} \right)\begin{pmatrix}
x & -\sqrt{1-x^2}\\
\sqrt{1-x^2} & x
\end{pmatrix}
\begin{pmatrix}
e^{\I \phi} & 0\\
0 & e^{-\I \phi}
\end{pmatrix}\\
=&\begin{pmatrix}
xP(x)-(1-x^2)Q(x) & -\sqrt{1-x^2}(P(x)+xQ(x))\\
-\sqrt{1-x^2}(P^*(x)+xQ^*(x)) & xP^*(x)-(1-x^2) Q^*(x)
\end{pmatrix}
\begin{pmatrix}
e^{\I \phi} & 0\\
0 & e^{-\I \phi}
\end{pmatrix}\\
=&\begin{pmatrix}
e^{\I \phi}(xP(x)-(1-x^2) Q(x)) & e^{-\I \phi}(-\sqrt{1-x^2}(P(x)+xQ(x)))\\
e^{\I \phi}((-\sqrt{1-x^2}(P^*(x)+xQ^*(x))) & e^{\I \phi}(xP^*(x)-(1-x^2) Q^*(x))
\end{pmatrix}.
\end{split}
\end{equation}
Therefore $U_{(\phi_0,\cdots,\phi_{d-1},\phi)}(x)$ satisfies conditions (1),(2).

$\Leftarrow:$

When $d=0$, the only possibility is $P(x)=e^{\I \phi_0}$ and $Q=0$, which satisfies \cref{eqn:QSP_representation}.

For $d>0$, when $d$ is even we may still have $\deg P=0$, i.e., $P(x)=e^{\I \phi_0}$ and $Q=0$. 
In this case, note that
\begin{equation}
O^{-1}(x)=O^{\dag}(x)=\begin{pmatrix}
x & \sqrt{1-x^2}\\
-\sqrt{1-x^2} & x
\end{pmatrix}=e^{-\I\frac{\pi}{2}Z}O(x)e^{+\I\frac{\pi}{2}Z},
\end{equation}
we may set $\phi_j=(-1)^{j}\frac{\pi}{2}, j=1,\ldots,d$, and 
\begin{equation}
e^{\I \phi_0 Z} \prod_{j=1}^{d} \left[ O(x) e^{\I \phi_j Z} \right]=e^{\I \phi_0 Z}(O^{\dag}(x)O(x))^{\frac{d}{2}}=e^{\I \phi_0 Z}.
\label{eqn:qsp_casezero}
\end{equation}
Thus the statement holds.

Now given $P,Q$ satisfying conditions (1)--(3), with $\deg P=\ell>0$, and $\ell\equiv d \pmod 2$. Then $\deg(\abs{P(x)}^2)=2\ell>0$, and according to the condition (3) we must have $\deg(Q)=\ell-1$.
Let $P,Q$ be expanded as
\begin{equation}
P(x)=\sum_{k=0}^{\ell} \alpha_k x^k, \quad Q(x)=\sum_{k=0}^{\ell-1} \beta_k x^k,
\end{equation}
then the leading term of $|P(x)|^2 + (1-x^2) |Q(x)|^2 $ is
\begin{equation}
\abs{\alpha_\ell}^2x^{2\ell}-x^2 \abs{\beta_{\ell-1}}^2 x^{2\ell-2}=(\abs{\alpha_\ell}^2-\abs{\beta_{\ell-1}}^2)x^{2\ell}=0,
\end{equation}
which implies $\abs{\alpha_{\ell}}=\abs{\beta_{\ell-1}}$.

For any $\phi\in \RR$, we have
\begin{equation}
\begin{split}
&\left( \begin{array}{cc}
        P(x) & -Q(x) \sqrt{1 - x^2}\\
        Q^*(x) \sqrt{1 - x^2} & P^*(x)
        \end{array} \right)
e^{-\I \phi Z}O^{\dag}(x)\\
=&\left( \begin{array}{cc}
        P(x) & -Q(x) \sqrt{1 - x^2}\\
        Q^*(x) \sqrt{1 - x^2} & P^*(x)
        \end{array} \right)
\begin{pmatrix}
e^{-\I \phi} & 0\\
0 & e^{\I \phi}
\end{pmatrix}\begin{pmatrix}
x & \sqrt{1-x^2}\\
-\sqrt{1-x^2} & x
\end{pmatrix}\\
=&\begin{pmatrix}
e^{-\I \phi}xP(x)+(1-x^2) Q(x) e^{\I \phi}& -\sqrt{1-x^2}(-e^{-\I \phi}P(x)+xQ(x)e^{\I \phi})\\
\sqrt{1-x^2}(-e^{\I \phi}P^*(x)+xQ^*(x)e^{-\I \phi}) & e^{\I \phi}xP^*(x)+(1-x^2)Q^*(x) e^{-\I \phi}
\end{pmatrix}\\
=:&  \left( \begin{array}{cc}
        \wt{P}(x) & -\wt{Q}(x) \sqrt{1 - x^2}\\
        \wt{Q}^*(x) \sqrt{1 - x^2} & \wt{P}^*(x)
        \end{array} \right).
\end{split}
\end{equation}
It may appear that $\deg{\wt{P}}=\ell+1$. However, by properly choosing $\phi$ we may obtain $\deg{\wt{P}}=\ell-1$. 
Let $e^{2\I \phi}=\alpha_{\ell}/\beta_{\ell-1}$. Then the coefficient of the $x^{\ell+1}$ term in $\wt{P}$ is 
\begin{equation}
e^{-\I \phi}\alpha_{\ell}-e^{\I \phi}\beta_{\ell-1}=0.
\end{equation}
Similarly, the coefficient of the $x^{\ell}$ term in $\wt{Q}$ is
\begin{equation}
-e^{-\I \phi}\alpha_{\ell}+e^{\I \phi}\beta_{\ell-1}=0.
\end{equation}
The coefficient of the $x^{\ell}$ term in $\wt{P}$, and the coefficient of the $x^{\ell-1}$ term in $\wt{Q}$ are both $0$ by the parity condition.
So we have 
\begin{enumerate}

\item $\deg(\wt{P}) \leq \ell-1\le d-1, \deg(Q) \leq \ell-2\le d-2$,

\item $\wt{P}$ has parity $d-1 \bmod 2$ and $\wt{Q}$ has parity $d-2 \bmod 2$, and

\item $|\wt{P}(x)|^2 + (1-x^2) |\wt{Q}(x)|^2 = 1, \forall x \in [-1, 1]$.
\end{enumerate}
Here the condition (3) is automatically satisfied due to unitarity.
The induction follows until $\ell=0$, and apply the argument in \cref{eqn:qsp_casezero} to represent the remaining constant phase factor if needed.
\end{proof}

\begin{rem}[$W$-convention of QSP]
\cite[Theorem 4]{GilyenSuLowEtAl2019} is stated slightly differently as
\begin{equation}
\begin{split}
        U_{\Phi^W}(x) &= e^{\I \phi^W_0 Z} \prod_{j=1}^{d} \left[ W(x) e^{\I \phi_j^W Z} \right] = \left( \begin{array}{cc}
        P(x) & \I Q(x) \sqrt{1 - x^2}\\
        \I Q^*(x) \sqrt{1 - x^2} & P^*(x)
        \end{array} \right)
\end{split},
\end{equation}
where
\begin{equation}
W(x)=e^{\I \arccos(x) X}=\begin{pmatrix}
x & \I\sqrt{1-x^2}\\
\I\sqrt{1-x^2} & x
\end{pmatrix}.
\end{equation}
This will be referred to as the $W$-convention. 
Correspondingly \cref{eqn:qsp_o_convention} will be referred to as the $O$-convention. The two conventions can be easily converted into one another, due to the relation
\begin{equation}
W(x)=e^{-\I\frac{\pi}{4}Z}O(x)e^{+\I\frac{\pi}{4}Z}.
\end{equation}
Correspondingly the relation between the phase angles using the $O$ and $W$ representations are related according to
\begin{equation}\label{eqn:phi_phi_W_convert}
\phi_j=\begin{cases}
\phi_0^W-\frac{\pi}{4},& j=0,\\
\phi_j^W, & j=1,\ldots,d-1,\\
\phi_d^W+\frac{\pi}{4},& j=d,\\
\end{cases}
\end{equation}
On the other hand, note that for any $\theta\in\RR$, $U_{\Phi}(x)$ and $e^{\I\theta Z} U_{\Phi}(x) e^{-\I\theta Z}$ both block encodes $P(x)$. Therefore WLOG we may as well take 
\begin{equation}
\Phi=\Phi^W.
\end{equation}

In many applications, we are only interested in $P\in\CC[x]$, and $Q\in\CC[x]$ is not provided \emph{a priori}. \cite[Theorem 4]{GilyenSuLowEtAl2018} states that under certain conditions $P$, the polynomial $Q$ can always be constructed. We omit the details here.
\end{rem}

\subsection{QSP for real polynomials}

Note that the normalization condition (3) in \cref{thm:qsp} imposes very strong constraints on the coefficients of $P,Q\in\CC[x]$.
If we are only interested in QSP for real polynomials, the conditions can be significantly relaxed.

\begin{thm}[Quantum signal processing for real polynomials]\label{thm:qsp_real}
Given a real polynomial $P_{\Re}(x)\in\RR[x]$, and $\deg P_{\Re}=d>0$, satisfying
\begin{enumerate}

\item $P_{\Re}$ has parity $d \bmod 2$,

\item $|P_{\Re}(x)|\le 1, \forall x \in [-1, 1]$,
\end{enumerate}
then there exists polynomials $P(x),Q(x)\in\CC[x]$ with $\Re P=P_{\Re}$ and a set of phase factors $\Phi := (\phi_0, \cdots, \phi_d) \in \RR^{d+1}$ such that the QSP representation \cref{eqn:QSP_representation} holds.
\end{thm}

Compared to \cref{thm:qsp}, the conditions in \cref{thm:qsp_real} is much easier to satisfy: given any polynomial $f(x)\in\RR[x]$ satisfying condition (1) on parity, we can always scale $f$ to satisfy the condition (2) on its magnitude.
Again the presentation is slightly modified compared to \cite[Corollary 5]{GilyenSuLowEtAl2019}.
We can now summarize the result of QET with real polynomials as follows.

\begin{cor}[Quantum eigenvalue transformation with real polynomials]
Let $A\in\CC^{N\times N}$ be encoded by its $(1,m)$-block-encoding $U_A$.
Given a polynomial $P_{\Re}(x)\in\RR[x]$ of degree $d$ satisfying the conditions in \cref{thm:qsp_real}, we can find a sequence of phase factors $\Phi\in\RR^{d+1}$, so that the circuit in \cref{fig:qet_circuit_general_real} denoted by $U_{\Phi}$ implements a $(1,m+1)$-block-encoding of $P_{\Re}(A)$.
$U_{\Phi}$ uses $U_A,U_A^{\dag}$, m-qubit controlled NOT, and single qubit rotation gates for $\Or(d)$ times.
\label{cor:qet_real}
\end{cor}

\begin{rem}[Relation between QSP representation and QET circuit]
Although $O(x)=U_A(x) Z$, we do not actually need to implement $Z$ separately in QET.
Note that $\I Z=e^{\I \frac{\pi}{2}Z}$, i.e., $Z e^{\I \phi Z}=(-\I ) e^{\I (\frac{\pi}{2}+\phi)Z}$, we obtain
\begin{equation}
U_{\Phi}(x)=
e^{\I \phi_0 Z} \prod_{j=1}^{d} \left[ O(x) e^{\I \phi_j Z} \right]=
(-\I)^d e^{\I \wt{\phi}_0 Z} \prod_{j=1}^{d} \left[ U_A(x) e^{\I \wt{\phi}_j Z} \right],
\label{eqn:modify_qsp}
\end{equation}
where $\wt{\phi}_0=\phi_0,\wt{\phi}_j=\phi_j+\pi/2,j=1,\ldots, d$.
For the purpose of block encoding $P(x)$, another equivalent, and more symmetric choice is 
\begin{equation}\label{eqn:symmetric_conversion_phi_phitilde}
\wt{\phi}_j=\begin{cases}
\phi_0+\frac{\pi}{4},& j=0,\\
\phi_j+\frac{\pi}{2}, & j=1,\ldots,d-1,\\
\phi_d+\frac{\pi}{4},& j=d.\\
\end{cases}
\end{equation}

When the phase factors are given in the $W$-convention, since we can perform a similarity transformation and take $\Phi=\Phi^W$, we can directly convert $\Phi^W$ to $\wt{\Phi}$ according to \cref{eqn:symmetric_conversion_phi_phitilde}, which is used in the QET circuit in \cref{fig:qet_circuit_general}.
\end{rem}

\begin{exam}[QSP for Chebyshev polynomial revisited]
In order to block encode the Chebyshev polynomial, we have
$\phi_j=0,j=0,\ldots,d$. This gives $\wt{\phi}_0=0$, $\wt{\phi}_j=\pi/2,j=1,\ldots,d$, and
\begin{equation}
U_{\Phi}(x)=[O(x)]^d=\begin{pmatrix}
T_d(x) & -\sqrt{1-x^2}U_{d-1}(x)\\
\sqrt{1-x^2}U_{d-1}(x) & T_d(x)
\end{pmatrix}=(-\I)^d \prod_{j=1}^{d} \left[ U_A(x) e^{\I \frac{\pi}{2} Z} \right].
\end{equation}
According to \cref{eqn:symmetric_conversion_phi_phitilde}, an equivalent symmetric choice for block encoding $T_d(x)$ is
\begin{equation}\wt{\phi}_j=\begin{cases}
\frac{\pi}{4},& j=0,\\
\frac{\pi}{2}, & j=1,\ldots,d-1,\\
\frac{\pi}{4},& j=d.\\
\end{cases}
\end{equation}
\end{exam}

\subsection{Optimization based method for finding phase factors}\label{sec:optimization_qsp}

QSP for real polynomials is the most useful version for many problems in scientific computation. 
Let us now summarize the problem of finding phase factors following the $W$-convention and identify $\Phi=\Phi^W$.

Given a target polynomial $f=P_{\Re}\in\RR[x]$ satisfying (1) $\deg(f)=d$, (2) the  parity of $f$ is $d \bmod 2$, (3) $\norm{f}_{\infty}:=\max_{x\in[-1,1]} \abs{f(x)}< 1$, we would like to find phase factors $\Phi:=(\phi_0,\cdots,\phi_d)\in [-\pi,\pi)^{d+1}$
so that
\begin{equation}\label{eqn:match_target}
f(x)=g(x,\Phi):=\Re[U(x,\Phi)_{11}], \quad x\in[-1,1],
\end{equation}
with
\begin{equation}\label{eqn:unitary-qsvt}
    U(x, \Phi) := e^{\I \phi_0 Z} W(x) e^{\I \phi_1 Z} W(x) \cdots e^{\I \phi_{d-1} Z} W(x) e^{\I \phi_d Z}.
\end{equation}
\cref{thm:qsp_real} shows the existence of the phase factors. Due to the parity constraint, the number of degrees of freedom in the target polynomial $f(x)$ is $\wt{d} := \lceil \frac{d+1}{2} \rceil$. Hence $f(x)$ is entirely determined by its values on $\wt{d}$ distinct points. 
Throughout the paper, we choose these points to be $x_k=\cos\left(\frac{2k-1}{4\wt{d}}\pi\right)$, $k=1,...,\wt{d}$, i.e., positive nodes of the Chebyshev polynomial $T_{2 \wt{d}}(x)$. The QSP problem can be equivalently solved via the following optimization problem
\begin{equation}
    \Phi^* = \argmin_{\substack{\Phi \in [-\pi,\pi)^{d+1}}} F(\Phi),\ F(\Phi) := \frac{1}{\wt{d}} \sum_{k=1}^{\wt{d}} \abs{g(x_k, \Phi) - f(x_k)}^2,
\end{equation}
i.e., any solution $\Phi^*$ to \cref{eqn:match_target} achieves the global minimum of the cost function with $F(\Phi^*)=0$, and vice versa.

However, note that the number of variables is larger than the number of equations and there should be an infinite number of global minima. \cite[Theorem 2]{DongMengWhaleyEtAl2021} shows that the existence of symmetric phase factors 
\begin{equation}
\Phi=(\phi_0,\phi_1,\phi_2,\ldots,\phi_2,\phi_1,\phi_0)\in [-\pi,\pi)^{d+1},
\label{eqn:symmetry_phase}
\end{equation}
Then the optimization problem is changed to
\begin{equation}\label{eqn:opt_symm_qsp}
    \Phi^* = \argmin_{\substack{\Phi \in [-\pi,\pi)^{d+1},\\
    \text{symmetric.}}} F(\Phi),\ F(\Phi) := \frac{1}{\wt{d}} \sum_{k=1}^{\wt{d}} \abs{g(x_k, \Phi) - f(x_k)}^2,
\end{equation}
This corresponds to choosing complementary polynomial $Q(x)\in\RR[x]$. 
With the symmetric constraint taken into account, the number of variables matches the number of constraints.

Unfortunately, the energy landscape of the cost function $F(\Phi)$ is very complex, and has numerous global as well as local minima. Starting from a random initial guess, an optimization algorithm can easily be trapped at a local minima already when $d$ is small.
It is therefore surprising that starting from a special symmetric initial  guess
\begin{equation}\label{eqn:phi0}
\Phi^0=(\pi/4,0,0,\ldots, 0,0,\pi/4),
\end{equation}
at least one global minimum can be robustly identified using standard unconstrained optimization algorithms even when $d$ is as large as $10,000$ using standard double precision arithmetic operations~\cite{DongMengWhaleyEtAl2021}, and the optimization method is observed to be free from being trapped by any local minima.
Direct calculation shows that $g(x,\Phi^0)=0$, and therefore $\Phi^0$ does not contain any \emph{a priori} information of the target polynomial $f(x)$.

This optimization based method is implemented in QSPPACK\footnote{\url{https://github.com/qsppack/QSPPACK}}.

\begin{rem}[Other methods for treating QSP with real polynomials]
The proof of \cite[Corollary 5]{GilyenSuLowEtAl2019} also gives a constructive algorithm for solving the QSP problem for real polynomials. 
Since $P_{\Re}=f\in\RR[x]$ is given, the idea is to first find \emph{complementary polynomials} $P_{\Im},Q\in\RR[x]$, so that the resulting $P(x)=P_{\Re}(x)+\I P_{\Im}(x)$ and $Q(x)$ satisfy the requirement in \cref{thm:qsp}.
Then the phase factors can be constructed following the recursion relation shown in the proof of \cref{thm:qsp}. 
We will not describe the details of the procedure here. 
It is worth noting that the method is not numerically stable. 
This is made more precise by \cite{Haah2019} that these algorithms require $\Or(d\log(d/\epsilon))$ bits of precision, where $d$ is the degree of $f(x)$ and $\epsilon$ is the target accuracy. 
It is worth mentioning that the extended precision needed in these algorithms is not an artifact of the proof technique. For instance, for $d\approx 500$, the number of bits needed to represent each floating point number can be as large as  $1000\sim 2000$.
In particular, such a task cannot be reliably performed using standard double precision arithmetic operations which only has $64$ bits.
\end{rem}

\subsection{A typical workflow preparing the circuit of QET}

Let us now use $f(x)=\cos(xt)$ as in the Hamiltonian simulation to demonstrate a typical workflow of QSP.
This function should be an even or odd function to satisfy the parity constraint (1) in \cref{thm:qsp_real}. 

\begin{enumerate}

\item Expand $f(x),x\in[-1,1]$ using a polynomial expansion (in this case, the Jacob--Anger expansion in \cref{eqn:jacobi_anger}), and truncate it to some finite order.

\item Scale the truncated polynomial by a suitable constant so that the resulting real polynomial $P_{\Re}(x)$ satisfies the maxnorm constraint (2) in \cref{thm:qsp_real}. 
\item Use the optimization based method to find phase factors $\Phi^W=\Phi$, and convert the result to $\wt{\Phi}$ according to the relation in \cref{eqn:symmetric_conversion_phi_phitilde}. $\wt{\Phi}$ can be directly used in the QET circuit in \cref{fig:qet_circuit_general_real,fig:qet_circuit_general}.
\end{enumerate}

\begin{rem}
When the function $f(x)$ of interest has singularity on $[-1,1]$, the function should first be mollified on a proper subinterval of interest, and then approximated by polynomials. A more streamlined method is to use the Remez exchange algorithm with parity constraint to directly approximate $f(x)$ on the subinterval.
We refer readers to \cite[Appendix E]{DongMengWhaleyEtAl2021} for more details.
\end{rem}

\section{Application: Time-independent Hamiltonian simulation}\label{sec:qsp_hamsim}

Using QET, let us now revisit the problem of time-independent Hamiltonian simulation problem $\mc{U}=e^{\I H t}$. Instead of Trotter splitting, we assume that we are given access to a block encoding $U_H\in\BE_{\alpha,m}(H)$.
Since $e^{\I H t}=e^{\I (H/\alpha) (\alpha t)}$, the subnormalization factor $\alpha$ can be factored into the simulation time $t$. So WLOG we assume $U_H\in\BE_{1,m}(H)$.
Since
\begin{equation}
\mc{U}=\cos(Ht)+\I \sin(Ht),
\end{equation} 
and that $\cos(xt),\sin(xt)$ are even and odd functions, respectively, we can construct a block encoding for the real and imaginary part directly using the circuit for real matrix polynomials in \cref{fig:qet_circuit_general_real}.

More specifically, we first use the Fourier--Chebyshev series of the trigonometric functions given by the Jacobi-Anger expansion $[-1,1]$:
\begin{equation}
\begin{split}
\cos (t x)=&J_{0}(t)+2 \sum_{k=1}^{\infty}(-1)^{k} J_{2 k}(t) T_{2 k}(x),\\
\sin (t x)=&2 \sum_{k=0}^{\infty}(-1)^{k} J_{2 k+1}(t) T_{2 k+1}(x).
\end{split}
\label{eqn:jacobi_anger}
\end{equation}
Here $J_{\nu}(t)$ denotes Bessel functions of the first kind.

This series converges very rapidly. With
\begin{equation}
r=\Theta\left(t+\frac{\log (1 / \epsilon)}{\log (e+\log (1 / \epsilon) / t)}\right)
\end{equation}
terms, the truncated Jacobi--Anger expansion with degree up to $d=2r+1$ can approximate $\cos (t x),\sin(tx)$ to precision $\epsilon/\sqrt{2}$, respectively~\cite[Corollary 32]{GilyenSuLowEtAl2019}.
Such a scaling also matches the complexity lower bound for Hamiltonian simulation~\cite{BerryAhokasCleveEtAl2007,LowChuang2017}.
Define
\begin{equation}\label{eqn:jacobi_anger_scale}
C_d(x)=\frac{1}{\beta}J_0(t)+\frac{2}{\beta} \sum_{k=1}^{r}(-1)^{k} J_{2 k}(t) T_{2 k}(x), \quad S_d(x)=\frac{2}{\beta} \sum_{k=0}^{r}(-1)^{k} J_{2 k+1}(t) T_{2 k+1}(x),
\end{equation} 
where $\beta>1$ is chosen so that $|C_d|,|S_d|\le 1$ on $[-1,1]$, and $\beta$ can be chosen to be as small as $1+\epsilon$. Also let $f_d(x)=C_d(x)+\I S_d(x)$, then 
\begin{equation}
\max_{x\in[-1,1]}\abs{\beta f_d(x)-e^{\I tx}}\le \epsilon.
\end{equation}

\cref{thm:qsp_real} guarantees the existence of phase factors $\wt{\Phi}_C,\wt{\Phi}_S$, using which we can construct $\mc{U}_C\in\BE_{1,m+1}(C_d(H))$ and $\mc{U}_s\in\BE_{1,m+1}(S_d(H))$.
Finally, we can use one more ancilla qubit and LCU in \cref{eqn:qet_general_polynomial} to construct a block encoding $\mc{U}_d\in\BE_{2,m+2}(f_d(H))$, or  $\mc{U}_d\in\BE_{2\beta,m+2}(\mc{U},\epsilon)$.
The circuit depth is $\Or\left(t + \log(1/\epsilon)\right)$.

As an example, \cref{fig:qsp_HS_maxorder24} shows the QSP representation of the quality of approximating $\cos(tx)$ using $P_{\Re}(x)=C_d(x)$ with $t=4\pi,\beta=1.001,d=24$. The quality of the approximation can be significantly improved with a larger degree $d=50$ (see \cref{fig:qsp_HS_maxorder50}).
The phase factors are obtained via QSPPACK.

\begin{figure}[H]
\begin{center}
\includegraphics[width=0.31\textwidth]{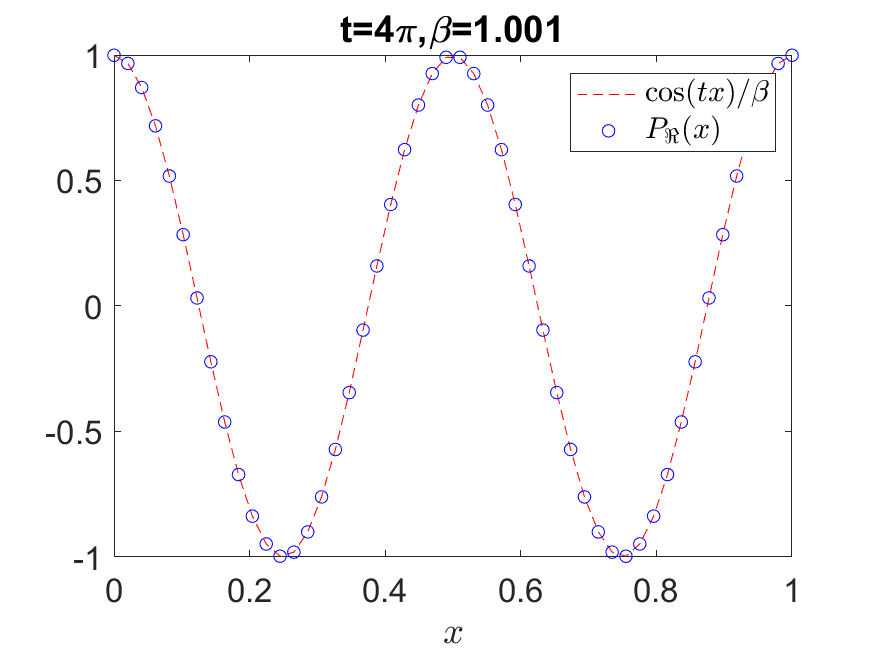}
\includegraphics[width=0.31\textwidth]{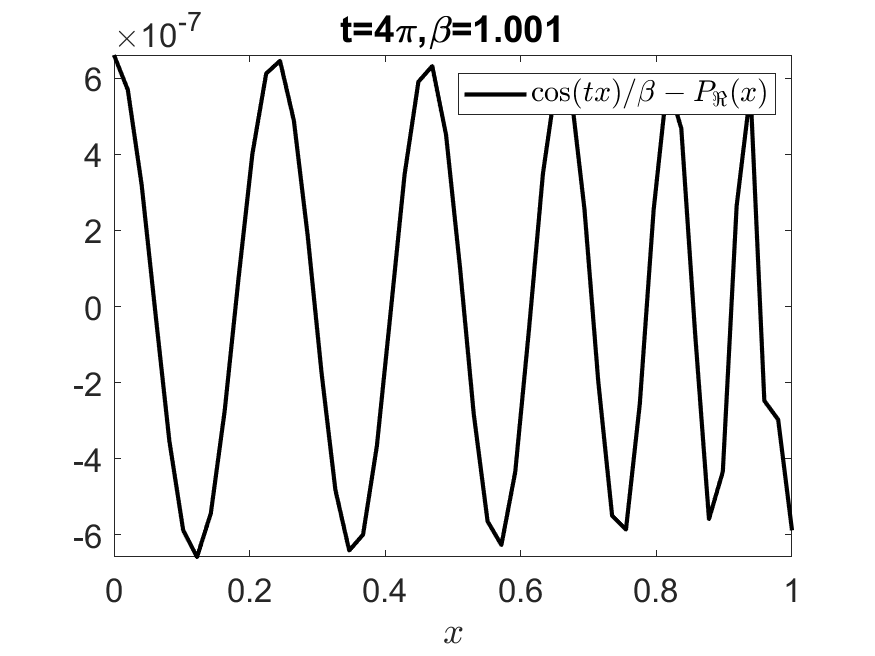}
\includegraphics[width=0.31\textwidth]{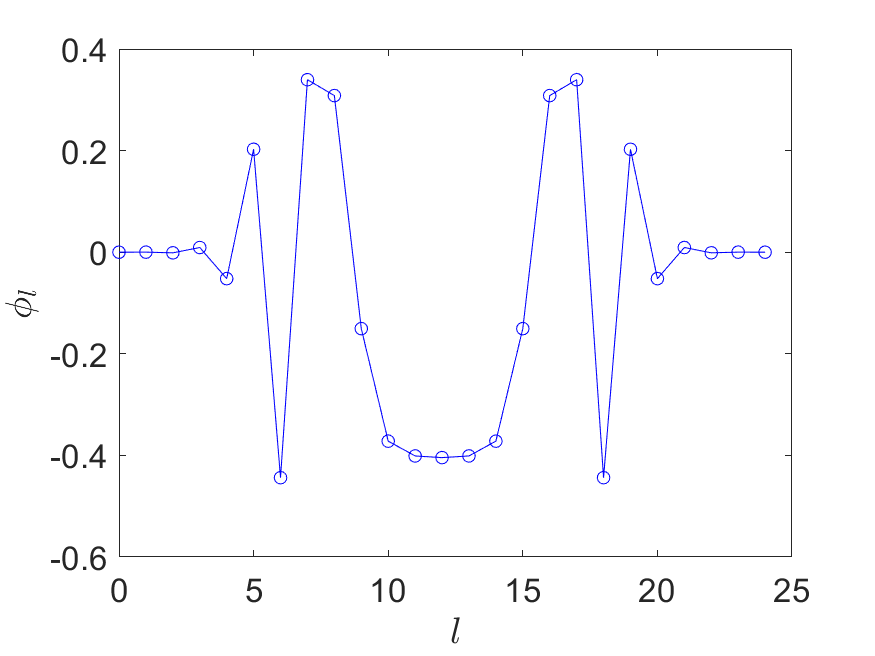}
\end{center}
\caption{QSP representation of $\cos(tx)/\beta$ with $t=4\pi,\beta=1.001,d=24$. The phase factors plotted removes a factor of $\pi/4$ on both ends (see \cref{eqn:phi0}).}
\label{fig:qsp_HS_maxorder24}
\end{figure}

\begin{figure}[H]
\begin{center}
\includegraphics[width=0.4\textwidth]{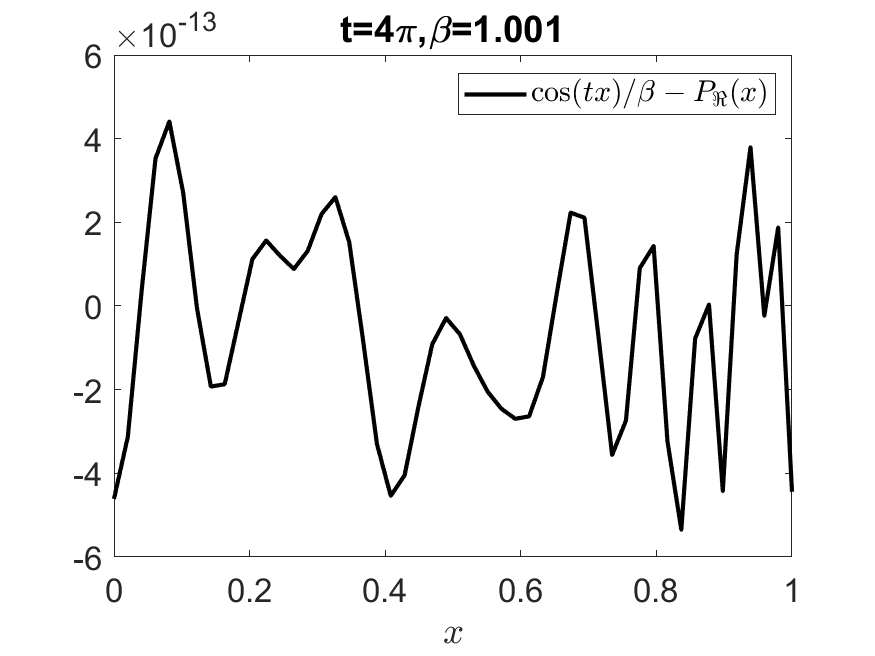}
\includegraphics[width=0.4\textwidth]{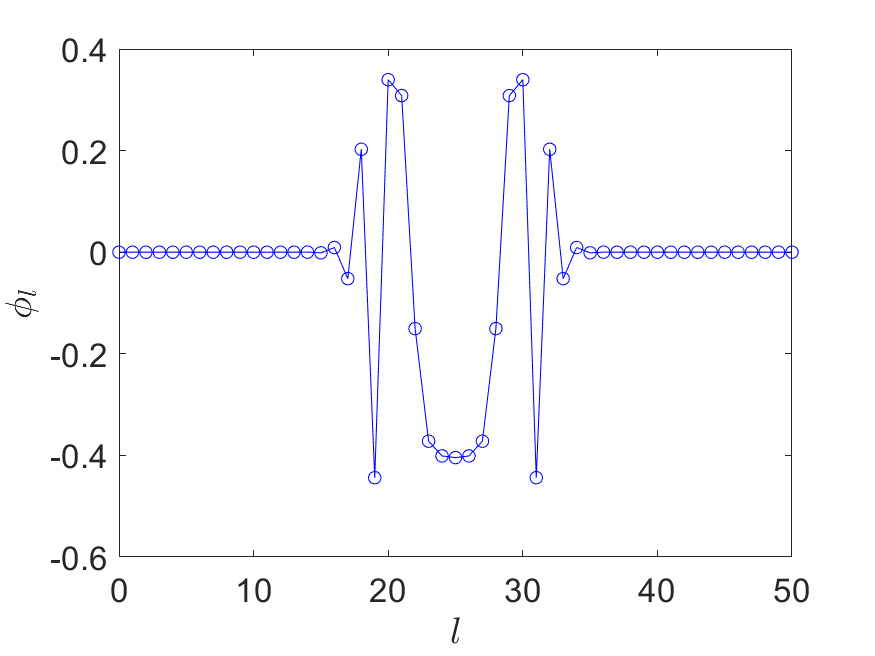}
\end{center}
\caption{Error of the QSP representation of $\cos(tx)/\beta$ with $t=4\pi,\beta=1.001,d=50$. The phase factors plotted removes a factor of $\pi/4$ on both ends  (see \cref{eqn:phi0}).}
\label{fig:qsp_HS_maxorder50}
\end{figure}

\section{Application: Ground state preparation}\label{sec:qsp_groundstate}
Given a block encoding $U_H\in\BE_{1,m}(H)$, and WLOG assume $0\preceq H\preceq 1$. We assume that we are provided an initial state $\ket{\varphi}$ so that the initial overlap $p_0=\abs{\braket{\varphi|\psi_0}}^2$ is not small.
To simplify the problem we also assume that ground and the first excited state energies $E_0,E_1$ are known with a positive gap $\Delta:=E_1-E_0>0$. 
Our goal is to use QET to prepare an approximate quantum state $\ket{\psi}\approx \ket{\psi_0}$.

To this end, we can construct a threshold polynomial approximating the sign function on $[0,1]$. 
Due to the assumption that $0\preceq H\preceq 1$, this function can be chosen to be an even function. 
Since the sign function is a discontinuous, the polynomial should only aim at approximating the sign function outside $(\mu-\Delta/2,\mu+\Delta/2)$, where $\mu=(E_0+E_1)/2$. 
The polynomial also needs to satisfy the conditions in \cref{thm:qsp_real}.
We need to find an even polynomial $P_{\Re}(x)$ satisfying
\begin{equation}
\abs{P_{\Re}(x)-1}\le \epsilon, \quad \forall x\in [0,\mu-\Delta/2]; \quad \abs{P_{\Re}(x)}\le \epsilon, \quad \forall x\in [\mu+\Delta/2,1].
\end{equation}
We can achieve this by approximating the sign function, with $\deg(P_{\Re})=\Or(\log(1/\epsilon) \Delta^{-1})$. (see e.g.~\cite[Corollary 7]{LowChuang2017a} and \cite[Corollary 16]{GilyenSuLowEtAl2019}). 
This construction is based on an approximation to the $\mathrm{erf}$ function.
\cref{fig:qsp_step_deg80} gives a concrete construction of the even polynomial obtained via numerical optimization, and the phase factors are obtained via QSPPACK.

\begin{figure}[H]
\begin{center}
\includegraphics[width=0.31\textwidth]{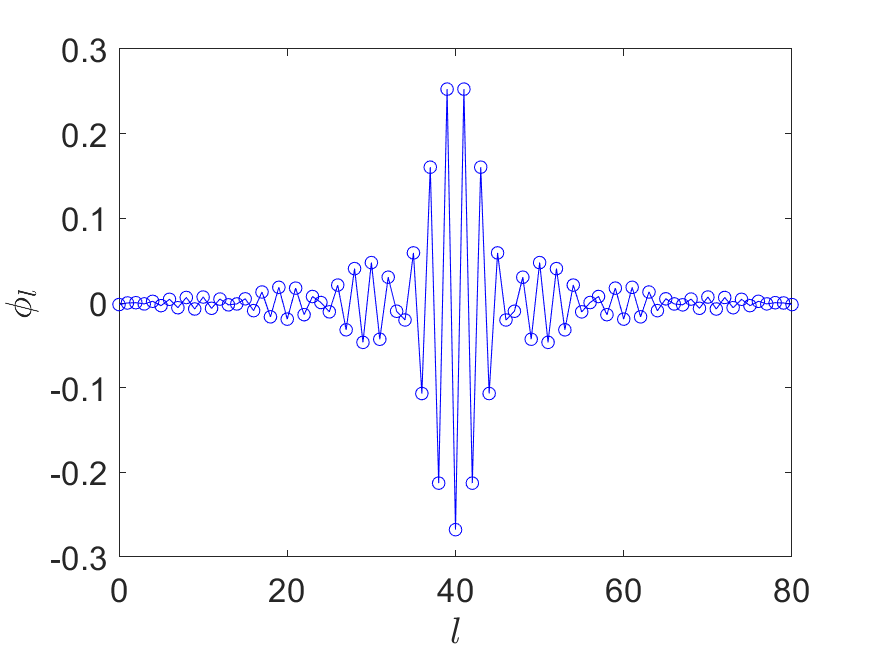}
\includegraphics[width=0.31\textwidth]{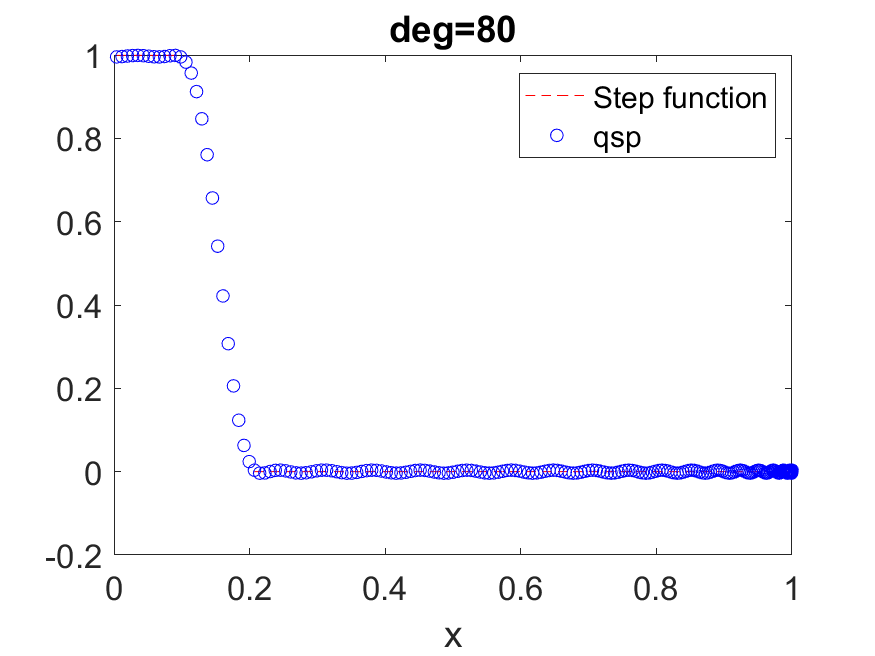}
\includegraphics[width=0.31\textwidth]{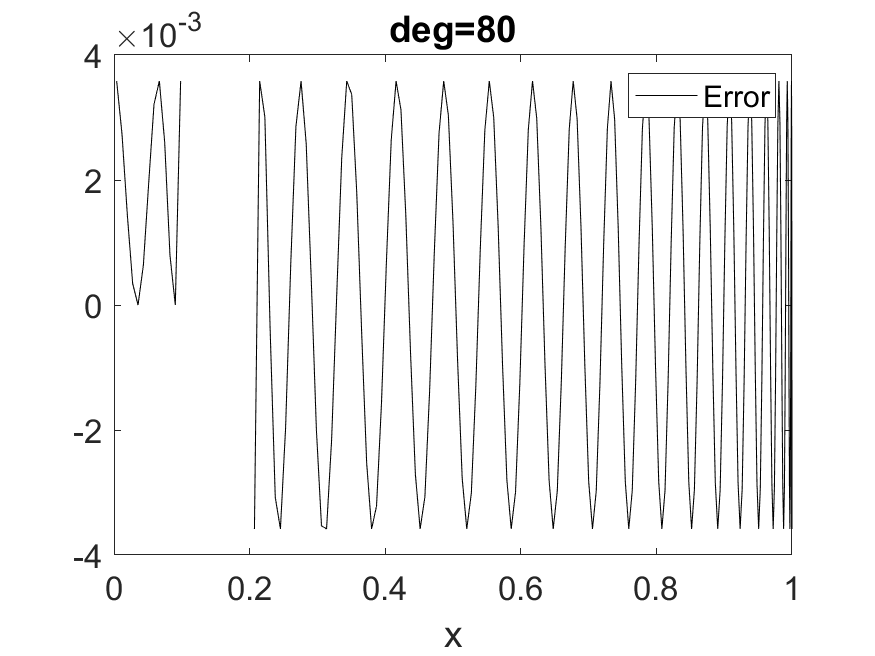}
\end{center}
\caption{QSP representation for approximating a step function using an even polynomial on $[0,0.1]\cup[0.2,1]$. The phase factors plotted removes a factor of $\pi/4$ on both ends (see \cref{eqn:phi0}).}
\label{fig:qsp_step_deg80}
\end{figure}

Applying the circuit $U_{\Phi}$ to the initial state $\ket{\varphi}$, we have
\begin{equation}
U_{\Phi}\ket{0^m}\ket{\varphi}=\sqrt{p_0} \ket{0^m} P_{\Re}(H) \ket{\psi_0}+\sqrt{1-p_0}\ket{0^m} P_{\Re}(H) \ket{\psi_{\perp}} + \ket{\perp}.
\end{equation}
Here $\ket{\psi_{\perp}}$ is a state in the system register orthogonal to $\ket{\psi_0}$, while $\ket{\perp}$ is orthogonal to all states $\ket{0^m}\ket{\psi}$. Note that 
\begin{equation}
\norm{P_{\Re}(H) \ket{\psi_{\perp}}}\le \epsilon, \quad \norm{P_{\Re}(H) \ket{\psi_0}}\ge 1-\epsilon,
\end{equation}
Therefore if we measure the ancilla qubits, the success probability of obtaining $0^m$ in the ancilla qubits, and the ground state $\ket{\psi_0}$ in the system register is at least $p_0 (1-\epsilon)$. So the total number of queries to to $U_A$ and $U_A^{\dag}$ is $\Or(\Delta^{-1}p_0^{-1}\log(1/\epsilon) )$.

Using amplitude amplification, the number of repetitions can be reduced to $\Or(\gamma^{-1})$, and the total number of queries to to $U_A$ and $U_A^{\dag}$ becomes $\Or(\Delta^{-1} p_0^{-\frac12}\log(1/\epsilon) )$. This also matches the lower bound~\cite{LinTong2020a}.

Once the ground state is prepared, we can estimate the ground state energy by measuring the expectation value $\braket{\psi_0|H|\psi_0}$. The number of samples needed is $\Or(1/\epsilon^2)$, which can be reduced to $\Or(1/\epsilon)$ using amplitude amplification. In summary, the best complexity for estimating the ground state energy $E_0$ to accuracy $\epsilon$ is $\Or(\Delta^{-1}p_0^{-\frac12}\epsilon^{-1}\log(1/\epsilon) )$. Note that the cost of estimating the ground state energy depends on the gap $\Delta$. 
This is because the algorithm first prepares the ground state and then estimates the ground state energy. 
If we are only interested in estimating $E_0$ to precision $\epsilon$, the gap dependence is not necessary (see~\cref{sec:groundenergy} as well as ~\cite{LinTong2020a}).

%\LL{Give examples and discuss the details later.} 
%
%
%Assume 
%Without loss of generality we assume , 
%Assume 
%Polynomial. QSPpack.

%\subsection{Quantum eigenvalue transformation viewed from the cosine-sine transformation*}

\vspace{2em}

\begin{exer}
Let $A,B$ be two $n$-qubit matrices. Construct a circuit to block encode $C=A+B$ with $U_A\in\BE_{\alpha_A,m}(A),U_B\in\BE_{\alpha_B,m}(B)$. 
\end{exer}

\begin{exer}
Use LCU to construct a block encoding of the TFIM model with periodic boundary conditions in \cref{eqn:ham_tfim}, with $g\ne 1$. 
\end{exer}

\begin{exer}
Prove \cref{prop:discriminant_reversiblewalk}.
\end{exer}

\begin{exer}
Let $A$ be an $n$-qubit Hermitian matrix. 
Write down the circuit for $U_{P_{\Re}(A)}\in\BE_{1,m+1}(P_{\Re}(A))$ with a block encoding $U_A\in\BE_{1,m}(A)$, where $P$ is characterized by the phase sequence $\wt{\Phi}$ specified in \cref{thm:qsp_simple}.
\end{exer}

\begin{exer}
Write down the circuit for LCU of Hamiltonian simulation.
\end{exer}

\begin{exer}
Using QET to prepare the Gibbs state.
\end{exer}

\chapter{Quantum singular value transformation}\label{chap:qsvt}

In \cref{chap:hermfunc}, we have found that using qubitization, we can effectively block encode the Chebyshev matrix polynomial $T_k(A)$ for a Hermitian matrix $A$.
Combined with LCU, we can construct a block encoding of any matrix polynomial of $A$.
The process is greatly simplified using QSP and QET, which allows the implementation a general class of matrix functions for Hermitian matrices.

In this section, we generalize the results of qubitization and QET to general non-Hermitian matrices. 
This is called the quantum singular value transformation (QSVT).
Throughout the chapter we assume $A\in\CC^{N\times N}$ is a square matrix.
QSVT is applicable to non-square matrices as well, and we will omit the discussions here.

\section{Generalized matrix functions}

For a square matrix $A \in \mathbb{C}^{N \times N}$, where for simplicity we assume $N=2^n$ for some positive integer $n$, the singular value decomposition (SVD) of the normalized matrix $A$ can be written as
\begin{equation}
A=W\Sigma V^{\dag},
\label{eqn:A_SVD}
\end{equation}
or equivalently
\begin{equation}
A\left|v_{i}\right\rangle= \sigma_{i}\left|w_{i}\right\rangle, \quad A^{\dagger}\left|w_{i}\right\rangle= \sigma_{i}\left|v_{i}\right\rangle, \quad i\in [N].
\label{eqn:A_svd_component}
\end{equation}
We may apply a function $f(\cdot)$ on its singular values and define the generalized matrix function \cite{HawkinsBen-Israel1973,ArrigoBenziFenu2016} as below.

\begin{defn}[Generalized matrix function {\cite[Definition 4]{ArrigoBenziFenu2016}}]
\label{def:gen_matrix_function}
Given $A\in\CC^{N\times N}$ with singular value decomposition \cref{eqn:A_SVD}, and let $f: \mathbb{R} \rightarrow \mathbb{\CC}$ be a scalar function such that $f\left(\sigma_{i}\right)$ is defined for all $i\in[N]$. The generalized matrix function is defined as
\begin{equation}
f^{\diamond}(A):=Wf(\Sigma)V^{\dag},
\label{eqn:gen_matrix_function}
\end{equation}
where 
$$
f\left(\Sigma\right)=\operatorname{diag}\left(f\left(\sigma_{0}\right), f\left(\sigma_{1}\right), \ldots, f\left(\sigma_{N-1}\right)\right).
$$
\end{defn}

Given the form in \cref{eqn:A_svd_component}, we also define two other types of generalized matrix functions.

\begin{defn}Under the conditions in \cref{def:gen_matrix_function}, the left and right generalized matrix function are defined respectively as
\begin{equation}
f^{\triangleleft}(A):=Wf(\Sigma)W^{\dag}, \quad 
f^{\triangleright}(A):=Vf(\Sigma)V^{\dag}.
\label{eqn:gen_matrix_function_lr}
\end{equation}
\label{def:gen_matrix_function_lr}
\end{defn}

Here the left and right pointing triangles reflects that the transformation only keeps the left and right singular vectors, respectively.
For a given $A$, somewhat confusingly, in the discussion below, the transformation $f^{\triangleright}(A),f^{\triangleleft}(A),f^{\diamond}(A)$ will all be referred to as \emph{singular value transformations}.
In particular, QSVT mainly concerns $f^{\triangleright}(A),f^{\diamond}(A)$.

\begin{prop}
The following relations hold:
\begin{equation}
f^{\diamond}(A^{\dag})=(f^{\diamond}(A))^{\dag}, \quad f^{\triangleright}(A)=f^{\triangleleft}(A^{\dag}),
\end{equation}
and
\begin{equation}
f^{\triangleright}(A)=f^{\diamond}(\sqrt{A^{\dag}A})=f(\sqrt{A^{\dag}A}), \quad
f^{\triangleleft}(A)=f^{\diamond}(\sqrt{A A^{\dag}})=f(\sqrt{A A^{\dag}}). \end{equation}
\label{prop:gen_mat_func_relation}
\end{prop}
\begin{proof}
Just note that $A^{\dag}A=V\Sigma^2 V^{\dag}$, we have $\sqrt{A^{\dag}A}=V\Sigma V^{\dag}$. So the eigenvalue and singular value decomposition coincide for both $\sqrt{A^{\dag}A}$ and $\sqrt{A A^{\dag}}$.
\end{proof}

WLOG we assume access to $U_A\in\BE_{1,m}(A)$, so that the singular values of $A$ are in $[0,1]$, i.e.,
\begin{equation}
0\le \sigma_i \le 1, \quad i\in[N].
\end{equation}

\section{Qubitization of general matrices}\label{sec:qubitize_general}

In \cref{sec:qubitize_genbe} we have observed that when $A$ is a Hermitian matrix, the qubitization procedure introduces two different subspaces $\mc{H}_i$ and $\mc{H}_i'$ associated with each eigenvector $\ket{v_i}$.
In particular, $U_A$ maps  $\mc{H}_i$ to $\mc{H}_i'$, and $U_A^{\dag}$ maps $\mc{H}_i'$ to $\mc{H}_i$. 
Furthermore, both $\mc{H}_i$ and $\mc{H}_i'$ are the invariant subspaces of the projection operator $\Pi$.
Therefore $\mc{H}_i$ is an invariant subspace of $U_A^{\dag} f(\Pi) U_A$ for any function $f$. 
Much of the same structure can be carried out to the quantum singular value transformation.
The only difference is that the qubitization is now defined with respect to the singular vectors. 
The procedure below almost entirely parallelizes that of \cref{sec:qubitize_genbe}, except that we need to work with the singular value decomposition instead of the eigenvalue decomposition.

Start from the SVD in \cref{eqn:A_svd_component}, we apply $U_A$ to $\ket{0^m}\ket{v_i}$ and obtain 
\begin{equation}
U_A\ket{0^m}\ket{v_i}=\sigma_i\ket{0^m}\ket{w_i}+\sqrt{1-\sigma_i^2}\ket{\perp_i'},
\label{eqn:UA_apply_vi_2_qsvt}
\end{equation}
where $\ket{\perp_i'}$ is a normalized state satisfying $\Pi\ket{\perp_i'}=0$.

Since $U_A$ block encodes a matrix $A$, we have
\begin{equation}
U_A^{\dag}=\begin{pmatrix}
{A}^{\dag} & {*} \\
{*} & {*}
\end{pmatrix},
\end{equation}
which implies that there exists another normalized state $\ket{\perp_i}$ satisfying $\Pi\ket{\perp_i}=0$ and
\begin{equation}
U^{\dag}_A\ket{0^m}\ket{w_i}=\sigma_i\ket{0^m}\ket{v_i}+\sqrt{1-\sigma_i^2}\ket{\perp_i}.
\label{eqn:Udag_apply_vi_qsvt}
\end{equation}
Now apply $U_A$ to both sides of \cref{eqn:Udag_apply_vi_qsvt}, we obtain
\begin{equation}
\ket{0^m}\ket{w_i}=\sigma^2_i\ket{0^m}\ket{w_i}+\sigma_i\sqrt{1-\sigma_i^2}\ket{\perp_i'} +\sqrt{1-\sigma_i^2}U_A\ket{\perp_i},
\end{equation}
which gives
\begin{equation}
U_A\ket{\perp_i}=\sqrt{1-\sigma_i^2}\ket{0^m}\ket{w_i}-\sigma_i\ket{\perp_i'}.
\label{eqn:UA_apply_perpi_2_qsvt}
\end{equation}
Define
\begin{equation}
\mc{B}_i=\{\ket{0^m}\ket{v_i},\ket{\perp_i}\}, \quad \mc{B}'_i=\{\ket{0^m}\ket{w_i},\ket{\perp_i'}\},
\end{equation}
and the associated two-dimensional subspaces $\mc{H}_i=\opr{span}{B_i}, \mc{H}'_i=\opr{span}{B_i'}$, we find that $U_A$ maps $\mc{H}_i$ to $\mc{H}_i'$.
Correspondingly $U_A^{\dag}$ maps $\mc{H}_i'$ to $\mc{H}_i$.

Then \cref{eqn:UA_apply_vi_2_qsvt,eqn:UA_apply_perpi_2_qsvt} give the matrix representation
\begin{equation}
[U_A]_{\mc{B}_i}^{\mc{B}_i'}=\begin{pmatrix}
\sigma_i & \sqrt{1-\sigma_i^2}\\
\sqrt{1-\sigma_i^2} & -\sigma_i
\end{pmatrix}.
\end{equation}
Similar calculation shows that
\begin{equation}
[U^{\dag}_A]_{\mc{B}_i'}^{\mc{B}_i}=\begin{pmatrix}
\sigma_i & \sqrt{1-\sigma_i^2}\\
\sqrt{1-\sigma_i^2} & -\sigma_i
\end{pmatrix}.
\end{equation}
Meanwhile both $\mc{H}_i$ and $\mc{H}_i'$ are the invariant subspaces of the projector $\Pi$, with matrix representation
\begin{equation}
[\Pi]_{\mc{B}_i}=[\Pi]_{\mc{B}_i'}=\begin{pmatrix}
1 & 0 \\
0 & 0
\end{pmatrix}.
\end{equation}
Therefore
\begin{equation}
[Z_\Pi]_{\mc{B}_i}=[Z_\Pi]_{\mc{B}_i'}=\begin{pmatrix}
1 & 0 \\
0 & -1
\end{pmatrix}.
\end{equation}
Hence $\mc{H}_i$ is an invariant subspace of $\wt{O}=U_A^{\dag}Z_{\Pi}U_A Z_{\Pi}$, with matrix representation
\begin{equation}
[\wt{O}]_{\mc{B}_i}=\begin{pmatrix}
\sigma_i & -\sqrt{1-\sigma_i^2}\\
\sqrt{1-\sigma_i^2} & \sigma_i
\end{pmatrix}^2.
\end{equation}
Repeating $k$ times, we have
\begin{equation}
\begin{split}
[\wt{O}^k]_{\mc{B}_i}=&(U_A^{\dag}Z_{\Pi}U_A Z_{\Pi})^{k}=\begin{pmatrix}
\sigma_i & -\sqrt{1-\sigma_i^2}\\
\sqrt{1-\sigma_i^2} & \sigma_i
\end{pmatrix}^{2k}\\
=&\begin{pmatrix}
T_{2k}(\sigma_i) & -\sqrt{1-\sigma_i^2}U_{2k-1}(\sigma_i)\\
\sqrt{1-\sigma_i^2}U_{2k-1}(\sigma_i) & T_{2k}(\sigma_i)
\end{pmatrix}.
\end{split}
\end{equation}
In other words, 
\begin{equation}
\wt{O}^k=\begin{pmatrix}
\sum_{i} v_i T_{2k}(\sigma_i) v_i^{\dag} & *\\
* & *
\end{pmatrix}
=\begin{pmatrix}
T^{\triangleright}_{2k}(A) & *\\
* & *
\end{pmatrix}.
\end{equation}
Therefore the circuit $(U_A^{\dag}Z_{\Pi}U_A Z_{\Pi})^{k}$ gives $(1,m)$-block-encoding of $T_{2k}^{\triangleright}(A)$.

Similarly,
\begin{equation}
[U_A Z_{\Pi}(U_A^{\dag}Z_{\Pi}U_A Z_{\Pi})^{k}]_{\mc{B}_i}^{\mc{B}_{i}'}=\begin{pmatrix}
T_{2k+1}(\sigma_i) & -\sqrt{1-\sigma_i^2}U_{2k}(\sigma_i)\\
\sqrt{1-\sigma_i^2}U_{2k}(\sigma_i) & T_{2k+1}(\sigma_i)
\end{pmatrix}.
\end{equation}
In other words, \begin{equation}
U_A Z_{\Pi}(U_A^{\dag}Z_{\Pi}U_A Z_{\Pi})^{k}=\begin{pmatrix}
\sum_{i} w_i T_{2k+1}(\sigma_i) v_i^{\dag} & *\\
* & *
\end{pmatrix}
=\begin{pmatrix}
T^{\diamond}_{2k+1}(A) & *\\
* & *
\end{pmatrix}.
\end{equation}
Therefore the circuit $U_A Z_{\Pi}(U_A^{\dag}Z_{\Pi}U_A Z_{\Pi})^{k}$ gives $(1,m)$-block-encoding of $T_{2k+1}^{\diamond}(A)$. 

\begin{rem}
By approximating any continuous function $f$ using polynomials, and using the LCU lemma, we can approximately evaluate $f^{\diamond}(A)$ for any odd function $f$, and $f^{\triangleright}(A)$ for any even function $f$.
This may seem somewhat restrictive.
However, note that all singular values are non-negative.
Hence when performing the polynomial approximation, if we are interested in $f^{\diamond}(A)$, we can always use first perform a polynomial approximation of an \emph{odd extension} of $f$, i.e.,
\begin{equation}
g(x)=\begin{cases}
f(x), & x>0,\\
0, & x=0, \\
-f(-x), & x<0,
\end{cases}
\end{equation}
and then evaluate $g^{\diamond}(A)$.
Similarly, if we are interested in $f^{\triangleright}(A)$ for a general $f$, we can perform polynomial approximation to its \emph{even extension}
\begin{equation}
g(x)=\begin{cases}
f(x), & x\ge0,\\
f(-x), & x<0,
\end{cases}
\end{equation}
and evaluate $g^{\triangleright}(A)$.
\end{rem}

\section{Quantum singular value transformation}\label{sec:qsvt}

\subsection{Quantum circuit}

Due to the close relation between eigenvalue and singular value transformation in terms of Chebyshev polynomials and qubitization in \cref{sec:qubitize_general}, we can obtain the QSVT circuit easily following the discussion in \cref{sec:qsp}.

First, there are no changes to the scalar case of QSP (in terms of SU(2) matrices), and in particular \cref{thm:qsp} and \cref{thm:qsp_real}.

For the matrix case, when $d$ is even,
\begin{equation}
U_{\Phi}=(-\I)^d e^{\I \wt{\phi}_0 Z_{\Pi}} 
\prod_{j=1}^{d/2}\left[ U^{\dag}_A e^{\I \wt{\phi}_{2j-1} Z_{\Pi}}U_A e^{\I \wt{\phi}_{2j} Z_{\Pi}} \right]
\end{equation}
gives a $(1,m+1)$-block-encoding of $P^{\triangleright}(A)$ for some even polynomial $P\in\CC[x]$.

When $d$ is odd, 
\begin{equation}
U_{\Phi}=(-\I)^d e^{\I \wt{\phi}_0 Z_{\Pi}} (U_A e^{\I \wt{\phi}_{1} Z_{\Pi}})
\prod_{j=1}^{(d-1)/2}\left[ U^{\dag}_A e^{\I \wt{\phi}_{2j} Z_{\Pi}}U_A e^{\I \wt{\phi}_{2j+1} Z_{\Pi}} \right]
\end{equation}
gives a $(1,m+1)$-block-encoding of $P^{\diamond}(A)$ for some odd polynomial $P\in\CC[x]$.

The quantum circuit is exactly the same as that in \cref{fig:qet_circuit_general}.
The phase factors can be adjusted so that all polynomials $P$ satisfying the conditions in \cref{thm:qsp} can be exactly represented.
If we are only interested in some real polynomial $P_{\Re}\in\RR[x]$ and $P_{\Re}^{\diamond}(A)$ (odd) and $P^{\triangleright}(A)$ (even), we can use \cref{thm:qsp_real} and the circuit in \cref{fig:qsvt_circuit_real} (which is simply a combination of \cref{fig:qet_circuit_general_real,fig:qet_circuit_general}) to implement its $(1,m+1)$-block-encoding.
We have the following theorem.
Since the conditions of QSP representation for real polynomials is simple to satisfy and is also most useful in practice, we only state the case with real polynomials.

\begin{thm}[Quantum singular value transformation with real polynomials]
Let $A\in\CC^{N\times N}$ be encoded by its $(1,m)$-block-encoding $U_A$.
Given a polynomial $P_{\Re}(x)\in\RR[x]$ of degree $d$ satisfying the conditions in \cref{thm:qsp_real}, we can find a sequence of phase factors $\Phi\in\RR^{d+1}$, so that the circuit in \cref{fig:qsvt_circuit_real} denoted by $U_{\Phi}$ implements a $(1,m+1)$-block-encoding of $P^{\diamond}_{\Re}(A)$ if $d$ is odd, and of $P^{\triangleright}_{\Re}(A)$ if $d$ is even.
$U_{\Phi}$ uses $U_A,U_A^{\dag}$, m-qubit controlled NOT, and single qubit rotation gates for $\Or(d)$ times.
\label{thm:qsvt_real}
\end{thm}

\begin{figure}[H]
  \begin{quantikz}
    \lstick{$\ket{0}$} & \gate{H}& \gate[2]{\opr{CR}_{\wt{\phi}_d}} & \qw & \gate[2]{\opr{CR}_{\wt{\phi}_{d-1}}} & \qw & \qw \raisebox{0em}{$\cdots$}&\qw& \gate[2]{\opr{CR}_{\wt{\phi}_0}}&\qw&\gate{H} &\qw\\
    \lstick{$\ket{0^m}$} &\qw&\qw& \gate[2]{U_A} &  \qw  & \gate[2]{U^{\dag}_A} &\qw\raisebox{0em}{$\cdots$} &\gate[2]{U_A}&\qw&\qw&\qw&\qw\\
    \lstick{$\ket{\psi}$}&\qw& \qw& \qw& \qw& \qw& \qw\raisebox{0em}{$\cdots$}&\qw&\qw&\qw&\qw&\qw
  \end{quantikz}
  \caption{Circuit of quantum singular value transformation to construct $U_{P^{\diamond}_{\Re}}\in\BE_{1,m+1}(P^{\diamond}_{\Re}(A))$, using $U_A\in\BE_{1,m}(A)$.
  Here $U_A,U_A^{\dag}$ should be applied alternately.  When $d$ is even, the last $U_A$ gate should be replaced $U_A^{\dag}$, and the circuit constructs $U_{P^{\triangleright}_{\Re}}\in\BE_{1,m+1}(P^{\triangleright}_{\Re}(A))$. This is simply a combination of \cref{fig:qet_circuit_general_real,fig:qet_circuit_general}.}
  \label{fig:qsvt_circuit_real}
\end{figure}

\subsection{QSVT applied to Hermitian matrices}

When $A$ is a Hermitian matrix, the quantum circuit for QET and QSVT are the same.
This means that the eigenvalue transformation and the singular value transformation are merely two different perspectives of the same object.

For a Hermitian matrix $A$, the eigenvalue decomposition and the singular value decomposition are connected as
\begin{equation}
A=\sum_{i} \ket{v_i}\lambda_i \bra{v_i}=\sum_{i} \ket{\opr{sgn}(\lambda_i)v_i} \abs{\lambda_i} \bra{v_i}:=\sum_{i} \ket{w_i} \sigma_i \bra{v_i}.
\end{equation}
Here 
\begin{equation}
\ket{w_i}=\ket{\opr{sgn}(\lambda_i)v_i}, \quad \sigma_i=\abs{\lambda_i}.
\end{equation}
So if $P\in\CC[x]$ is an even polynomial,
\begin{equation}
P(A)=\sum_{i} \ket{v_i}P(\lambda_i)\bra{v_i}=\sum_{i} \ket{v_i}P(\abs{\lambda_i})\bra{v_i}=P^{\triangleright}(A).
\end{equation}
Similarly, if $P\in\CC[x]$ is an odd polynomial, then
\begin{equation}
P(A)=\sum_{i} \ket{v_i}P(\lambda_i)\bra{v_i}=\sum_{i} \ket{\opr{sgn}(\lambda_i)v_i}P(\abs{\lambda_i})\bra{v_i}=P^{\diamond}(A).
\end{equation}
These relations indeed verify that the eigenvalue decomposition and singular value decomposition are indeed the same when $P$ has a definite parity.
When the parity of $P$ is indefinite, the two objects are in general not the same, and in particular cannot be directly implemented using the QET circuit. 

\subsection{QSVT and matrix dilation}

For general matrices, we have seen in the context of solving linear equations in \cref{sec:HHL} that the matrix dilation method in \cref{eqn:dilation_hermitian} can be used to convert the non-Hermitian problem to a Hermitian problem.
Here we study the relation between QSP applied to the dilated Hermitian matrix, and QSVT for the general matrix.

Recall the definition of the dilated Hermitian matrix 
\begin{equation}
\label{eqn:dilation_trick}
    \wt{A}= \begin{bmatrix} 0 & A^\dagger\\ A & 0 \end{bmatrix}.
\end{equation}
When $A$ is given by its block encoding $U_A\in\BE_{1,m}(A)$, the dilated Hermitian matrix $\wt{A}$ can be obtained with one ancilla qubit through $U_{\wt{A}}=\ket{0}\bra{1}\otimes U_A + \ket{1}\bra{0}\otimes U_A^\dagger$, i.e., $U_{\wt{A}}\in\BE_{1,m+1}(\wt{A})$. Note that this requires the controlled version of $U_A,U_A^{\dag}$. 

From the SVD in \cref{eqn:A_svd_component}, we can construct 
\begin{equation}
\ket{z^{\pm}_i}=\frac{1}{\sqrt{2}}(\ket{0}\ket{v_i}\pm\ket{1}\ket{w_i}).
\end{equation}
Direct calculation shows
\begin{equation}
\wt{A}\ket{z^{\pm}_i}=\pm \sigma_i \ket{z^{\pm}_i},
\end{equation}
i.e., $\{\ket{z^{\pm}_i}\}$ are all the eigenvectors of $\wt{A}$.

For an arbitrary polynomial $f\in\CC[x]$, the matrix function applied to $\wt{A}$ can be computed as
\begin{equation}
\begin{split}
f(\wt{A})&= \sum_{i} \ket{z^{+}_i} f(\sigma_i) \bra{z^{+}_i}+
\ket{z^{-}_i} f(-\sigma_i) \bra{z^{-}_i}\\
&=\sum_{i} \begin{pmatrix}
\ket{v_i} f_{\mathrm{even}}(\sigma_k)\bra{v_i} & \ket{v_i} f_{\mathrm{odd}}(\sigma_k)\bra{w_i}\\
\ket{w_i} f_{\mathrm{odd}}(\sigma_k)\bra{v_i} & \ket{w_i} f_{\mathrm{even}}(\sigma_k)\bra{w_i}\\
\end{pmatrix}\\
&=\begin{pmatrix}
f^{\triangleright}_{\mathrm{even}}(A) & f^{\diamond}_{\mathrm{odd}}(A^{\dag})\\
f^{\diamond}_{\mathrm{odd}}(A)& f^{\triangleleft}_{\mathrm{even}}(A)
\end{pmatrix}.
\end{split}
\end{equation}
Here 
\begin{equation}
f_{\mathrm{even}}(x)=\frac12(f(x)+f(-x)), \quad f_{\mathrm{odd}}(x)=\frac12(f(x)-f(-x)).
\end{equation}
Therefore applying the QSP to the dilated matrix $\wt{A}$ automatically implements QSVT of $A$ using polynomials of even and odd parities.

In particular, if $f$ is an even function, then 
\begin{equation}
f(\wt{A})\ket{0}\ket{\psi}=\ket{0}f^{\triangleright}(A)\ket{\psi},
\end{equation}
i.e., measuring the signal qubit we obtain $0$ with certainty, and the system register is $f^{\triangleright}_{\mathrm{even}}(A)\ket{\psi}$.
Similarly, if $f$ is odd, then
\begin{equation}
f(\wt{A})\ket{0}\ket{\psi}=\ket{1}f^{\diamond}(A)\ket{\psi},
\end{equation}
i.e., measuring the signal qubit we obtain $1$ with certainty.

\section{Application: Solve linear systems of equations}\label{sec:appqsvt_qlsp}

In this section, we revisit the problem of solving linear systems of equations $Ax=b$. With QSVT we can solve QLSP for general matrices without the need of dilating the matrix into a Hermitian matrix.
Assume $A=W\Sigma V^{\dag}$ is invertible, i.e., $\forall i, \Sigma_{ii}>0$, then
\begin{equation}
A^{-1}=V \Sigma^{-1} W^\dagger=f^{\diamond}(A^\dagger),
\label{eqn:ainv_qsvt}
\end{equation}
where $f(x)=x^{-1}$ is an odd function.

WLOG assume $A$ can be accessed by $U_A\in\BE_{1,m}(A)$. For simplicity we also assume $\norm{A}=1$ (though in general $\norm{A}\le \alpha=1$, and we may not always be able to set $\norm{A}=1$).
Let $\kappa$ be the condition number of $A$, then the singular values of $A$ are contained in the interval $[\delta,1]$, with $\delta =\kappa^{-1}$.

Note that $f(\cdot)$  is not bounded by 1 and in particular singular at $x=0$. Therefore instead of approximating $f$ on the whole interval $[-1,1]$ we consider an odd polynomial $p(x)$ such that
\begin{equation}
\left|p(x)-\frac{\delta}{\beta x}\right|\leq \epsilon',\quad \forall x\in [-1,-\delta]\cup[\delta,1].
\end{equation}
The $\beta$ factor is chosen arbitrarily so that $|p(x)|\leq 1$ for all $x\in[-1,1]$ to satisfy the requirement of the condition (2) in \cref{thm:qsp_real}. For instance, we may choose $\beta=4/3$. The precision parameter $\epsilon'$ will be chosen later.
The degree of the odd polynomial can be chosen to be $\Or(\frac{1}{\delta}\log(\frac{1}{\epsilon'}))$ is guaranteed by e.g. \cite[Corollary 69]{GilyenSuLowEtAl2018}.
This construction is not explicit (see an explicit construction in~\cite{ChildsKothariSomma2017}).
\cref{fig:qsp_inv_deg81} gives a concrete construction of an odd polynomial obtained via numerical optimization, and the phase factors are obtained via QSPPACK.

\begin{figure}[H]
\begin{center}
\includegraphics[width=0.31\textwidth]{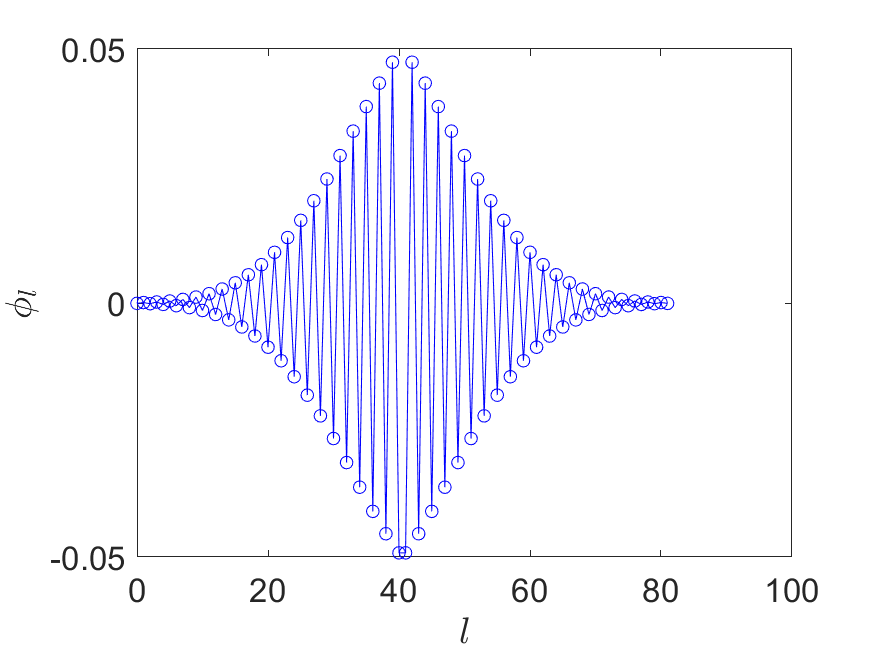}
\includegraphics[width=0.31\textwidth]{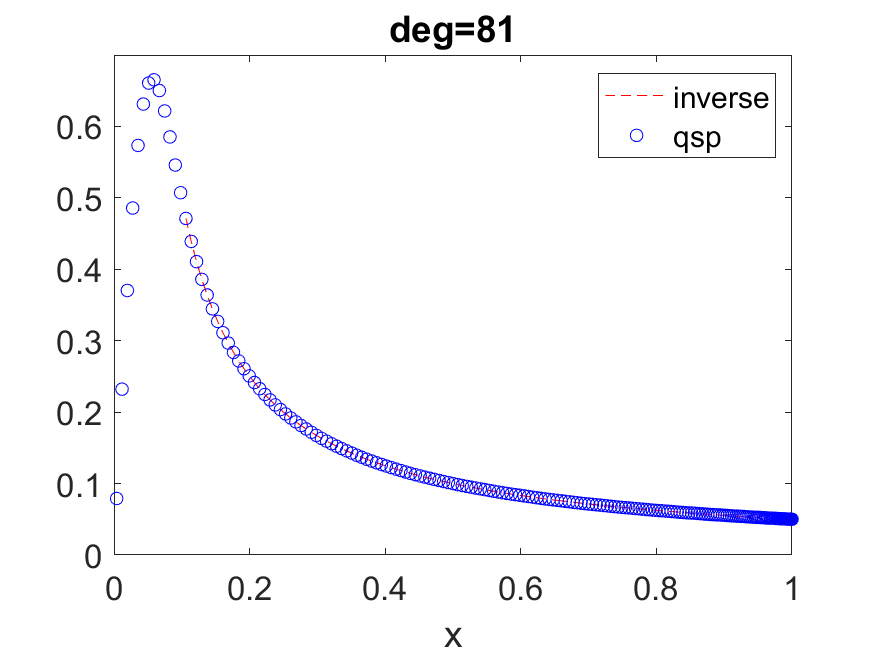}
\includegraphics[width=0.31\textwidth]{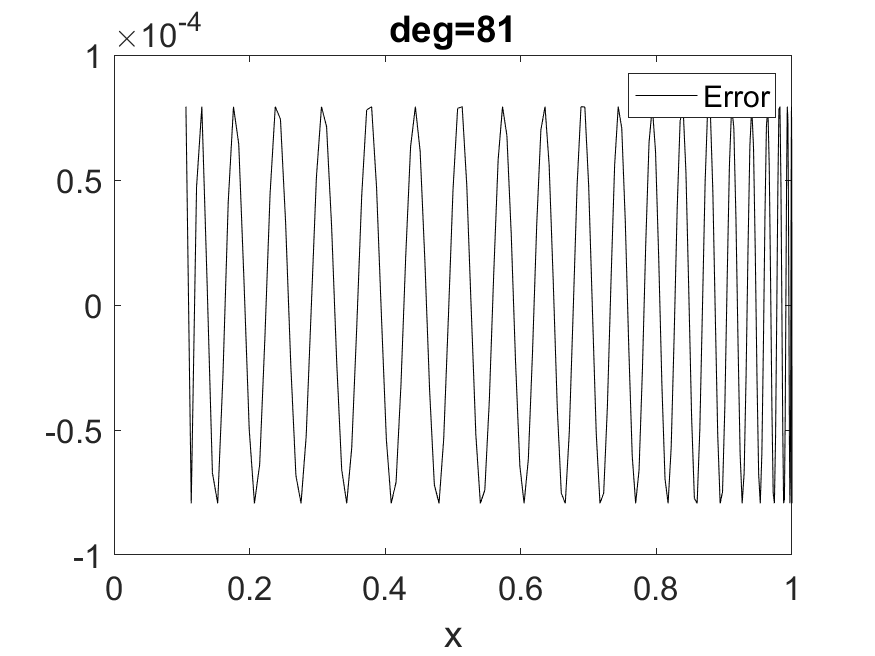}
\end{center}
\caption{QSP representation of an odd polynomial approximating the inverse function on $[\kappa^{-1},1]$ with $\kappa=10$. The phase factors plotted removes a factor of $\pi/4$ on both ends (see \cref{eqn:phi0}).}
\label{fig:qsp_inv_deg81}
\end{figure}

Then \cref{fig:qsvt_circuit_real} implements a $(1,m+1)$-block-encoding of $p^{\diamond}(A^\dagger)=V p(\Sigma)W^\dagger$ denoted by $U_{\Phi}$. We have
\begin{equation}
\label{eq:block_encoding_err_A_inv}
    \|(p^{\diamond}(A^\dagger)-(\delta/\beta)A^{-1}\|=\|p(\Sigma)-(\delta/\beta)\Sigma^{-1}\|\leq \epsilon'.
\end{equation}
The total number of queries to to $U_A$ and $U_A^{\dag}$ is
\begin{equation}
\Or\left(\frac{1}{\delta}\log\left(\frac{1}{\epsilon'}\right)\right)
=\Or\left(\kappa\log\left(\frac{1}{\epsilon'}\right)\right).
\end{equation}

To solve QLSP, we assume the right hand side vector $\ket{b}$ can be accessed through the oracle $U_b$ such that
\begin{equation}
U_b\ket{0^n}=\ket{b}.
\end{equation}
We introduce the parameter 
\begin{equation}
\xi=\norm{A^{-1}\ket{b}},
\end{equation}
which plays an important part in the success probability of the procedure. Let 
\begin{equation}
x=(\delta/\beta)A^{-1}\ket{b}
\end{equation}
be the unnormalized true solution, and the normalized solution state is  $\ket{x}=x/\norm{x}$.
Now denote $\wt{x}=p^{\diamond}(A^{\dagger})\ket{b}$, and $\ket{\wt{x}}=\wt{x}/\norm{\wt{x}}$. Then the unnormalized solution satisfies $\|x-\wt{x}\|\leq\epsilon'$. 
For the normalized state $\ket{y}$, this error is scaled accordingly. When $\epsilon'\ll \|\ket{\wt{x}}\|$, we have
\begin{equation}
\|\ket{x}-\ket{x}\|\approx \frac{\|\wt{x}-x\|}{\|\wt{x}\|}\leq \frac{\epsilon'}{\|\wt{x}\|}.
\end{equation}
Also we have
\begin{equation}
\|\wt{x}\| = \left\|\frac{\delta}{\beta}A^{-1}\ket{b}\right\| = \frac{\delta\xi}{\beta} = \frac{\xi}{\beta\kappa}.
\end{equation}
Therefore in order for the normalized output quantum state to be $\epsilon$-close to the normalized solution state $\ket{x}$, we need to set $\epsilon'=\Or(\epsilon  \xi /\kappa)$.
This is similar to the case of the HHL algorithm in \cref{sec:HHL}, where QPE needs to achieve $\epsilon$ multiplicative accuracy, which means that the additive accuracy parameter $\epsilon'$ should be set to $\Or(\epsilon/\kappa)$.

The success probability of the above procedure is $\Omega(\|\wt{x}\|^2)=\Omega(\xi^2/\kappa^2)$. With amplitude amplification we can boost the success probability to be greater than $1/2$ with one qubit serving as a witness, i.e., if measuring this qubit we get an outcome 0 it means the procedure has succeeded, and if 1 it  means the procedure has failed. 
It takes $\Or(\kappa/\xi)$ rounds of amplitude amplification, i.e., using $U_{\Phi}^\dagger$, $U_{\Phi}$, $U_b$, and $U_b^\dagger$ for $\Or(\kappa/\xi)$ times.
A single $U_{\Phi}$ uses $U_A$ and its inverse
\begin{equation}
\Or\left(\frac{1}{\delta}\log\left(\frac{1}{\epsilon'}\right)\right)
=\Or\left(\kappa\log\left(\frac{\kappa}{\epsilon\xi}\right)\right)
\end{equation}
times.
Therefore the total number of queries to $U_A$ and its inverse is 
\begin{equation}
\Or\left(\frac{\kappa^2}{\xi} \log\left(\frac{\kappa}{\xi\epsilon}\right)\right).
\end{equation}
The number of queries to the $U_b$ and its inverse is $\Or(\kappa/\xi)$. 
We consider the following two cases for the magnitude of $\xi$. 
\begin{enumerate}\label{enum:two_cases}
    \item  In general if no further promise is given, then $\xi\geq 1$. The total query complexity of $U_A$ is therefore $\Or(\kappa^2 \log(\kappa/\epsilon))$. This is the typical complexity referred to in the literature.
    \item If $\ket{b}$ has a $\Omega(1)$ overlap with the left-singular vector of $A$ with the smallest singular value, then $\xi=\Omega(\kappa)$. This is the best case scenario, and the total query complexity of $U_A$ is  $\Or(\kappa \log(1/\epsilon))$, and the number of queries to the right hand side vector $\ket{b}$ is $\Or(1)$, which is independent of the condition number. 
\end{enumerate}

\section{Quantum singular value transformation with basis transformation*}

So far we have assumed that we have a matrix $A$ in mind, and the access to $A$ is provided via the block encoding matrix $U_A$.
The QSVT circuit takes the form
\begin{equation}
U_{\Phi}= \opr{CR}_{\wt{\phi}_{0}}\cdots \opr{CR}_{\wt{\phi}_{d-2}}U_A^{\dag}\opr{CR}_{\wt{\phi}_{d-1}}U_A \opr{CR}_{\wt{\phi}_d}.
\end{equation}
If we are further given two unitary matrices $P,Q$, we can equivalently rewrite the QSVT transformation as
\begin{equation}
U_{\Phi}= \cdots (Q^{\dag}U_A^{\dag}P)^{\dag}(Q\opr{CR}_{\wt{\phi}_{d-1}}Q^{\dag})(QU_A P^{\dag})(P\opr{CR}_{\wt{\phi}_d}P^{\dag})P.
\end{equation}
The beginning of the equation ends with $P$ or $Q$ depending on the parity of $d$. The insertion of $P,Q$ amounts to a \textit{basis transformation}. Assume that we have access to $\wt{U}_A=QU_A P^{\dag}$, 
then $U_A$ is the matrix representation of $\wt{U}_A$ with respect to the bases given by $P,Q$, respectively.
What is different is that the controlled rotations before $U_A$ and $U_A^{\dag}$ are now expressed with respect to two different basis sets, i.e., $P\opr{CR}_{\wt{\phi}_d}P^{\dag}$, and $Q\opr{CR}_{\wt{\phi}_{d-1}}Q^{\dag}$, respectively. 
This can be useful for certain applications. 
Let us now express these ideas more formally.

Assume that we are given an $(n+m)$-qubit unitary $\wt{U}_A$, and two $(n+m)$-qubit projectors $\Pi,\Pi'$.  
For simplicity we assume $\opr{rank}(\Pi)=\opr{rank}(\Pi')=N$. 
Define an orthonormal basis set
\begin{equation}
\mc{B}=\{\ket{\varphi_0},\ldots,\ket{\varphi_{N-1}},\ket{v_N},\ldots,\ket{v_{NM-1}}\},
\label{eqn:qsvt_proj_basis_b}
\end{equation}
where the vectors $\ket{\varphi_0},\ldots,\ket{\varphi_{N-1}}$ span the range of $\Pi$, and all states $\ket{v_i}$ are orthogonal to $\ket{\varphi_j}$.
Similarly define an orthonormal basis set
\begin{equation}
\mc{B}'=\{\ket{\psi_0},\ldots,\ket{\psi_{N-1}},\ket{w_N},\ldots,\ket{w_{NM-1}}\}
\label{eqn:qsvt_proj_basis_bp}
\end{equation}
where the vectors $\ket{\psi_0},\ldots,\ket{\psi_{N-1}}$ span the range of $\Pi'$, and all states $\ket{w_i}$ are orthogonal to $\ket{\psi_j}$.
We can think that the columns of $\mc{B},\mc{B}'$ form the basis transformation matrix $P,Q$, respectively.

Then the matrix $A$ is defined implicitly in terms of its matrix representation
\begin{equation}
[\wt{U}_A]_{\mc{B}}^{\mc{B}'}=U_A=\begin{pmatrix}
A & * \\
* & *
\end{pmatrix}.
\label{eqn:qsvt_proj_blockencode}
\end{equation}
Note that 
\begin{equation}
[\Pi]_{\mc{B}}^{\mc{B}}=[\Pi']_{\mc{B'}}^{\mc{B'}}=\begin{pmatrix}
I_n & 0 \\
0 & 0
\end{pmatrix}=\ket{0^m}\bra{0^m}\otimes I_n,
\end{equation}
we find that 
\begin{equation}
\Pi'\wt{U}_A\Pi=\sum_{i,j\in [N]}\ket{\psi_i} A_{ij}\bra{\varphi_j}.
\end{equation}
Therefore \cref{thm:qsvt_real} can be viewed as the singular value transformation of $A$, which is a submatrix of the \emph{matrix representation} of $U_A$ with respect to bases $\mc{B},\mc{B}'$. 

The implementation of the controlled rotation $P\opr{CR}_{\wt{\phi}}P^{\dag}$ relies on the implementation of $\opr{C}_{\Pi}\opr{NOT}$.
The projectors $\Pi,\Pi'$ can be accessed directly, and WLOG we focus on one projector $\Pi$. 
Motivated from Grover's search, we may assume access to a reflection operator
\begin{equation}
R_{\Pi}=I_m-2\Pi.
\end{equation}
via the controlled NOT gates $\opr{C}_{\Pi}\opr{NOT},\opr{C}_{\Pi'}\opr{NOT}$ respectively. We can then define an $m$-qubit controlled NOT gate as
\begin{equation}
\opr{C}_{\Pi}\opr{NOT}:=X\otimes \Pi+I\otimes (I_{m}-\Pi),
\end{equation}
which can be constructed using $R_{\Pi}$ as
\begin{equation}
\begin{split}
\opr{C}_{\Pi}\opr{NOT}&=X\otimes \frac{I_m-R_{\Pi}}{2}+I\otimes \frac{I_m+R_{\Pi}}{2}\\
&=\frac{I+X}{2} \otimes I_m + \frac{I-X}{2} \otimes R_{\Pi}\\
&=\ket{+}\bra{+}\otimes I_m+\ket{-}\bra{-}\otimes R_\Pi\\
&=(H\otimes I_m)(\ket{0}\bra{0}\otimes I_m+\ket{1}\bra{1}\otimes R_\Pi)(H\otimes I_m).
\end{split}
\label{eqn:reflection_cpinot}
\end{equation}
Therefore assuming access to $R_{\Pi}$, the $\opr{C}_{\Pi}\opr{NOT}$ gate can be implemented using the circuit in \cref{fig:implement_CPINOT}.

\begin{figure}[H]
\begin{center}
\begin{quantikz}
\qw&\gate{H}& \ctrl{1}\qw& \gate{H}&\qw\\
\qw&\qw& \gate{R_{\Pi}}& \qw&\qw\\
\end{quantikz}
\end{center}
\caption{Circuit for implementing $\opr{C}_{\Pi}\opr{NOT}$ using a reflector $R_{\Pi}$.}
\label{fig:implement_CPINOT}
\end{figure}
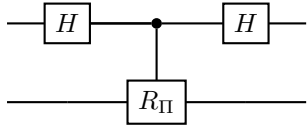

Then according to \cref{thm:qsvt_real}, we can implement the QSVT using $\wt{U}_A,\wt{U}_A^{\dag}$, $\opr{C}_{\ket{0^m}\bra{0^m}}\opr{NOT}$ and single qubit rotation gates, where $\ket{0^m}\bra{0^m}\otimes I_n$ is a rank-$2^n$ projector. In this generalized setting, $\opr{C}_{\ket{0^m}\bra{0^m}\otimes I_n}\opr{NOT}=\opr{C}_{\ket{0^m}\bra{0^m}}\opr{NOT}\otimes I_n$ in the $\mc{B},\mc{B}'$ basis should be implemented using $\opr{C}_{\Pi}\opr{NOT},\opr{C}_{\Pi'}\opr{NOT}$, respectively.
We arrive at the following theorem:
\begin{thm}[Quantum singular value transformation with real polynomials and projectors]
Let $\wt{U}_A$ be a $(n+m)$-qubit unitary, and $\Pi,\Pi'$ be two $(n+m)$-qubit projectors of rank $2^n$.
Define the basis $\mc{B},\mc{B'}$ according to \cref{eqn:qsvt_proj_basis_b,eqn:qsvt_proj_basis_bp}, and let $A\in\CC^{N\times N}$ be defined in terms of the matrix representation in \cref{eqn:qsvt_proj_blockencode}. 
Given a polynomial $P_{\Re}(x)\in\RR[x]$ of degree $d$ satisfying the conditions in \cref{thm:qsp_real}, we can find a sequence of phase factors $\Phi\in\RR^{d+1}$ to define a unitary $U_{\Phi}$ satisfying
\begin{equation}
U_{\Phi}\ket{\varphi_j}=\sum_{i\in[N]}\ket{\psi_i} [P^{\diamond}_{\Re}(A)]_{ij}+\ket{\perp_j'},
\end{equation}
if $d$ is odd, and 
\begin{equation}
U_{\Phi}\ket{\varphi_j}=\sum_{i\in[N]}\ket{\varphi_i} [P^{\triangleright}_{\Re}(A)]_{ij}+\ket{\perp_j}.
\end{equation}
if $d$ is even. 
Here  $\Pi'\ket{\perp_j'}=0$, $\Pi\ket{\perp_j}=0$, and
$U_{\Phi}$ uses $\wt{U}_A,\wt{U}_A^{\dag}$, $\opr{C}_{\Pi}\opr{NOT},\opr{C}_{\Pi'}\opr{NOT}$, and single qubit rotation gates for $\Or(d)$ times.
\label{thm:qsvt_real_proj}
\end{thm}

\section{Application: Grover's search revisited, and fixed-point amplitude amplification*}\label{sec:appqsvt_grover}

As an application, let us revisit the Grover search problem from the perspective of QSVT.
Again let $\ket{x_0}$ be the desired marked state, and $\ket{\psi_0}$ be the uniform superposition of states.
Now let us perform a basis transformation. 
We define an orthonormal basis set $\mc{B}=\{\ket{\psi_0},\ket{v_1},\ldots,\ket{v_{N-1}}\}$, where all states $\ket{v_i}$ are orthogonal to $\ket{\psi_0}$. 
Similarly define an orthonormal basis set $\mc{B}'=\{\ket{x_0},\ket{w_1},\ldots,\ket{w_{N-1}}\}$, where all states $\ket{w_i}$ are orthogonal to $\ket{x_0}$.
Then the matrix of reflection operator $R_{\psi_0}$ with respect to $\mc{B},\mc{B'}$ is (let $a=1/\sqrt{N}$)
\begin{equation}
[R_{\psi_0}]_{\mc{B}}^{\mc{B}'}=\begin{pmatrix}
a & *\\
* & *
\end{pmatrix},
\end{equation}
Let $A=a$ be a $1\times 1$ matrix, and then $R_{\psi_0}\in \BE_{1,n}(A)$.
Furthermore, the projectors $\Pi=\ket{\psi_0}\bra{\psi_0}$ and $\Pi'=\ket{x_0}\bra{x_0}$ can be accessed via the reflection operator $R_{\psi_0},R_{x_0}$, respectively.
According to \cref{eqn:reflection_cpinot}, this defines $\opr{C}_{\Pi}\opr{NOT},\opr{C}_{\Pi'}\opr{NOT}$.
Let $\wt{U}_A=R_{\psi_0}$, we have
\begin{equation}
\Pi' \wt{U}_A\Pi=a\ket{x_0}\bra{\psi_0},
\end{equation}
and we would like to use \cref{thm:qsvt_real_proj} to find $U_{\Phi}$ that block encodes
\begin{equation}
\ket{x_0}P_{\Re}^{\diamond}(a)\bra{\psi_0}\approx \ket{x_0}\bra{\psi_0}.
\end{equation}
To this end, we need to find an \emph{odd}, real polynomial $P_{\Re}(x)$ satisfying $P_{\Re}(a)\approx 1$.
More specifically, we need to find $P_{\Re}(x)$ satisfying
\begin{equation}
\abs{P_{\Re}(x)-1}\le \epsilon^2, \quad \forall x\in [a,1].
\end{equation}
We can achieve this by approximating the sign function, with $\deg(P_{\Re})=\Or(\log(1/\epsilon^2) a^{-1})=\Or(\log(1/\epsilon)\sqrt{N})$ (see e.g.~\cite[Corollary 6]{LowChuang2017a}). 
This construction is also based on an approximation to the $\mathrm{erf}$ function.
\cref{fig:qsp_sign_deg81} gives a concrete construction of an odd polynomial obtained via numerical optimization, and the phase factors are obtained via QSPPACK.

\begin{figure}[H]
\begin{center}
\includegraphics[width=0.31\textwidth]{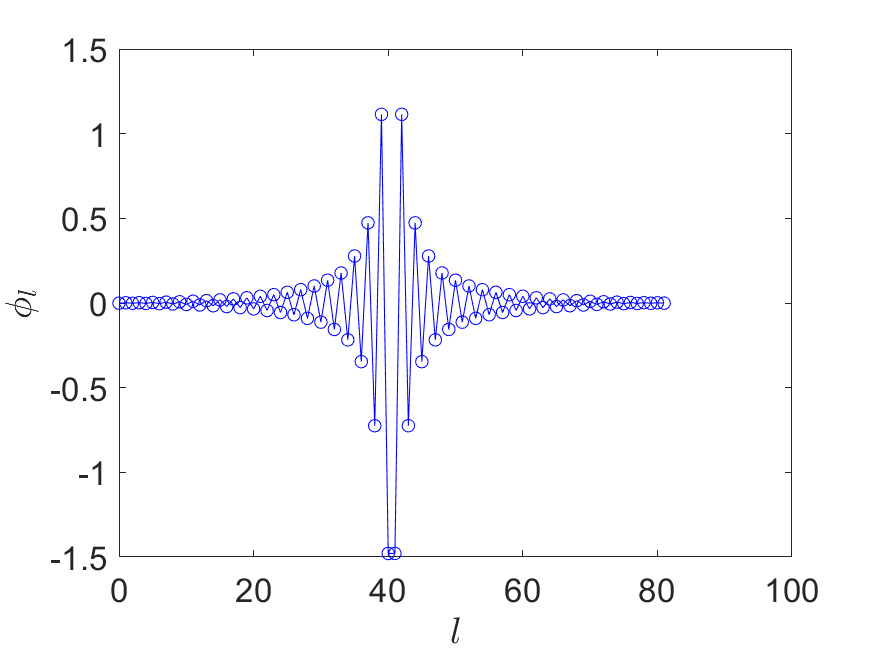}
\includegraphics[width=0.31\textwidth]{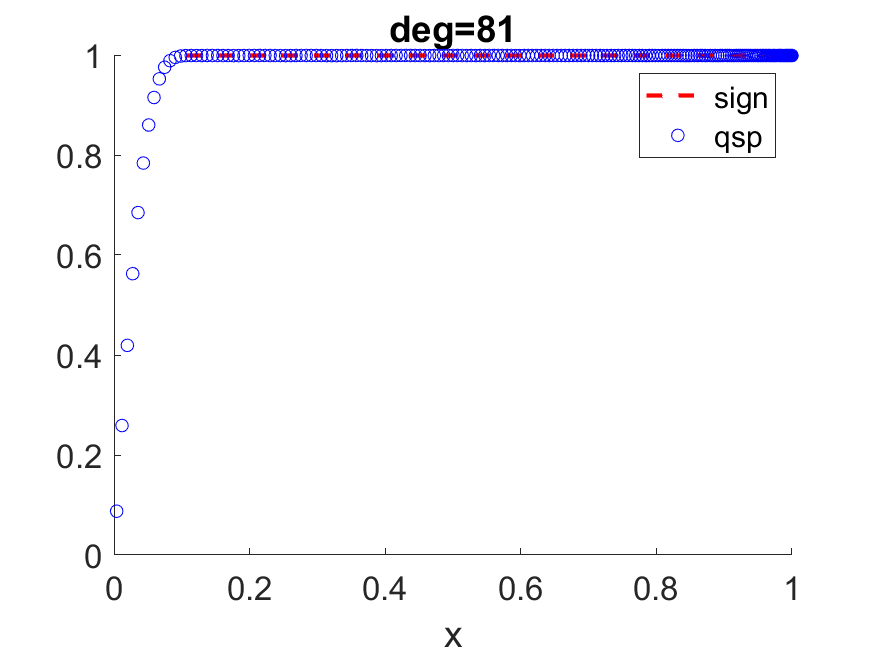}
\includegraphics[width=0.31\textwidth]{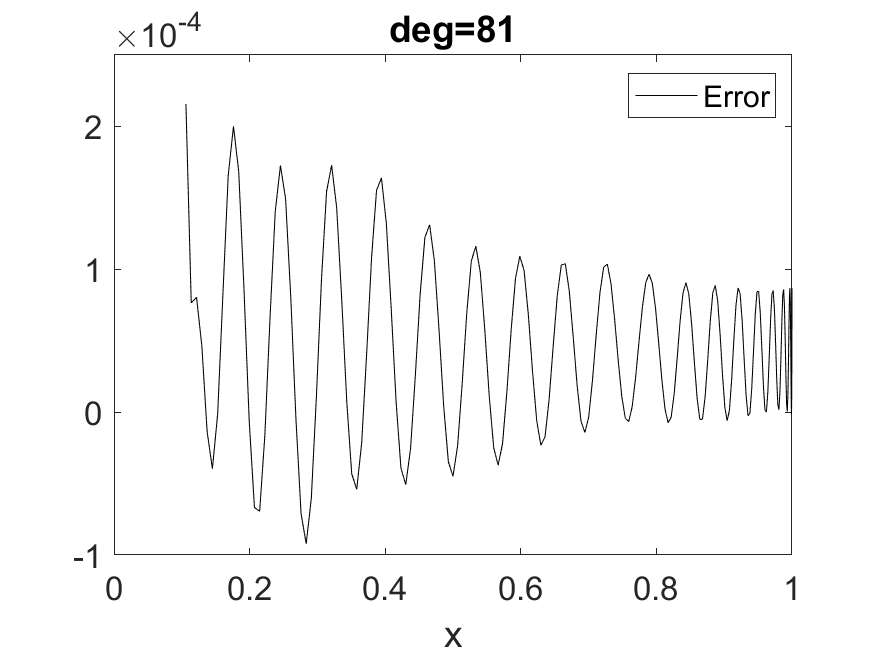}
\end{center}
\caption{QSP representation of an odd polynomial approximating the sign function on $[0.1,1]$. The phase factors plotted removes a factor of $\pi/4$ on both ends (see \cref{eqn:phi0}).}
\label{fig:qsp_sign_deg81}
\end{figure}

Then
\begin{equation}
U_{\Phi}\ket{\psi_0}\approx \ket{x_0}.
\end{equation}
More specifically, note that
\begin{equation}
U_{\Phi}\ket{\psi_0} - \ket{x_0}=(P_{\Re}(a)-1)\ket{\psi_0}+\ket{\wt{\perp}}
\end{equation}
for some unnormalized state $\ket{\wt{\perp}}$. Moreover
\begin{equation}
\norm{\ket{\wt{\perp}}}^2+P_{\Re}^2(a)=1,
\end{equation}
which gives
\begin{equation}
\norm{\ket{\wt{\perp}}}= \sqrt{1-P_{\Re}^2(a)}\le \sqrt{1-(1-\epsilon^{2})^2}\le \sqrt{2\epsilon^2}.
\end{equation}
So 
\begin{equation}
\norm{U_{\Phi}\ket{\psi_0} - \ket{x_0}}\le \epsilon^2+\sqrt{2}\epsilon=\Or(\epsilon).
\end{equation}
Therefore we can measure the system register to find $x_0$, and we achieve the same Grover type speedup.

Note that the approximation can be arbitrarily accurate, and there is no overshooting problem (though this is a small problem) as in the standard Grover search. 
While Grover's search does not require the output quantum state to be exactly $\ket{x_0}$, this could be desirable when it is used as a quantum subroutine, such as amplitude amplification.

The immediate generalization of the procedure above is called the fixed point amplitude amplification. 

\begin{prop}[Fixed-point amplitude amplification]
Let $\wt{U}_A$ be an $n$-qubit unitary and $\Pi'$ be an $n$-qubit orthogonal projector such that 
\begin{equation}
\Pi' \wt{U}_A\ket{\varphi_0}=a\ket{\psi}, \quad a> \delta>0.
\end{equation}
Then there is a $(n+1)$-qubit unitary circuit $U_{\Phi}$ such that
\begin{equation}
\norm{\ket{0}\ket{\psi}-U_{\Phi}\ket{0}\ket{\varphi_0}}\le \epsilon.
\end{equation}
here $U_{\Phi}$ uses the gates $\wt{U}_A,\wt{U}_A^{\dag},\opr{C}_{\Pi'}\opr{NOT},\opr{C}_{\ket{\varphi_0}\bra{\varphi_0}}\opr{NOT}$ and single qubit rotation gates for $\Or(\log(1/\epsilon) \delta^{-1})$ times.
\label{prop:fixedpt_aa}
\end{prop}
\begin{proof}
The procedure is very similar to Grover's search.
We only prove the case when $\Pi$ is of rank $1$, though the statement is also correct when the rank of $\Pi$ is larger than $1$.
We can construct $\mc{B}=\{\ket{\varphi_0},\ket{v_1},\ldots,\ket{v_{N-1}}\}$, where all states $\ket{v_i}$ are orthogonal to $\ket{\varphi_0}$. 
Similarly define an orthonormal basis set $\mc{B}'=\{\ket{\psi},\ket{w_1},\ldots,\ket{w_{N-1}}\}$, where all states $\ket{w_i}$ are orthogonal to the target state $\ket{\psi}$.
Since the target state $\ket{\psi}$ belongs to the range of $\Pi'$, 
\begin{equation}
\braket{\psi|\wt{U}_A|\varphi_0}=\braket{\psi|\Pi '\wt{U}_A|\varphi_0}=a,
\end{equation}
i.e.,
\begin{equation}
[\wt{U}_A]_{\mc{B}}^{\mc{B}'}=\begin{pmatrix}
a & *\\
* & *
\end{pmatrix}.
\end{equation}
 Now let $\Pi=\ket{\varphi_0}\bra{\varphi_0}$, we can use the same choice of $P_{\Re}(x)$ as in Grover's search so that $\abs{P_{\Re}(x)-1}=\Or(\epsilon^2)$ for any $x\ge \delta$, and $\deg(P_{\Re})=\Or(\log(1/\epsilon) \delta^{-1})$.
The corresponding $U_{\Phi}$ uses one ancilla qubit to block encode
\begin{equation}
\ket{\psi} P_{\Re}(\delta)\bra{\varphi_0}\approx \ket{\psi} \bra{\varphi_0}.
\end{equation} 
 \end{proof}

Note that the ranks of $\Pi',\Pi$ are different. This does not affect the proof.

\vspace{2em}

\begin{exer}[Robust oblivious amplitude amplification]
Consider a quantum circuit consisting of two registers denoted by $a$ and $s$. Suppose we have a block encoding $V$ of $A$: $A=(\bra{0}_a\otimes I_s)V(\ket{0}_a\otimes I_s)$. Let $W=-V(\mathrm{REF}\otimes I_s)V^{\dagger}(\mathrm{REF}\otimes I_s)V$, where $\mathrm{REF}=I_a-2\ket{0}_a\bra{0}_a$. (1) Within the framework of QSVT, what is the polynomial associated with the singular value transformation implemented by $W$? (2) Suppose $A=U/2$ for some unitary $U$. What is $(\bra{0}_a\otimes I_s)W(\ket{0}_a\otimes I_s)$? (3) Explain the construction of $W$ in terms of a singular value transformation $f^{\diamond}(A)$ with $f(x)=3x-4x^3$. Draw the picture of $f(x)$ and mark its values at $x=0,\frac12,1$.
\end{exer}

\begin{exer}[Logarithm of unitaries]
  Given access to a unitary $U=e^{\I H}$ where $\norm{H}\le \pi/2$. Use QSVT to design a quantum algorithm to approximately implement a block encoding of $H$, using controlled $U$ and its inverses, as well as elementary quantum gates.
\end{exer}

\bibliographystyle{alpha}
% Must use relative path
%\bibliography{../../../Dropbox/Bibliography/reference}
\bibliography{ref_arxiv}

\end{document}